\documentclass[review=false]{jfp-epi}
\usepackage{jfp-cup2epi}

\usepackage[dvipsnames]{xcolor}
\usepackage{amsmath}
\usepackage{amsthm}
\usepackage{subcaption}
\usepackage{mathpartir}
\usepackage{newunicodechar}
\usepackage[bibliography=common]{apxproof}
\usepackage{cleveref}
\newunicodechar{ }{{\,}}
\newunicodechar{¬}{\ensuremath{\neg}} %
\newunicodechar{²}{^2}
\newunicodechar{³}{^3}
\newunicodechar{·}{\ensuremath{\cdot}}
\newunicodechar{¹}{^1}
\newunicodechar{×}{\ensuremath{\times}} %
\newunicodechar{÷}{\ensuremath{\div}} %
\newunicodechar{ʲ}{^j}
\newunicodechar{ʳ}{^r}
\newunicodechar{ˡ}{^l}
\newunicodechar{̷}{\not} %
\newunicodechar{Γ}{\ensuremath{\Gamma}}   %
\newunicodechar{Δ}{\ensuremath{\Delta}} %
\newunicodechar{Η}{\ensuremath{\textrm{H}}} %
\newunicodechar{Θ}{\ensuremath{\Theta}} %
\newunicodechar{Λ}{\ensuremath{\Lambda}} %
\newunicodechar{Ξ}{\ensuremath{\Xi}} %
\newunicodechar{Π}{\ensuremath{\Pi}}   %
\newunicodechar{Σ}{\ensuremath{\Sigma}} %
\newunicodechar{Φ}{\ensuremath{\Phi}} %
\newunicodechar{Ψ}{\ensuremath{\Psi}} %
\newunicodechar{Ω}{\ensuremath{\Omega}} %
\newunicodechar{α}{\ensuremath{\mathnormal{\alpha}}}
\newunicodechar{β}{\ensuremath{\beta}} %
\newunicodechar{γ}{\ensuremath{\mathnormal{\gamma}}} %
\newunicodechar{δ}{\ensuremath{\mathnormal{\delta}}} %
\newunicodechar{ε}{\ensuremath{\mathnormal{\varepsilon}}} %
\newunicodechar{ζ}{\ensuremath{\mathnormal{\zeta}}} %
\newunicodechar{η}{\ensuremath{\mathnormal{\eta}}} %
\newunicodechar{θ}{\ensuremath{\mathnormal{\theta}}} %
\newunicodechar{ι}{\ensuremath{\mathnormal{\iota}}} %
\newunicodechar{κ}{\ensuremath{\mathnormal{\kappa}}} %
\newunicodechar{λ}{\ensuremath{\mathnormal{\lambda}}} %
\newunicodechar{μ}{\ensuremath{\mathnormal{\mu}}} %
\newunicodechar{ν}{\ensuremath{\mathnormal{\mu}}} %
\newunicodechar{ξ}{\ensuremath{\mathnormal{\xi}}} %
\newunicodechar{π}{\ensuremath{\mathnormal{\pi}}}
\newunicodechar{π}{\ensuremath{\mathnormal{\pi}}} %
\newunicodechar{ρ}{\ensuremath{\mathnormal{\rho}}} %
\newunicodechar{σ}{\ensuremath{\mathnormal{\sigma}}} %
\newunicodechar{τ}{\ensuremath{\mathnormal{\tau}}} %
\newunicodechar{φ}{\ensuremath{\mathnormal{\varphi}}} %
\newunicodechar{χ}{\ensuremath{\mathnormal{\chi}}} %
\newunicodechar{ψ}{\ensuremath{\mathnormal{\psi}}} %
\newunicodechar{ω}{\ensuremath{\mathnormal{\omega}}} %
\newunicodechar{ϕ}{\ensuremath{\mathnormal{\phi}}} %
\newunicodechar{ϵ}{\ensuremath{\mathnormal{\epsilon}}} %
\newunicodechar{ᵏ}{^k}
\newunicodechar{ᵢ}{\ensuremath{_i}} %
\newunicodechar{ }{\quad}
\newunicodechar{ }{{\,}}
\newunicodechar{′}{\ensuremath{^\prime}}  %
\newunicodechar{″}{\ensuremath{^\second}} %
\newunicodechar{‴}{\ensuremath{^\third}}  %
\newunicodechar{ⁱ}{^i}
\newunicodechar{⁵}{\ensuremath{^5}}
\newunicodechar{⁺}{\ensuremath{^+}} 
\newunicodechar{⁻}{\ensuremath{^-}} 
\newunicodechar{ⁿ}{^n}
\newunicodechar{₀}{\ensuremath{_0}} %
\newunicodechar{₁}{\ensuremath{_1}} 
\newunicodechar{₂}{\ensuremath{_2}} 
\newunicodechar{₃}{\ensuremath{_3}}
\newunicodechar{₊}{\ensuremath{_+}} 
\newunicodechar{₋}{\ensuremath{_-}} 
\newunicodechar{ₙ}{_n} %
\newunicodechar{ℂ}{\ensuremath{\mathbb{C}}} %
\newunicodechar{ℒ}{\ensuremath{\mathscr{L}}}
\newunicodechar{ℕ}{\mathbb{N}} %
\newunicodechar{ℚ}{\ensuremath{\mathbb{Q}}}
\newunicodechar{ℝ}{\ensuremath{\mathbb{R}}} %
\newunicodechar{ℤ}{\ensuremath{\mathbb{Z}}} %
\newunicodechar{ℳ}{\mathscr{M}}
\newunicodechar{⅋}{\ensuremath{\parr}} %
\newunicodechar{←}{\ensuremath{\leftarrow}} %
\newunicodechar{↑}{\ensuremath{\uparrow}} %
\newunicodechar{→}{\ensuremath{\rightarrow}} %
\newunicodechar{↔}{\ensuremath{\leftrightarrow}} %
\newunicodechar{↖}{\nwarrow} %
\newunicodechar{↗}{\nearrow} %
\newunicodechar{↝}{\ensuremath{\leadsto}}
\newunicodechar{↦}{\ensuremath{\mapsto}}
\newunicodechar{⇆}{\ensuremath{\leftrightarrows}} %
\newunicodechar{⇐}{\ensuremath{\Leftarrow}} %
\newunicodechar{⇒}{\ensuremath{\Rightarrow}} %
\newunicodechar{⇔}{\ensuremath{\Leftrightarrow}} %
\newunicodechar{∀}{\ensuremath{\forall}}   %
\newunicodechar{∂}{\ensuremath{\partial}}
\newunicodechar{∃}{\ensuremath{\exists}} %
\newunicodechar{∅}{\ensuremath{\varnothing}} %
\newunicodechar{∈}{\ensuremath{\in}}
\newunicodechar{∉}{\ensuremath{\not\in}} %
\newunicodechar{∋}{\ensuremath{\ni}}  %
\newunicodechar{∎}{\ensuremath{\qed}}
\newunicodechar{∏}{\prod}
\newunicodechar{∑}{\sum}
\newunicodechar{∗}{\ensuremath{\ast}} %
\newunicodechar{∘}{\ensuremath{\circ}} %
\newunicodechar{∙}{\ensuremath{\bullet}} 
\newunicodechar{∞}{\ensuremath{\infty}} %
\newunicodechar{∣}{\ensuremath{\mid}} %
\newunicodechar{∧}{\wedge}%
\newunicodechar{∨}{\vee}%
\newunicodechar{∩}{\ensuremath{\cap}} %
\newunicodechar{∪}{\ensuremath{\cup}} %
\newunicodechar{∫}{\int}
\newunicodechar{∷}{::} %
\newunicodechar{∼}{\ensuremath{\sim}} %
\newunicodechar{≃}{\ensuremath{\simeq}} %
\newunicodechar{≅}{\ensuremath{\cong}} %
\newunicodechar{≈}{\ensuremath{\approx}} %
\newunicodechar{≜}{\ensuremath{\stackrel{\scriptscriptstyle {\triangle}}{=}}} 
\newunicodechar{≝}{\ensuremath{\stackrel{\scriptscriptstyle {\text{def}}}{=}}}
\newunicodechar{≟}{\ensuremath{\stackrel {_\text{\textbf{?}}}{\text{\textbf{=}}\negthickspace\negthickspace\text{\textbf{=}}}}} 
\newunicodechar{≠}{\ensuremath{\neq}}%
\newunicodechar{≡}{\ensuremath{\equiv}}%
\newunicodechar{≤}{\ensuremath{\le}} %
\newunicodechar{≥}{\ensuremath{\ge}} %
\newunicodechar{⊂}{\ensuremath{\subset}} %
\newunicodechar{⊃}{\ensuremath{\supset}} %
\newunicodechar{⊆}{\ensuremath{\subseteq}} %
\newunicodechar{⊇}{\ensuremath{\supseteq}} %
\newunicodechar{⊎}{\ensuremath{\uplus}} %
\newunicodechar{⊑}{\ensuremath{\sqsubseteq}} %
\newunicodechar{⊒}{\ensuremath{\sqsupseteq}} %
\newunicodechar{⊓}{\ensuremath{\sqcap}} %
\newunicodechar{⊔}{\ensuremath{\sqcup}} %
\newunicodechar{⊕}{\ensuremath{\oplus}} %
\newunicodechar{⊗}{\ensuremath{\otimes}} %
\newunicodechar{⊛}{\ensuremath{\circledast}}
\newunicodechar{⊢}{\ensuremath{\vdash}} %
\newunicodechar{⊤}{\ensuremath{\top}}
\newunicodechar{⊥}{\ensuremath{\bot}} 
\newunicodechar{⊧}{\models} %
\newunicodechar{⊨}{\models} %
\newunicodechar{⊩}{\Vdash}
\newunicodechar{⊸}{\ensuremath{\multimap}} %
\newunicodechar{⋁}{\ensuremath{\bigvee}}
\newunicodechar{⋃}{\ensuremath{\bigcup}} %
\newunicodechar{⋄}{\ensuremath{\diamond}} %
\newunicodechar{⋅}{\ensuremath{\cdot}}
\newunicodechar{⋆}{\ensuremath{\star}} %
\newunicodechar{⋮}{\ensuremath{\vdots}} %
\newunicodechar{⋯}{\ensuremath{\cdots}} %
\newunicodechar{─}{---}
\newunicodechar{■}{\ensuremath{\blacksquare}} 
\newunicodechar{□}{\ensuremath{\square}} %
\newunicodechar{▵}{\ensuremath{\triangle}} %
\newunicodechar{▹}{\ensuremath{\rhd}} %
\newunicodechar{◃}{\triangleleft{}}
\newunicodechar{◅}{\ensuremath{\triangleleft{}}}
\newunicodechar{◇}{\ensuremath{\diamond}} %
\newunicodechar{◽}{\ensuremath{\square}}
\newunicodechar{★}{\ensuremath{\star}}   %
\newunicodechar{♭}{\ensuremath{\flat}} %
\newunicodechar{♯}{\ensuremath{\sharp}} %
\newunicodechar{✓}{\ensuremath{\checkmark}} %
\newunicodechar{⟂}{\ensuremath{^\bot}} 
\newunicodechar{⟦}{\ensuremath{\llbracket}}
\newunicodechar{⟧}{\ensuremath{\rrbracket}}
\newunicodechar{⟨}{\ensuremath{\langle}} %
\newunicodechar{⟩}{\ensuremath{\rangle}} %
\newunicodechar{⟶}{{\longrightarrow}}
\newunicodechar{⟷} {\ensuremath{\leftrightarrow}}
\newunicodechar{⟹}{\ensuremath{\Longrightarrow}} %
\newunicodechar{ⱼ}{_j} %
\newunicodechar{𝒟}{\ensuremath{\mathcal{D}}} %
\newunicodechar{𝒢}{\ensuremath{\mathcal{G}}}
\newunicodechar{𝒦}{\ensuremath{\mathcal{K}}} %
\newunicodechar{𝒫}{\ensuremath{\mathcal{P}}}
\newunicodechar{𝔸}{\ensuremath{\mathds{A}}} %
\newunicodechar{𝔹}{\ensuremath{\mathds{B}}} %
\newunicodechar{𝟙}{\ensuremath{\mathds{1}}} %

\usepackage[supertabular,implicitLineBreakHack]{ottalt}

\newenvironment{ottdefnblock}[3][]{ \framebox{\mbox{#2}} \quad #3 \\[0pt]}{}

\newcommand{\ottnt}[1]{\mathit{#1}}
\newcommand{\ottmv}[1]{\mathit{#1}}
\newcommand{\ottkw}[1]{\mathbf{#1}}
\newcommand{\ottsym}[1]{#1}

\newcommand{\ottafterlastrule}{\\}

  \renewottcommands[ott]

\makeatletter
\@ifundefined{lhs2tex.lhs2tex.sty.read}%
  {\@namedef{lhs2tex.lhs2tex.sty.read}{}%
   \newcommand\SkipToFmtEnd{}%
   \newcommand\EndFmtInput{}%
   \long\def\SkipToFmtEnd#1\EndFmtInput{}%
  }\SkipToFmtEnd

\newcommand\ReadOnlyOnce[1]{\@ifundefined{#1}{\@namedef{#1}{}}\SkipToFmtEnd}
\usepackage{amstext}
\usepackage{amssymb}
\usepackage{stmaryrd}
\DeclareFontFamily{OT1}{cmtex}{}
\DeclareFontShape{OT1}{cmtex}{m}{n}
  {<5><6><7><8>cmtex8
   <9>cmtex9
   <10><10.95><12><14.4><17.28><20.74><24.88>cmtex10}{}
\DeclareFontShape{OT1}{cmtex}{m}{it}
  {<-> ssub * cmtt/m/it}{}

\DeclareFontShape{OT1}{cmtt}{bx}{n}
  {<5><6><7><8>cmtt8
   <9>cmbtt9
   <10><10.95><12><14.4><17.28><20.74><24.88>cmbtt10}{}
\DeclareFontShape{OT1}{cmtex}{bx}{n}
  {<-> ssub * cmtt/bx/n}{}

\newcommand{\Conid}[1]{\mathit{#1}}
\newcommand{\Varid}[1]{\mathit{#1}}
\newcommand{\anonymous}{\kern0.06em \vbox{\hrule\@width.5em}}

\newcommand{\bind}{\mathbin{>\!\!\!>\mkern-6.7mu=}}
\newcommand{\sequ}{\mathbin{>\!\!\!>}}
\renewcommand{\leq}{\leqslant}

\usepackage{polytable}

\@ifundefined{mathindent}%
  {\newdimen\mathindent\mathindent\leftmargini}%
  {}%

\def\resethooks{%
  \global\let\SaveRestoreHook\empty
  \global\let\ColumnHook\empty}
\newcommand*{\savecolumns}[1][default]%
  {\g@addto@macro\SaveRestoreHook{\savecolumns[#1]}}
\newcommand*{\restorecolumns}[1][default]%
  {\g@addto@macro\SaveRestoreHook{\restorecolumns[#1]}}
\newcommand*{\aligncolumn}[2]%
  {\g@addto@macro\ColumnHook{\column{#1}{#2}}}

\resethooks

\newcommand{\onelinecommentchars}{\quad-{}- }
\newcommand{\commentbeginchars}{\enskip\{-}
\newcommand{\commentendchars}{-\}\enskip}

\newcommand{\visiblecomments}{%
  \let\onelinecomment=\onelinecommentchars
  \let\commentbegin=\commentbeginchars
  \let\commentend=\commentendchars}

\newcommand{\invisiblecomments}{%
  \let\onelinecomment=\empty
  \let\commentbegin=\empty
  \let\commentend=\empty}

\visiblecomments

\newlength{\blanklineskip}
\newcommand{\hsindent}[1]{\quad}
\let\hspre\empty
\let\hspost\empty

\EndFmtInput
\makeatother
\ReadOnlyOnce{polycode.fmt}%
\makeatletter

\newcommand{\hsnewpar}[1]%
  {{\parskip=0pt\parindent=0pt\par\vskip #1\noindent}}

\newcommand{\hscodestyle}{}

\newcommand{\sethscode}[1]%
  {\expandafter\let\expandafter\hscode\csname #1\endcsname
   \expandafter\let\expandafter\endhscode\csname end#1\endcsname}

  {\par\noindent
   \advance\leftskip\mathindent
   \hscodestyle
   \let\\=\@normalcr
   \let\hspre\(\let\hspost\)%
   \pboxed}%
  {\endpboxed\)%
   \par\noindent
   \ignorespacesafterend}

  {\hsnewpar\abovedisplayskip
   \advance\leftskip\mathindent
   \hscodestyle
   \let\hspre\(\let\hspost\)%
   \pboxed}%
  {\endpboxed%
   \hsnewpar\belowdisplayskip
   \ignorespacesafterend}

  {\hsnewpar\abovedisplayskip
   \advance\leftskip\mathindent
   \hscodestyle
   \let\\=\@normalcr
   \(\pboxed}%
  {\endpboxed\)%
   \hsnewpar\belowdisplayskip
   \ignorespacesafterend}

\newcommand{\plainhs}{\sethscode{plainhscode}}

\plainhs

  {\hsnewpar\abovedisplayskip
   \advance\leftskip\mathindent
   \hscodestyle
   \let\\=\@normalcr
   \(\parray}%
  {\endparray\)%
   \hsnewpar\belowdisplayskip
   \ignorespacesafterend}

  {\parray}{\endparray}

  {\(\parray}{\endparray\)}

\def\codeframewidth{\arrayrulewidth}
\RequirePackage{calc}

  {\parskip=\abovedisplayskip\par\noindent
   \hscodestyle
   \arrayrulewidth=\codeframewidth
   \tabular{@{}|p{\linewidth-2\arraycolsep-2\arrayrulewidth-2pt}|@{}}%
   \hline\framedhslinecorrect\\{-1.5ex}%
   \let\endoflinesave=\\
   \let\\=\@normalcr
   \(\pboxed}%
  {\endpboxed\)%
   \framedhslinecorrect\endoflinesave{.5ex}\hline
   \endtabular
   \parskip=\belowdisplayskip\par\noindent
   \ignorespacesafterend}

\newcommand{\framedhslinecorrect}[2]%
  {#1[#2]}

  {\(\def\column##1##2{}%
   \let\>\undefined\let\<\undefined\let\\\undefined
   \newcommand\>[1][]{}\newcommand\<[1][]{}\newcommand\\[1][]{}%
   \def\fromto##1##2##3{##3}%
   }{\) }%

  {\let\orighscode=\hscode
   \let\origendhscode=\endhscode
   \def\endhscode{\def\hscode{\endgroup\def\@currenvir{hscode}\\}\begingroup}
   \orighscode\def\hscode{\endgroup\def\@currenvir{hscode}}}%
  {\origendhscode
   \global\let\hscode=\orighscode
   \global\let\endhscode=\origendhscode}%

\makeatother
\EndFmtInput
  \InputIfFileExists{no-editing-marks}{
    
  }

  \usepackage{xargs}
  \usepackage{tikz}
  \usetikzlibrary{automata, positioning}
  \usepackage[colorinlistoftodos,prependcaption,textsize=tiny]{todonotes}

      \newcommand{\note}[1]{}

      \newcommand{\unsure}[2][1=]{}
      \newcommand{\info}[2][1=]{}
      \newcommand{\change}[2]{}
      \newcommand{\inconsistent}[2]{}
      \newcommand{\critical}[2]{}
      \newcommand{\improvement}[1]{}
      \newcommand{\resolved}[2]{}

      \newcommand{\csongor}[2][1=]{}
      \newcommand{\jp}[2][1=]{}
      \newcommandx{\rae}[2][1=]{}
      \newcommandx{\nw}[2][1=]{}

  \newcommand{\constraintcolour}{\color{RoyalBlue}}
  \newcommand{\constraintfont}[1]{{\constraintcolour#1}}
  \newcommand{\multiplicitycolour}{\color{RoyalBlue}}
  \newcommand{\multiplicityfont}[1]{{\multiplicitycolour#1}}

  \newcommand{\vdashi}{%
      \mathrel{%
          \vdash\hspace*{-4pt}%
          \raisebox{0.9pt}{\scalebox{.66}{\(\blacktriangleright\)}}%
      }%
  }
  \newcommand{\vdashs}{⊢_{\mathsf{s}}}
  \newcommand{\vdashd}{⊢_{\mathsf{diff}}}
  
  \newcommand{\scale}{\constraintfont{\cdot}}
  \newcommand{\constraintop}[1]{\mathop{\constraintfont{#1}}}
  \newcommand{\aand}{\constraintop{\&}}
  \DeclareMathOperator*{\bigaand}{\vcenter{\hbox{\Large\&}}}
  \newcommand{\lollycirc}{\raisebox{-0.255ex}{\scalebox{1.4}{$\circ$}}}
  \newcommand{\Lolly}{\constraintop{=\kern-0.9ex \lollycirc}}
  \newcommand{\FatArrow}{\constraintop{\Rightarrow}}
  \newcommand{\RLolly}{\mathop{\constraintfont \circledless}}
  
  \newcommand{\qtensor}{\constraintop{\otimes}}
  \newcommand{\dsterm}[2]{\llbracket #2 \rrbracket_{#1}}
  \newcommand{\dstype}[1]{\llbracket #1 \rrbracket}
  \newcommand{\dsevidence}[1]{\llbracket #1 \rrbracket^{\mathbf{ev}}}

  \newcommand{\keyword}[1]{\mathbf{#1}}
  \newcommand{\klet}{\keyword{let}}
  \newcommand{\kcase}{\keyword{case}}

  \newcommand{\packbox}{\raisebox{-0.08ex}{\scalebox{0.88}{$\square$}}}
  \newcommand{\kin}{\keyword{in}}
  
  \newcommand{\kunpack}{\klet\packbox}

  \newcommand{\bnfeq}{\mathrel{\Coloneqq}}
  \newcommand{\bnfor}{\mathrel{\mid}}

  \newtheoremrep{theorem}{Theorem}[section]
  \newtheoremrep{lemma}[theorem]{Lemma}
  \newtheoremrep{corollary}[theorem]{Corollary}
  \newtheoremrep{property}[theorem]{Property}

  \newtheorem{definition}[theorem]{Definition}

\begin{document}

\jfpYear{2026}

\title{Linear Constraints}

\author{Arnaud Spiwack}
\orcid{0000-0002-5985-2086}
\affiliation{
  \institution{Tweag}
  \city{Paris}
  \country{France}
  \authoremail{arnaud.spiwack@tweag.io}
 }
\author{Csongor Kiss}
\orcid{0000-0002-0195-2420}
\affiliation{
  \institution{Imperial College London}
  \city{London}
  \country{United Kingdom}
  \authoremail{csongor.kiss14@imperial.ac.uk}
 }
\author{Jean-Philippe Bernardy}
\orcid{0000-0002-8469-5617}
\affiliation{
  \institution{University of Gothenburg}
  \city{Gothenburg}
  \country{Sweden}
  \authoremail{jean-philippe.bernardy@gu.se}
 }
\author{Nicolas Wu}
\orcid{0000-0002-4161-985X}
\affiliation{
  \institution{Imperial College London}
  \city{London}
  \country{United Kingdom}
  \authoremail{n.wu@imperial.ac.uk}
 }
\author{Richard A. Eisenberg}
\orcid{0000-0002-7669-9781}
\affiliation{
  \institution{Tweag}
  \city{Paris}
  \country{France}
  \authoremail{rae@richarde.dev}
 }
\authornote{Since working on this paper, Richard Eisenberg's affiliation has
changed to Jane Street, USA}

\begin{abstract}
  Linear constraints are the linear counterpart of Haskell's class
  constraints.
  Linearly typed parameters allow the programmer to control resources
  such as file handles and manually managed memory as linear
  arguments. Indeed, a linear type system can verify that these
  resources are used safely.  However, writing code with explicit
  linear arguments requires bureaucracy. Linear constraints address
  this shortcoming: a linear constraint acts as an implicit linear
  argument that can be filled in automatically by the compiler.
  
  We present this new feature as a qualified type system,
  together with an inference algorithm which extends
  \textsc{ghc}'s existing constraint solver algorithm. Soundness of
  linear constraints is ensured by the fact that they desugar into
  Linear Haskell.

  This paper is a revised and extended version of a previous paper by
  the same authors~\citep{SpiwackKBWE22}. The formal system and the
  constraint solver have been significantly simplified and numerous
  additional applications are described.
\end{abstract}

\maketitle

\newcommand{\maybesmall}{\small}

\section{Introduction}
\label{sec:introduction}

Linear type systems have seen a renaissance in recent years in
various programming communities. Rust's ownership system guarantees
memory safety for systems programmers, Haskell's \textsc{ghc} includes
support for linear types, and even
dependently typed programmers can now use linear types with Idris 2.
All of these systems are vastly different in ergonomics and
scope. Rust uses dedicated syntax and code generation
to support management of resources, while Linear Haskell is a
type system change without requiring any other impact on the compiler, such as in the
code generator or runtime system.
Linear Haskell is powerful enough to emulate Rust's
ownership, but using its linear arguments to do so is a tedious exercise,
since it requires
the programmer to carefully thread resource tokens.

To get a sense of the power and the tedium of programming with linear types
for memory management, consider the following function which
reads the first two elements \ensuremath{\Varid{x}} and \ensuremath{\Varid{y}} of an array \ensuremath{\Varid{arr}_{0}} and
returns them after deallocating the array:
\begin{hscode}\SaveRestoreHook
\column{B}{@{}>{\hspre}l<{\hspost}@{}}%
\column{3}{@{}>{\hspre}l<{\hspost}@{}}%
\column{8}{@{}>{\hspre}l<{\hspost}@{}}%
\column{19}{@{}>{\hspre}l<{\hspost}@{}}%
\column{E}{@{}>{\hspre}l<{\hspost}@{}}%
\>[B]{}\Varid{read2AndDiscard}\mathbin{::}\Conid{MArray}\;\Varid{a}⊸(\Conid{Ur}\;\Varid{a},\Conid{Ur}\;\Varid{a}){}\<[E]%
\\
\>[B]{}\Varid{read2AndDiscard}\;\Varid{arr}_{0}\mathrel{=}{}\<[E]%
\\
\>[B]{}\hsindent{3}{}\<[3]%
\>[3]{}\mathbf{let}\;{}\<[8]%
\>[8]{}(\Varid{arr}_{1},\Varid{x}){}\<[19]%
\>[19]{}\mathrel{=}\Varid{read}\;\Varid{arr}_{0}\;\mathrm{0}{}\<[E]%
\\
\>[8]{}(\Varid{arr}_{2},\Varid{y}){}\<[19]%
\>[19]{}\mathrel{=}\Varid{read}\;\Varid{arr}_{1}\;\mathrm{1}{}\<[E]%
\\
\>[8]{}(){}\<[19]%
\>[19]{}\mathrel{=}\Varid{free}\;\Varid{arr}_{2}{}\<[E]%
\\
\>[B]{}\hsindent{3}{}\<[3]%
\>[3]{}\mathbf{in}\;(\Varid{x},\Varid{y}){}\<[E]%
\ColumnHook
\end{hscode}\resethooks
The power of linearity enables the array library to ensure that there is only
one reference to the array, and therefore it can be mutated in-place without
violating referential transparency.
Linearity also ensures that after the array has been freed, it is no longer
possible to read or write to it.
The values \ensuremath{\Varid{x}} and \ensuremath{\Varid{y}} read from the array are returned; their
types include elements wrapped by the \ensuremath{\Conid{Ur}}
(pronounced ``unrestricted'',  and corresponding to the
``!'' operator of \citet{girard-linear-logic} and defined precisely by \citet{LinearHaskell}) type, allowing them to be used
arbitrarily many times. This works because \ensuremath{\Varid{read2AndDiscard}} takes a restricted-use array
containing unrestricted elements.

On the other hand, there is significant tedium in manipulating
the array through \ensuremath{\Varid{arr}_{0}}, \ensuremath{\Varid{arr}_{1}}, and \ensuremath{\Varid{arr}_{2}}.
This arises because the \ensuremath{\Varid{read}} function consumes the array and returns a fresh
array, to be used in future operations.
Operationally, the array remains
the same, but each operation assigns a new name to it, to facilitate tracking
references statically.

In a non-linear language, one would have to use a monadic interface, or allow arbitrary effects, foregoing referential transparency.
Compare the above function with what one would write in a non-linear, impure
language:
\begin{hscode}\SaveRestoreHook
\column{B}{@{}>{\hspre}l<{\hspost}@{}}%
\column{3}{@{}>{\hspre}l<{\hspost}@{}}%
\column{9}{@{}>{\hspre}l<{\hspost}@{}}%
\column{14}{@{}>{\hspre}l<{\hspost}@{}}%
\column{E}{@{}>{\hspre}l<{\hspost}@{}}%
\>[B]{}\Varid{read2AndDiscard}\mathbin{::}\Conid{MArray}\;\Varid{a}\to (\Varid{a},\Varid{a}){}\<[E]%
\\
\>[B]{}\Varid{read2AndDiscard}\;\Varid{arr}\mathrel{=}{}\<[E]%
\\
\>[B]{}\hsindent{3}{}\<[3]%
\>[3]{}\mathbf{let}\;{}\<[9]%
\>[9]{}\Varid{x}{}\<[14]%
\>[14]{}\mathrel{=}\Varid{read}\;\Varid{arr}\;\mathrm{0}{}\<[E]%
\\
\>[9]{}\Varid{y}{}\<[14]%
\>[14]{}\mathrel{=}\Varid{read}\;\Varid{arr}\;\mathrm{1}{}\<[E]%
\\
\>[9]{}(){}\<[14]%
\>[14]{}\mathrel{=}\Varid{unsafeFree}\;\Varid{arr}{}\<[E]%
\\
\>[B]{}\hsindent{3}{}\<[3]%
\>[3]{}\mathbf{in}\;(\Varid{x},\Varid{y}){}\<[E]%
\ColumnHook
\end{hscode}\resethooks
This non-linear version does not guarantee that there is a unique reference to
the array, so freeing the array is a potentially unsafe operation.
However, it is simpler because there is less bureaucracy to manage: we are
clearly interacting with the \emph{same} array throughout, and this version makes
that apparent.
We see here a clear tension between extra safety and clarity of code---one
we wish, as language designers, to avoid. 
%
How can we get the compiler to see that the array is used safely
without explicit threading?

Following well-known ideas~\citep{DBLP:conf/popl/CraryWM99,DBLP:conf/esop/SmithWM00,DBLP:conf/tic/WalkerM00}, our approach is to let arrays be unrestricted, but
associate linear capabilities (such as \ensuremath{\constraintfont{\Conid{Read}}}, \ensuremath{\constraintfont{\Conid{Write}}}) to them.
In fact, we show in this paper that such linear capabilities
are the natural analogue of Haskell's type class
constraints to the setting of linear types.
We call these new constraints \emph{linear constraints}.
Like class constraints~\citep{OutsideIn},
linear constraints are propagated implicitly by the compiler.
Like linear arguments, they can be used to safely track resources such as arrays
or file handles. Thus, linear constraints are the combination of these two
concepts.

With our extension, we can write a new pure version of \ensuremath{\Varid{read2AndDiscard}} which does
not require explicit threading of the array. This instead uses a \emph{unique} array
\ensuremath{\Conid{UArray}\;\Varid{a}\;\Varid{n}} of values of type \ensuremath{\Varid{a}}, which is tagged with a type value \ensuremath{\Varid{n}}:
\begin{hscode}\SaveRestoreHook
\column{B}{@{}>{\hspre}l<{\hspost}@{}}%
\column{3}{@{}>{\hspre}l<{\hspost}@{}}%
\column{7}{@{}>{\hspre}c<{\hspost}@{}}%
\column{7E}{@{}l@{}}%
\column{11}{@{}>{\hspre}l<{\hspost}@{}}%
\column{E}{@{}>{\hspre}l<{\hspost}@{}}%
\>[B]{}\Varid{read2AndDiscard}\mathbin{::}\constraintfont{(\Conid{Read}\;\Varid{n},\Conid{Write}\;\Varid{n})}\Lolly \Conid{UArray}\;\Varid{a}\;\Varid{n}\to (\Conid{Ur}\;\Varid{a},\Conid{Ur}\;\Varid{a}){}\<[E]%
\\
\>[B]{}\Varid{read2AndDiscard}\;\Varid{arr}\mathrel{=}\Conid{Linearly}.\mathbf{do}{}\<[E]%
\\
\>[B]{}\hsindent{3}{}\<[3]%
\>[3]{}\Varid{x}{}\<[7]%
\>[7]{}\leftarrow {}\<[7E]%
\>[11]{}\Varid{read}\;\Varid{arr}\;\mathrm{0}{}\<[E]%
\\
\>[B]{}\hsindent{3}{}\<[3]%
\>[3]{}\Varid{y}{}\<[7]%
\>[7]{}\leftarrow {}\<[7E]%
\>[11]{}\Varid{read}\;\Varid{arr}\;\mathrm{1}{}\<[E]%
\\
\>[B]{}\hsindent{3}{}\<[3]%
\>[3]{}\Varid{free}\;\Varid{arr}{}\<[E]%
\\
\>[B]{}\hsindent{3}{}\<[3]%
\>[3]{}\Varid{\Conid{Linearly}.return}\;(\Varid{x},\Varid{y}){}\<[E]%
\ColumnHook
\end{hscode}\resethooks
The only changes from the impure version are that this version explicitly constrains
the function to require read and write access to the array, and the
use of the qualified \ensuremath{\Conid{Linearly}.\mathbf{do}} notation to sequence the actions (see \cref{sec:packing-unpacking-do}).
Crucially, the resource representing the ownership of the
array is a linear constraint and is separate from the array itself, which no
longer needs to be threaded manually.
In this approach the array is non-linear but its capabilities
(\textit{i.e.} read or write access) are given by \emph{linear constraints}.
Once these capabilities are consumed, the array can no longer be read from
or written to without triggering a compile time error.
Additionally, using capabilities in this way prevents problems
that arise from aliasing.


This paper is an extended version of a conference
publication~\citep{SpiwackKBWE22}. Overall our contributions are:
\begin{itemize}
\item A system of qualified types that allows a constraint assumption
  to be given a multiplicity (linear or unrestricted). Linear assumptions are used precisely
  once in the body of a definition
  (\cref{sec:qualified-type-system}). This system supports examples
  that have motivated the design of several resource-aware systems,
  such as ownership \textit{à la} Rust (\cref{sec:memory-ownership}), or
  capabilities in the style of Mezzo~\cite{mezzo-permissions} or
  \textsc{ats}~\cite{AtsLinearViews}; accordingly, our system points
  towards a possible unification of these lines of research.

\item Applications of this qualified type system to allow writing
  \begin{itemize}
  \item functions whose result can only be used linearly (\cref{sec:Unique-constraint})
  \item resource-aware algorithms without explicit threading (\cref{sec:memory-ownership}); and
  \end{itemize}

\item An inference algorithm that respects the multiplicity of
  assumptions. We prove that this algorithm is sound with respect to
  our type system~(\cref{sec:type-inference}). It consists of
  \begin{itemize}
  \item a constraint generation
    algorithm~(\cref{sec:constraint-generation}). The language of
    generated constraints tracks multiplicities.
  \item a solver~(\cref{sec:constraint-solver}) for the generated
    constraints, which restricts proof-search algorithms for linear
    logic in order to be \emph{guess
      free}~\cite[Section~6.4]{OutsideIn}. A guess-free algorithm
    ensures that constraint inference is predictable and insensitive
    to small changes in the source program; it is necessarily
    incomplete.
  \end{itemize}
\end{itemize}
Our language is given semantics by desugaring into a core language based
on that of \citet{LinearHaskell}.
Our design is intended to work well with other features of Haskell and
\textsc{ghc} extensions. Indeed, we have a prototype implementation (\cref{sec:implementation}).

The specific contributions of  this extended version are:
\begin{itemize}
\item numerous new examples and applications~(\cref{sec:memory-ownership})
\item simplified requirements for the entailment relation of simple
        constraints~(\cref{sec:qualified-type-system})
\item a new and simplified constraint solver~(\cref{sec:constraint-solver})
\item a discussion of the design decisions~(\cref{sec:design-considerations})
\end{itemize}

\section{Background: Linear Haskell}
\label{sec:linear-types}

This section follows \citet[Section 2.1]{LinearHaskell},
to describe our baseline approach, as released in \textsc{ghc}~9.0.
Linear Haskell adds a new type of functions,
dubbed \emph{linear functions}, and written \ensuremath{\Varid{a}\mathbin{⊸}\Varid{b}}.\footnote{The linear function
  type and its notation come from linear
  logic~\citep{girard-linear-logic}, to which the phrase \emph{linear
    types} refers. All the various designs of linear typing in the
  literature amount to adding such a linear function type, but details
  can vary wildly. See~\citet[Section 6]{LinearHaskell} for an analysis
  of alternative approaches.} A linear function consumes its
  argument exactly once. Linear Haskell defines it as follows:

\begin{quote}
\ensuremath{\Varid{f}\mathbin{::}\Varid{a}\mathbin{⊸}\Varid{b}} guarantees that if \ensuremath{(\Varid{f}\;\Varid{u})} is consumed exactly once, \\
then the argument \ensuremath{\Varid{u}} is consumed exactly once.
\end{quote}
To make precise sense of this statement, we need to know what ``consumed exactly once'' means
formally, which requires an understanding of the type system that we introduce later.
However, an intuitive understanding can be given by
inspecting the type of the value concerned:
\begin{definition}[Consume exactly once]~ \label{def:consume}
\begin{itemize}
  \item To consume a value of atomic base type (e.g. \ensuremath{\Conid{Int}}) evaluate it exactly once.
  \item To consume a function exactly once, apply it to one argument, and then consume its result exactly once.
  \item To consume a pair exactly once, pattern-match on it, and then consume each component exactly once.
  \item In general, to consume a value of an algebraic datatype exactly once, pattern-match on it,
  and then consume all its linear components exactly once.
\end{itemize}
\end{definition}
\noindent
Note that in this paper, we consider arrays to be generalised unrestricted containers.
In fact, defining exactly what ``evaluate exactly once'' means for arrays is
a difficult task. In this paper, we avoid the difficulty by considering instead the
capabilities associated with such arrays.

\subsection{Multiplicities}
\label{sec:multiplicities}
The usual arrow type \ensuremath{\Varid{a}\to \Varid{b}} can be recovered using \ensuremath{\Conid{Ur}}, as \ensuremath{\Conid{Ur}\;\Varid{a}\mathbin{⊸}\Varid{b}}, but Linear
Haskell provides a first-class treatment of \ensuremath{\Varid{a}\to \Varid{b}}, thus ensuring
backwards compatibility with Haskell. In practice, the type-checker
records the \emph{multiplicity} of every introduced variable: \(1\) for linear
arguments and \(ω\) for unrestricted ones. Multiplicity is also recorded
in the arrows:
\[
  \ensuremath{\Varid{a}\mathop{\to_{\multiplicityfont{\Varid{\pi}}}}\Varid{b}}
\]
This way, one can give a unified treatment of both arrow types~\citep{linearity-and-pi-calculus} by
instantiating $\pi$ to the appropriate multiplicity:
linear arrows \ensuremath{\Varid{a}⊸\Varid{b}} are given by
\ensuremath{\Varid{a}\mathop{\to_{\multiplicityfont{\mathrm{1}}}}\Varid{b}} and unrestricted arrows \ensuremath{\Varid{a}\to \Varid{b}} are given by \ensuremath{\Varid{a}\mathop{\to_{\multiplicityfont{\omega}}}\Varid{b}}.

We stress that a multiplicity of \(1\) restricts \emph{how the variable can be used}. It does not
restrict \emph{which values can be substituted for it}.
In particular, a linear function cannot assume that it is given the
unique pointer to its argument.  For example, if \ensuremath{\Varid{f}\mathbin{::}\Varid{a}\mathbin{⊸}\Varid{b}}, then
the following is fine because the type of \ensuremath{\Varid{duplicate}} makes no guarantees about how it uses \ensuremath{\Varid{x}}:
\begin{hscode}\SaveRestoreHook
\column{B}{@{}>{\hspre}l<{\hspost}@{}}%
\column{E}{@{}>{\hspre}l<{\hspost}@{}}%
\>[B]{}\Varid{duplicate}\mathbin{::}(\Varid{a}\mathbin{⊸}\Varid{b})\to \Varid{a}\to (\Varid{b},\Varid{b}){}\<[E]%
\\
\>[B]{}\Varid{duplicate}\;\Varid{f}\;\Varid{x}\mathrel{=}(\Varid{f}\;\Varid{x},\Varid{f}\;\Varid{x}){}\<[E]%
\ColumnHook
\end{hscode}\resethooks
In particular, \ensuremath{\Varid{duplicate}} can pass two copies of \ensuremath{\Varid{x}} to two copies of \ensuremath{\Varid{f}}.

Pattern matching on a value of type \ensuremath{\Conid{Ur}\;\Varid{a}} yields a payload of multiplicity
$  \multiplicityfont{ \omega }  $, even when the scrutinee has multiplicity $  \multiplicityfont{ \ottsym{1} }  $.

\section{Working With Linear Constraints}
\label{sec:what-it-looks-like}

The use of type classes to handle \textit{ad-hoc} polymorphism
was one of the innovative features of Haskell's design~\citep{WalderB89,HudakHJW07}.
This allows for a constraint to appear in the signature of a type
that the compiler attempts to satisfy.
For instance, consider the Haskell function \ensuremath{\Varid{show}}:
\begin{hscode}\SaveRestoreHook
\column{B}{@{}>{\hspre}l<{\hspost}@{}}%
\column{E}{@{}>{\hspre}l<{\hspost}@{}}%
\>[B]{}\Varid{show}\mathbin{::}\constraintfont{\Conid{Show}\;\Varid{a}}\FatArrow \Varid{a}\to \Conid{String}{}\<[E]%
\ColumnHook
\end{hscode}\resethooks
In addition to the function arrow \ensuremath{\to }, common to all functional
programming languages, the type of this function features a constraint arrow \ensuremath{\FatArrow }.
Everything to the
left of a constraint arrow is called a \emph{constraint}, and will be
highlighted in \constraintfont{blue} throughout the paper. Here
\ensuremath{\constraintfont{\Conid{Show}\;\Varid{a}}} is a
class constraint.

Constraints are handled implicitly by
the typechecker. That is, if we want to \ensuremath{\Varid{show}} the integer \ensuremath{\Varid{n}\mathbin{::}\Conid{Int}} we would write \ensuremath{\Varid{show}\;\Varid{n}}, and the typechecker is responsible for proving that \ensuremath{\constraintfont{\Conid{Show}\;\Conid{Int}}} holds, without
intervention from the programmer.
Note that \ensuremath{\FatArrow } arrows are \emph{unrestricted}: they create a demand
from the compiler to supply the required constraint, which it achieves by
providing the appropriate instance.
Since \ensuremath{\Conid{Show}\;\Conid{Int}} is a valid instance, its use is entirely unrestricted.

Just as \ensuremath{\FatArrow } arrows allow a value to depend on \emph{unrestricted} implicit parameters,
linear constraints allow a value to depend on \emph{linear} implicit parameters.
In order to manage linearity implicitly, this paper introduces a
linear constraint arrow (\ensuremath{\Lolly }), much like Linear Haskell introduces a linear
function arrow (\ensuremath{⊸}). Constraints to the left of a linear constraint
arrow are \emph{linear constraints}. There are thus four different arrows that can be used:
\[
\begin{array}{c|cc}
      & \text{explicit argument} & \text{implicit argument} \\ \hline
\text{unrestricted} & \ensuremath{\to }   & \ensuremath{\FatArrow } \\
\text{linear}        & \ensuremath{⊸} & \ensuremath{\Lolly }
\end{array}
\]
The fact that linear constraints are implicit and inferred by the compiler is
key to avoiding the tedium of passing resource variables around.
This section proceeds by demonstrating
how linear constraints are consumed (\cref{sec:consuming-linear-constraints}),
how linear constraints are threaded through programs (\cref{sec:borrowing-linear-constraints}),
and how these constraints are produced (\cref{sec:producing-linear-constraints}).

\subsection{Consuming Linear Constraints}
\label{sec:consuming-linear-constraints}

\begin{figure}%
  \maybesmall
  \centering
  \begin{subfigure}{.3\linewidth}%
    \noindent%
\begin{hscode}\SaveRestoreHook
\column{B}{@{}>{\hspre}l<{\hspost}@{}}%
\column{6}{@{}>{\hspre}l<{\hspost}@{}}%
\column{E}{@{}>{\hspre}l<{\hspost}@{}}%
\>[B]{}\phantom{(\mathbf{type}\;\constraintfont{\Conid{RW}\;\Varid{n}}\mathrel{=}\constraintfont{(\Conid{Read}\;\Varid{n},\Conid{Write}\;\Varid{n})})}{}\<[E]%
\\[\blanklineskip]%
\>[B]{}\Varid{new}{}\<[6]%
\>[6]{}\mathbin{::}\Conid{Int}\to (\Conid{MArray}\;\Varid{a}⊸\Conid{Ur}\;\Varid{r})⊸\Conid{Ur}\;\Varid{r}{}\<[E]%
\\
\>[B]{}\Varid{write}\mathbin{::}\Conid{MArray}\;\Varid{a}⊸\Conid{Int}\to \Varid{a}\to \Conid{MArray}\;\Varid{a}{}\<[E]%
\\
\>[B]{}\Varid{read}\mathbin{::}\Conid{MArray}\;\Varid{a}⊸\Conid{Int}\to (\Conid{MArray}\;\Varid{a},\Conid{Ur}\;\Varid{a}){}\<[E]%
\\
\>[B]{}\Varid{free}\mathbin{::}\Conid{MArray}\;\Varid{a}⊸(){}\<[E]%
\ColumnHook
\end{hscode}\resethooks
\caption{Linear Types}
\label{fig:linear-interface}
  \end{subfigure}
  \hfill
  \begin{subfigure}{.55\linewidth}
\begin{hscode}\SaveRestoreHook
\column{B}{@{}>{\hspre}l<{\hspost}@{}}%
\column{6}{@{}>{\hspre}l<{\hspost}@{}}%
\column{E}{@{}>{\hspre}l<{\hspost}@{}}%
\>[B]{}\mathbf{type}\;\constraintfont{\Conid{RW}\;\Varid{n}}\mathrel{=}\constraintfont{(\Conid{Read}\;\Varid{n},\Conid{Write}\;\Varid{n})}{}\<[E]%
\\[\blanklineskip]%
\>[B]{}\Varid{new}{}\<[6]%
\>[6]{}\mathbin{::}\constraintfont{\constraintfont{\Conid{Linearly}}}\Lolly \Conid{Int}\to \exists\;\Varid{n}.\Conid{Ur}\;(\Conid{UArray}\;\Varid{a}\;\Varid{n})\RLolly\constraintfont{\Conid{RW}\;\Varid{n}}{}\<[E]%
\\
\>[B]{}\Varid{write}\mathbin{::}\constraintfont{\Conid{RW}\;\Varid{n}}\Lolly \Conid{UArray}\;\Varid{a}\;\Varid{n}\to \Conid{Int}\to \Varid{a}\to ()\RLolly\constraintfont{\Conid{RW}\;\Varid{n}}{}\<[E]%
\\
\>[B]{}\Varid{read}\mathbin{::}\constraintfont{\Conid{Read}\;\Varid{n}}\Lolly \Conid{UArray}\;\Varid{a}\;\Varid{n}\to \Conid{Int}\to \Conid{Ur}\;\Varid{a}\RLolly\constraintfont{\Conid{Read}\;\Varid{n}}{}\<[E]%
\\
\>[B]{}\Varid{free}\mathbin{::}\constraintfont{\Conid{RW}\;\Varid{n}}\Lolly \Conid{UArray}\;\Varid{a}\;\Varid{n}\to (){}\<[E]%
\ColumnHook
\end{hscode}\resethooks
\caption{Linear Constraints}
\label{fig:constraints-interface}
  \end{subfigure}
\caption{Interfaces for mutable arrays}
\end{figure}

One way to understand constraints is that they demand evidence to be
supplied before the value they constrain can be used.
From this perspective, linear constraints consume that evidence exactly
once.

In the context of mutuble array access, the capabilities of that
array --such as the ability to read or write from the array--
act as suitable evidence that operations may act on that array.
These capabilities exist at the type level and must refer to a specific
array. To achieve this, an array \ensuremath{\Varid{arr}} of values of type \ensuremath{\Varid{a}} is tagged with a type variable \ensuremath{\Varid{n}}:
\begin{hscode}\SaveRestoreHook
\column{B}{@{}>{\hspre}l<{\hspost}@{}}%
\column{3}{@{}>{\hspre}l<{\hspost}@{}}%
\column{E}{@{}>{\hspre}l<{\hspost}@{}}%
\>[3]{}\Varid{arr}\mathbin{::}\Conid{MArray}\;\Varid{a}\;\Varid{n}{}\<[E]%
\ColumnHook
\end{hscode}\resethooks
Then, the constraints \ensuremath{\constraintfont{\Conid{Read}\;\Varid{n}}} and \ensuremath{\constraintfont{\Conid{Write}\;\Varid{n}}} become evidence that the array can be
read from or written to. As a convenience we also have a synonym for their product:
\begin{hscode}\SaveRestoreHook
\column{B}{@{}>{\hspre}l<{\hspost}@{}}%
\column{E}{@{}>{\hspre}l<{\hspost}@{}}%
\>[B]{}\mathbf{type}\;\constraintfont{\Conid{RW}\;\Varid{n}}\mathrel{=}\constraintfont{(\Conid{Read}\;\Varid{n},\Conid{Write}\;\Varid{n})}{}\<[E]%
\ColumnHook
\end{hscode}\resethooks
That is, the constraint \ensuremath{\constraintfont{\Conid{RW}\;\Varid{n}}}
is provable if and only if the array tagged with \ensuremath{\Varid{n}} is readable and writable.
This constraint is linear: it must be consumed (that is, used as an assumption
in a function call) exactly once.

The \ensuremath{\Varid{free}} function, which frees the contents of a given array will consume the \ensuremath{\constraintfont{\Conid{RW}\;\Varid{n}}} constraint
entirely, making it no longer available:
\begin{hscode}\SaveRestoreHook
\column{B}{@{}>{\hspre}l<{\hspost}@{}}%
\column{E}{@{}>{\hspre}l<{\hspost}@{}}%
\>[B]{}\Varid{free}\mathbin{::}\constraintfont{\Conid{RW}\;\Varid{n}}\Lolly \Conid{UArray}\;\Varid{a}\;\Varid{n}\to (){}\<[E]%
\ColumnHook
\end{hscode}\resethooks
This function can only be called when \ensuremath{\constraintfont{\Conid{RW}\;\Varid{n}}} can be satisfied for the given array of type \ensuremath{\Conid{UArray}\;\Varid{a}\;\Varid{n}}.

There are a few things to notice:
\begin{itemize}
\item We have introduced a new type variable \ensuremath{\Varid{n}}, which is a type-level tag
  used to identify the array. In contrast, the version in
  \cref{fig:linear-interface} without linear constraints  has type
  \ensuremath{\Varid{free}\mathbin{::}\Conid{MArray}\;\Varid{a}⊸()}.
\item The run-time variable representing the array can now be used multiple
  times. Instead of restricting the use of this variable,
  the linear constraint \ensuremath{\constraintfont{\Conid{RW}\;\Varid{n}}} controls access to the
  array.
\item If we have a linear \ensuremath{\constraintfont{\Conid{RW}\;\Varid{n}}}
  available, then after \ensuremath{\Varid{free}} there will not be any \ensuremath{\constraintfont{\Conid{RW}\;\Varid{n}}}
  left to use, thus preventing the array from being used after freeing.
  This is precisely what we were trying to achieve.
\end{itemize}

Let us now see some examples programs that use \ensuremath{\Varid{free}}. We reject \ensuremath{\Varid{dithering}}
because it does not unconditionally consume \ensuremath{\constraintfont{\Conid{RW}\;\Varid{n}}}:
\begin{hscode}\SaveRestoreHook
\column{B}{@{}>{\hspre}l<{\hspost}@{}}%
\column{E}{@{}>{\hspre}l<{\hspost}@{}}%
\>[B]{}\Varid{dithering}\mathbin{::}\constraintfont{\Conid{RW}\;\Varid{n}}\Lolly \Conid{UArray}\;\Varid{a}\;\Varid{n}\to \Conid{Bool}\to (){}\<[E]%
\\
\>[B]{}\Varid{dithering}\;\Varid{arr}\;\Varid{x}\mathrel{=}\mathbf{if}\;\Varid{x}\;\mathbf{then}\;\Varid{free}\;\Varid{arr}\;\mathbf{else}\;(){}\<[E]%
\ColumnHook
\end{hscode}\resethooks
The branch where \ensuremath{\Varid{x}\equiv \Conid{True}} uses the resource \ensuremath{\constraintfont{\Conid{RW}\;\Varid{n}}},
whereas the other branch does not.

Linear constraints can interact with linear functions.
Consider the type of the linear version of \ensuremath{\Varid{const}}:
\begin{hscode}\SaveRestoreHook
\column{B}{@{}>{\hspre}l<{\hspost}@{}}%
\column{E}{@{}>{\hspre}l<{\hspost}@{}}%
\>[B]{}\Varid{const}\mathbin{::}\Varid{a}⊸\Varid{b}\to \Varid{a}{}\<[E]%
\ColumnHook
\end{hscode}\resethooks
This function uses its first argument linearly, and ignores the second. Thus,
the second arrow is unrestricted.
One way to improperly use the linear \ensuremath{\Varid{const}} is by neglecting a linear variable:
\begin{hscode}\SaveRestoreHook
\column{B}{@{}>{\hspre}l<{\hspost}@{}}%
\column{E}{@{}>{\hspre}l<{\hspost}@{}}%
\>[B]{}\Varid{neglecting}\mathbin{::}\constraintfont{\Conid{RW}\;\Varid{n}}\Lolly \Conid{UArray}\;\Varid{a}\;\Varid{n}\to (){}\<[E]%
\\
\>[B]{}\Varid{neglecting}\;\Varid{arr}\mathrel{=}\Varid{const}\;()\;(\Varid{free}\;\Varid{arr}){}\<[E]%
\ColumnHook
\end{hscode}\resethooks
The problem with \ensuremath{\Varid{neglecting}} is that, although \ensuremath{\Varid{free}} is mentioned in this program,
it is never consumed: \ensuremath{\Varid{const}} does not use its second argument.
The constraint \ensuremath{\constraintfont{\Conid{RW}\;\Varid{n}}} is not consumed exactly once, and
thus this program is rejected.
The rule is that a linear constraint can only be consumed (linearly)
in a linear context. For example,
\begin{hscode}\SaveRestoreHook
\column{B}{@{}>{\hspre}l<{\hspost}@{}}%
\column{E}{@{}>{\hspre}l<{\hspost}@{}}%
\>[B]{}\Varid{notNeglecting}\mathbin{::}\constraintfont{\Conid{RW}\;\Varid{n}}\Lolly \Conid{UArray}\;\Varid{a}\;\Varid{n}\to (){}\<[E]%
\\
\>[B]{}\Varid{notNeglecting}\;\Varid{arr}\mathrel{=}\Varid{const}\;(\Varid{free}\;\Varid{arr})\;(){}\<[E]%
\ColumnHook
\end{hscode}\resethooks
is accepted, because the \ensuremath{\constraintfont{\Conid{RW}\;\Varid{n}}} constraint is passed on to \ensuremath{\Varid{free}}
which itself appears as an argument to a linear function (whose result is
itself consumed linearly).

\label{sec:overusing}
Finally, the following program is rejected because it uses \ensuremath{\Conid{RW}\;\Varid{n}} twice:
\begin{hscode}\SaveRestoreHook
\column{B}{@{}>{\hspre}l<{\hspost}@{}}%
\column{E}{@{}>{\hspre}l<{\hspost}@{}}%
\>[B]{}\Varid{indulging}\mathbin{::}\constraintfont{\Conid{RW}\;\Varid{n}}\Lolly \Conid{UArray}\;\Varid{a}\;\Varid{n}\to ((),()){}\<[E]%
\\
\>[B]{}\Varid{indulging}\;\Varid{arr}\mathrel{=}(\Varid{free}\;\Varid{arr},\Varid{free}\;\Varid{arr}){}\<[E]%
\ColumnHook
\end{hscode}\resethooks
Although freeing a resource multiple times is sometimes only
considered bad form, our type system rules this out completely.

\subsection{Threading Linear Constraints}
\label{sec:borrowing-linear-constraints}

Working with mutable arrays requires operations such as \ensuremath{\Varid{read}} and \ensuremath{\Varid{write}} to coordinate
on their access. Using ordinary linear types this coordination is achieved by linearly
consuming the input array of type \ensuremath{\Conid{MArray}\;\Varid{a}}, and making it available as output
for the next operation.
This can be readily seen in the types, where \ensuremath{\Conid{MArray}\;\Varid{a}} appears twice -- once as input
into a linear function and then again in the output. Consider the type of \ensuremath{\Varid{write}}:
\begin{hscode}\SaveRestoreHook
\column{B}{@{}>{\hspre}l<{\hspost}@{}}%
\column{E}{@{}>{\hspre}l<{\hspost}@{}}%
\>[B]{}\Varid{write}\mathbin{::}\Conid{MArray}\;\Varid{a}⊸\Conid{Int}\to \Varid{a}\to \Conid{MArray}\;\Varid{a}{}\<[E]%
\ColumnHook
\end{hscode}\resethooks
%
This interface is at the heart of the bureaucracy that we wish to avoid
in programs like the first \ensuremath{\Varid{read2AndDicard}} implementation (\cref{sec:introduction}),
since the resulting array is an explicit value that must be threaded
forward through to the next operation. The frustration is that this is essentially
a fresh variable for the array which was passed in as input.

Using linear constraints allows the appropriate capability to be passed around implicitly,
so that the array variable can be reused and no fresh variables need to be emitted
for threading.
In terms of capabilities, the operation for writing to an array requires
exclusive access to both read and write capabilities so that any other
reads or writes cannot occur concurrently.
This is demonstrated in the type of the \ensuremath{\Varid{write}} function that our library
enables, where \ensuremath{\constraintfont{\Conid{RW}\;\Varid{n}}} is the capability to access an array, which
is to say that both read and write capabilities are needed:
\begin{hscode}\SaveRestoreHook
\column{B}{@{}>{\hspre}l<{\hspost}@{}}%
\column{E}{@{}>{\hspre}l<{\hspost}@{}}%
\>[B]{}\Varid{write}\mathbin{::}\constraintfont{\Conid{RW}\;\Varid{n}}\Lolly \Conid{UArray}\;\Varid{a}\;\Varid{n}\to \Conid{Int}\to \Varid{a}\to ()\RLolly\constraintfont{\Conid{RW}\;\Varid{n}}{}\<[E]%
\ColumnHook
\end{hscode}\resethooks
Here there are two occurrences of the \ensuremath{\constraintfont{\Conid{RW}\;\Varid{n}}} constraint:
The first occurrence in the usual constraint position on the left, \ensuremath{\constraintfont{\Conid{RW}\;\Varid{n}}\Lolly \mathbin{...}},
indicates that this function \emph{consumes} a \ensuremath{\constraintfont{\Conid{RW}\;\Varid{n}}} constraint.
The second occurrence at the end of the signature, \ensuremath{\mathbin{...}\RLolly\constraintfont{\Conid{RW}\;\Varid{n}}},
indicates that this function \emph{produces} a \ensuremath{\constraintfont{\Conid{RW}\;\Varid{n}}} constraint.\footnote{In general \ensuremath{\Conid{A}\RLolly\constraintfont{\Conid{C}}} pairs a linear dictionary of \ensuremath{\constraintfont{\Conid{C}}} with \ensuremath{\Conid{A}}, see \cref{sec:typing-rules} for the formal treatment.}

We must also ensure that \ensuremath{\Varid{read}} can both promise to operate only on a readable array
and that the array persists in its readable state:
\begin{hscode}\SaveRestoreHook
\column{B}{@{}>{\hspre}l<{\hspost}@{}}%
\column{E}{@{}>{\hspre}l<{\hspost}@{}}%
\>[B]{}\Varid{read}\mathbin{::}\constraintfont{\Conid{Read}\;\Varid{n}}\Lolly \Conid{UArray}\;\Varid{a}\;\Varid{n}\to \Conid{Int}\to \Conid{Ur}\;\Varid{a}\RLolly\constraintfont{\Conid{Read}\;\Varid{n}}{}\<[E]%
\ColumnHook
\end{hscode}\resethooks
That is, \ensuremath{\Varid{read}} must both consume
a linear constraint \ensuremath{\constraintfont{\Conid{Read}\;\Varid{n}}} and also produce a fresh linear constraint
\ensuremath{\constraintfont{\Conid{Read}\;\Varid{n}}}.

As a consequence, the array \ensuremath{\Varid{arr}} can be safely reused as a parameter
between different invocations of \ensuremath{\Varid{read}\;\Varid{arr}} and \ensuremath{\Varid{write}\;\Varid{arr}}, while it is the
capabilities that are being implicitly threaded around.
When the work with the array is finished, the capabilities can be consumed with \ensuremath{\Varid{free}\;\Varid{arr}},
again using the same variable name.

\subsection{Producing Linear Contexts}
\label{sec:linear-contexts}
\label{sec:Unique-constraint}
\label{sec:producing-linear-constraints}

The above deals with freeing an array and ensuring that it cannot be used afterwards.
However, we still need to explain how a constraint \ensuremath{\constraintfont{\Conid{RW}\;\Varid{n}}} can come
into scope. First though, it will be beneficial to understand the
mechanics of creating new arrays using linear types without constraints.

The creation of an \ensuremath{\Conid{MArray}\;\Varid{a}} must be carefully protected,
because a naive constructor function can be abused. For instance, consider
the following function that creates a new array:
\begin{hscode}\SaveRestoreHook
\column{B}{@{}>{\hspre}l<{\hspost}@{}}%
\column{E}{@{}>{\hspre}l<{\hspost}@{}}%
\>[B]{}\Varid{newNaive}\mathbin{::}\Conid{Int}\to \Conid{MArray}\;\Varid{a}{}\<[E]%
\ColumnHook
\end{hscode}\resethooks
Nothing prevents the result of a call
to \ensuremath{\Varid{naive}} from being duplicated with \ensuremath{\Varid{diagonal}}:
\begin{hscode}\SaveRestoreHook
\column{B}{@{}>{\hspre}l<{\hspost}@{}}%
\column{E}{@{}>{\hspre}l<{\hspost}@{}}%
\>[B]{}\Varid{diagonal}\mathbin{::}\Varid{a}\to (\Varid{a},\Varid{a}){}\<[E]%
\\
\>[B]{}\Varid{diagonal}\;\Varid{x}\mathrel{=}(\Varid{x},\Varid{x}){}\<[E]%
\ColumnHook
\end{hscode}\resethooks
Now the first and second components of \ensuremath{\Varid{diagonal}\;(\Varid{newNaive}\;\mathrm{10})} point to
the same array and there is nothing preventing \ensuremath{\Varid{free}} from
being evaluated on one component before the other is written to.
Similarly, the array produced can be lifted an unrestricted occurrence
with a simple use of \ensuremath{\Conid{Ur}\;(\Varid{newNaive}\;\Varid{n})}.
With the above interface it possible to perform repeated operations on an array,
completely circumventing the linear safety mechanism.

The original interface for linear types
makes use of the \ensuremath{\Varid{new}} function to create new arrays,
following the pattern of using a linear continuation
advocated by \citet[Fig.~2]{LinearHaskell},
(which is similar style to the use of the \ensuremath{\Conid{ST}} monad~\cite{st-monad}):
\begin{hscode}\SaveRestoreHook
\column{B}{@{}>{\hspre}l<{\hspost}@{}}%
\column{6}{@{}>{\hspre}l<{\hspost}@{}}%
\column{E}{@{}>{\hspre}l<{\hspost}@{}}%
\>[B]{}\Varid{new}{}\<[6]%
\>[6]{}\mathbin{::}\Conid{Int}\to (\Conid{MArray}\;\Varid{a}⊸\Conid{Ur}\;\Varid{r})⊸\Conid{Ur}\;\Varid{r}{}\<[E]%
\ColumnHook
\end{hscode}\resethooks
The result of \ensuremath{\Varid{new}\;\Varid{k}\;(\lambda \Varid{arr}\to \Varid{p})} is essentially that of the program \ensuremath{\Varid{p}},
where it has been able to make use of the array \ensuremath{\Varid{arr}} of size \ensuremath{\Varid{k}},
with the guarantee that \ensuremath{\Varid{arr}} is treated linearly.
In particular, only one reference to \ensuremath{\Varid{arr}} is produced by \ensuremath{\Varid{new}},
and this must be consumed linearly, thus avoiding the problems that arise from using the naive constructor.

Similarly, it is possible to create a \ensuremath{\Varid{new}} function
for working with \ensuremath{\Conid{UArray}\;\Varid{a}\;\Varid{n}} types using the same trick:
\begin{hscode}\SaveRestoreHook
\column{B}{@{}>{\hspre}l<{\hspost}@{}}%
\column{6}{@{}>{\hspre}l<{\hspost}@{}}%
\column{51}{@{}>{\hspre}l<{\hspost}@{}}%
\column{E}{@{}>{\hspre}l<{\hspost}@{}}%
\>[B]{}\Varid{new}{}\<[6]%
\>[6]{}\mathbin{::}\Conid{Int}\to (\forall\;\Varid{n}.\constraintfont{\Conid{RW}\;\Varid{n}}\Lolly {}\<[51]%
\>[51]{}\Conid{UArray}\;\Varid{a}\;\Varid{n}\to \Conid{Ur}\;\Varid{r})⊸\Conid{Ur}\;\Varid{r}{}\<[E]%
\ColumnHook
\end{hscode}\resethooks
Now, when \ensuremath{\Varid{new}\;\Varid{k}\;(\lambda \Varid{arr}\to \Varid{p})} is called, the \ensuremath{\constraintfont{\Conid{RW}\;\Varid{n}}} constraint is made
available to the computation \ensuremath{\Varid{p}} that is provided unrestricted access to
the new array \ensuremath{\Varid{arr}} of size \ensuremath{\Varid{k}}.

\subsection{Linear Contexts with Linearly}

\begin{figure}
\begin{hscode}\SaveRestoreHook
\column{B}{@{}>{\hspre}l<{\hspost}@{}}%
\column{E}{@{}>{\hspre}l<{\hspost}@{}}%
\>[B]{}\mathbf{type}\;\constraintfont{\Conid{Linearly}}\mathbin{::}\Conid{Constraint}{}\<[E]%
\\[\blanklineskip]%
\>[B]{}\Varid{dup}\mathbin{::}\constraintfont{\Conid{Linearly}}\Lolly ()\RLolly(\constraintfont{\Conid{Linearly}},\constraintfont{\Conid{Linearly}}){}\<[E]%
\\
\>[B]{}\Varid{dis}\mathbin{::}\constraintfont{\Conid{Linearly}}\Lolly (){}\<[E]%
\\[\blanklineskip]%
\>[B]{}\Varid{linearly}\mathbin{::}(\constraintfont{\Conid{Linearly}}\Lolly \Conid{Ur}\;\Varid{r})⊸\Conid{Ur}\;\Varid{r}{}\<[E]%
\ColumnHook
\end{hscode}\resethooks
\caption{Interface for the built-in \ensuremath{\constraintfont{\Conid{Linearly}}} constraint}
\label{fig:unique-interface}
\end{figure}
Working within a linear context provided using a continuation-based interface
restricts the scope within which the bound variable can occur.
Thus, when multiple resources are required, instances of \ensuremath{\Varid{new}} must
be nested, which can result in unwieldy code.

Linear constraints provide an alternative solution to this problem
by working with a special built-in \ensuremath{\constraintfont{\Conid{Linearly}}} constraint.
Using this special constraint, it is possible to create a new
array variable for the explicitly threaded \ensuremath{\Conid{MArray}} type that can be safely used:
\begin{hscode}\SaveRestoreHook
\column{B}{@{}>{\hspre}l<{\hspost}@{}}%
\column{E}{@{}>{\hspre}l<{\hspost}@{}}%
\>[B]{}\Varid{new}\mathbin{::}\constraintfont{\constraintfont{\Conid{Linearly}}}\Lolly \Conid{Int}\to \Conid{MArray}\;\Varid{a}{}\<[E]%
\ColumnHook
\end{hscode}\resethooks
This works in a manner that mirrors \ensuremath{\Varid{newNaive}}, but nevertheless
ensures that no misuse is possible.
Suppose that we have assumed the \ensuremath{\constraintfont{\constraintfont{\Conid{Linearly}}}} constraint linearly; that is,
we must use the \ensuremath{\constraintfont{\constraintfont{\Conid{Linearly}}}} assumption exactly once.
Now, \ensuremath{\Varid{diagonal}\;(\Varid{new}\;\mathrm{10})} is rejected: either we infer
the new array to have multiplicity~$  \multiplicityfont{ \omega }  $,
in which case its definition uses \ensuremath{\constraintfont{\constraintfont{\Conid{Linearly}}}} $  \multiplicityfont{ \omega }  $ times; or
we infer it to have multiplicity~$  \multiplicityfont{ \ottsym{1} }  $, in which case its use (twice) violates
the linearity restriction. Likewise, the use of \ensuremath{\Conid{Ur}\;(\Varid{new}\;\mathrm{10})} requires using
the \ensuremath{\constraintfont{\constraintfont{\Conid{Linearly}}}} assumption $  \multiplicityfont{ \omega }  $ times and is thus rejected.

A \ensuremath{\Varid{new}} function that works for the linear constraints
interface for arrays using \ensuremath{\Conid{UArray}} can also be defined:
\begin{hscode}\SaveRestoreHook
\column{B}{@{}>{\hspre}l<{\hspost}@{}}%
\column{6}{@{}>{\hspre}l<{\hspost}@{}}%
\column{E}{@{}>{\hspre}l<{\hspost}@{}}%
\>[B]{}\Varid{new}{}\<[6]%
\>[6]{}\mathbin{::}\constraintfont{\constraintfont{\Conid{Linearly}}}\Lolly \Conid{Int}\to \exists\;\Varid{n}.\Conid{Ur}\;(\Conid{UArray}\;\Varid{a}\;\Varid{n})\RLolly\constraintfont{\Conid{RW}\;\Varid{n}}{}\<[E]%
\ColumnHook
\end{hscode}\resethooks
This time, the new array comes equipped with the \ensuremath{\constraintfont{\Conid{RW}\;\Varid{n}}} capability,
ready to be consumed by future operations.

This behaviour suggests that \ensuremath{\Conid{Linearly}} gives us precisely the discipline that we
want: a way to assume a constraint exactly once, thereby guaranteeing that any
resource created under that assumption is itself used linearly.
However, the picture is not yet complete.
Assuming \ensuremath{\constraintfont{\constraintfont{\Conid{Linearly}}}} exactly once is not, by itself, enough:
real programs allocate different numbers of arrays, may allocate none at all,
and must ensure that no linear resource escapes its intended scope.
These requirements are fulfilled by the interface to \ensuremath{\constraintfont{\constraintfont{\Conid{Linearly}}}} constraints (\cref{fig:unique-interface}).

The remainder of this section explains why each of these operations is required,
and how they allow \ensuremath{\constraintfont{\constraintfont{\Conid{Linearly}}}} to serve as a usable linear capability.
We discuss, in turn, how to \emph{duplicate} a linear constraint, how to
\emph{discard} it, and finally how to \emph{introduce} it.

\paragraph*{Duplicating \ensuremath{\constraintfont{\constraintfont{\Conid{Linearly}}}}}
Programs that require multiple arrays each with different unique
pointers need to interact well with \ensuremath{\constraintfont{\constraintfont{\Conid{Linearly}}}}.
If \ensuremath{\constraintfont{\constraintfont{\Conid{Linearly}}}} is assumed
linearly, then \ensuremath{\mathbf{let}\;\Varid{arr}_{1}\mathrel{=}\Varid{new}\;\mathrm{5};\Varid{arr}_{2}\mathrel{=}\Varid{new}\;\mathrm{6}} will fail, because it uses our
\ensuremath{\constraintfont{\constraintfont{\Conid{Linearly}}}} assumption twice. We thus stipulate that \ensuremath{\constraintfont{\constraintfont{\Conid{Linearly}}}}
must itself be duplicable: from one assumption of \ensuremath{\constraintfont{\constraintfont{\Conid{Linearly}}}}, we
must be able to satisfy any arbitrary fixed number of demands on that constraint.
Thus, the interface to \ensuremath{\constraintfont{\constraintfont{\Conid{Linearly}}}} provides \ensuremath{\Varid{dup}}, which
allows this duplication:
\begin{hscode}\SaveRestoreHook
\column{B}{@{}>{\hspre}l<{\hspost}@{}}%
\column{E}{@{}>{\hspre}l<{\hspost}@{}}%
\>[B]{}\Varid{dup}\mathbin{::}\constraintfont{\Conid{Linearly}}\Lolly ()\RLolly(\constraintfont{\Conid{Linearly}},\constraintfont{\Conid{Linearly}}){}\<[E]%
\ColumnHook
\end{hscode}\resethooks
By ``arbitrary fixed number'', we mean to say that we can duplicate \ensuremath{\constraintfont{\constraintfont{\Conid{Linearly}}}}
a finite number of times, but we may not use an assumption of \ensuremath{\constraintfont{\constraintfont{\Conid{Linearly}}}} with
multiplicity 1 to satisfy \ensuremath{\constraintfont{\constraintfont{\Conid{Linearly}}}} at multiplicity $  \multiplicityfont{ \omega }  $.

\paragraph*{Discarding \ensuremath{\constraintfont{\constraintfont{\Conid{Linearly}}}}} Similarly to allowing
  duplication, we must allow discarding, in case a function allocates
  no arrays at all. This is made possible with a function \ensuremath{\Varid{dis}}:
\begin{hscode}\SaveRestoreHook
\column{B}{@{}>{\hspre}l<{\hspost}@{}}%
\column{E}{@{}>{\hspre}l<{\hspost}@{}}%
\>[B]{}\Varid{dis}\mathbin{::}\constraintfont{\Conid{Linearly}}\Lolly (){}\<[E]%
\ColumnHook
\end{hscode}\resethooks
  Accordingly, we allow a linear assumption of \ensuremath{\constraintfont{\constraintfont{\Conid{Linearly}}}} to
  be accepted even if the constraint is never used.  Our type-system
  not only supports \ensuremath{\Varid{dup}} and \ensuremath{\Varid{dis}}, but will also infer when they
  need to be used.  In sum, a single \ensuremath{\constraintfont{\constraintfont{\Conid{Linearly}}}} is
  sufficient to indicate that the program can call any number of
  functions that require this linear constraint.

\paragraph*{Introducing \ensuremath{\constraintfont{\constraintfont{\Conid{Linearly}}}}} For this approach to work,
we must have an assumption of \ensuremath{\constraintfont{\constraintfont{\Conid{Linearly}}}} of multiplicity 1. We can
achieve this via the following primitive:
\begin{hscode}\SaveRestoreHook
\column{B}{@{}>{\hspre}l<{\hspost}@{}}%
\column{E}{@{}>{\hspre}l<{\hspost}@{}}%
\>[B]{}\Varid{linearly}\mathbin{::}(\constraintfont{\Conid{Linearly}}\Lolly \Conid{Ur}\;\Varid{r})⊸\Conid{Ur}\;\Varid{r}{}\<[E]%
\ColumnHook
\end{hscode}\resethooks
The argument to \ensuremath{\Varid{linearly}} will be a continuation that assumes \ensuremath{\constraintfont{\constraintfont{\Conid{Linearly}}}}
with multiplicity 1. Because \ensuremath{\Varid{linearly}} returns an unrestricted value, no restricted
values from the continuation can escape the scope of the \ensuremath{\constraintfont{\constraintfont{\Conid{Linearly}}}}
assumption. Thus, the continuation has exactly the condition we need: a
linear assumption of \ensuremath{\constraintfont{\constraintfont{\Conid{Linearly}}}}.

\vspace{\baselineskip} 

With just these simple ingredients---a duplicable, discardable constraint
that can be assumed linearly---we can write \textsc{api}s that require
uniqueness without heavy use of continuations.
The pattern of using a continuation in \ensuremath{\Varid{linearly}} mirrors the use
of that technique by \citet[Fig.~2]{LinearHaskell}. But
\ensuremath{\Varid{linearly}} is, now, the only place where we need a continuation: once
we have our linear \ensuremath{\constraintfont{\constraintfont{\Conid{Linearly}}}} assumption, we can use it to produce
new values that must be unique. We compare the two approaches further in
\cref{sec:linearly-vs-sticky}.

\subsection{Chaining with the ``do'' notation}
\label{sec:packing-unpacking-do}

The last ingredient that we need is the ability to string
linear-constraint-threading functions together. For this we will rely on a ``do''
notation. We want to write functions like:

\begin{hscode}\SaveRestoreHook
\column{B}{@{}>{\hspre}l<{\hspost}@{}}%
\column{3}{@{}>{\hspre}l<{\hspost}@{}}%
\column{E}{@{}>{\hspre}l<{\hspost}@{}}%
\>[B]{}\Conid{Linearly}.\mathbf{do}{}\<[E]%
\\
\>[B]{}\hsindent{3}{}\<[3]%
\>[3]{}\Varid{arr}\leftarrow \Varid{new}\;\mathrm{57}{}\<[E]%
\\
\>[B]{}\hsindent{3}{}\<[3]%
\>[3]{}\Varid{write}\;\mathrm{42}\;\text{\ttfamily \char34 answer\char34}{}\<[E]%
\\
\>[B]{}\hsindent{3}{}\<[3]%
\>[3]{}\Conid{Ur}\;\Varid{x}\leftarrow \Varid{read}\;\mathrm{0}{}\<[E]%
\\
\>[B]{}\hsindent{3}{}\<[3]%
\>[3]{}\Varid{\Conid{Linearly}.return}\;\Varid{x}{}\<[E]%
\ColumnHook
\end{hscode}\resethooks

The simplest part of this expression is the \ensuremath{\Varid{\Conid{Linearly}.return}} function in the
final line. Its role is merely to package a value with the necessary linear
constraints and existential quantifiers. It has type

\begin{hscode}\SaveRestoreHook
\column{B}{@{}>{\hspre}l<{\hspost}@{}}%
\column{E}{@{}>{\hspre}l<{\hspost}@{}}%
\>[B]{}\Varid{\Conid{Linearly}.return}\mathbin{::}\constraintfont{\Conid{Q}}\Lolly \Varid{t}\to \exists\;\Varid{a}_{1}\mathbin{...}\Varid{a}_{\Varid{n}}.\Varid{t}\RLolly\constraintfont{\Conid{Q}}{}\<[E]%
\ColumnHook
\end{hscode}\resethooks

This isn't a real Haskell type, but it's very convenient to pretend that it is.
Admitting such types, questions about syntax notwithstanding, is a mild
extension of Haskell. In real Haskell you would have to rely on a family of
\textsc{gadt}s instead (one for each \ensuremath{\Varid{t}} and \ensuremath{\constraintfont{\Conid{Q}}}). Though we discuss a more ambitious extension in
\cref{sec:implicit-existentials}.

The rest of the ``do'' notation is understood by desugaring to more atomic
function. Specifically, \textsc{ghc}'s ``qualified do'' feature will translate
our example to the following:

\begin{hscode}\SaveRestoreHook
\column{B}{@{}>{\hspre}l<{\hspost}@{}}%
\column{23}{@{}>{\hspre}l<{\hspost}@{}}%
\column{E}{@{}>{\hspre}l<{\hspost}@{}}%
\>[B]{}\Varid{new}\;\mathrm{57}{}\<[23]%
\>[23]{}\Conid{Linearly}{.\!}\bind \lambda \Varid{arr}\to {}\<[E]%
\\
\>[B]{}\Varid{write}\;\mathrm{42}\;\text{\ttfamily \char34 answer\char34}\,{}\<[23]%
\>[23]{}\Conid{Linearly}{.\!}\sequ {}\<[E]%
\\
\>[B]{}\Varid{read}\;\mathrm{0}{}\<[23]%
\>[23]{}\Conid{Linearly}{.\!}\bind \lambda !(\Conid{Ur}\;\Varid{x})\to {}\<[E]%
\\
\>[B]{}\Varid{\Conid{Linearly}.return}\;\Varid{x}{}\<[E]%
\ColumnHook
\end{hscode}\resethooks

More generally, this follows the standard translation:

\begin{hscode}\SaveRestoreHook
\column{B}{@{}>{\hspre}l<{\hspost}@{}}%
\column{32}{@{}>{\hspre}c<{\hspost}@{}}%
\column{32E}{@{}l@{}}%
\column{36}{@{}>{\hspre}l<{\hspost}@{}}%
\column{E}{@{}>{\hspre}l<{\hspost}@{}}%
\>[B]{}\Conid{Linearly}.\mathbf{do}\;\{\mskip1.5mu \Varid{x}\leftarrow \Varid{u};\Varid{stmts}\mskip1.5mu\}{}\<[32]%
\>[32]{}\ensuremath{\leadsto}{}\<[32E]%
\>[36]{}\Varid{u}\,\,\Conid{Linearly}{.\!}\bind \lambda \Varid{x}\to \Conid{Linearly}.\mathbf{do}\;\{\mskip1.5mu \Varid{stmts}\mskip1.5mu\}{}\<[E]%
\\
\>[B]{}\Conid{Linearly}.\mathbf{do}\;\{\mskip1.5mu \Varid{pat}\leftarrow \Varid{u};\Varid{stmts}\mskip1.5mu\}{}\<[32]%
\>[32]{}\ensuremath{\leadsto}{}\<[32E]%
\>[36]{}\Varid{u}\,\,\Conid{Linearly}{.\!}\bind \lambda !\Varid{pat}\to \Conid{Linearly}.\mathbf{do}\;\{\mskip1.5mu \Varid{stmts}\mskip1.5mu\}{}\<[E]%
\\
\>[B]{}\Conid{Linearly}.\mathbf{do}\;\{\mskip1.5mu \Varid{u};\Varid{stmts}\mskip1.5mu\}{}\<[32]%
\>[32]{}\ensuremath{\leadsto}{}\<[32E]%
\>[36]{}\Varid{u}\,\,\Conid{Linearly}{.\!}\sequ \Conid{Linearly}.\mathbf{do}\;\{\mskip1.5mu \Varid{stmts}\mskip1.5mu\}{}\<[E]%
\ColumnHook
\end{hscode}\resethooks
Where \ensuremath{\Varid{pat}} is a non-variable pattern.

The role of \ensuremath{\Conid{Linearly}{.\!}\bind } and \ensuremath{\Conid{Linearly}{.\!}\sequ } is to pass both the
returned value and returned constraints to their continuation. Thus,
the \ensuremath{\Conid{Linearly}.\mathbf{do}} notation does not strictly sequence commands in
textual order.  It only sequences commands when the one consumes the
output constraint of another one.

\begin{hscode}\SaveRestoreHook
\column{B}{@{}>{\hspre}l<{\hspost}@{}}%
\column{17}{@{}>{\hspre}l<{\hspost}@{}}%
\column{45}{@{}>{\hspre}l<{\hspost}@{}}%
\column{65}{@{}>{\hspre}l<{\hspost}@{}}%
\column{114}{@{}>{\hspre}l<{\hspost}@{}}%
\column{E}{@{}>{\hspre}l<{\hspost}@{}}%
\>[B]{}(\Conid{Linearly}{.}\bind ){}\<[17]%
\>[17]{}\mathbin{::}(\exists\;\Varid{a}_{1}\mathbin{...}\Varid{a}_{\Varid{n}}.\,\Varid{t}{}\<[45]%
\>[45]{}\RLolly\constraintfont{\Conid{Q}}){}\<[65]%
\>[65]{}⊸(\forall\;\Varid{a}_{1}\mathbin{...}\Varid{a}_{\Varid{n}}.\,\constraintfont{\Conid{Q}}\Lolly \Varid{t}\to {}\<[114]%
\>[114]{}\Varid{s})⊸\Varid{s}{}\<[E]%
\\
\>[B]{}(\Conid{Linearly}{.}\sequ ){}\<[17]%
\>[17]{}\mathbin{::}(\exists\;\Varid{a}_{1}\mathbin{...}\Varid{a}_{\Varid{n}}.\,(){}\<[45]%
\>[45]{}\RLolly\constraintfont{\Conid{Q}}){}\<[65]%
\>[65]{}⊸(\forall\;\Varid{a}_{1}\mathbin{...}\Varid{a}_{\Varid{n}}.\,\constraintfont{\Conid{Q}}\Lolly {}\<[114]%
\>[114]{}\Varid{s})⊸\Varid{s}{}\<[E]%
\ColumnHook
\end{hscode}\resethooks

For the time being, we shall be considering these functions as primitives which
package and unpackage values with constraints and possibly existential
quantifiers. We will give them a more explicit definition in
\cref{sec:typing-rules,sec:desugaring}.

The \ensuremath{\Conid{Linearly}.\mathbf{do}} notation does not always behave exactly like the ordinary ``do''
notation for monads. In particular we will see idioms like the following:
\begin{hscode}\SaveRestoreHook
\column{B}{@{}>{\hspre}l<{\hspost}@{}}%
\column{3}{@{}>{\hspre}l<{\hspost}@{}}%
\column{E}{@{}>{\hspre}l<{\hspost}@{}}%
\>[B]{}\Conid{Linearly}.\mathbf{do}{}\<[E]%
\\
\>[B]{}\hsindent{3}{}\<[3]%
\>[3]{}\mathbin{...}{}\<[E]%
\\
\>[B]{}\hsindent{3}{}\<[3]%
\>[3]{}\Varid{foo}{}\<[E]%
\\
\>[B]{}\hsindent{3}{}\<[3]%
\>[3]{}\Varid{\Conid{Linearly}.return}\;(){}\<[E]%
\ColumnHook
\end{hscode}\resethooks
In the ``do'' notation for monads, the last line would be superfluous. But here
it lets the type system rejig the constraints. The call to \ensuremath{\Varid{foo}} will consume
\emph{some} of the linear constraints in scope, and return some new constraints.
The call to \ensuremath{\Varid{\Conid{Linearly}.return}} will package together the remaining constraints
that \ensuremath{\Varid{foo}} left untouched, and the new constraints returned by \ensuremath{\Varid{foo}}.

We can also use familiar combinators adapted to the \ensuremath{\Conid{Linearly}.\mathbf{do}}
notation. For instance we will be using:

\begin{hscode}\SaveRestoreHook
\column{B}{@{}>{\hspre}l<{\hspost}@{}}%
\column{13}{@{}>{\hspre}l<{\hspost}@{}}%
\column{16}{@{}>{\hspre}l<{\hspost}@{}}%
\column{E}{@{}>{\hspre}l<{\hspost}@{}}%
\>[B]{}\Varid{when}\mathbin{::}\constraintfont{\Conid{Q}}\Lolly \Conid{Bool}\to (\constraintfont{\Conid{Q}}\Lolly ()\RLolly\constraintfont{\Conid{Q}})\to ()\RLolly\constraintfont{\Conid{Q}}{}\<[E]%
\\
\>[B]{}\Varid{when}\;\Conid{True}\;{}\<[13]%
\>[13]{}\Varid{a}{}\<[16]%
\>[16]{}\mathrel{=}\Varid{a}{}\<[E]%
\\
\>[B]{}\Varid{when}\;\Conid{False}\;{}\<[13]%
\>[13]{}\anonymous {}\<[16]%
\>[16]{}\mathrel{=}\Varid{\Conid{Linearly}.return}\;(){}\<[E]%
\ColumnHook
\end{hscode}\resethooks
Other combinators follow a similar pattern.

\section{Applications}
\label{sec:memory-ownership}

There are many different applications of linear constraints.
In this section, we consider
borrowing slices~(\cref{sec:borrowing-slices}),
blockwise matrices~(\cref{sec:valiant}),
nested arrays~(\cref{sec:nested-arrays}) and their implementation~(\cref{sec:narray-impl}),
freezing nested structures~(\cref{sec:o1-freeze}),
splay trees~(\cref{sec:splay-trees}),
non-lexical lifetimes~(\cref{sec:nl-lifetimes}),
self-referential borrows~(\cref{sec:self-borrow}), and
typestates in monadic computations~(\cref{sec:monad-capability}).

\subsection{Borrowing slices}
\label{sec:borrowing-slices}

Managing mutuable arrays is an important use case for linearly qualified types,
since many algorithms are split into subalgorithms that work on a particular
given subarray, or \emph{slice}, before returning control to the main algorithm.

This section first covers simple slices where work only needs to
be carried out on a restricted part of an array~(\cref{sec:simple-slices}).
Then the more general case is considered, where independent parts of
an array should be worked on independently~(\cref{sec:slicing-two}).


\subsubsection{Simple slices}
\label{sec:simple-slices}

A first intuition may be to add a function that looks like this:
\begin{hscode}\SaveRestoreHook
\column{B}{@{}>{\hspre}l<{\hspost}@{}}%
\column{10}{@{}>{\hspre}l<{\hspost}@{}}%
\column{E}{@{}>{\hspre}l<{\hspost}@{}}%
\>[B]{}\Varid{constrict}\mathbin{::}{}\<[10]%
\>[10]{}\constraintfont{\Conid{RW}\;\Varid{n}}\Lolly \Conid{UArray}\;\Varid{a}\;\Varid{n}\to \Conid{Int}\to \Conid{Int}\to {}\<[E]%
\\
\>[10]{}\exists\;\Varid{p}.\Conid{Ur}\;(\Conid{UArray}\;\Varid{a}\;\Varid{p})\RLolly\constraintfont{\Conid{RW}\;\Varid{p}}{}\<[E]%
\ColumnHook
\end{hscode}\resethooks
Given an array \ensuremath{\Varid{arr}} together with and index \ensuremath{\Varid{i}} and a length \ensuremath{\Varid{m}},
the result of \ensuremath{\Varid{constrict}\;\Varid{arr}\;\Varid{i}\;\Varid{m}} is the corresponding slice of \ensuremath{\Varid{m}}
elements of the array starting at index \ensuremath{\Varid{i}}.
Referential transparency is preserved because
the slice is returned with a fresh type-level name \ensuremath{\Varid{p}},
and all the capabilities on the original array with type-level name \ensuremath{\Varid{n}}
are consumed, so that the slice on the original array cannot be written to concurrently.

There is one problem, though: this is a one way transformation since
\ensuremath{\Varid{n}} is not recovered in the output constraint. It is rare that
one wants to abandon the original array altogether when taking a slice: the goal
is usually to focus on the array for a portion of an algorithm, and return to
the global view later.

To be able to return to the global view, we only need to be able to recover
the original \ensuremath{\constraintfont{\Conid{RW}\;\Varid{n}}} capabilities. A simple modification of \ensuremath{\Varid{constrict}}
achieves this result:
\begin{hscode}\SaveRestoreHook
\column{B}{@{}>{\hspre}l<{\hspost}@{}}%
\column{9}{@{}>{\hspre}l<{\hspost}@{}}%
\column{E}{@{}>{\hspre}l<{\hspost}@{}}%
\>[B]{}\Varid{restrict}\mathbin{::}{}\<[9]%
\>[9]{}\constraintfont{\Conid{RW}\;\Varid{n}}\Lolly \Conid{UArray}\;\Varid{a}\;\Varid{n}\to \Conid{Int}\to \Conid{Int}\to {}\<[E]%
\\
\>[9]{}\exists\;\Varid{p}.(\Conid{Ur}\;(\Conid{UArray}\;\Varid{a}\;\Varid{p}),(\constraintfont{\Conid{RW}\;\Varid{p}}\Lolly ()\RLolly\constraintfont{\Conid{RW}\;\Varid{n}}))\RLolly\constraintfont{\Conid{RW}\;\Varid{p}}{}\<[E]%
\ColumnHook
\end{hscode}\resethooks
That is, \ensuremath{\Varid{restrict}} returns an additional value of type \ensuremath{\constraintfont{\Conid{RW}\;\Varid{p}}\Lolly ()\RLolly(\constraintfont{\Conid{RW}\;\Varid{n}})}. We call this additional value a \emph{release operator}. We
can then call this release operator when we're done with the slice, and return
to the original array.

For instance we can write the following in-place insertion sort function, where
\ensuremath{\Varid{insert}\;\Varid{a}\;\Varid{bs}} inserts \ensuremath{\Varid{a}} into the already-sorted array \ensuremath{\Varid{bs}} in place. To
preserve the size of the array, the smallest element of $\{\ensuremath{\Varid{a}}\} ⊎ \ensuremath{\Varid{bs}}$ is removed
from the result and returned. In the interesting case, \ensuremath{\Varid{a}} isn't the smallest
element of $\{\ensuremath{\Varid{a}}\} ⊎ \ensuremath{\Varid{bs}}$, so \ensuremath{\Varid{b}\mathrel{=}\Varid{read}\;\Varid{bs}\;\mathrm{0}} is, we return \ensuremath{\Varid{b}} and insert \ensuremath{\Varid{a}} to the
tail of \ensuremath{\Varid{bs}}, that is \ensuremath{\Varid{restrict}\;\Varid{bs}\;\mathrm{1}\;(\Varid{length}\;\Varid{bs}\mathbin{-}\mathrm{1})}.

\begin{hscode}\SaveRestoreHook
\column{B}{@{}>{\hspre}l<{\hspost}@{}}%
\column{5}{@{}>{\hspre}l<{\hspost}@{}}%
\column{7}{@{}>{\hspre}l<{\hspost}@{}}%
\column{E}{@{}>{\hspre}l<{\hspost}@{}}%
\>[B]{}\Varid{insertSort}\mathbin{::}\constraintfont{\Conid{RW}\;\Varid{n}}\Lolly \Conid{UArray}\;\Varid{a}\;\Varid{n}\to ()\RLolly\constraintfont{\Conid{RW}\;\Varid{n}}{}\<[E]%
\\
\>[B]{}\Varid{insertSort}\;\Varid{as}\mathrel{=}\Varid{when}\;(\Varid{length}\;\Varid{as}\mathbin{>}\mathrm{1})\mathbin{\$}\Conid{Linearly}.\mathbf{do}{}\<[E]%
\\
\>[B]{}\hsindent{5}{}\<[5]%
\>[5]{}\Conid{Ur}\;\Varid{a}\leftarrow \Varid{read}\;\mathrm{0}\;\Varid{as}{}\<[E]%
\\
\>[B]{}\hsindent{5}{}\<[5]%
\>[5]{}(\Conid{Ur}\;\Varid{tail},\Varid{release})\leftarrow \Varid{restrict}\;\Varid{as}\;\mathrm{1}\;(\Varid{len}\mathbin{-}\mathrm{1}){}\<[E]%
\\
\>[B]{}\hsindent{5}{}\<[5]%
\>[5]{}\Varid{insertSort}\;\Varid{tail}{}\<[E]%
\\
\>[B]{}\hsindent{5}{}\<[5]%
\>[5]{}\Conid{Ur}\;\Varid{a'}\leftarrow \Varid{insert}\;\Varid{a}\;\Varid{tail}{}\<[E]%
\\
\>[B]{}\hsindent{5}{}\<[5]%
\>[5]{}\Varid{release}{}\<[E]%
\\
\>[B]{}\hsindent{5}{}\<[5]%
\>[5]{}\Varid{write}\;\Varid{as}\;\mathrm{0}\;\Varid{a'}{}\<[E]%
\\
\>[B]{}\hsindent{5}{}\<[5]%
\>[5]{}\Varid{\Conid{Linearly}.return}\;(){}\<[E]%
\\[\blanklineskip]%
\>[B]{}\Varid{insert}\mathbin{::}\constraintfont{\Conid{RW}\;\Varid{n}}\Lolly \Varid{a}\to \Conid{UArray}\;\Varid{a}\;\Varid{n}\to \Conid{Ur}\;\Varid{a}\RLolly\constraintfont{\Conid{RW}\;\Varid{n}}{}\<[E]%
\\
\>[B]{}\Varid{insert}\;\Varid{a}\;\Varid{bs}\mathrel{=}\Varid{when}\;(\Varid{length}\;\Varid{bs}\mathbin{>}\mathrm{0})\mathbin{\$}\Conid{Linearly}.\mathbf{do}{}\<[E]%
\\
\>[B]{}\hsindent{5}{}\<[5]%
\>[5]{}\Conid{Ur}\;\Varid{b}\leftarrow \Varid{read}\;\mathrm{0}\;\Varid{bs}{}\<[E]%
\\
\>[B]{}\hsindent{5}{}\<[5]%
\>[5]{}\mathbf{if}\;\Varid{a}\leq \Varid{b}\;\mathbf{then}{}\<[E]%
\\
\>[5]{}\hsindent{2}{}\<[7]%
\>[7]{}\Varid{\Conid{Linearly}.return}\;(\Conid{Ur}\;\Varid{a}){}\<[E]%
\\
\>[B]{}\hsindent{5}{}\<[5]%
\>[5]{}\mathbf{else}\;\Conid{Linearly}.\mathbf{do}{}\<[E]%
\\
\>[5]{}\hsindent{2}{}\<[7]%
\>[7]{}(\Conid{Ur}\;\Varid{tail},\Varid{release})\leftarrow \Varid{restrict}\;\Varid{bs}\;\mathrm{1}\;(\Varid{len}\mathbin{-}\mathrm{1}){}\<[E]%
\\
\>[5]{}\hsindent{2}{}\<[7]%
\>[7]{}\Conid{Ur}\;\Varid{b'}\leftarrow \Varid{insert}\;\Varid{a}\;\Varid{tail}{}\<[E]%
\\
\>[5]{}\hsindent{2}{}\<[7]%
\>[7]{}\Varid{release}{}\<[E]%
\\
\>[5]{}\hsindent{2}{}\<[7]%
\>[7]{}\Varid{write}\;\Varid{bs}\;\mathrm{0}\;\Varid{b'}{}\<[E]%
\\
\>[5]{}\hsindent{2}{}\<[7]%
\>[7]{}\Varid{\Conid{Linearly}.return}\;(\Conid{Ur}\;\Varid{b}){}\<[E]%
\ColumnHook
\end{hscode}\resethooks

The release operator returned by \ensuremath{\Varid{restrict}} is linear, which means it must be called
(see \cref{sec:unrestricted-release} for alternative designs).
When an array has a linear release operator attached, it can't be freed (with
the \ensuremath{\Varid{free}} function of \cref{fig:constraints-interface}, for instance), because
this would violate linearity: the release operator would be left unconsumed, and the type system
would reject the program.

When a value has such a linear release operator we say that the value \emph{has
  a lifetime}. When a value doesn't have a lifetime, we say it is \emph{owned}.
In addition, we refer to any function which returns one or more
values with lifetimes as a \emph{borrow}.

\subsubsection{Independent slices}
\label{sec:slicing-two}

A more advanced form of borrowing slices of arrays is to borrow two independent
slices from the same array simultaneously. As long as two slices do not overlap,
mutating one does not affect the other, and it is therefore safe to mutate them
independently.
The ability to work on disjoint array slices is crucial for certain algorithms,
such as the \ensuremath{\Varid{merge}} function in a merge sort.

Borrowing two independent slices is achieved using \ensuremath{\Varid{slice}}, which slices
an array at a given index and returns two arrays:
\begin{hscode}\SaveRestoreHook
\column{B}{@{}>{\hspre}l<{\hspost}@{}}%
\column{11}{@{}>{\hspre}l<{\hspost}@{}}%
\column{E}{@{}>{\hspre}l<{\hspost}@{}}%
\>[B]{}\Varid{slice}\mathbin{::}{}\<[11]%
\>[11]{}\constraintfont{\Conid{RW}\;\Varid{n}}\Lolly \Conid{UArray}\;\Varid{a}\;\Varid{n}\to \Conid{Int}\to {}\<[E]%
\\
\>[11]{}\exists\;\Varid{p}\;\Varid{q}.(\Conid{Ur}\;(\Conid{UArray}\;\Varid{a}\;\Varid{p},\Conid{UArray}\;\Varid{a}\;\Varid{q}),(\constraintfont{\Conid{RW}\;\Varid{p},\Conid{RW}\;\Varid{q}}\Lolly ()\RLolly\constraintfont{\Conid{RW}\;\Varid{n}}))\RLolly\constraintfont{(\Conid{RW}\;\Varid{p},\Conid{RW}\;\Varid{q})}{}\<[E]%
\ColumnHook
\end{hscode}\resethooks
The result of \ensuremath{\Varid{slice}\;\Varid{as}\;\Varid{n}} is two slices, one with the first \ensuremath{\Varid{n}} elements of \ensuremath{\Varid{as}},
and one with the remaining \ensuremath{\Varid{length}\;\Varid{as}\mathbin{-}\Varid{n}} elements. The two slices can be used
simultaneously. Crucially, both slices must be released for the original
array to be accessible again. This is represented as a single release operator
which consumes the capabilities for both slices.

The \ensuremath{\Varid{merge}} function takes two sorted arrays \ensuremath{\Varid{ls}} and \ensuremath{\Varid{rs}}, and mutates them so that the
concatenation of \ensuremath{\Varid{ls}} and \ensuremath{\Varid{rs}} becomes sorted as a whole.%
\footnote{The implementation of \ensuremath{\Varid{merge}} is chosen for simplicity over
efficiency for the sake of this exposition. Inserting elements into \ensuremath{\Varid{rs}} with
the \ensuremath{\Varid{insert}} function of \cref{sec:simple-slices} is convenient but not
particularly efficient.
In fact, in the current implementation, \ensuremath{\Varid{insert}} is not truly in-place:
while is uses $O(1)$ heap space,  it actually uses $O(n)$ stack
space, since it is not tail recursive. The \ensuremath{\Varid{merge}} may be compiled to a tail-recursive
function if the compiler is able to optimise away the call to \ensuremath{\Varid{release}}.}
It is crucial that the arrays \ensuremath{\Varid{ls}} and \ensuremath{\Varid{rs}} be disjoint, if they are not then
referential transparency would be lost and \ensuremath{\Varid{merge}} would not be a pure function
(for instance reordering \ensuremath{\Varid{write}\;\Varid{ls}\;\mathrm{0}\;\Varid{r}} and \ensuremath{\Varid{insert}\;\Varid{l}\;\Varid{rs}}, which is a legal
transformation, would change the result).

\begin{hscode}\SaveRestoreHook
\column{B}{@{}>{\hspre}l<{\hspost}@{}}%
\column{5}{@{}>{\hspre}l<{\hspost}@{}}%
\column{15}{@{}>{\hspre}l<{\hspost}@{}}%
\column{52}{@{}>{\hspre}l<{\hspost}@{}}%
\column{E}{@{}>{\hspre}l<{\hspost}@{}}%
\>[B]{}\Varid{mergeSort}\mathbin{::}{}\<[15]%
\>[15]{}\constraintfont{\Conid{RW}\;\Varid{n}}\Lolly \Conid{UArray}\;\Varid{a}\;\Varid{n}\to ()\RLolly\constraintfont{\Conid{RW}\;\Varid{n}}{}\<[E]%
\\
\>[B]{}\Varid{mergeSort}\;\Varid{as}\mathrel{=}\Varid{when}\;(\Varid{length}\;\Varid{as}\mathbin{>}\mathrm{1})\mathbin{\$}\Conid{Linearly}.\mathbf{do}\;{}\<[52]%
\>[52]{}\mbox{\onelinecomment  \ensuremath{\constraintfont{(\Conid{RW}\;\Varid{n})}}}{}\<[E]%
\\
\>[B]{}\hsindent{5}{}\<[5]%
\>[5]{}(\Conid{Ur}\;(\Varid{ls},\Varid{rs}),\Varid{release})\leftarrow \Varid{slice}\;\Varid{as}\;(\Varid{len}\mathbin{/}\mathrm{2}){}\<[52]%
\>[52]{}\mbox{\onelinecomment  \ensuremath{\constraintfont{(\Conid{RW}\;\Varid{p},\Conid{RW}\;\Varid{q})}}}{}\<[E]%
\\
\>[B]{}\hsindent{5}{}\<[5]%
\>[5]{}\Varid{mergeSort}\;\Varid{ls}{}\<[52]%
\>[52]{}\mbox{\onelinecomment  \ensuremath{\constraintfont{(\Conid{RW}\;\Varid{p},\Conid{RW}\;\Varid{q})}}}{}\<[E]%
\\
\>[B]{}\hsindent{5}{}\<[5]%
\>[5]{}\Varid{mergeSort}\;\Varid{rs}{}\<[52]%
\>[52]{}\mbox{\onelinecomment  \ensuremath{\constraintfont{(\Conid{RW}\;\Varid{p},\Conid{RW}\;\Varid{q})}}}{}\<[E]%
\\
\>[B]{}\hsindent{5}{}\<[5]%
\>[5]{}\Varid{merge}\;\Varid{ls}\;\Varid{rs}{}\<[52]%
\>[52]{}\mbox{\onelinecomment  \ensuremath{\constraintfont{(\Conid{RW}\;\Varid{p},\Conid{RW}\;\Varid{q})}}}{}\<[E]%
\\
\>[B]{}\hsindent{5}{}\<[5]%
\>[5]{}\Varid{release}{}\<[52]%
\>[52]{}\mbox{\onelinecomment  \ensuremath{\constraintfont{(\Conid{RW}\;\Varid{n})}}}{}\<[E]%
\\
\>[B]{}\hsindent{5}{}\<[5]%
\>[5]{}\Varid{\Conid{Linear}.return}\;(){}\<[52]%
\>[52]{}\mbox{\onelinecomment  \ensuremath{\constraintfont{(\Conid{RW}\;\Varid{n})}}}{}\<[E]%
\ColumnHook
\end{hscode}\resethooks
As comments, we write the linear constraints which are in context
after the call on their left. Most notably we can see the effect of
the \ensuremath{\Varid{slice}} function (splitting the \ensuremath{\constraintfont{\Conid{RW}\;\Varid{n}}} into two), and
the reverse effect of the \ensuremath{\Varid{release}} operator. Also notable is that
other operations keep linear constraints intact.

\begin{hscode}\SaveRestoreHook
\column{B}{@{}>{\hspre}l<{\hspost}@{}}%
\column{5}{@{}>{\hspre}l<{\hspost}@{}}%
\column{7}{@{}>{\hspre}l<{\hspost}@{}}%
\column{11}{@{}>{\hspre}l<{\hspost}@{}}%
\column{E}{@{}>{\hspre}l<{\hspost}@{}}%
\>[B]{}\Varid{merge}\mathbin{::}{}\<[11]%
\>[11]{}\constraintfont{(\Conid{RW}\;\Varid{p},\Conid{RW}\;\Varid{q})}\Lolly \Conid{UArray}\;\Varid{a}\;\Varid{p}\to \Conid{UArray}\;\Varid{a}\;\Varid{q}\to ()\RLolly\constraintfont{(\Conid{RW}\;\Varid{p},\Conid{RW}\;\Varid{q})}{}\<[E]%
\\
\>[B]{}\Varid{merge}\;\Varid{ls}\;\Varid{rs}\mathrel{=}\Varid{when}\;(\Varid{length}\;\Varid{ls}\mathbin{>}\mathrm{0}\mathrel{\wedge}\Varid{length}\;\Varid{rs}\mathbin{>}\mathrm{0})\mathbin{\$}\Conid{Linearly}.\mathbf{do}{}\<[E]%
\\
\>[B]{}\hsindent{5}{}\<[5]%
\>[5]{}\Conid{Ur}\;\Varid{l}\leftarrow \Varid{read}\;\mathrm{0}\;\Varid{ls}{}\<[E]%
\\
\>[B]{}\hsindent{5}{}\<[5]%
\>[5]{}\Conid{Ur}\;\Varid{r}\leftarrow \Varid{read}\;\mathrm{0}\;\Varid{rs}{}\<[E]%
\\
\>[B]{}\hsindent{5}{}\<[5]%
\>[5]{}\mathbf{if}\;\Varid{l}\leq \Varid{r}\;\mathbf{then}\;\Conid{Linearly}.\mathbf{do}{}\<[E]%
\\
\>[5]{}\hsindent{2}{}\<[7]%
\>[7]{}(\Conid{Ur}\;(\anonymous ,\Varid{ls'}),\Varid{release})\leftarrow \Varid{slice}\;\Varid{ls}\;\mathrm{1}{}\<[E]%
\\
\>[5]{}\hsindent{2}{}\<[7]%
\>[7]{}\Varid{merge}\;\Varid{ls'}\;\Varid{rs}{}\<[E]%
\\
\>[5]{}\hsindent{2}{}\<[7]%
\>[7]{}\Varid{release}{}\<[E]%
\\
\>[5]{}\hsindent{2}{}\<[7]%
\>[7]{}\Varid{\Conid{Linear}.return}\;(){}\<[E]%
\\
\>[B]{}\hsindent{5}{}\<[5]%
\>[5]{}\mathbf{else}\;\Conid{Linearly}.\mathbf{do}{}\<[E]%
\\
\>[5]{}\hsindent{2}{}\<[7]%
\>[7]{}\Varid{write}\;\Varid{ls}\;\mathrm{0}\;\Varid{r}{}\<[E]%
\\
\>[5]{}\hsindent{2}{}\<[7]%
\>[7]{}(\Conid{Ur}\;\Varid{ls'},\Varid{release})\leftarrow \Varid{slice}\;\mathrm{1}\;\Varid{ls}{}\<[E]%
\\
\>[5]{}\hsindent{2}{}\<[7]%
\>[7]{}\Conid{Ur}\;\Varid{r'}\leftarrow \Varid{insert}\;\Varid{l}\;\Varid{rs}{}\<[E]%
\\
\>[5]{}\hsindent{2}{}\<[7]%
\>[7]{}\Varid{write}\;\Varid{rs}\;\mathrm{0}\;\Varid{r'}{}\<[E]%
\\
\>[5]{}\hsindent{2}{}\<[7]%
\>[7]{}\Varid{merge}\;\Varid{ls'}\;\Varid{rs}{}\<[E]%
\\
\>[5]{}\hsindent{2}{}\<[7]%
\>[7]{}\Varid{release}{}\<[E]%
\\
\>[5]{}\hsindent{2}{}\<[7]%
\>[7]{}\Varid{\Conid{Linear}.return}\;(){}\<[E]%
\ColumnHook
\end{hscode}\resethooks
This code requires the two slices to be accessed simultaneously, and this is
achieved using the \ensuremath{\Varid{slice}} function in \ensuremath{\Varid{mergeSort}} which then goes
on to work with \ensuremath{\Varid{merge}}.

The \ensuremath{\Varid{slice}} function is both simple and general. Indeed, it can be used to
define the \ensuremath{\Varid{restrict}} function of \cref{sec:simple-slices}.
\begin{hscode}\SaveRestoreHook
\column{B}{@{}>{\hspre}l<{\hspost}@{}}%
\column{3}{@{}>{\hspre}l<{\hspost}@{}}%
\column{9}{@{}>{\hspre}l<{\hspost}@{}}%
\column{31}{@{}>{\hspre}l<{\hspost}@{}}%
\column{33}{@{}>{\hspre}l<{\hspost}@{}}%
\column{E}{@{}>{\hspre}l<{\hspost}@{}}%
\>[B]{}\Varid{restrict}\mathbin{::}{}\<[9]%
\>[9]{}\constraintfont{\Conid{RW}\;\Varid{n}}\Lolly \Conid{UArray}\;\Varid{a}\;\Varid{n}\to \Conid{Int}\to \Conid{Int}\to {}\<[E]%
\\
\>[9]{}\exists\;\Varid{p}.(\Conid{Ur}\;(\Conid{UArray}\;\Varid{a}\;\Varid{p}),(\constraintfont{\Conid{RW}\;\Varid{p}}\Lolly ()\RLolly\constraintfont{\Conid{RW}\;\Varid{n}}))\RLolly\constraintfont{\Conid{RW}\;\Varid{p}}{}\<[E]%
\\
\>[B]{}\Varid{restrict}\;\Varid{as}\;\Varid{start}\;\Varid{len}\mathrel{=}\Conid{Linearly}.\mathbf{do}{}\<[E]%
\\
\>[B]{}\hsindent{3}{}\<[3]%
\>[3]{}(\Conid{Ur}\;(\anonymous ,\Varid{rght}),\Varid{release\char95 r})\leftarrow {}\<[33]%
\>[33]{}\Varid{slice}\;\Varid{as}\;\Varid{start}{}\<[E]%
\\
\>[B]{}\hsindent{3}{}\<[3]%
\>[3]{}(\Conid{Ur}\;(\Varid{sb},\anonymous ),\Varid{release\char95 s})\leftarrow {}\<[31]%
\>[31]{}\Varid{slice}\;\Varid{rght}\;\Varid{len}{}\<[E]%
\\
\>[B]{}\hsindent{3}{}\<[3]%
\>[3]{}\Varid{\Conid{Linearly}.return}\;(\Conid{Ur}\;\Varid{sb},\Conid{Linearly}.\mathbf{do}\;\{\mskip1.5mu \Varid{release\char95 s};\Varid{release\char95 r}\mskip1.5mu\}){}\<[E]%
\ColumnHook
\end{hscode}\resethooks
Notice how the unneeded
slices are released, by partially discharging the linear constraint of the
release operators.

Conversely, The \ensuremath{\Varid{restrict}} function from cannot be used to define \ensuremath{\Varid{slice}} since it
cannot be used to borrow two independent slices: when a slice is borrowed,
access to the original array is lost until the slice is released and so a
second slice cannot exist simultaneously.

\subsection{Blockwise matrices: in-place Valiant algorithm}
\label{sec:valiant}

We have seen that linear constraints are an effective way of defining
in-place algorithms on arrays. A question that arises naturally is
whether they apply generally on matrices (and in general
multi-dimensional arrays). We answer positively, and illustrate with
in-place modification of matrices.  Specifically, we consider the
context-free parsing algorithm invented by
\citet{valiant_general_1975}. The original motivation was to show that
context-free parsing was sub-cubic, but the algorithm has also found
practical applications
\citep{bernardy_efficient_2013,bernardy_efficient_2015}.  Valiant's
algorithm is a divide-and conquer algorithm which works on matrices,
and we use it illustrate how such corresponding ownership patterns can
be encoded with our system.

We direct the curious reader to the cited works, but at its core,
Valiant's algorithm consists in reducing context-free parsing to the
computation of a (non-associative) transitive closure of a cycle-free
directed graph. Each node corresponds to a position in the string, and
edges are labeled by the set non-terminals generating the corresponding substring. The grammar
controls whether two paths can be combined into a longer path, depending on their labels.
In the closed graph, the eventual labels of the edge
between the start node and the end node give the parsing result.  The
graph and its closure are represented by adjacency matrices which are
both upper-triangular.  We will show how to update this matrix in-place,
building up closures, in a divide and conquer pattern.
By nature of the closure, the closure is built incrementally: one add labels to the entries of the matrix, but never remove them.

In fact, Valiant's algorithm consists in two divide-and-conquer recursive functions:
\begin{hscode}\SaveRestoreHook
\column{B}{@{}>{\hspre}l<{\hspost}@{}}%
\column{7}{@{}>{\hspre}l<{\hspost}@{}}%
\column{E}{@{}>{\hspre}l<{\hspost}@{}}%
\>[B]{}\Varid{v}\mathbin{::}\constraintfont{\Conid{RW}\;\Varid{n}}\Lolly \Conid{Matrix}\;\Varid{a}\;\Varid{n}\RLolly\constraintfont{\Conid{RW}\;\Varid{n}}{}\<[E]%
\\
\>[B]{}\Varid{w}\mathbin{::}{}\<[7]%
\>[7]{}\constraintfont{(\Conid{Read}\;\Varid{p},\Conid{RW}\;\Varid{n},\Conid{Read}\;\Varid{q})}\Lolly {}\<[E]%
\\
\>[7]{}\Conid{Matrix}\;\Varid{a}\;\Varid{p}\to \Conid{Matrix}\;\Varid{a}\;\Varid{n}\to \Conid{Matrix}\;\Varid{a}\;\Varid{q}\RLolly\constraintfont{(\Conid{Read}\;\Varid{p},\Conid{RW}\;\Varid{n},\Conid{Read}\;\Varid{q})}{}\<[E]%
\ColumnHook
\end{hscode}\resethooks
Here \ensuremath{\Varid{v}} takes a writeable matrix \ensuremath{\Varid{n}}. When done \ensuremath{\Varid{n}} will be transitively
closed. The \ensuremath{\Varid{w}} function is a
helper, which takes two fully closed sub-matrices \ensuremath{\Varid{p}} and \ensuremath{\Varid{q}}, and a
partial subgraph between the nodes of \ensuremath{\Varid{p}} and the nodes of \ensuremath{\Varid{q}}, represented as a rectangular matrix \ensuremath{\Varid{n}}. When done, \ensuremath{\Varid{n}} contains
all the possible labeled paths from \ensuremath{\Varid{p}} to \ensuremath{\Varid{q}}.

\providecommand{\upperRightTriangle}[4]{
  \begin{scope}[shift={(#1 , -#1)}]
    \pgfmathsetmacro{\triSize}{#2}
    \draw[#4] (0,0) -- (\triSize,0) -- (\triSize,-\triSize) -- cycle;
    \node at (2*\triSize/3, -\triSize/3) {#3};
  \end{scope}
}
\providecommand{\matrixRectangle}[5]{
  \begin{scope}[shift={#1}]
    \draw [#5] (0,0) rectangle (#2,-#3) ;
    \node at (#2 / 2,-#3 / 2) {#4};
  \end{scope}
}

\begin{figure}
  \centering
  \begin{tikzpicture}
    \def\extraSpace{0.25}
    \def\asize{1}
    \def\bsize{1.1}
    \def\qsize{0.6}
    \pgfmathsetmacro{\psize}{\asize + \bsize}
    \pgfmathsetmacro{\bpos}{\asize}
    \pgfmathsetmacro{\qpos}{\psize + \extraSpace}
    \pgfmathsetmacro{\totsize}{\psize + \qsize}
    \upperRightTriangle{0}{\asize}{a}{white}
    \upperRightTriangle{\bpos}{\bsize}{b}{white}
    \upperRightTriangle{0}{\psize}{}{}
    \upperRightTriangle{\qpos}{\qsize}{q}{}
    \matrixRectangle{(\qpos,0)}{\qsize}{\psize}{}{}
    \matrixRectangle{(\bpos,0)}{\bsize}{\asize}{x}{dotted}
    \matrixRectangle{(\qpos,0)}{\qsize}{\asize}{y}{dotted}
    \matrixRectangle{(\qpos,-\bpos)}{\qsize}{\bsize}{z}{dotted}
  \end{tikzpicture}
  \caption{Division of the inputs of the \ensuremath{\Varid{w}} function when \ensuremath{\Varid{p}} is larger than \ensuremath{\Varid{q}}. That is, \ensuremath{\Varid{p}\mathrel{=}(\Varid{a},\Varid{x},\Varid{b})} and \ensuremath{\Varid{n}\mathrel{=}(\Varid{y},\Varid{z})}.}
  \label{fig:valiant-w}
\end{figure}
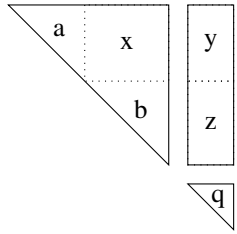

If its input is small enough, then there is nothing to do for \ensuremath{\Varid{v}}, but
otherwise it splits \ensuremath{\Varid{n}} into two new triangular matrices \ensuremath{\Varid{a}} and
\ensuremath{\Varid{b}} and the top-right matrix \ensuremath{\Varid{x}}. (The bottom left is always empty.)
The read-write ownership is transferred to each of the components.
After recursively computing the closure of \ensuremath{\Varid{a}} and \ensuremath{\Varid{b}}, the matrix is
completed using the helper function (\ensuremath{\Varid{w}\;\Varid{a}\;\Varid{x}\;\Varid{b}}).
\begin{hscode}\SaveRestoreHook
\column{B}{@{}>{\hspre}l<{\hspost}@{}}%
\column{3}{@{}>{\hspre}l<{\hspost}@{}}%
\column{5}{@{}>{\hspre}l<{\hspost}@{}}%
\column{E}{@{}>{\hspre}l<{\hspost}@{}}%
\>[B]{}\Varid{v}\;\Varid{n}\mathrel{=}{}\<[E]%
\\
\>[B]{}\hsindent{3}{}\<[3]%
\>[3]{}\mathbf{if}\;\Varid{size}\;\Varid{n}\leq \mathrm{2}{}\<[E]%
\\
\>[3]{}\hsindent{2}{}\<[5]%
\>[5]{}\Varid{return}\;(){}\<[E]%
\\
\>[B]{}\hsindent{3}{}\<[3]%
\>[3]{}\mathbf{else}\;\Conid{Linearly}.\mathbf{do}{}\<[E]%
\\
\>[3]{}\hsindent{2}{}\<[5]%
\>[5]{}\mathbf{let}\;\Varid{i}\mathrel{=}\Varid{size}\;\Varid{n}\mathbin{\Varid{`div`}}\mathrm{2}{}\<[E]%
\\
\>[3]{}\hsindent{2}{}\<[5]%
\>[5]{}(\Conid{Ur}\;(\Varid{a},\Varid{x},\Varid{b}),\Varid{release\char95 n})\leftarrow \Varid{splitUpperMatrix}\;\Varid{n}\;\Varid{i}{}\<[E]%
\\
\>[3]{}\hsindent{2}{}\<[5]%
\>[5]{}\Varid{v}\;\Varid{a}{}\<[E]%
\\
\>[3]{}\hsindent{2}{}\<[5]%
\>[5]{}\Varid{v}\;\Varid{b}{}\<[E]%
\\
\>[3]{}\hsindent{2}{}\<[5]%
\>[5]{}\Varid{w}\;\Varid{a}\;\Varid{x}\;\Varid{b}{}\<[E]%
\\
\>[3]{}\hsindent{2}{}\<[5]%
\>[5]{}\Varid{release\char95 n}{}\<[E]%
\ColumnHook
\end{hscode}\resethooks
The splitting upper triangular matrices is prodived by the function:
\begin{hscode}\SaveRestoreHook
\column{B}{@{}>{\hspre}l<{\hspost}@{}}%
\column{3}{@{}>{\hspre}l<{\hspost}@{}}%
\column{22}{@{}>{\hspre}l<{\hspost}@{}}%
\column{E}{@{}>{\hspre}l<{\hspost}@{}}%
\>[B]{}\Varid{splitUpperMatrix}\mathbin{::}{}\<[22]%
\>[22]{}\constraintfont{\Conid{RW}\;\Varid{n}}\Lolly \Conid{Matrix}\;\Varid{a}\;\Varid{n}\to \Conid{Int}\to {}\<[E]%
\\
\>[B]{}\hsindent{3}{}\<[3]%
\>[3]{}\exists\;\Varid{m}\;\Varid{l}\;\Varid{r}.(\Conid{Ur}\;(\Conid{Matrix}\;\Varid{a}\;\Varid{l},\Conid{Matrix}\;\Varid{a}\;\Varid{m},\Conid{Matrix}\;\Varid{a}\;\Varid{r}),\constraintfont{(\Conid{RW}\;\Varid{m},\Conid{RW}\;\Varid{l},\Conid{RW}\;\Varid{r})\Lolly ()\RLolly\Conid{RW}\;\Varid{n}}){}\<[E]%
\\
\>[B]{}\hsindent{3}{}\<[3]%
\>[3]{}\RLolly\constraintfont{(\Conid{RW}\;\Varid{m},\Conid{RW}\;\Varid{l},\Conid{RW}\;\Varid{r})}{}\<[E]%
\ColumnHook
\end{hscode}\resethooks
The helper \ensuremath{\Varid{w}\;\Varid{p}\;\Varid{n}\;\Varid{q}} is somewhat more involved. It won't touch the
completed matrices \ensuremath{\Varid{p}} and \ensuremath{\Varid{q}}, as declared in its signature. If both
of the upper-triangular matrices \ensuremath{\Varid{p}} and \ensuremath{\Varid{q}} are one-by-one, then
we're done. Otherwise, we split the larger of the two (say \ensuremath{\Varid{p}}) into
three parts, and \ensuremath{\Varid{n}} is split along the corresponding dimension, at
the same position, into \ensuremath{\Varid{y}} and \ensuremath{\Varid{z}}. The matrix \ensuremath{\Varid{z}} is closed by a recursive
call. Then the multiplication of \ensuremath{\Varid{x}} and \ensuremath{\Varid{z}} is performed (the
multiplication of elements, non-terminals, determines the grammar) and
added (using set union) to \ensuremath{\Varid{y}}. Finally \ensuremath{\Varid{z}} is recursively closed. The case of splitting \ensuremath{\Varid{q}} is symmetric.
\begin{hscode}\SaveRestoreHook
\column{B}{@{}>{\hspre}l<{\hspost}@{}}%
\column{3}{@{}>{\hspre}l<{\hspost}@{}}%
\column{5}{@{}>{\hspre}l<{\hspost}@{}}%
\column{7}{@{}>{\hspre}l<{\hspost}@{}}%
\column{E}{@{}>{\hspre}l<{\hspost}@{}}%
\>[B]{}\Varid{w}\;\Varid{p}\;\Varid{n}\;\Varid{q}\mathrel{=}{}\<[E]%
\\
\>[B]{}\hsindent{3}{}\<[3]%
\>[3]{}\mathbf{if}\;\Varid{size}\;\Varid{p}\leq \mathrm{1}\mathrel{\wedge}\Varid{size}\;\Varid{q}\leq \mathrm{1}{}\<[E]%
\\
\>[B]{}\hsindent{3}{}\<[3]%
\>[3]{}\mathbf{then}\;\Varid{return}\;(){}\<[E]%
\\
\>[B]{}\hsindent{3}{}\<[3]%
\>[3]{}\mathbf{else}\;\mathbf{if}\;\Varid{size}\;\Varid{p}\mathbin{>}\Varid{size}\;\Varid{q}{}\<[E]%
\\
\>[3]{}\hsindent{2}{}\<[5]%
\>[5]{}\mathbf{then}\;\Conid{Linearly}.\mathbf{do}{}\<[E]%
\\
\>[5]{}\hsindent{2}{}\<[7]%
\>[7]{}\mathbf{let}\;\Varid{i}\mathrel{=}\Varid{size}\;\Varid{p}\mathbin{\Varid{`div`}}\mathrm{2}{}\<[E]%
\\
\>[5]{}\hsindent{2}{}\<[7]%
\>[7]{}(\Conid{Ur}\;(\Varid{a},\Varid{x},\Varid{b}),\Varid{release}_{\Varid{p}})\leftarrow \Varid{splitUpperMatrix}\;\Varid{p}\;\Varid{i}{}\<[E]%
\\
\>[5]{}\hsindent{2}{}\<[7]%
\>[7]{}(\Conid{Ur}\;(\Varid{y},\Varid{z}),\Varid{release}_{\Varid{n}})\leftarrow \Varid{sliceMatrixV}\;\Varid{n}\;\Varid{i}{}\<[E]%
\\
\>[5]{}\hsindent{2}{}\<[7]%
\>[7]{}\Varid{w}\;\Varid{b}\;\Varid{z}\;\Varid{q}{}\<[E]%
\\
\>[5]{}\hsindent{2}{}\<[7]%
\>[7]{}\Varid{multiplyMatricesInto}\;\Varid{y}\;\Varid{x}\;\Varid{z}{}\<[E]%
\\
\>[5]{}\hsindent{2}{}\<[7]%
\>[7]{}\Varid{w}\;\Varid{a}\;\Varid{y}\;\Varid{q}{}\<[E]%
\\
\>[5]{}\hsindent{2}{}\<[7]%
\>[7]{}\Varid{release}_{\Varid{p}};\Varid{release}_{\Varid{n}}{}\<[E]%
\\
\>[3]{}\hsindent{2}{}\<[5]%
\>[5]{}\mathbf{else}\;\Conid{Linearly}.\mathbf{do}{}\<[E]%
\\
\>[5]{}\hsindent{2}{}\<[7]%
\>[7]{}\mathbf{let}\;\Varid{i}\mathrel{=}\Varid{size}\;\Varid{q}\mathbin{\Varid{`div`}}\mathrm{2}{}\<[E]%
\\
\>[5]{}\hsindent{2}{}\<[7]%
\>[7]{}(\Conid{Ur}\;(\Varid{a},\Varid{z},\Varid{b}),\Varid{release}_{\Varid{q}})\leftarrow \Varid{splitUpperMatrix}\;\Varid{q}\;\Varid{i}{}\<[E]%
\\
\>[5]{}\hsindent{2}{}\<[7]%
\>[7]{}(\Conid{Ur}\;(\Varid{x},\Varid{y}),\Varid{release}_{\Varid{n}})\leftarrow \Varid{sliceMatrixH}\;\Varid{n}\;\Varid{i}{}\<[E]%
\\
\>[5]{}\hsindent{2}{}\<[7]%
\>[7]{}\Varid{w}\;\Varid{p}\;\Varid{x}\;\Varid{a}{}\<[E]%
\\
\>[5]{}\hsindent{2}{}\<[7]%
\>[7]{}\Varid{multiplyMatricesInto}\;\Varid{y}\;\Varid{x}\;\Varid{z}{}\<[E]%
\\
\>[5]{}\hsindent{2}{}\<[7]%
\>[7]{}\Varid{w}\;\Varid{p}\;\Varid{y}\;\Varid{b}{}\<[E]%
\\
\>[5]{}\hsindent{2}{}\<[7]%
\>[7]{}\Varid{release}_{\Varid{q}};\Varid{release}_{\Varid{n}}{}\<[E]%
\ColumnHook
\end{hscode}\resethooks

The splitting and joining of general rectangular matrices follows the usual pattern (and in fact, can be used to implement such operations on triangular matrices).
\begin{hscode}\SaveRestoreHook
\column{B}{@{}>{\hspre}l<{\hspost}@{}}%
\column{3}{@{}>{\hspre}l<{\hspost}@{}}%
\column{5}{@{}>{\hspre}l<{\hspost}@{}}%
\column{E}{@{}>{\hspre}l<{\hspost}@{}}%
\>[3]{}\Varid{sliceMatrixH}\mathbin{::}\constraintfont{\Conid{RW}\;\Varid{n}}\Lolly \Conid{Matrix}\;\Varid{a}\;\Varid{n}\to \Conid{Int}\to {}\<[E]%
\\
\>[3]{}\hsindent{2}{}\<[5]%
\>[5]{}\exists\;\Varid{l}\;\Varid{r}.(\Conid{Ur}\;(\Conid{Matrix}\;\Varid{a}\;\Varid{l},\Conid{Matrix}\;\Varid{a}\;\Varid{r}),\constraintfont{\Conid{RW}\;\Varid{l},\Conid{RW}\;\Varid{r}}\Lolly ()\RLolly\constraintfont{\Conid{RW}\;\Varid{n}})\RLolly\constraintfont{(\Conid{RW}\;\Varid{l},\Conid{RW}\;\Varid{r})}{}\<[E]%
\\
\>[3]{}\Varid{sliceMatrixV}\mathbin{::}\constraintfont{\Conid{RW}\;\Varid{n}}\Lolly \Conid{Matrix}\;\Varid{a}\;\Varid{n}\to \Conid{Int}\to {}\<[E]%
\\
\>[3]{}\hsindent{2}{}\<[5]%
\>[5]{}\exists\;\Varid{l}\;\Varid{r}.(\Conid{Ur}\;(\Conid{Matrix}\;\Varid{a}\;\Varid{l},\Conid{Matrix}\;\Varid{a}\;\Varid{r}),\constraintfont{\Conid{RW}\;\Varid{l},\Conid{RW}\;\Varid{r}}\Lolly ()\RLolly\constraintfont{\Conid{RW}\;\Varid{n}})\RLolly\constraintfont{(\Conid{RW}\;\Varid{l},\Conid{RW}\;\Varid{r})}{}\<[E]%
\ColumnHook
\end{hscode}\resethooks
We won't detail the underlying matrix multiplication; but its signature is the following.\begin{hscode}\SaveRestoreHook
\column{B}{@{}>{\hspre}l<{\hspost}@{}}%
\column{3}{@{}>{\hspre}l<{\hspost}@{}}%
\column{5}{@{}>{\hspre}l<{\hspost}@{}}%
\column{E}{@{}>{\hspre}l<{\hspost}@{}}%
\>[3]{}\Varid{multiplyMatricesInto}\mathbin{::}\constraintfont{(\Conid{RW}\;\Varid{x},\Conid{Read}\;\Varid{y},\Conid{Read}\;\Varid{z})}\Lolly {}\<[E]%
\\
\>[3]{}\hsindent{2}{}\<[5]%
\>[5]{}\Conid{Matrix}\;\Conid{NTSet}\;\Varid{x}\to \Conid{Matrix}\;\Conid{NTSet}\;\Varid{y}\to \Conid{Matrix}\;\Conid{NTSet}\;\Varid{z}\RLolly\constraintfont{\Conid{RW}\;\Varid{x},\Conid{Read}\;\Varid{y},\Conid{Read}\;\Varid{z}}{}\<[E]%
\ColumnHook
\end{hscode}\resethooks
The structure of the algorithm follows exactly that outlined by
\citet{valiant_general_1975}. The overhead is the
explicitly calling of join functions to re-construct ownership
of the analysed matrices. In return, one gets additional static
checks: we do not demand full read-write capabilities on complete
matrices. Besides, even though the algorithm is imperative, it can be
called from pure functions, where the resulting matrix would be
\emph{frozen}, see \cref{sec:o1-freeze}.

\subsection{Matrices from first principles: a theory of nested arrays}
\label{sec:nested-arrays}

In \cref{sec:borrowing-slices}, we split arrays in two at a position, in
\cref{sec:valiant} we split matrices either at a vertical position or a
horizontal position. Higher-dimensional tensors could, in general, be split
along any dimension. The good news is that we don't have to define a new
interface from unsafe primitives for each dimension.
Capabilities-as-linear-constraints are flexible enough that we can build a
single interface which covers all dimensions.

\paragraph*{Arrays with nesting}

Specifically, we're going to represent matrices as arrays of arrays (say an
array of columns). Higher-dimensional tensors are a straightforward generalisation. To whet the
reader's appetite, slicing a matrix in two sets of columns is, rather
straightforward: just consider two halves of the root array; but slicing a
matrix in two sets of rows is much less obvious as it requires slicing each
column in two.

Let us first notice that we cannot use the \ensuremath{\Conid{UArray}} type of
\cref{sec:borrowing-slices}. Indeed, nesting \ensuremath{\Conid{UArray}} would give us a type like
\ensuremath{\Conid{UArray}\;(\Conid{UArray}\;\Varid{a}\;\Varid{p})\;\Varid{n}}. But this is not what we want: presumably \ensuremath{\Varid{p}} should be
different for each array. Nor can we use something like \ensuremath{\Conid{UArray}\;(\exists\;\Varid{p}.\Conid{UArray}\;\Varid{a}\;\Varid{p}\RLolly\Conid{RW}\;\Varid{p})\;\Varid{n}}, since \ensuremath{\Conid{RW}\;\Varid{p}} is a linear constraint, and an \ensuremath{\Conid{UArray}}
can only store linear values.

So we need to introduce a new type \ensuremath{\Conid{NArray}} (the \ensuremath{\Conid{N}} stands for ``nested'').
What we want to express, with \ensuremath{\Conid{NArray}} is that all the arrays have the same
owner, so that when we write an array \ensuremath{\Varid{bs}} in a cell of an array \ensuremath{\Varid{as}}, we
relinquish the ownership of \ensuremath{\Varid{bs}} to that of \ensuremath{\Varid{as}}. We do that by abstracting over
the type-level name of \ensuremath{\Varid{bs}}.


\begin{hscode}\SaveRestoreHook
\column{B}{@{}>{\hspre}l<{\hspost}@{}}%
\column{E}{@{}>{\hspre}l<{\hspost}@{}}%
\>[B]{}\mathbf{type}\;\Conid{NArray}\mathbin{::}(\Conid{Reg}\to \Conid{Type})\to \Conid{Reg}\to \Conid{Type}{}\<[E]%
\ColumnHook
\end{hscode}\resethooks
Notice that it is convenient for the region parameter to be the last
parameter, in order for the types to align when nesting arrays.

In consequence, we cannot have \ensuremath{\Conid{NArray}\;\Conid{Int}}: we must wrap atomic
values in a type of kind \ensuremath{\Conid{Reg}\to \Conid{Type}}. We introduce the
type constructor \ensuremath{\Conid{Val}} to that effect.

\begin{hscode}\SaveRestoreHook
\column{B}{@{}>{\hspre}l<{\hspost}@{}}%
\column{E}{@{}>{\hspre}l<{\hspost}@{}}%
\>[B]{}\mathbf{type}\;\Conid{Val}\mathbin{::}\Conid{Type}\to \Conid{Reg}\to \Conid{Type}{}\<[E]%
\\
\>[B]{}\mathbf{newtype}\;\Conid{Val}\;\Varid{a}\;\Varid{n}\mathrel{=}\Conid{MkVal}\;\Varid{a}{}\<[E]%
\\
\>[B]{}\mathbf{instance}\;\constraintfont{\Conid{R}\;(\Conid{Val}\;\Varid{a}\;\Varid{n})}{}\<[E]%
\\
\>[B]{}\mathbf{instance}\;\constraintfont{\Conid{W}\;(\Conid{Val}\;\Varid{a}\;\Varid{n})}{}\<[E]%
\ColumnHook
\end{hscode}\resethooks

Since \ensuremath{\Conid{Val}} isn't an abstract type, it is free to ignore its capabilities (which
is fine since no mutation can occur in a \ensuremath{\Conid{Val}\;\Varid{a}\;\Varid{r}}). We provide unrestricted
read and write access to \ensuremath{\Conid{Val}} since neither do anything, but they'll be needed
to interface with the array functions\info{This means that we can have both
  unrestricted and linear capabilities for \ensuremath{\Conid{Val}} in scope at the same time (the
  toplevel ones, and the one given from the borrow). So this design actually
  uses the new solver.}. With this, a matrix can have type:
\begin{hscode}\SaveRestoreHook
\column{B}{@{}>{\hspre}l<{\hspost}@{}}%
\column{E}{@{}>{\hspre}l<{\hspost}@{}}%
\>[B]{}\mathbf{type}\;\Varid Matrix \ensuremath{_0}\;\Varid{a}\;\Varid{n}\mathrel{=}\Conid{NArray}\;(\Conid{NArray}\;(\Conid{Val}\;\Varid{a}))\;\Varid{n}{}\<[E]%
\ColumnHook
\end{hscode}\resethooks
This isn't quite the type that we use to represent the matrices of
\cref{sec:valiant}, because in this section it is convenient to have a
separate type for full arrays (\ensuremath{\Conid{NArray}}) and slices (\ensuremath{\Conid{Slice}} below).

\begin{hscode}\SaveRestoreHook
\column{B}{@{}>{\hspre}l<{\hspost}@{}}%
\column{12}{@{}>{\hspre}c<{\hspost}@{}}%
\column{12E}{@{}l@{}}%
\column{16}{@{}>{\hspre}l<{\hspost}@{}}%
\column{E}{@{}>{\hspre}l<{\hspost}@{}}%
\>[B]{}\Varid{newNArray}{}\<[12]%
\>[12]{}\mathbin{::}{}\<[12E]%
\>[16]{}\constraintfont{\Conid{Linearly}}\Lolly {}\<[E]%
\\
\>[16]{}\Conid{Int}\to (\constraintfont{\Conid{Linearly}}\Lolly \Conid{Int}\to \exists\;\Varid{p}.\Varid{a}\;\Varid{p}\RLolly\constraintfont{\Conid{RW}\;\Varid{p}})\to \exists\;\Varid{n}.\Conid{NArray}\;\Varid{a}\;\Varid{n}\RLolly\constraintfont{\Conid{RW}\;\Varid{n}}{}\<[E]%
\\
\>[B]{}\Varid{borrowNA}{}\<[12]%
\>[12]{}\mathbin{::}{}\<[12E]%
\>[16]{}\constraintfont{\Conid{RW}\;\Varid{n}}\Lolly \Conid{NArray}\;\Varid{a}\;\Varid{n}\to \Conid{Int}\to \exists\;\Varid{p}.(\Conid{Ur}\;(\Varid{a}\;\Varid{p}),(\constraintfont{\Conid{RW}\;\Varid{p}}\Lolly ()\RLolly\constraintfont{\Conid{RW}\;\Varid{n}}))\RLolly\constraintfont{\Conid{RW}\;\Varid{p}}{}\<[E]%
\\
\>[B]{}\Varid{borrowNAR}{}\<[12]%
\>[12]{}\mathbin{::}{}\<[12E]%
\>[16]{}\constraintfont{\Conid{R}\;\Varid{n}}\Lolly \Conid{NArray}\;\Varid{a}\;\Varid{n}\to \Conid{Int}\to \exists\;\Varid{p}.(\Conid{Ur}\;(\Varid{a}\;\Varid{p}),(\constraintfont{\Conid{R}\;\Varid{p}}\Lolly ()\RLolly(\constraintfont{\Conid{R}\;\Varid{n}})))\RLolly\constraintfont{\Conid{R}\;\Varid{p}}{}\<[E]%
\\
\>[B]{}\Varid{writeNA}{}\<[12]%
\>[12]{}\mathbin{::}{}\<[12E]%
\>[16]{}\constraintfont{(\Conid{RW}\;\Varid{n},\Conid{RW}\;\Varid{p})}\Lolly \Conid{NArray}\;\Varid{a}\;\Varid{n}\to \Conid{Int}\to \Varid{a}\;\Varid{p}\to ()\RLolly\constraintfont{\Conid{RW}\;\Varid{n}}{}\<[E]%
\ColumnHook
\end{hscode}\resethooks
\info{Maybe something to be aware of here is that since \ensuremath{\Varid{a}\;\Varid{p}} is taken with
  multiplicity $  \multiplicityfont{ \omega }  $, any \ensuremath{\Conid{Val}\;\Varid{a}\;\Varid{p}} written in the array is
  unrestricted, hence points to an unrestricted \ensuremath{\Varid{a}}. So we have nested arrays,
  here, but they are ultimately arrays of unrestricted values. Which is why the
  types of the \ensuremath{\Varid{write}} functions work.}

To write a value \ensuremath{\Varid{u}} into an array \ensuremath{\Varid{a}}, one needs read and write
capabilities to \ensuremath{\Varid{u}}, because these capabilities will be given to
whomever already has them for \ensuremath{\Varid{a}}. For the same
reason we need to produce read and write capabilities to the elements we
initialise the array with in \ensuremath{\Varid{newNArray}}.

Accessing inner arrays, both for reading and writing, is more interesting: it
introduces a new kind of borrow, in the form of the functions \ensuremath{\Varid{borrowNA}} and
\ensuremath{\Varid{borrowNAR}}. We don't borrow slices here, we borrow inner values.
If you ignore
all the capabilities, these borrows have the type signature we'd expect of a
read function. They're borrows because they need to provide capabilities for the
inner value, and a release operator. The fresh capabilities and the release
operator prevent aliasing mutable cells.

\paragraph*{Slice type}
\label{sec:slice-type}

We introduce a separate type for slices. The intuition is that
is that a block of a matrix is a \ensuremath{\Conid{Slice}\;(\Conid{Slice}\;(\Conid{Val}\;\Varid{a}))\;\Varid{n}}. The inner \ensuremath{\Conid{Slice}\;(\Conid{Val}\;\Varid{a})} (columns, say) has a \ensuremath{\Conid{RW}} capability so that we can write its scalar,
but it can't itself be overwritten, as the slice itself isn't stored in the row.
What's actually stored in the row in the full column of the full matrix, and
overwriting it would modify the matrix outside of the current block. This would
be unsound. So \ensuremath{\Conid{Slice}} is a marker to mean that the cells of the arrays can be
mutated but not the array itself.

\begin{hscode}\SaveRestoreHook
\column{B}{@{}>{\hspre}l<{\hspost}@{}}%
\column{E}{@{}>{\hspre}l<{\hspost}@{}}%
\>[B]{}\mathbf{type}\;\Conid{Slice}\mathbin{::}(\Conid{Reg}\to \Conid{Type})\to \Conid{Reg}\to \Conid{Type}{}\<[E]%
\ColumnHook
\end{hscode}\resethooks

To create a slice, we can simply convert a regular array, forgetting about its
extra power.
\begin{hscode}\SaveRestoreHook
\column{B}{@{}>{\hspre}l<{\hspost}@{}}%
\column{E}{@{}>{\hspre}l<{\hspost}@{}}%
\>[B]{}\Varid{fullSlice}\mathbin{::}\constraintfont{\Conid{RW}\;\Varid{n}}\Lolly \Conid{NArray}\;\Varid{a}\;\Varid{n}\to \exists\;\Varid{p}.(\Conid{Ur}\;(\Conid{Slice}\;\Varid{a}\;\Varid{p}),\constraintfont{\Conid{RW}\;\Varid{p}}\Lolly ()\RLolly\constraintfont{\Conid{RW}\;\Varid{n}})\RLolly\constraintfont{\Conid{RW}\;\Varid{p}}{}\<[E]%
\ColumnHook
\end{hscode}\resethooks
A slice can be split in two, as expected.
\begin{hscode}\SaveRestoreHook
\column{B}{@{}>{\hspre}l<{\hspost}@{}}%
\column{12}{@{}>{\hspre}l<{\hspost}@{}}%
\column{E}{@{}>{\hspre}l<{\hspost}@{}}%
\>[B]{}\Varid{sliceS}\mathbin{::}{}\<[12]%
\>[12]{}\constraintfont{\Conid{RW}\;\Varid{n}}\Lolly \Conid{Slice}\;\Varid{a}\;\Varid{n}\to \Conid{Int}\to {}\<[E]%
\\
\>[12]{}\exists\;\Varid{p}\;\Varid{q}.(\Conid{Ur}\;(\Conid{Slice}\;\Varid{a}\;\Varid{p},\Conid{Slice}\;\Varid{a}\;\Varid{q}),(\constraintfont{\Conid{RW}\;\Varid{p},\Conid{RW}\;\Varid{q}}\Lolly ()\RLolly\constraintfont{\Conid{RW}\;\Varid{n}}))\RLolly\constraintfont{\Conid{RW}\;\Varid{p},\Conid{RW}\;\Varid{q}}{}\<[E]%
\ColumnHook
\end{hscode}\resethooks

\paragraph*{Implementing UArray}
\label{sec:implementing-uarray}

With these primitives in hands we can give an implementation to the \ensuremath{\Conid{UArray}}
type of \cref{sec:borrowing-slices} and up. Since in
\cref{sec:borrowing-slices}, we use the same a type for owned arrays and slice,
we implement it as a sum type.

\begin{hscode}\SaveRestoreHook
\column{B}{@{}>{\hspre}l<{\hspost}@{}}%
\column{3}{@{}>{\hspre}l<{\hspost}@{}}%
\column{E}{@{}>{\hspre}l<{\hspost}@{}}%
\>[B]{}\mathbf{data}\;\Conid{UArray}\;\Varid{a}\;\Varid{n}{}\<[E]%
\\
\>[B]{}\hsindent{3}{}\<[3]%
\>[3]{}\mathrel{=}\Conid{OwnedArray}\;(\Conid{NArray}\;(\Conid{Val}\;\Varid{a})\;\Varid{n}){}\<[E]%
\\
\>[B]{}\hsindent{3}{}\<[3]%
\>[3]{}\mid \Conid{SliceArray}\;(\Conid{Slice}\;(\Conid{Val}\;\Varid{a})\;\Varid{n}){}\<[E]%
\ColumnHook
\end{hscode}\resethooks

We can, for instance, write the \ensuremath{\Varid{write}} function

\begin{hscode}\SaveRestoreHook
\column{B}{@{}>{\hspre}l<{\hspost}@{}}%
\column{E}{@{}>{\hspre}l<{\hspost}@{}}%
\>[B]{}\Varid{write}\mathbin{::}\constraintfont{\Conid{RW}\;\Varid{n}}\Lolly \Conid{UArray}\;\Varid{a}\;\Varid{n}\to \Conid{Int}\to \Varid{a}\to ()\RLolly\constraintfont{\Conid{RW}\;\Varid{n}}{}\<[E]%
\\
\>[B]{}\Varid{write}\;(\Conid{OwnedArray}\;\Varid{as})\;\Varid{i}\;\Varid{a}\mathrel{=}\Varid{writeNA}\;\Varid{as}\;\Varid{i}\;(\Conid{Val}\;\Varid{a}){}\<[E]%
\\
\>[B]{}\Varid{write}\;(\Conid{SliceArray}\;\Varid{as})\;\Varid{i}\;\Varid{a}\mathrel{=}\Varid{writeS}\;\Varid{as}\;\Varid{i}\;(\Conid{Val}\;\Varid{a}){}\<[E]%
\ColumnHook
\end{hscode}\resethooks

The \ensuremath{\Varid{read}} function is a little more involved, since it involves a borrow. But
the borrow is degenerate since we're borrowing a value of type \ensuremath{\Conid{Val}\;\Varid{a}}.

\begin{hscode}\SaveRestoreHook
\column{B}{@{}>{\hspre}l<{\hspost}@{}}%
\column{3}{@{}>{\hspre}l<{\hspost}@{}}%
\column{E}{@{}>{\hspre}l<{\hspost}@{}}%
\>[B]{}\Varid{read}\mathbin{::}\Conid{UArray}\;\Varid{a}⊸\Conid{Int}\to (\Conid{UArray}\;\Varid{a},\Conid{Ur}\;\Varid{a}){}\<[E]%
\\
\>[B]{}\Varid{read}\;(\Conid{OwnedArray}\;\Varid{as})\;\Varid{i}\mathrel{=}\Conid{Linearly}.\mathbf{do}{}\<[E]%
\\
\>[B]{}\hsindent{3}{}\<[3]%
\>[3]{}(\Conid{Ur}\;(\Conid{Val}\;\Varid{a}),\Varid{release})\leftarrow \Varid{borrowNA}\;\Varid{as}\;\Varid{i}{}\<[E]%
\\
\>[B]{}\hsindent{3}{}\<[3]%
\>[3]{}\Varid{release}{}\<[E]%
\\
\>[B]{}\hsindent{3}{}\<[3]%
\>[3]{}\Varid{\Conid{Linear}.return}\;(\Conid{Ur}\;\Varid{a}){}\<[E]%
\\
\>[B]{}\Varid{read}\;(\Conid{SliceArray}\;\Varid{as})\;\Varid{i}\mathrel{=}\Conid{Linearly}.\mathbf{do}{}\<[E]%
\\
\>[B]{}\hsindent{3}{}\<[3]%
\>[3]{}(\Conid{Ur}\;(\Conid{Val}\;\Varid{a}),\Varid{release})\leftarrow \Varid{borrowS}\;\Varid{as}\;\Varid{i}{}\<[E]%
\\
\>[B]{}\hsindent{3}{}\<[3]%
\>[3]{}\Varid{release}{}\<[E]%
\\
\>[B]{}\hsindent{3}{}\<[3]%
\>[3]{}\Varid{\Conid{Linear}.return}\;(\Conid{Ur}\;\Varid{a}){}\<[E]%
\ColumnHook
\end{hscode}\resethooks

As a final example, here is the implementation of the \ensuremath{\Varid{slice}} function.

\begin{hscode}\SaveRestoreHook
\column{B}{@{}>{\hspre}l<{\hspost}@{}}%
\column{3}{@{}>{\hspre}l<{\hspost}@{}}%
\column{11}{@{}>{\hspre}l<{\hspost}@{}}%
\column{E}{@{}>{\hspre}l<{\hspost}@{}}%
\>[B]{}\Varid{slice}\mathbin{::}{}\<[11]%
\>[11]{}\constraintfont{\Conid{RW}\;\Varid{n}}\Lolly \Conid{UArray}\;\Varid{a}\;\Varid{n}\to \Conid{Int}\to {}\<[E]%
\\
\>[11]{}\exists\;\Varid{p}\;\Varid{q}.(\Conid{Ur}\;(\Conid{UArray}\;\Varid{a}\;\Varid{p},\Conid{UArray}\;\Varid{a}\;\Varid{q}),(\constraintfont{\Conid{RW}\;\Varid{p},\Conid{RW}\;\Varid{q}}\Lolly ()\RLolly\constraintfont{\Conid{RW}\;\Varid{n}}))\RLolly\constraintfont{(\Conid{RW}\;\Varid{p},\Conid{RW}\;\Varid{q})}{}\<[E]%
\\
\>[B]{}\Varid{slice}\;(\Conid{OwnedArray}\;\Varid{as})\;\Varid{i}\mathrel{=}\Conid{Linearly}.\mathbf{do}{}\<[E]%
\\
\>[B]{}\hsindent{3}{}\<[3]%
\>[3]{}(\Conid{Ur}\;\Varid{as'},\Varid{release\char95 top})\leftarrow \Varid{fullSlice}\;\Varid{as}{}\<[E]%
\\
\>[B]{}\hsindent{3}{}\<[3]%
\>[3]{}(\Conid{Ur}\;(\Conid{SliceArray}\;\Varid{l},\Conid{SliceArray}\;\Varid{r}),\Varid{release\char95 slices})\leftarrow \Varid{sliceS}\;\Varid{as'}\;\Varid{i}{}\<[E]%
\\
\>[B]{}\hsindent{3}{}\<[3]%
\>[3]{}\Varid{\Conid{Linearly}.return}\;(\Conid{Ur}\;(\Varid{l},\Varid{r}),\mathbf{do}\;\{\mskip1.5mu \Varid{release\char95 slices};\Varid{release\char95 top}\mskip1.5mu\}){}\<[E]%
\\
\>[B]{}\Varid{slice}\;(\Conid{SliceArray}\;\Varid{as})\;\Varid{i}\mathrel{=}\Conid{Linearly}.\mathbf{do}{}\<[E]%
\\
\>[B]{}\hsindent{3}{}\<[3]%
\>[3]{}(\Conid{Ur}\;(\Varid{l},\Varid{r}),\Varid{release\char95 slices})\leftarrow \Varid{sliceS}\;\Varid{as}\;\Varid{i}{}\<[E]%
\\
\>[B]{}\hsindent{3}{}\<[3]%
\>[3]{}\Varid{\Conid{Linearly}.return}\;(\Conid{Ur}\;(\Conid{SliceArray}\;\Varid{l},\Conid{SliceArray}\;\Varid{r}),\Varid{release\char95 slices}){}\<[E]%
\ColumnHook
\end{hscode}\resethooks

\paragraph*{Deep borrowing and slicing}
\label{sec:deep-slicing}

The whole reason for the \ensuremath{\Conid{Slice}} type, however, is that we want to borrow and slice deep
inside a structure. To this effect, we introduce a function \ensuremath{\Varid{borrowDeep}} (resp.
\ensuremath{\Varid{sliceDeep}}) which takes a borrowing function (resp. a slicing function) on the inner type of a slice, and
returns a borrowing function (resp. a slicing function) on the slice itself.

\begin{hscode}\SaveRestoreHook
\column{B}{@{}>{\hspre}l<{\hspost}@{}}%
\column{3}{@{}>{\hspre}c<{\hspost}@{}}%
\column{3E}{@{}l@{}}%
\column{7}{@{}>{\hspre}l<{\hspost}@{}}%
\column{16}{@{}>{\hspre}l<{\hspost}@{}}%
\column{17}{@{}>{\hspre}l<{\hspost}@{}}%
\column{18}{@{}>{\hspre}l<{\hspost}@{}}%
\column{19}{@{}>{\hspre}l<{\hspost}@{}}%
\column{E}{@{}>{\hspre}l<{\hspost}@{}}%
\>[B]{}\Varid{borrowDeep}\mathbin{::}{}\<[16]%
\>[16]{}\constraintfont{\Conid{RW}\;\Varid{n}}\Lolly \Conid{Slice}\;\Varid{a}\;\Varid{n}\to {}\<[E]%
\\
\>[16]{}(\forall\;\Varid{p}.\constraintfont{\Conid{RW}\;\Varid{p}}\Lolly \Varid{a}\;\Varid{p}\to \exists\;\Varid{q}.(\Conid{Ur}\;(\Varid{b}\;\Varid{p}),(\constraintfont{\Conid{RW}\;\Varid{b}}\Lolly ()\RLolly\constraintfont{\Conid{RW}\;\Varid{p}})))\to {}\<[E]%
\\
\>[16]{}\exists\;\Varid{r}.(\Conid{Ur}\;(\Conid{Slice}\;\Varid{b}\;\Varid{r}),\constraintfont{\Conid{RW}\;\Varid{r}}\Lolly ()\RLolly\constraintfont{\Conid{RW}\;\Varid{n}}){}\<[E]%
\\
\>[B]{}\Varid{borrowDeepR}\mathbin{::}{}\<[17]%
\>[17]{}\constraintfont{\Conid{Read}\;\Varid{n}}\Lolly \Conid{Slice}\;\Varid{a}\;\Varid{n}\to {}\<[E]%
\\
\>[17]{}(\forall\;\Varid{p}.\constraintfont{\Conid{Read}\;\Varid{p}}\Lolly \Varid{a}\;\Varid{p}\to \exists\;\Varid{q}.(\Conid{Ur}\;(\Varid{b}\;\Varid{p}),(\constraintfont{\Conid{Read}\;\Varid{b}}\Lolly ()\RLolly\constraintfont{\Conid{Read}\;\Varid{p}})))\to {}\<[E]%
\\
\>[17]{}\exists\;\Varid{r}.(\Conid{Ur}\;(\Conid{Slice}\;\Varid{b}\;\Varid{r}),\constraintfont{\Conid{Read}\;\Varid{r}}\Lolly ()\RLolly\constraintfont{\Conid{Read}\;\Varid{n}}){}\<[E]%
\\
\>[B]{}\Varid{sliceDeep}{}\<[E]%
\\
\>[B]{}\hsindent{3}{}\<[3]%
\>[3]{}\mathbin{::}{}\<[3E]%
\>[7]{}\constraintfont{\Conid{RW}\;\Varid{n}}\Lolly \Conid{Slice}\;\Varid{a}\;\Varid{n}\to {}\<[E]%
\\
\>[7]{}(\forall\;\Varid{p}.{}\<[19]%
\>[19]{}\constraintfont{\Conid{RW}\;\Varid{p}}\Lolly \Varid{a}\;\Varid{p}\to {}\<[E]%
\\
\>[19]{}\exists\;\Varid{r}\;\Varid{s}.\constraintfont{(\Conid{RW}\;\Varid{r},\Conid{RW}\;\Varid{s})}\Lolly (\Conid{Ur}\;(\Varid{a}\;\Varid{r},\Varid{a}\;\Varid{s}),\constraintfont{(\Conid{RW}\;\Varid{r},\Conid{RW}\;\Varid{s})}\Lolly ()\RLolly\constraintfont{\Conid{RW}\;\Varid{p}}))\to {}\<[E]%
\\
\>[7]{}\exists\;\Varid{q}.{}\<[18]%
\>[18]{}\constraintfont{(\Conid{RW}\;\Varid{o},\Conid{RW}\;\Varid{q})}\Lolly {}\<[E]%
\\
\>[18]{}(\Conid{Ur}\;(\Conid{Slice}\;\Varid{a}\;\Varid{o},\Conid{Slice}\;\Varid{a}\;\Varid{q}),\constraintfont{(\Conid{RW}\;\Varid{o},\Conid{RW}\;\Varid{q})}\Lolly ()\RLolly\constraintfont{\Conid{RW}\;\Varid{n}})\RLolly\constraintfont{(\Conid{RW}\;\Varid{o},\Conid{RW}\;\Varid{q})}{}\<[E]%
\\
\>[B]{}\Varid{sliceDeepR}{}\<[E]%
\\
\>[B]{}\hsindent{3}{}\<[3]%
\>[3]{}\mathbin{::}{}\<[3E]%
\>[7]{}\constraintfont{\Conid{Read}\;\Varid{n}}\Lolly \Conid{Slice}\;\Varid{a}\;\Varid{n}\to {}\<[E]%
\\
\>[7]{}(\forall\;\Varid{p}.{}\<[19]%
\>[19]{}\constraintfont{\Conid{Read}\;\Varid{p}}\Lolly \Varid{a}\;\Varid{p}\to {}\<[E]%
\\
\>[19]{}\exists\;\Varid{r}\;\Varid{s}.\constraintfont{(\Conid{Read}\;\Varid{r},\Conid{Read}\;\Varid{s})}\Lolly (\Conid{Ur}\;(\Varid{a}\;\Varid{r},\Varid{a}\;\Varid{s}),\constraintfont{(\Conid{Read}\;\Varid{r},\Conid{Read}\;\Varid{s})}\Lolly ()\RLolly\constraintfont{\Conid{Read}\;\Varid{p}}))\to {}\<[E]%
\\
\>[7]{}\exists\;\Varid{q}.{}\<[18]%
\>[18]{}\constraintfont{(\Conid{Read}\;\Varid{o},\Conid{Read}\;\Varid{q})}\Lolly {}\<[E]%
\\
\>[18]{}(\Conid{Ur}\;(\Conid{Slice}\;\Varid{a}\;\Varid{o},\Conid{Slice}\;\Varid{a}\;\Varid{q}),\constraintfont{(\Conid{Read}\;\Varid{o},\Conid{Read}\;\Varid{q})}\Lolly ()\RLolly\constraintfont{\Conid{Read}\;\Varid{n}})\RLolly\constraintfont{(\Conid{Read}\;\Varid{o},\Conid{Read}\;\Varid{q})}{}\<[E]%
\ColumnHook
\end{hscode}\resethooks
Note that the type
of the inner slicing function, in \ensuremath{\Varid{sliceDeep}} is slightly abstracted compared to the type of
\ensuremath{\Varid{sliceS}}: it doesn't take an integer argument, instead it assumes that the
integer is captured by the inner slicing function's closure.

The key, here, is that \ensuremath{\Varid{borrowDeep}} and \ensuremath{\Varid{sliceDeep}} can be nested,
so that we can borrow (or slice) at an arbitrary depth.

\paragraph*{Implementing Matrix}
\label{sec:implementing-matrix}

\begin{figure}
  \maybesmall
\begin{hscode}\SaveRestoreHook
\column{B}{@{}>{\hspre}l<{\hspost}@{}}%
\column{3}{@{}>{\hspre}l<{\hspost}@{}}%
\column{7}{@{}>{\hspre}l<{\hspost}@{}}%
\column{12}{@{}>{\hspre}l<{\hspost}@{}}%
\column{15}{@{}>{\hspre}l<{\hspost}@{}}%
\column{16}{@{}>{\hspre}l<{\hspost}@{}}%
\column{17}{@{}>{\hspre}l<{\hspost}@{}}%
\column{18}{@{}>{\hspre}l<{\hspost}@{}}%
\column{19}{@{}>{\hspre}l<{\hspost}@{}}%
\column{E}{@{}>{\hspre}l<{\hspost}@{}}%
\>[B]{}\mathbf{type}\;\Conid{NArray}\mathbin{::}(\Conid{Reg}\to \Conid{Type})\to \Conid{Reg}\to \Conid{Type}{}\<[E]%
\\
\>[B]{}\Varid{newNArray}{}\<[12]%
\>[12]{}\mathbin{::}\constraintfont{\Conid{Linearly}}\Lolly \Conid{Int}\to (\constraintfont{\Conid{Linearly}}\Lolly \Conid{Int}\to \exists\;\Varid{p}.\Varid{a}\;\Varid{p}\RLolly\constraintfont{\Conid{RW}\;\Varid{p}})\to \exists\;\Varid{n}.\Conid{NArray}\;\Varid{a}\;\Varid{n}\RLolly\constraintfont{\Conid{RW}\;\Varid{n}}{}\<[E]%
\\
\>[B]{}\Varid{borrowNA}{}\<[12]%
\>[12]{}\mathbin{::}\constraintfont{\Conid{RW}\;\Varid{n}}\Lolly \Conid{NArray}\;\Varid{a}\;\Varid{n}\to \Conid{Int}\to \exists\;\Varid{p}.(\Conid{Ur}\;(\Varid{a}\;\Varid{p}),(\constraintfont{\Conid{RW}\;\Varid{p}}\Lolly ()\RLolly\constraintfont{\Conid{RW}\;\Varid{n}}))\RLolly\constraintfont{\Conid{RW}\;\Varid{p}}{}\<[E]%
\\
\>[B]{}\Varid{borrowNAR}{}\<[12]%
\>[12]{}\mathbin{::}\constraintfont{\Conid{R}\;\Varid{n}}\Lolly \Conid{NArray}\;\Varid{a}\;\Varid{n}\to \Conid{Int}\to \exists\;\Varid{p}.(\Conid{Ur}\;(\Varid{a}\;\Varid{p}),(\constraintfont{\Conid{R}\;\Varid{p}}\Lolly ()\RLolly(\constraintfont{\Conid{R}\;\Varid{n}})))\RLolly\constraintfont{\Conid{R}\;\Varid{p}}{}\<[E]%
\\
\>[B]{}\Varid{writeNA}{}\<[12]%
\>[12]{}\mathbin{::}\constraintfont{(\Conid{RW}\;\Varid{n},\Conid{RW}\;\Varid{p})}\Lolly \Conid{NArray}\;\Varid{a}\;\Varid{n}\to \Conid{Int}\to \Varid{a}\;\Varid{p}\to ()\RLolly\constraintfont{\Conid{RW}\;\Varid{n}}{}\<[E]%
\\[\blanklineskip]%
\>[B]{}\mathbf{type}\;\Conid{Val}\mathbin{::}\Conid{Type}\to \Conid{Reg}\to \Conid{Type}{}\<[E]%
\\
\>[B]{}\mathbf{newtype}\;\Conid{Val}\;\Varid{a}\;\Varid{n}\mathrel{=}\Conid{MkVal}\;\Varid{a}{}\<[E]%
\\
\>[B]{}\mathbf{instance}\;\constraintfont{\Conid{R}\;(\Conid{Val}\;\Varid{a}\;\Varid{n})}{}\<[E]%
\\
\>[B]{}\mathbf{instance}\;\constraintfont{\Conid{W}\;(\Conid{Val}\;\Varid{a}\;\Varid{n})}{}\<[E]%
\\[\blanklineskip]%
\>[B]{}\mathbf{type}\;\Conid{Slice}\mathbin{::}(\Conid{Reg}\to \Conid{Type})\to \Conid{Reg}\to \Conid{Type}{}\<[E]%
\\
\>[B]{}\Varid{fullSlice}\mathbin{::}\constraintfont{\Conid{RW}\;\Varid{n}}\Lolly \Conid{NArray}\;\Varid{a}\;\Varid{n}\to \exists\;\Varid{p}.(\Conid{Ur}\;(\Conid{Slice}\;\Varid{a}\;\Varid{p}),\constraintfont{\Conid{RW}\;\Varid{p}}\Lolly ()\RLolly\constraintfont{\Conid{RW}\;\Varid{n}})\RLolly\constraintfont{\Conid{RW}\;\Varid{p}}{}\<[E]%
\\
\>[B]{}\Varid{sliceS}\mathbin{::}{}\<[12]%
\>[12]{}\constraintfont{\Conid{RW}\;\Varid{n}}\Lolly \Conid{Slice}\;\Varid{a}\;\Varid{n}\to \Conid{Int}\to {}\<[E]%
\\
\>[12]{}\exists\;\Varid{p}\;\Varid{q}.(\Conid{Ur}\;(\Conid{Slice}\;\Varid{a}\;\Varid{p},\Conid{Slice}\;\Varid{a}\;\Varid{q}),(\constraintfont{\Conid{RW}\;\Varid{p},\Conid{RW}\;\Varid{q}}\Lolly ()\RLolly\constraintfont{\Conid{RW}\;\Varid{n}}))\RLolly\constraintfont{\Conid{RW}\;\Varid{p},\Conid{RW}\;\Varid{q}}{}\<[E]%
\\
\>[B]{}\Varid{borrowDeep}\mathbin{::}{}\<[16]%
\>[16]{}\constraintfont{\Conid{RW}\;\Varid{n}}\Lolly \Conid{Slice}\;\Varid{a}\;\Varid{n}\to {}\<[E]%
\\
\>[16]{}(\forall\;\Varid{p}.\constraintfont{\Conid{RW}\;\Varid{p}}\Lolly \Varid{a}\;\Varid{p}\to \exists\;\Varid{q}.(\Conid{Ur}\;(\Varid{b}\;\Varid{p}),(\constraintfont{\Conid{RW}\;\Varid{b}}\Lolly ()\RLolly\constraintfont{\Conid{RW}\;\Varid{p}})))\to {}\<[E]%
\\
\>[16]{}\exists\;\Varid{r}.(\Conid{Ur}\;(\Conid{Slice}\;\Varid{b}\;\Varid{r}),\constraintfont{\Conid{RW}\;\Varid{r}}\Lolly ()\RLolly\constraintfont{\Conid{RW}\;\Varid{n}}){}\<[E]%
\\
\>[B]{}\Varid{borrowDeepR}\mathbin{::}{}\<[17]%
\>[17]{}\constraintfont{\Conid{Read}\;\Varid{n}}\Lolly \Conid{Slice}\;\Varid{a}\;\Varid{n}\to {}\<[E]%
\\
\>[17]{}(\forall\;\Varid{p}.\constraintfont{\Conid{Read}\;\Varid{p}}\Lolly \Varid{a}\;\Varid{p}\to \exists\;\Varid{q}.(\Conid{Ur}\;(\Varid{b}\;\Varid{p}),(\constraintfont{\Conid{Read}\;\Varid{b}}\Lolly ()\RLolly\constraintfont{\Conid{Read}\;\Varid{p}})))\to {}\<[E]%
\\
\>[17]{}\exists\;\Varid{r}.(\Conid{Ur}\;(\Conid{Slice}\;\Varid{b}\;\Varid{r}),\constraintfont{\Conid{Read}\;\Varid{r}}\Lolly ()\RLolly\constraintfont{\Conid{Read}\;\Varid{n}}){}\<[E]%
\\
\>[B]{}\Varid{sliceDeep}{}\<[E]%
\\
\>[B]{}\hsindent{3}{}\<[3]%
\>[3]{}\mathbin{::}{}\<[7]%
\>[7]{}\constraintfont{\Conid{RW}\;\Varid{n}}\Lolly \Conid{Slice}\;\Varid{a}\;\Varid{n}\to {}\<[E]%
\\
\>[7]{}(\forall\;\Varid{p}.{}\<[19]%
\>[19]{}\constraintfont{\Conid{RW}\;\Varid{p}}\Lolly \Varid{a}\;\Varid{p}\to {}\<[E]%
\\
\>[19]{}\exists\;\Varid{r}\;\Varid{s}.\constraintfont{(\Conid{RW}\;\Varid{r},\Conid{RW}\;\Varid{s})}\Lolly (\Conid{Ur}\;(\Varid{a}\;\Varid{r},\Varid{a}\;\Varid{s}),\constraintfont{(\Conid{RW}\;\Varid{r},\Conid{RW}\;\Varid{s})}\Lolly ()\RLolly\constraintfont{\Conid{RW}\;\Varid{p}}))\to {}\<[E]%
\\
\>[7]{}\exists\;\Varid{q}.{}\<[18]%
\>[18]{}\constraintfont{(\Conid{RW}\;\Varid{o},\Conid{RW}\;\Varid{q})}\Lolly {}\<[E]%
\\
\>[18]{}(\Conid{Ur}\;(\Conid{Slice}\;\Varid{a}\;\Varid{o},\Conid{Slice}\;\Varid{a}\;\Varid{q}),\constraintfont{(\Conid{RW}\;\Varid{o},\Conid{RW}\;\Varid{q})}\Lolly ()\RLolly\constraintfont{\Conid{RW}\;\Varid{n}})\RLolly\constraintfont{(\Conid{RW}\;\Varid{o},\Conid{RW}\;\Varid{q})}{}\<[E]%
\\
\>[B]{}\Varid{sliceDeepR}{}\<[E]%
\\
\>[B]{}\hsindent{3}{}\<[3]%
\>[3]{}\mathbin{::}{}\<[7]%
\>[7]{}\constraintfont{\Conid{Read}\;\Varid{n}}\Lolly \Conid{Slice}\;\Varid{a}\;\Varid{n}\to {}\<[E]%
\\
\>[B]{}\hsindent{3}{}\<[3]%
\>[3]{}(\forall\;\Varid{p}.{}\<[15]%
\>[15]{}\constraintfont{\Conid{Read}\;\Varid{p}}\Lolly \Varid{a}\;\Varid{p}\to {}\<[E]%
\\
\>[15]{}\exists\;\Varid{r}\;\Varid{s}.\constraintfont{(\Conid{Read}\;\Varid{r},\Conid{Read}\;\Varid{s})}\Lolly (\Conid{Ur}\;(\Varid{a}\;\Varid{r},\Varid{a}\;\Varid{s}),\constraintfont{(\Conid{Read}\;\Varid{r},\Conid{Read}\;\Varid{s})}\Lolly ()\RLolly\constraintfont{\Conid{Read}\;\Varid{p}}))\to {}\<[E]%
\\
\>[3]{}\hsindent{4}{}\<[7]%
\>[7]{}\exists\;\Varid{q}.{}\<[18]%
\>[18]{}\constraintfont{(\Conid{Read}\;\Varid{o},\Conid{Read}\;\Varid{q})}\Lolly {}\<[E]%
\\
\>[18]{}(\Conid{Ur}\;(\Conid{Slice}\;\Varid{a}\;\Varid{o},\Conid{Slice}\;\Varid{a}\;\Varid{q}),\constraintfont{(\Conid{Read}\;\Varid{o},\Conid{Read}\;\Varid{q})}\Lolly ()\RLolly\constraintfont{\Conid{Read}\;\Varid{n}})\RLolly\constraintfont{(\Conid{Read}\;\Varid{o},\Conid{Read}\;\Varid{q})}{}\<[E]%
\ColumnHook
\end{hscode}\resethooks
  \caption{Nested array \textsc{api}}
  \label{fig:narray-api}
\end{figure}

\unsure{In the figure: A \ensuremath{\Conid{SliceSR}} function is missing. We don't show its use, but should it be
  there for symmetry? Functions such as \ensuremath{\Conid{WriteS}} are also missing}
This \textsc{api}, which we repeat in full in \cref{fig:narray-api},
is enough to implement the \ensuremath{\Conid{Matrix}} type from \cref{sec:valiant}.
Like arrays, matrices are implemented as a disjunction
\begin{hscode}\SaveRestoreHook
\column{B}{@{}>{\hspre}l<{\hspost}@{}}%
\column{3}{@{}>{\hspre}l<{\hspost}@{}}%
\column{E}{@{}>{\hspre}l<{\hspost}@{}}%
\>[B]{}\mathbf{data}\;\Conid{Matrix}\;\Varid{a}\;\Varid{n}{}\<[E]%
\\
\>[B]{}\hsindent{3}{}\<[3]%
\>[3]{}\mathrel{=}\Conid{OwnedMatrix}\;(\Conid{NArray}\;(\Conid{NArray}\;(\Conid{Val}\;\Varid{a}))\;\Varid{n}){}\<[E]%
\\
\>[B]{}\hsindent{3}{}\<[3]%
\>[3]{}\mid \Conid{SliceMatrix}\;(\Conid{Slice}\;(\Conid{Slice}\;(\Conid{Val}\;\Varid{a}))\;\Varid{n}){}\<[E]%
\ColumnHook
\end{hscode}\resethooks

We can see the deep borrowing and slicing functions at work, for instance, in
the \ensuremath{\Varid{sliceMatrixV}} function, which splits a matrix vertically into a top matrix
and a bottom matrix.
\begin{hscode}\SaveRestoreHook
\column{B}{@{}>{\hspre}l<{\hspost}@{}}%
\column{3}{@{}>{\hspre}l<{\hspost}@{}}%
\column{5}{@{}>{\hspre}l<{\hspost}@{}}%
\column{18}{@{}>{\hspre}l<{\hspost}@{}}%
\column{30}{@{}>{\hspre}l<{\hspost}@{}}%
\column{E}{@{}>{\hspre}l<{\hspost}@{}}%
\>[B]{}\Varid{sliceMatrixV}\mathbin{::}{}\<[18]%
\>[18]{}\constraintfont{\Conid{RW}\;\Varid{n}}\Lolly \Conid{Matrix}\;\Varid{a}\;\Varid{n}\to \Conid{Int}\to {}\<[E]%
\\
\>[18]{}\exists\;\Varid{p}\;\Varid{q}.(\Conid{Ur}\;(\Conid{Matrix}\;\Varid{p},\Conid{Matrix}\;\Varid{q}),\constraintfont{(\Conid{RW}\;\Varid{p},\Conid{RW}\;\Varid{q})}\Lolly {}\<[E]%
\\
\>[18]{}\hsindent{12}{}\<[30]%
\>[30]{}()\RLolly\constraintfont{\Conid{RW}\;\Varid{n}})\RLolly\constraintfont{(\Conid{RW}\;\Varid{p},\Conid{RW}\;\Varid{q})}{}\<[E]%
\\
\>[B]{}\Varid{sliceMatrixV}\;(\Conid{OwnedMatrix}\;\Varid{m})\;\Varid{row\char95 num}\mathrel{=}\Conid{Linearly}.\mathbf{do}{}\<[E]%
\\
\>[B]{}\hsindent{3}{}\<[3]%
\>[3]{}(\Conid{Ur}\;\Varid{m'},\Varid{release\char95 top})\leftarrow \Varid{fullSlice}\;\Varid{m}{}\<[E]%
\\
\>[B]{}\hsindent{3}{}\<[3]%
\>[3]{}(\Conid{Ur}\;\Varid{m''},\Varid{release\char95 top\char95 deep})\leftarrow \Varid{borrowDeep}\;\Varid{m'}\;\Varid{fullSlice}{}\<[E]%
\\
\>[B]{}\hsindent{3}{}\<[3]%
\>[3]{}(\Conid{Ur}\;(\Varid{t},\Varid{b}),\Varid{release\char95 slices})\leftarrow \Varid{sliceDeep}\;\Varid{m''}\;(\lambda \Varid{col}\to \Varid{slice}\;\Varid{col}\;\Varid{row\char95 num}){}\<[E]%
\\
\>[B]{}\hsindent{3}{}\<[3]%
\>[3]{}\Varid{\Conid{Linearly}.return}\;({}\<[E]%
\\
\>[3]{}\hsindent{2}{}\<[5]%
\>[5]{}\Conid{Ur}\;(\Conid{SliceMatrix}\;\Varid{t},\Conid{SliceMatrix}\;\Varid{b}),{}\<[E]%
\\
\>[3]{}\hsindent{2}{}\<[5]%
\>[5]{}\Conid{Linearly}.\mathbf{do}\;\{\mskip1.5mu \Varid{release\char95 slices};\Varid{release\char95 top\char95 deep};\Varid{release\char95 top}\mskip1.5mu\}){}\<[E]%
\\
\>[B]{}\Varid{sliceMatrixV}\;(\Conid{SliceMatrix}\;\Varid{m})\;\Varid{row\char95 num}\mathrel{=}\Conid{Linearly}.\mathbf{do}{}\<[E]%
\\
\>[B]{}\hsindent{3}{}\<[3]%
\>[3]{}(\Conid{Ur}\;(\Varid{t},\Varid{b}),\Varid{release\char95 slices})\leftarrow \Varid{sliceDeep}\;\Varid{m}\;(\lambda \Varid{col}\to \Varid{slice}\;\Varid{col}\;\Varid{row\char95 num}){}\<[E]%
\\
\>[B]{}\hsindent{3}{}\<[3]%
\>[3]{}\Varid{\Conid{Linearly}.return}\;(\Conid{Ur}\;(\Conid{SliceMatrix}\;\Varid{t},\Conid{SliceMatrix}\;\Varid{b}),\Varid{release\char95 slices}){}\<[E]%
\ColumnHook
\end{hscode}\resethooks
Notice how we use \ensuremath{\Varid{borrowDeep}} in the \ensuremath{\Conid{OwnedMatrix}} case to convert an \ensuremath{\Conid{Array}\;(\Conid{Array}\;\anonymous )} into a \ensuremath{\Conid{Slice}\;(\Conid{Slice}\;\anonymous )}. The rest of the matrix functions follow
easily.

As promised our nested array \textsc{api} is general enough to support slicing
along both dimension of matrices. Of course, there's nothing special about
matrices, higher dimension tensors would work as well. The key to all this is
that all this flexibility follows from the theory of linear constraints, which
is a natural extension of Haskell's constraints. Constraints can be captured in
closure, and in existential types (or \textsc{gadt}s in real-life Haskell),
which affords us a lot of programming techniques. We don't have to develop any
\emph{ad hoc} theory of capabilities.

\subsection{Behind the curtain: implementing slices}
\label{sec:narray-impl}

One may naturally ask how a library providing the capabilities such as the
nested array \textsc{api} of \cref{fig:narray-api} can be implemented. In this
section we will do exactly that.

Capabilities like \ensuremath{\Conid{Read}} and \ensuremath{\Conid{Write}} carry no intrinsic meaning for the
compiler. They are simply constraints, much like type-class constraints. Their
interpretation is determined by the library that implements them. To users of
the library these constraints remain abstract, but now we look inside the
library to see how to implement their concrete definitions. In particular, we
will require the following primitives:

\begin{hscode}\SaveRestoreHook
\column{B}{@{}>{\hspre}l<{\hspost}@{}}%
\column{E}{@{}>{\hspre}l<{\hspost}@{}}%
\>[B]{}\Varid{unsafeRead}\mathbin{::}()\RLolly\constraintfont{\Conid{Read}\;\Varid{n}}{}\<[E]%
\\
\>[B]{}\Varid{unsafeWrite}\mathbin{::}()\RLolly\constraintfont{\Conid{Write}\;\Varid{n}}{}\<[E]%
\\
\>[B]{}\Varid{unsafeRW}\mathbin{::}()\RLolly\constraintfont{\Conid{RW}\;\Varid{n}}{}\<[E]%
\\
\>[B]{}\Varid{unsafeRW}\mathrel{=}\Conid{Linearly}.\mathbf{do}\;\{\mskip1.5mu \Varid{unsafeRead};\Varid{unsafeWrite};\Varid{\Conid{Linear}.return}\;()\mskip1.5mu\}{}\<[E]%
\\[\blanklineskip]%
\>[B]{}\Varid{unsafeConsumeRead}\mathbin{::}\constraintfont{\Conid{Read}\;\Varid{n}}\Lolly (){}\<[E]%
\\
\>[B]{}\Varid{unsafeConsumeWrite}\mathbin{::}\constraintfont{\Conid{Write}\;\Varid{n}}\Lolly (){}\<[E]%
\\
\>[B]{}\Varid{unsafeConsumeRW}\mathbin{::}\constraintfont{\Conid{RW}\;\Varid{n}}\Lolly (){}\<[E]%
\\
\>[B]{}\Varid{unsafeConsumeRW}\mathrel{=}\Conid{Linearly}.\mathbf{do}\;\{\mskip1.5mu \Varid{unsafeConsumeRead};\Varid{unsafeConsumeWrite};\Varid{\Conid{Linear}.return}\;()\mskip1.5mu\}{}\<[E]%
\ColumnHook
\end{hscode}\resethooks

In this section, we will assume the functions \ensuremath{\Varid{newNArray}}, \ensuremath{\Varid{borrowNA}},
\ensuremath{\Varid{borrowNAR}}, and \ensuremath{\Varid{writeNA}}, and will implement the rest of
\cref{fig:narray-api}.

To implement slices, we will use an object-oriented definition functions are a
powerful and convenient abstraction. An implementation aiming for
efficiency would have to use a more specialised implementation (like above,
\ensuremath{\Conid{UnsafeMkSlice}} isn't part of the public interface).
\begin{hscode}\SaveRestoreHook
\column{B}{@{}>{\hspre}l<{\hspost}@{}}%
\column{3}{@{}>{\hspre}l<{\hspost}@{}}%
\column{12}{@{}>{\hspre}c<{\hspost}@{}}%
\column{12E}{@{}l@{}}%
\column{16}{@{}>{\hspre}l<{\hspost}@{}}%
\column{E}{@{}>{\hspre}l<{\hspost}@{}}%
\>[B]{}\mathbf{type}\;\Conid{Slice}\mathbin{::}(\Conid{Reg}\to \Conid{Type})\to \Conid{Reg}\to \Conid{Type}{}\<[E]%
\\
\>[B]{}\mathbf{data}\;\Conid{Slice}\;\Varid{a}\;\Varid{n}\mathrel{=}\Conid{UnsafeMkSlice}\;\{\mskip1.5mu {}\<[E]%
\\
\>[B]{}\hsindent{3}{}\<[3]%
\>[3]{}\Varid{borrowS}{}\<[12]%
\>[12]{}\mathbin{::}{}\<[12E]%
\>[16]{}\constraintfont{\Conid{RW}\;\Varid{n}}\Lolly \Conid{Int}\to \exists\;\Varid{p}.(\Conid{Ur}\;(\Varid{a}\;\Varid{p}),(\constraintfont{\Conid{RW}\;\Varid{p}}\Lolly ()\RLolly\constraintfont{\Conid{RW}\;\Varid{n}}))\RLolly\constraintfont{\Conid{RW}\;\Varid{p}},{}\<[E]%
\\
\>[B]{}\hsindent{3}{}\<[3]%
\>[3]{}\Varid{writeS}{}\<[12]%
\>[12]{}\mathbin{::}{}\<[12E]%
\>[16]{}\forall\;\Varid{p}.\constraintfont{(\Conid{RW}\;\Varid{n},\Conid{RW}\;\Varid{p})}\Lolly \Conid{Int}\to \Varid{a}\;\Varid{p}\to ()\RLolly\constraintfont{\Conid{RW}\;\Varid{n}}\mskip1.5mu\}{}\<[E]%
\ColumnHook
\end{hscode}\resethooks
Converting an array into a slice is straightforward.
\begin{hscode}\SaveRestoreHook
\column{B}{@{}>{\hspre}l<{\hspost}@{}}%
\column{3}{@{}>{\hspre}c<{\hspost}@{}}%
\column{3E}{@{}l@{}}%
\column{6}{@{}>{\hspre}l<{\hspost}@{}}%
\column{E}{@{}>{\hspre}l<{\hspost}@{}}%
\>[B]{}\Varid{fullSlice}\mathbin{::}\constraintfont{\Conid{RW}\;\Varid{n}}\Lolly \Conid{NArray}\;\Varid{a}\;\Varid{n}\to \exists\;\Varid{p}.(\Conid{Ur}\;(\Conid{Slice}\;\Varid{a}\;\Varid{p}),\constraintfont{\Conid{RW}\;\Varid{p}}\Lolly ()\RLolly\constraintfont{\Conid{RW}\;\Varid{n}})\RLolly\constraintfont{\Conid{RW}\;\Varid{p}}{}\<[E]%
\\
\>[B]{}\Varid{fullSlice}\;\Varid{as}\mathrel{=}\Varid{\Conid{Linearly}.return}{}\<[E]%
\\
\>[B]{}\hsindent{3}{}\<[3]%
\>[3]{}({}\<[3E]%
\>[6]{}(\Conid{UnsafeMkSlice}\;\{\mskip1.5mu \Varid{borrowS}\mathrel{=}\Varid{borrowNA}\;\Varid{as},\Varid{writeS}\mathrel{=}\Varid{writeNA}\;\Varid{as}\mskip1.5mu\}),{}\<[E]%
\\
\>[6]{}\Varid{\Conid{Linearly}.return}\;()){}\<[E]%
\ColumnHook
\end{hscode}\resethooks
In this implementation $\ensuremath{\Varid{p}}=\ensuremath{\Varid{n}}$, but this information is, crucially,
unavailable to client code, because \ensuremath{\Varid{p}} is existentially quantified. This
explains why the release function is trivial.

The implementation of \ensuremath{\Varid{sliceS}} is more interesting\unsure{The type applications
  are a bit of a lie here. But this is already quite a complicated function, so
  I don't know whether we should go even more in the weeds. In fact, I don't
  even think the text of this function really helps at anything. It was meant as
  evidence that, yes, this can be done, but I don't like it anymore.}
\begin{hscode}\SaveRestoreHook
\column{B}{@{}>{\hspre}l<{\hspost}@{}}%
\column{4}{@{}>{\hspre}l<{\hspost}@{}}%
\column{7}{@{}>{\hspre}l<{\hspost}@{}}%
\column{9}{@{}>{\hspre}l<{\hspost}@{}}%
\column{11}{@{}>{\hspre}l<{\hspost}@{}}%
\column{12}{@{}>{\hspre}l<{\hspost}@{}}%
\column{13}{@{}>{\hspre}l<{\hspost}@{}}%
\column{14}{@{}>{\hspre}l<{\hspost}@{}}%
\column{15}{@{}>{\hspre}l<{\hspost}@{}}%
\column{E}{@{}>{\hspre}l<{\hspost}@{}}%
\>[B]{}\Varid{sliceS}\mathbin{::}{}\<[12]%
\>[12]{}\constraintfont{\Conid{RW}\;\Varid{n}}\Lolly \Conid{Slice}\;\Varid{a}\;\Varid{n}\to \Conid{Int}\to {}\<[E]%
\\
\>[12]{}\exists\;\Varid{p}\;\Varid{q}.(\Conid{Ur}\;(\Conid{Slice}\;\Varid{a}\;\Varid{p},\Conid{Slice}\;\Varid{a}\;\Varid{q}),(\constraintfont{\Conid{RW}\;\Varid{p},\Conid{RW}\;\Varid{q}}\Lolly ()\RLolly\constraintfont{\Conid{RW}\;\Varid{n}}))\RLolly\constraintfont{\Conid{RW}\;\Varid{p},\Conid{RW}\;\Varid{q}}{}\<[E]%
\\
\>[B]{}\Varid{sliceS}\;\Varid{as}\;\Varid{i}\mathrel{=}\Varid{\Conid{Linearly}.return}{}\<[E]%
\\
\>[B]{}({}\<[4]%
\>[4]{}\Conid{Ur}{}\<[E]%
\\
\>[4]{}(\Conid{UnsafeMkSlice}\;\{\mskip1.5mu {}\<[E]%
\\
\>[4]{}\hsindent{7}{}\<[11]%
\>[11]{}\Varid{borrowS}\mathrel{=}(\lambda \Varid{j}\to {}\<[E]%
\\
\>[11]{}\hsindent{2}{}\<[13]%
\>[13]{}\mathbf{if}\;\Varid{j}\mathbin{>}\Varid{i}\;\mathbf{then}{}\<[E]%
\\
\>[13]{}\hsindent{2}{}\<[15]%
\>[15]{}\Varid{error}\;\text{\ttfamily \char34 Index~out~of~bound\char34}{}\<[E]%
\\
\>[11]{}\hsindent{2}{}\<[13]%
\>[13]{}\mathbf{else}\;\Conid{Linearly}.\mathbf{do}{}\<[E]%
\\
\>[13]{}\hsindent{2}{}\<[15]%
\>[15]{}\Varid{unsafeRw}\mathord{@}\Varid{n}{}\<[E]%
\\
\>[13]{}\hsindent{2}{}\<[15]%
\>[15]{}\Varid{borrowS}\;\Varid{as}\;\Varid{j}),{}\<[E]%
\\
\>[4]{}\hsindent{7}{}\<[11]%
\>[11]{}\Varid{writeS}\mathrel{=}(\lambda \Varid{j}\;\Varid{v}\to {}\<[E]%
\\
\>[11]{}\hsindent{2}{}\<[13]%
\>[13]{}\mathbf{if}\;\Varid{j}\mathbin{>}\Varid{i}\;\mathbf{then}{}\<[E]%
\\
\>[13]{}\hsindent{2}{}\<[15]%
\>[15]{}\Varid{error}\;\text{\ttfamily \char34 Index~out~of~bound\char34}{}\<[E]%
\\
\>[11]{}\hsindent{2}{}\<[13]%
\>[13]{}\mathbf{else}{}\<[E]%
\\
\>[13]{}\hsindent{2}{}\<[15]%
\>[15]{}\Varid{unsafeRw}\mathord{@}\Varid{n}{}\<[E]%
\\
\>[13]{}\hsindent{2}{}\<[15]%
\>[15]{}\Varid{writeS}\;\Varid{as}\;\Varid{j}\;\Varid{v})\mskip1.5mu\},{}\<[E]%
\\
\>[4]{}\hsindent{5}{}\<[9]%
\>[9]{}\Conid{UnsafeMkSlice}\;\{\mskip1.5mu {}\<[E]%
\\
\>[9]{}\hsindent{2}{}\<[11]%
\>[11]{}\Varid{borrowS}\mathrel{=}(\lambda \Varid{j}\to \Conid{Linearly}.\mathbf{do}{}\<[E]%
\\
\>[11]{}\hsindent{3}{}\<[14]%
\>[14]{}\Varid{unsafeRW}\mathord{@}\Varid{n}{}\<[E]%
\\
\>[11]{}\hsindent{3}{}\<[14]%
\>[14]{}\Varid{borrowS}\;\Varid{as}\;(\Varid{i}\mathbin{+}\Varid{j})),{}\<[E]%
\\
\>[9]{}\hsindent{2}{}\<[11]%
\>[11]{}\Varid{writeS}\mathrel{=}(\lambda \Varid{j}\;\Varid{v}\to \Conid{Linearly}.\mathbf{do}{}\<[E]%
\\
\>[11]{}\hsindent{2}{}\<[13]%
\>[13]{}\Varid{unsafeRW}\mathord{@}\Varid{n}{}\<[E]%
\\
\>[11]{}\hsindent{2}{}\<[13]%
\>[13]{}\Varid{writeS}\;\Varid{as}\;(\Varid{i}\mathbin{+}\Varid{j})\;\Varid{v})\mskip1.5mu\}),{}\<[E]%
\\
\>[4]{}\hsindent{3}{}\<[7]%
\>[7]{}(\Conid{Linearly}.\mathbf{do}\;\{\mskip1.5mu \Varid{unsafeConsumeRW}\mathord{@}\Varid{p};\Varid{unsafeConsumeRW}\mathord{@}\Varid{q};\Varid{\Conid{Linearly}.return}\;()\mskip1.5mu\})){}\<[E]%
\ColumnHook
\end{hscode}\resethooks
Note that the implementation uses unsafe creation and consumption of
capabilities, since it creates a new access pattern. As the library author, we
must, in exchange, ensure that this access pattern doesn't create aliases.

The deep borrows is where we really leverage the object-oriented definition of
slices. The idea is
simple: we delay actually calling the inner borrow function until we borrow an
inner value from the slice.
\begin{hscode}\SaveRestoreHook
\column{B}{@{}>{\hspre}l<{\hspost}@{}}%
\column{3}{@{}>{\hspre}l<{\hspost}@{}}%
\column{7}{@{}>{\hspre}l<{\hspost}@{}}%
\column{9}{@{}>{\hspre}l<{\hspost}@{}}%
\column{10}{@{}>{\hspre}l<{\hspost}@{}}%
\column{16}{@{}>{\hspre}l<{\hspost}@{}}%
\column{17}{@{}>{\hspre}l<{\hspost}@{}}%
\column{28}{@{}>{\hspre}c<{\hspost}@{}}%
\column{28E}{@{}l@{}}%
\column{32}{@{}>{\hspre}l<{\hspost}@{}}%
\column{E}{@{}>{\hspre}l<{\hspost}@{}}%
\>[B]{}\Varid{borrowDeep}\mathbin{::}{}\<[16]%
\>[16]{}\constraintfont{\Conid{RW}\;\Varid{n}}\Lolly \Conid{Slice}\;\Varid{a}\;\Varid{n}\to {}\<[E]%
\\
\>[16]{}(\forall\;\Varid{p}.\constraintfont{\Conid{RW}\;\Varid{p}}\Lolly \Varid{a}\;\Varid{p}\to \exists\;\Varid{q}.(\Conid{Ur}\;(\Varid{b}\;\Varid{p}),(\constraintfont{\Conid{RW}\;\Varid{b}}\Lolly ()\RLolly\constraintfont{\Conid{RW}\;\Varid{p}})))\to {}\<[E]%
\\
\>[16]{}\exists\;\Varid{r}.(\Conid{Ur}\;(\Conid{Slice}\;\Varid{b}\;\Varid{r}),\constraintfont{\Conid{RW}\;\Varid{r}}\Lolly ()\RLolly\constraintfont{\Conid{RW}\;\Varid{n}}){}\<[E]%
\\
\>[B]{}\Varid{borrowDeep}\;\Varid{as}\;\Varid{brw}\mathrel{=}\Varid{\Conid{Linearly}.return}\;({}\<[E]%
\\
\>[B]{}\hsindent{3}{}\<[3]%
\>[3]{}\Conid{UnsafeMkSlice}\;\{\mskip1.5mu {}\<[E]%
\\
\>[3]{}\hsindent{4}{}\<[7]%
\>[7]{}\Varid{borrowS}\mathrel{=}\lambda \Varid{i}\to \Conid{Linearly}.\mathbf{do}{}\<[E]%
\\
\>[7]{}\hsindent{2}{}\<[9]%
\>[9]{}(\Conid{Ur}\;\Varid{a},\Varid{release\char95 a}){}\<[28]%
\>[28]{}\leftarrow {}\<[28E]%
\>[32]{}\Varid{borrowS}\;\Varid{as}\;\Varid{i}{}\<[E]%
\\
\>[7]{}\hsindent{2}{}\<[9]%
\>[9]{}(\Conid{Ur}\;\Varid{b},\Varid{release\char95 b}){}\<[28]%
\>[28]{}\leftarrow {}\<[28E]%
\>[32]{}\Varid{brw}\;\Varid{a}{}\<[E]%
\\
\>[7]{}\hsindent{2}{}\<[9]%
\>[9]{}\Varid{\Conid{Linearly}.return}\;(\Conid{Ur}\;\Varid{b},\Conid{Linearly}.\mathbf{do}\;\{\mskip1.5mu \Varid{release\char95 b};\Varid{release\char95 a}\mskip1.5mu\}),{}\<[E]%
\\
\>[3]{}\hsindent{4}{}\<[7]%
\>[7]{}\Varid{writeS}\mathrel{=}{}\<[E]%
\\
\>[7]{}\hsindent{3}{}\<[10]%
\>[10]{}\Varid{error}\;\text{\ttfamily \char34 Attempt~to~write~to~a~borrowed~space\char34}\mskip1.5mu\},{}\<[E]%
\\
\>[B]{}\hsindent{3}{}\<[3]%
\>[3]{}(\Varid{\Conid{Linearly}.return}\;())){}\<[E]%
\\[\blanklineskip]%
\>[B]{}\Varid{borrowDeepR}\mathbin{::}{}\<[17]%
\>[17]{}\constraintfont{\Conid{Read}\;\Varid{n}}\Lolly \Conid{Slice}\;\Varid{a}\;\Varid{n}\to {}\<[E]%
\\
\>[17]{}(\forall\;\Varid{p}.\constraintfont{\Conid{Read}\;\Varid{p}}\Lolly \Varid{a}\;\Varid{p}\to \exists\;\Varid{q}.(\Conid{Ur}\;(\Varid{b}\;\Varid{p}),(\constraintfont{\Conid{Read}\;\Varid{b}}\Lolly ()\RLolly\constraintfont{\Conid{Read}\;\Varid{p}})))\to {}\<[E]%
\\
\>[17]{}\exists\;\Varid{r}.(\Conid{Ur}\;(\Conid{Slice}\;\Varid{b}\;\Varid{r}),\constraintfont{\Conid{Read}\;\Varid{r}}\Lolly ()\RLolly\constraintfont{\Conid{Read}\;\Varid{n}}){}\<[E]%
\\
\>[B]{}\Varid{borrowDeepR}\mathrel{=}\mathbin{…}\mbox{\onelinecomment  Mutatis mutandis}{}\<[E]%
\ColumnHook
\end{hscode}\resethooks
Note how the write method is an error after a deep borrow. It doesn't make sense
to try to write a \ensuremath{\Varid{b}} in a cell that really holds an \ensuremath{\Varid{a}}. In our case \ensuremath{\Varid{b}} will
always be \ensuremath{\Conid{Slice}\;\Varid{c}} for some \ensuremath{\Varid{c}}, so \ensuremath{\Varid{writeS}} would have to replace a slice of
an array wholesale, this isn't an operation which can be performed without
copying the inner array, while \ensuremath{\Varid{writeS}} is supposed to be of $O(1)$ write
operation. Fortunately, slices are always borrowed, and attempting to write a
slice to an array cell would result in a static error: the release operator
hasn't been called.

Similarly we can lift a slicing function on the inner \ensuremath{\Varid{a}} with the following
functions
\begin{hscode}\SaveRestoreHook
\column{B}{@{}>{\hspre}l<{\hspost}@{}}%
\column{3}{@{}>{\hspre}l<{\hspost}@{}}%
\column{5}{@{}>{\hspre}l<{\hspost}@{}}%
\column{7}{@{}>{\hspre}l<{\hspost}@{}}%
\column{9}{@{}>{\hspre}l<{\hspost}@{}}%
\column{11}{@{}>{\hspre}l<{\hspost}@{}}%
\column{15}{@{}>{\hspre}l<{\hspost}@{}}%
\column{16}{@{}>{\hspre}l<{\hspost}@{}}%
\column{18}{@{}>{\hspre}l<{\hspost}@{}}%
\column{19}{@{}>{\hspre}l<{\hspost}@{}}%
\column{39}{@{}>{\hspre}l<{\hspost}@{}}%
\column{E}{@{}>{\hspre}l<{\hspost}@{}}%
\>[B]{}\Varid{sliceDeep}{}\<[E]%
\\
\>[B]{}\hsindent{3}{}\<[3]%
\>[3]{}\mathbin{::}{}\<[7]%
\>[7]{}\constraintfont{\Conid{RW}\;\Varid{n}}\Lolly \Conid{Slice}\;\Varid{a}\;\Varid{n}\to {}\<[E]%
\\
\>[7]{}(\forall\;\Varid{p}.{}\<[19]%
\>[19]{}\constraintfont{\Conid{RW}\;\Varid{p}}\Lolly \Varid{a}\;\Varid{p}\to {}\<[E]%
\\
\>[19]{}\exists\;\Varid{r}\;\Varid{s}.\constraintfont{(\Conid{RW}\;\Varid{r},\Conid{RW}\;\Varid{s})}\Lolly (\Conid{Ur}\;(\Varid{a}\;\Varid{r},\Varid{a}\;\Varid{s}),\constraintfont{(\Conid{RW}\;\Varid{r},\Conid{RW}\;\Varid{s})}\Lolly ()\RLolly\constraintfont{\Conid{RW}\;\Varid{p}}))\to {}\<[E]%
\\
\>[7]{}\exists\;\Varid{q}.{}\<[18]%
\>[18]{}\constraintfont{(\Conid{RW}\;\Varid{o},\Conid{RW}\;\Varid{q})}\Lolly {}\<[E]%
\\
\>[18]{}(\Conid{Ur}\;(\Conid{Slice}\;\Varid{a}\;\Varid{o},\Conid{Slice}\;\Varid{a}\;\Varid{q}),\constraintfont{(\Conid{RW}\;\Varid{o},\Conid{RW}\;\Varid{q})}\Lolly ()\RLolly\constraintfont{\Conid{RW}\;\Varid{n}})\RLolly\constraintfont{(\Conid{RW}\;\Varid{o},\Conid{RW}\;\Varid{q})}{}\<[E]%
\\
\>[B]{}\Varid{sliceDeep}\;\Varid{as}\;\Varid{slc}\mathrel{=}\Varid{\Conid{Linearly}.return}\;({}\<[E]%
\\
\>[B]{}\hsindent{3}{}\<[3]%
\>[3]{}(\Conid{Ur}{}\<[E]%
\\
\>[3]{}\hsindent{2}{}\<[5]%
\>[5]{}(\Conid{UnsafeMkSlice}\;{}\<[E]%
\\
\>[5]{}\hsindent{2}{}\<[7]%
\>[7]{}\{\mskip1.5mu \Varid{borrowS}\mathrel{=}(\lambda \Varid{rwn'}\;\Varid{i}\to \Conid{Linearly}.\mathbf{do}{}\<[E]%
\\
\>[7]{}\hsindent{4}{}\<[11]%
\>[11]{}(\Conid{Ur}\;\Varid{a},\Varid{release\char95 a})\leftarrow \Varid{borrowS}\;\Varid{as}\;\Varid{rwn'}\;\Varid{i}{}\<[E]%
\\
\>[7]{}\hsindent{4}{}\<[11]%
\>[11]{}(\Conid{Ur}\;(\Varid{l},\anonymous ),\Varid{release\char95 lr})\leftarrow \Varid{slc}\;\Varid{a}{}\<[E]%
\\
\>[7]{}\hsindent{4}{}\<[11]%
\>[11]{}\Varid{\Conid{Linearly}.return}\;(\Conid{Ur}\;\Varid{l},\Conid{Linearly}.\mathbf{do}\;\{\mskip1.5mu \Varid{release\char95 lr};\Varid{release\char95 a}\mskip1.5mu\})),{}\<[E]%
\\
\>[7]{}\hsindent{2}{}\<[9]%
\>[9]{}\Varid{writeS}\mathrel{=}{}\<[E]%
\\
\>[9]{}\hsindent{2}{}\<[11]%
\>[11]{}\Varid{error}\;\text{\ttfamily \char34 Attempt~to~write~to~a~borrowed~space\char34}\mskip1.5mu\}){}\<[E]%
\\
\>[3]{}\hsindent{2}{}\<[5]%
\>[5]{}(\Conid{UnsafeMkSlice}\;{}\<[E]%
\\
\>[5]{}\hsindent{2}{}\<[7]%
\>[7]{}\{\mskip1.5mu \Varid{borrowS}\mathrel{=}(\lambda \Varid{rwn'}\;\Varid{i}\to \Conid{Linearly}.\mathbf{do}{}\<[E]%
\\
\>[7]{}\hsindent{4}{}\<[11]%
\>[11]{}(\Conid{Ur}\;\Varid{a},\Varid{release\char95 a})\leftarrow \Varid{borrowS}\;\Varid{as}\;\Varid{rwn'}\;\Varid{i}{}\<[E]%
\\
\>[7]{}\hsindent{4}{}\<[11]%
\>[11]{}(\Conid{Ur}\;(\anonymous ,\Varid{r}),\Varid{release\char95 lr})\leftarrow \Varid{slc}\;\Varid{a}{}\<[E]%
\\
\>[7]{}\hsindent{4}{}\<[11]%
\>[11]{}\Varid{\Conid{Linearly}.return}\;(\Conid{Ur}\;\Varid{r},\Conid{Linearly}.\mathbf{do}\;\{\mskip1.5mu \Varid{release\char95 lr};\Varid{release\char95 a}\mskip1.5mu\})),{}\<[E]%
\\
\>[7]{}\hsindent{2}{}\<[9]%
\>[9]{}\Varid{writeS}\mathrel{=}{}\<[E]%
\\
\>[9]{}\hsindent{2}{}\<[11]%
\>[11]{}\Varid{error}\;\text{\ttfamily \char34 Attempt~to~write~to~a~borrowed~space\char34}\mskip1.5mu\}){}\<[E]%
\\
\>[B]{}\hsindent{3}{}\<[3]%
\>[3]{}),{}\<[E]%
\\
\>[B]{}\hsindent{3}{}\<[3]%
\>[3]{}(\Conid{Linearly}.\mathbf{do}\;\{\mskip1.5mu \Varid{unsafeConsumeRW}\mathord{@}\Varid{p};{}\<[39]%
\>[39]{}\Varid{unsafeConsumeRW}\mathord{@}\Varid{q};\Varid{\Conid{Linearly}.return}\;()\mskip1.5mu\})){}\<[E]%
\\[\blanklineskip]%
\>[B]{}\Varid{sliceDeepR}\mathbin{::}{}\<[16]%
\>[16]{}\constraintfont{\Conid{Read}\;\Varid{n}}\Lolly \Conid{Slice}\;\Varid{a}\;\Varid{n}\to {}\<[E]%
\\
\>[B]{}\hsindent{3}{}\<[3]%
\>[3]{}\mathbin{::}{}\<[7]%
\>[7]{}\constraintfont{\Conid{Read}\;\Varid{n}}\Lolly \Conid{Slice}\;\Varid{a}\;\Varid{n}\to {}\<[E]%
\\
\>[B]{}\hsindent{3}{}\<[3]%
\>[3]{}(\forall\;\Varid{p}.{}\<[15]%
\>[15]{}\constraintfont{\Conid{Read}\;\Varid{p}}\Lolly \Varid{a}\;\Varid{p}\to {}\<[E]%
\\
\>[15]{}\exists\;\Varid{r}\;\Varid{s}.\constraintfont{(\Conid{Read}\;\Varid{r},\Conid{Read}\;\Varid{s})}\Lolly (\Conid{Ur}\;(\Varid{a}\;\Varid{r},\Varid{a}\;\Varid{s}),\constraintfont{(\Conid{Read}\;\Varid{r},\Conid{Read}\;\Varid{s})}\Lolly ()\RLolly\constraintfont{\Conid{Read}\;\Varid{p}}))\to {}\<[E]%
\\
\>[3]{}\hsindent{4}{}\<[7]%
\>[7]{}\exists\;\Varid{q}.{}\<[18]%
\>[18]{}\constraintfont{(\Conid{Read}\;\Varid{o},\Conid{Read}\;\Varid{q})}\Lolly {}\<[E]%
\\
\>[18]{}(\Conid{Ur}\;(\Conid{Slice}\;\Varid{a}\;\Varid{o},\Conid{Slice}\;\Varid{a}\;\Varid{q}),\constraintfont{(\Conid{Read}\;\Varid{o},\Conid{Read}\;\Varid{q})}\Lolly ()\RLolly\constraintfont{\Conid{Read}\;\Varid{n}})\RLolly\constraintfont{(\Conid{Read}\;\Varid{o},\Conid{Read}\;\Varid{q})}{}\<[E]%
\\
\>[B]{}\Varid{sliceDeepR}\mathrel{=}\mathbin{…}\mbox{\onelinecomment  Mutatis mutandis}{}\<[E]%
\ColumnHook
\end{hscode}\resethooks
The idea is the same as deep borrowing: the slicing function is
suspended until a borrows is called. But when we borrow, we only need one of the
two produce slices, so we drop the other one immediately. Note that the \ensuremath{\Conid{Read}}
and \ensuremath{\Conid{Write}} capabilities for the dropped slice is captured by the corresponding
release function (that is the constraints of \ensuremath{\Varid{release\char95 lr}} are partially discharged).

\subsection{Freezing nested structures}
\label{sec:o1-freeze}

When working with matrices in a functional language the most typical pattern is
to use immutable matrices. Some algorithms, though, like Valiant's algorithm in
\cref{sec:valiant}, need to mutate matrices. In which case a mutable
matrix is created temporarily, then \emph{frozen} into an immutable matrix when
the algorithm is over.

In Haskell, this is embodied by the following (simplified) \textsc{api}.

\begin{hscode}\SaveRestoreHook
\column{B}{@{}>{\hspre}l<{\hspost}@{}}%
\column{E}{@{}>{\hspre}l<{\hspost}@{}}%
\>[B]{}\Varid{new}\mathbin{::}\Conid{Int}\to \Varid{a}\to \Conid{ST}\;\Varid{s}\;(\Conid{MArray}\;\Varid{s}\;\Varid{a}){}\<[E]%
\\
\>[B]{}\Varid{read}\mathbin{::}\Conid{MArray}\;\Varid{s}\;\Varid{a}\to \Conid{Int}\to \Conid{ST}\;\Varid{s}\;\Varid{a}{}\<[E]%
\\
\>[B]{}\Varid{write}\mathbin{::}\Conid{MArray}\;\Varid{s}\;\Varid{a}\to \Conid{Int}\to \Varid{a}\to \Conid{ST}\;\Varid{s}\;(){}\<[E]%
\\
\>[B]{}\Varid{unsafeFreeze}\mathbin{::}\Conid{MArray}\;\Varid{s}\;\Varid{a}\to \Conid{ST}\;\Varid{s}\;(\Conid{Array}\;\Varid{a}){}\<[E]%
\ColumnHook
\end{hscode}\resethooks

For instance the function \ensuremath{\Varid{array}} to build an immutable array from an
associating list looks like this:

\begin{hscode}\SaveRestoreHook
\column{B}{@{}>{\hspre}l<{\hspost}@{}}%
\column{3}{@{}>{\hspre}l<{\hspost}@{}}%
\column{5}{@{}>{\hspre}l<{\hspost}@{}}%
\column{E}{@{}>{\hspre}l<{\hspost}@{}}%
\>[B]{}\Varid{array}\mathbin{::}\Conid{Int}\to [\mskip1.5mu (\Conid{Int},\Varid{a})\mskip1.5mu]\to \Conid{Array}\;\Varid{a}{}\<[E]%
\\
\>[B]{}\Varid{array}\;\Varid{n}\;\Varid{assocs}\mathrel{=}\Varid{runST}\mathbin{\$}\mathbf{do}{}\<[E]%
\\
\>[B]{}\hsindent{3}{}\<[3]%
\>[3]{}\Varid{buffer}\leftarrow \Varid{newArray}\;\Varid{n}\;\bot {}\<[E]%
\\
\>[B]{}\hsindent{3}{}\<[3]%
\>[3]{}\Varid{forM\char95 }\;\Varid{assocs}\mathbin{\$}\lambda (\Varid{i},\Varid{a})\to {}\<[E]%
\\
\>[3]{}\hsindent{2}{}\<[5]%
\>[5]{}\Varid{write}\;\Varid{buffer}\;\Varid{i}\;\Varid{a}{}\<[E]%
\\
\>[B]{}\hsindent{3}{}\<[3]%
\>[3]{}\Varid{unsafeFreeze}\;\Varid{buffer}{}\<[E]%
\ColumnHook
\end{hscode}\resethooks

In \ensuremath{\Varid{array}} we see the pattern, albeit on arrays rather than matrices: create a
mutable \ensuremath{\Varid{buffer}}, mutate it as needed then freeze it into an immutable array.
Freezing is $O(1)$ because the frozen immutable array is the same memory block
as the input mutable array.

This has two shortcomings, however:
\begin{itemize}
  \item The freezing function is unsafe, in that there's nothing preventing you
        from writing \ensuremath{\Varid{buffer}} to buffer after calling \ensuremath{\Varid{unsafeFreeze}}, in which
        case you'd mutate the immutable array which shares its memory.
  \item It doesn't give a $O(1)$ way to freeze arrays of arrays. Such as our
        matrices, as described in \cref{sec:implementing-matrix}. If you have an
        \ensuremath{\Conid{MArray}\;\Varid{s}\;(\Conid{MArray}\;\Varid{s}\;\Varid{a})}, then you need to first map over the outer array
        to freeze the inner arrays.
\end{itemize}

Traditional linear types give a good solution to unsafe
freeze~\citep{LinearHaskell}. With linear constraints, we can solve both
problems handily (we compare linear solutions with and without linear
constraints in \cref{sec:freezing-in-details}). We can provide a
primitive:

\begin{hscode}\SaveRestoreHook
\column{B}{@{}>{\hspre}l<{\hspost}@{}}%
\column{E}{@{}>{\hspre}l<{\hspost}@{}}%
\>[B]{}\Varid{freezeNA}\mathbin{::}\constraintfont{(\Conid{Read}\;\Varid{n},\Conid{Write}\;\Varid{n})}\FatArrow \Conid{NArray}\;\Varid{a}\;\Varid{n}\to ()\RLolly\constraintfont{(\Conid{Ur}\;(\Conid{Read}\;\Varid{n}))}{}\<[E]%
\ColumnHook
\end{hscode}\resethooks

After calling \ensuremath{\Varid{freeze}}, we no longer have a \ensuremath{\Conid{Write}} capability on the array: it
has become immutable. In exchange we gain an unrestricted \ensuremath{\Conid{Read}} access meaning
that the array can safely be shared.

If we are willing to fix a particular nesting, we can even hide the capabilities
altogether and use a traditional Haskell style. For instance for matrices as
defined in \cref{sec:implementing-matrix}:

\begin{hscode}\SaveRestoreHook
\column{B}{@{}>{\hspre}l<{\hspost}@{}}%
\column{3}{@{}>{\hspre}l<{\hspost}@{}}%
\column{E}{@{}>{\hspre}l<{\hspost}@{}}%
\>[B]{}\mathbf{data}\;\Conid{FMatrix}\;\Varid{a}\;\mathbf{where}\;\{\mskip1.5mu \Conid{MkMatrix}\mathbin{::}\constraintfont{\Conid{Read}\;\Varid{n}}\FatArrow \Conid{Matrix}\;\Varid{a}\;\Varid{n}\to \Conid{FMatrix}\;\Varid{a}\mskip1.5mu\}{}\<[E]%
\\[\blanklineskip]%
\>[B]{}\Varid{readF}\mathbin{::}\Conid{FMatrix}\to \Conid{Int}\to \Conid{Int}\to \Conid{Ur}\;\Varid{a}{}\<[E]%
\\
\>[B]{}\Varid{readF}\;(\Conid{MkMatrix}\;\Varid{m})\;\Varid{row\char95 num}\;\Varid{col\char95 num}\mathrel{=}\mathbf{do}{}\<[E]%
\\
\>[B]{}\hsindent{3}{}\<[3]%
\>[3]{}(\Conid{Ur}\;\Varid{col},\Varid{release})\leftarrow \Varid{borrowNA}\;\Varid{m}\;\Varid{col\char95 num}{}\<[E]%
\\
\>[B]{}\hsindent{3}{}\<[3]%
\>[3]{}(\Conid{Ur}\;(\Conid{Val}\;\Varid{a}),\Varid{release\char95 col})\leftarrow \Varid{borrowNA}\;\Varid{col}\;\Varid{row\char95 num}{}\<[E]%
\\
\>[B]{}\hsindent{3}{}\<[3]%
\>[3]{}\Varid{release\char95 col}{}\<[E]%
\\
\>[B]{}\hsindent{3}{}\<[3]%
\>[3]{}\Varid{release}{}\<[E]%
\\
\>[B]{}\hsindent{3}{}\<[3]%
\>[3]{}\Conid{Ur}\;\Varid{a}{}\<[E]%
\ColumnHook
\end{hscode}\resethooks

\subsection{Beyond arrays: splay Trees}
\label{sec:splay-trees}

Another application of in-place updates are \emph{splay trees}. A
splay tree acts like a binary search tree, but at every access, the
tree structure is modified so that the recently accessed elements are
moved towards the root of the tree. This way, recently accessed
elements are faster to reach in subsequent lookups.

A tree associated with a region \ensuremath{\Varid{n}} is defined as follows:
\begin{hscode}\SaveRestoreHook
\column{B}{@{}>{\hspre}l<{\hspost}@{}}%
\column{3}{@{}>{\hspre}l<{\hspost}@{}}%
\column{E}{@{}>{\hspre}l<{\hspost}@{}}%
\>[3]{}\mathbf{data}\;\Conid{Node}\;\Varid{n}\mathrel{=}\Conid{Node}\;(\Conid{Ref}\;\Conid{Node}\;\Varid{n})\;\Conid{Int}\;(\Conid{Ref}\;\Conid{Node}\;\Varid{n})\mid \Conid{Leaf}{}\<[E]%
\ColumnHook
\end{hscode}\resethooks
The noteworthy feature is that each child is accessible via an
updateable reference. To access a tree, we need a borrowing function. We'll
assume a borrowing function which decomposes the structure as follows:\begin{hscode}\SaveRestoreHook
\column{B}{@{}>{\hspre}l<{\hspost}@{}}%
\column{3}{@{}>{\hspre}l<{\hspost}@{}}%
\column{14}{@{}>{\hspre}l<{\hspost}@{}}%
\column{25}{@{}>{\hspre}l<{\hspost}@{}}%
\column{E}{@{}>{\hspre}l<{\hspost}@{}}%
\>[3]{}\Varid{splitTree}{}\<[14]%
\>[14]{}\mathbin{::}\constraintfont{\Conid{RW}\;\Varid{n}}\FatArrow \Conid{Ref}\;\Conid{Node}\;\Varid{n}{}\<[E]%
\\
\>[14]{}\to \mathbin{∃}\Varid{p}\;\Varid{q}.{}\<[25]%
\>[25]{}(\Conid{Ref}\;\Conid{Node}\;\Varid{p},\Conid{Int},\Conid{Ref}\;\Conid{Node}\;\Varid{q},{}\<[E]%
\\
\>[25]{}\constraintfont{(\Conid{RW}\;\Varid{p},\Conid{RW}\;\Varid{q})}\Lolly ()\RLolly\constraintfont{\Conid{RW}\;\Varid{n}})\RLolly\constraintfont{(\Conid{RW}\;\Varid{p},\Conid{RW}\;\Varid{q})}{}\<[E]%
\ColumnHook
\end{hscode}\resethooks
We will regard the type \ensuremath{\Conid{Ref}} as being defined merely as an array whose length
is known to be one:\begin{hscode}\SaveRestoreHook
\column{B}{@{}>{\hspre}l<{\hspost}@{}}%
\column{3}{@{}>{\hspre}l<{\hspost}@{}}%
\column{13}{@{}>{\hspre}l<{\hspost}@{}}%
\column{E}{@{}>{\hspre}l<{\hspost}@{}}%
\>[3]{}\mathbf{type}\;\Conid{Ref}{}\<[13]%
\>[13]{}\mathbin{::}(\Conid{Reg}\to \Conid{Type})\to \Conid{Reg}\to \Conid{Type}{}\<[E]%
\\
\>[3]{}\mathbf{newtype}\;\Conid{Ref}\;\Varid{f}\;\Varid{n}\mathrel{=}\Conid{UnsafeMkRef}\;(\Conid{NArray}\;\Varid{f}\;\Varid{n}){}\<[E]%
\ColumnHook
\end{hscode}\resethooks
In this section we will be using the following functions on references:\begin{hscode}\SaveRestoreHook
\column{B}{@{}>{\hspre}l<{\hspost}@{}}%
\column{3}{@{}>{\hspre}l<{\hspost}@{}}%
\column{13}{@{}>{\hspre}l<{\hspost}@{}}%
\column{E}{@{}>{\hspre}l<{\hspost}@{}}%
\>[3]{}\Varid{newRef}{}\<[13]%
\>[13]{}\mathbin{::}\constraintfont{(\constraintfont{\Conid{Linearly}},\Conid{RW}\;\Varid{p})}\Lolly \Varid{f}\;\Varid{p}\to \exists\;\Varid{n}.\Conid{Ref}\;\Varid{f}\;\Varid{n}\RLolly\constraintfont{\Conid{RW}\;\Varid{n}}{}\<[E]%
\\
\>[3]{}\Varid{swapRefs}{}\<[13]%
\>[13]{}\mathbin{::}\constraintfont{(\Conid{RW}\;\Varid{p},\Conid{RW}\;\Varid{q})}\Lolly \Conid{Ref}\;\Varid{f}\;\Varid{p}\to \Conid{Ref}\;\Varid{f}\;\Varid{q}\to ()\RLolly\constraintfont{(\Conid{RW}\;\Varid{p},\Conid{RW}\;\Varid{q})}{}\<[E]%
\\
\>[3]{}\Varid{freeRefs}{}\<[13]%
\>[13]{}\mathbin{::}\constraintfont{\Conid{RW}\;\Varid{n}}\Lolly \Conid{Ref}\;\Varid{f}\;\Varid{n}\to (){}\<[E]%
\ColumnHook
\end{hscode}\resethooks
The implementation of \ensuremath{\Varid{newRef}} and \ensuremath{\Varid{freeRefs}} are straightforward wrappers
around \ensuremath{\Varid{newNArray}} (\cref{fig:narray-api}) and \ensuremath{\Varid{freezeNA}} (\cref{sec:o1-freeze}) respectively.
On the other hand, \ensuremath{\Varid{swapRefs}} is new: \ensuremath{\Varid{swapRefs}} exchanges the contents of two
references. It can't be implemented in terms of borrows and writes, because
writes consume a \ensuremath{\Conid{RW}\;\Varid{p}} capability but borrows require the same \ensuremath{\Conid{RW}\;\Varid{p}} to be
passed to their release function.

When accessing a node in the splay tree, it is brought towards the
root by repeated use of \ensuremath{\Varid{rotate}} operations which effectively exchange
a node with either its left or right child. There follows a schematic
representation of the rotation with the left child:
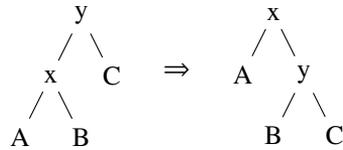
\begin{center}
  \begin{tikzpicture}[baseline=(current bounding box.center),
    level distance=.8cm,
    level 1/.style={sibling distance=.8cm},
    level 2/.style={sibling distance=.8cm}]
    \node {y}
    child {node {x}
      child {node {A}}
      child {node {B}}
    }
    child {node {C}};
  \end{tikzpicture}
  \(\quad \Rightarrow \quad\)
  \begin{tikzpicture}[baseline=(current bounding box.center),
    level distance=.8cm,
    level 1/.style={sibling distance=.8cm},
    level 2/.style={sibling distance=.8cm}]
    \node {x}
    child {node {A}}
    child {node {y}
      child {node {B}}
      child {node {C}}
    };
  \end{tikzpicture}
\end{center}

Using our ownership system, the implementation of the above rotation would
be:
\begin{hscode}\SaveRestoreHook
\column{B}{@{}>{\hspre}l<{\hspost}@{}}%
\column{3}{@{}>{\hspre}l<{\hspost}@{}}%
\column{E}{@{}>{\hspre}l<{\hspost}@{}}%
\>[B]{}\Varid{rotate}\mathbin{::}\constraintfont{(\Conid{Linearly},\Conid{RW}\;\Varid{n})}\Lolly \Conid{Ref}\;\Conid{Node}\;\Varid{n}\to ()\RLolly\constraintfont{\Conid{RW}\;\Varid{n}}{}\<[E]%
\\
\>[B]{}\Varid{rotate}\;\Varid{root}\mathrel{=}\Conid{Linearly}.\mathbf{do}{}\<[E]%
\\
\>[B]{}\hsindent{3}{}\<[3]%
\>[3]{}\mbox{\onelinecomment  create a new temporary tree for the exchange}{}\<[E]%
\\
\>[B]{}\hsindent{3}{}\<[3]%
\>[3]{}\Varid{newRoot}\leftarrow \Varid{newRef}\;\Conid{Leaf}{}\<[E]%
\\[\blanklineskip]%
\>[B]{}\hsindent{3}{}\<[3]%
\>[3]{}\mbox{\onelinecomment  Move $x$'s subtree to the new root}{}\<[E]%
\\
\>[B]{}\hsindent{3}{}\<[3]%
\>[3]{}(\Varid{yl},\Varid{y},\Varid{yr},\Varid{release\char95 root})\leftarrow \Varid{splitTree}\;\Varid{root}{}\<[E]%
\\
\>[B]{}\hsindent{3}{}\<[3]%
\>[3]{}\Varid{swapRefs}\;\Varid{newRoot}\;\Varid{l}{}\<[E]%
\\[\blanklineskip]%
\>[B]{}\hsindent{3}{}\<[3]%
\>[3]{}\mbox{\onelinecomment  Move $B$ to $y$'s left child}{}\<[E]%
\\
\>[B]{}\hsindent{3}{}\<[3]%
\>[3]{}(\Varid{xl},\Varid{x},\Varid{xr},\Varid{release\char95 new})\leftarrow \Varid{splitTree}\;\Varid{newRoot}{}\<[E]%
\\
\>[B]{}\hsindent{3}{}\<[3]%
\>[3]{}\Varid{swapRefs}\;\Varid{xr}\;\Varid{yl}{}\<[E]%
\\[\blanklineskip]%
\>[B]{}\hsindent{3}{}\<[3]%
\>[3]{}\mbox{\onelinecomment  Move $y$'s subtree to $x$'s right child}{}\<[E]%
\\
\>[B]{}\hsindent{3}{}\<[3]%
\>[3]{}\Varid{release\char95 root}{}\<[E]%
\\
\>[B]{}\hsindent{3}{}\<[3]%
\>[3]{}\Varid{swapRefs}\;\Varid{xr}\;\Varid{root}{}\<[E]%
\\[\blanklineskip]%
\>[B]{}\hsindent{3}{}\<[3]%
\>[3]{}\mbox{\onelinecomment  cleanup}{}\<[E]%
\\
\>[B]{}\hsindent{3}{}\<[3]%
\>[3]{}\Varid{release\char95 new}{}\<[E]%
\\
\>[B]{}\hsindent{3}{}\<[3]%
\>[3]{}\Varid{swapRefs}\;\Varid{root}\;\Varid{newRoot}{}\<[E]%
\\
\>[B]{}\hsindent{3}{}\<[3]%
\>[3]{}\Varid{freeRef}\;\Varid{newRoot}{}\<[E]%
\ColumnHook
\end{hscode}\resethooks

The above code mixes aspects of functional and imperative styles. Like in a
functional language, we're explicitly pattern matching on nodes structures when
accessing them, using \ensuremath{\Varid{splitTree}}. Like in an imperative programing language, we
perform assignments (via \ensuremath{\Varid{swapRefs}}) rather than allocating new data.

\subsection{Non-lexical lifetimes}
\label{sec:nl-lifetimes}
In early versions of the Rust programming language, borrows were, for all
intents and purposes, bound to lexical scopes. This turned out to be too coarse
grained, rejecting natural sound programs. This became known as the problem of
non-lexical lifetimes (see \emph{e.g.}~\cite{non-lexical-lifetime-blog}). Rust
has since implemented non-lexical lifetime, however in capability-based systems,
including capabilities encoded as linear constraints, non-lexical lifetimes are
for free.

A typical example of non-lexical lifetimes consists in the following: we have a
map attaching references to keys, if a key is present we want to update the
corresponding reference, otherwise we want to create a new reference at that
key.

We can model this problem with the following \textsc{api} (where \ensuremath{\Conid{Ref}} is the
type of mutable references from \cref{sec:splay-trees}):

\begin{hscode}\SaveRestoreHook
\column{B}{@{}>{\hspre}l<{\hspost}@{}}%
\column{9}{@{}>{\hspre}l<{\hspost}@{}}%
\column{E}{@{}>{\hspre}l<{\hspost}@{}}%
\>[B]{}\mathbf{type}\;\Conid{Map}\mathbin{::}\Conid{Type}\to \Conid{Reg}\to \Conid{Type}{}\<[E]%
\\
\>[B]{}\Varid{get}\mathbin{::}{}\<[9]%
\>[9]{}\constraintfont{\Conid{RW}\;\Varid{s}}\Lolly \Conid{Map}\;\Varid{v}\;\Varid{s}\to \Conid{Int}\to {}\<[E]%
\\
\>[9]{}\Conid{Either}\;(()\RLolly\constraintfont{\Conid{RW}\;\Varid{s}})\;(\exists\;\Varid{r}.(\Conid{Ref}\;(\Conid{Val}\;\Varid{v})\;\Varid{r},\constraintfont{\Conid{RW}\;\Varid{r}}\Lolly ()\RLolly\constraintfont{\Conid{RW}\;\Varid{s}})\RLolly\constraintfont{\Conid{RW}\;\Varid{r}}){}\<[E]%
\\
\>[B]{}\Varid{add}\mathbin{::}\constraintfont{\Conid{RW}\;\Varid{s}}\Lolly \Conid{Map}\;\Varid{v}\;\Varid{s}\to \Conid{Int}\to \Varid{v}\to ()\RLolly\constraintfont{\Conid{RW}\;\Varid{s}}{}\<[E]%
\ColumnHook
\end{hscode}\resethooks

The \ensuremath{\Varid{get}} function is the most interesting: if the key (here an \ensuremath{\Conid{Int}} is
absent), \ensuremath{\Varid{get}} returns the read and write capabilities to the map, otherwise, it
borrows the corresponding cell. With this, we can easily write a \ensuremath{\Varid{set}} function
with the above specification.

\begin{hscode}\SaveRestoreHook
\column{B}{@{}>{\hspre}l<{\hspost}@{}}%
\column{3}{@{}>{\hspre}l<{\hspost}@{}}%
\column{E}{@{}>{\hspre}l<{\hspost}@{}}%
\>[B]{}\Varid{set}\mathbin{::}\constraintfont{\Conid{RW}\;\Varid{s}}\Lolly \Conid{Map}\;\Varid{v}\;\Varid{s}\to \Conid{Int}\to \Varid{v}\to ()\RLolly\constraintfont{\Conid{RW}\;\Varid{s}}{}\<[E]%
\\
\>[B]{}\Varid{set}\;\Varid{vs}\;\Varid{i}\;\Varid{v}\mathrel{=}\mathbf{case}\;\Varid{get}\;\Varid{vs}\;\Varid{i}\;\mathbf{of}{}\<[E]%
\\
\>[B]{}\hsindent{3}{}\<[3]%
\>[3]{}\Conid{Right}\;\Varid{ev}\to \Conid{Linearly}.\mathbf{do}\;\{\mskip1.5mu (\Varid{rv},\Varid{release})\leftarrow \Varid{ev};\Varid{write}\;\Varid{rv}\;\mathrm{0}\;\Varid{v};\Varid{release}\mskip1.5mu\}{}\<[E]%
\\
\>[B]{}\hsindent{3}{}\<[3]%
\>[3]{}\Conid{Left}\;\Varid{ev}\to \Conid{Linearly}.\mathbf{do}\;\{\mskip1.5mu \Varid{ev};\Varid{add}\;\Varid{vs}\;\Varid{i}\;\Varid{v}\mskip1.5mu\}{}\<[E]%
\ColumnHook
\end{hscode}\resethooks

The difficulty, when lifetimes are lexical, is that the borrow from \ensuremath{\Varid{vs}}, which
comes from the \ensuremath{\mathbf{case}} scrutinee, remains in the \ensuremath{\Conid{Right}} branch (hence \ensuremath{\Varid{vs}} is
inaccessible here), but is written to directly in the \ensuremath{\Conid{Left}} branch. There is no
lexical scope which encompasses the scrutinee and the \ensuremath{\Conid{Right}} branch but not the
\ensuremath{\Conid{Left}} branch. With linear constraints and capabilities, it is explicit in the
type of \ensuremath{\Varid{get}} that the \ensuremath{\Conid{Left}} branch doesn't borrow, so we get this non-lexical
pattern at no cost.

\subsection{Self-referential borrows}
\label{sec:self-borrow}

Another, this one current, pattern which Rust doesn't support is what has been
known as ``self-referential borrows'' by the community. The problem is as
follow: borrowed value \ensuremath{\Varid{b}} can't outlive the (owned) value \ensuremath{\Varid{o}} it borrows from, but as
long as \ensuremath{\Varid{o}} and \ensuremath{\Varid{b}} are paired together in a record, we know that \ensuremath{\Varid{o}} doesn't
die, so \ensuremath{\Varid{b}} remains valid. A typical example is an array \ensuremath{\Varid{o}} paired with a
slice of \ensuremath{\Varid{o}}. Such a pair shouldn't have restrictions on its lifetime, but Rust
can't express the relationship between the fields.

In our presentation, this \emph{can't-outlive} requirement is enforced by the release
operator: as the release operator is linear we must use the release operator to
return, say, read and write capabilities to \ensuremath{\Varid{o}}, which must then be disposed of
some way. Until that happens, \ensuremath{\Varid{o}} lives. Making a self-referential borrow is as
easy as storing not only the original value and the borrow, but also the release
operator.

\begin{hscode}\SaveRestoreHook
\column{B}{@{}>{\hspre}l<{\hspost}@{}}%
\column{3}{@{}>{\hspre}l<{\hspost}@{}}%
\column{E}{@{}>{\hspre}l<{\hspost}@{}}%
\>[B]{}\mathbf{data}\;\Conid{Buffer}\;\Varid{s}\;\Varid{a}\;\Varid{n}\mathrel{=}{}\<[E]%
\\
\>[B]{}\hsindent{3}{}\<[3]%
\>[3]{}\Conid{MkBuffer}\mathbin{::}\Conid{NArray}\;\Varid{a}\;\Varid{s}\to \Conid{Slice}\;\Varid{a}\;\Varid{n}\to (\constraintfont{\Conid{RW}\;\Varid{n}}\Lolly ()\RLolly\constraintfont{\Conid{RW}\;\Varid{s}})⊸\Conid{Buffer}\;\Varid{s}\;\Varid{a}\;\Varid{n}{}\<[E]%
\ColumnHook
\end{hscode}\resethooks
It may be convenient to existentially quantify over the region argument \ensuremath{\Varid{s}} of
the \ensuremath{\Conid{NArray}}, making a \ensuremath{\Conid{Buffer}} look very similar to a slice.

Note that buffers define in this manner are always going to be linear values,
rather than unrestricted value with linear constraints like our arrays and
matrices so far. Here both the value and the constraints are linear. This is
because a \ensuremath{\Conid{Buffer}} contains a release operator, which is always linear. It may
make \ensuremath{\Conid{Buffer}} less convenient to use than it should, we revisit this aspect
in~\cref{sec:unrestricted-release}.

\subsection{Typestates: capabilities in monadic computations}
\label{sec:monad-capability}

So far, we have used linear constraints in conjunction with pure
algorithms, but they are just as effective to annotate effectful code.
To be able to do this, we must use a linear monad. For instance linear
\ensuremath{\Conid{IO}}, whose \ensuremath{\Varid{bind}} and \ensuremath{\Varid{return}} have the following signatures:

\begin{hscode}\SaveRestoreHook
\column{B}{@{}>{\hspre}l<{\hspost}@{}}%
\column{E}{@{}>{\hspre}l<{\hspost}@{}}%
\>[B]{}\Varid{return}\mathbin{::}\Varid{a}\mathbin{⊸}\Conid{L}\;\Varid{a}{}\<[E]%
\\
\>[B]{}\Varid{bind}\mathbin{::}\Conid{L}\;\Varid{a}\mathbin{⊸}(\Varid{a}\mathbin{⊸}\Conid{L}\;\Varid{b})\mathbin{⊸}\Conid{L}\;\Varid{b}{}\<[E]%
\ColumnHook
\end{hscode}\resethooks

In what follows, we illustrate effectful use of linear constraints
with a modified the API for a ATM machine presented by
\citet{DBLP:conf/ecoop/Brady21}. This API has three states, and
function calls can effect a transition between those, as schematized
below:
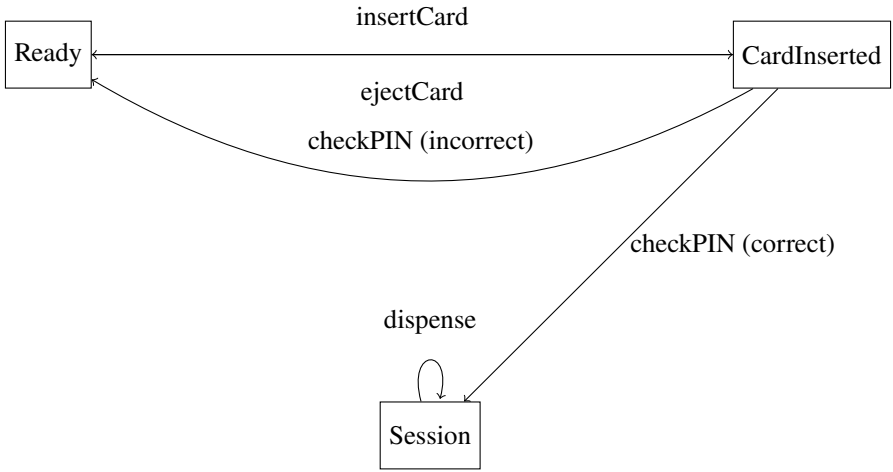
\begin{center}
  \begin{tikzpicture}[minimum width=2cm, minimum height=1cm,
    every state/.style={rectangle}]
    
    \draw (0,0) node (ready) [state] {Ready};
    \draw (10,0) node (cardinserted) [state] {CardInserted};
    \draw (5,-5) node (session) [state] {Session};

    \draw[->] (ready) --  node[above] {insertCard} (cardinserted);
    \draw[->] (cardinserted) to[bend left] node[above] {checkPIN (incorrect)} (ready);
    \draw[->] (cardinserted) -- node[right] {checkPIN (correct)} (session);
    \draw[->] (session) edge[loop above] node {dispense} ();
    \draw[->] (cardinserted) -- node[below] {ejectCard} (ready);
  \end{tikzpicture}
\end{center}
While Brady's API uses linear types, we can capture the (effectful)
transitions using linear constraints, as follows:
\begin{hscode}\SaveRestoreHook
\column{B}{@{}>{\hspre}l<{\hspost}@{}}%
\column{14}{@{}>{\hspre}l<{\hspost}@{}}%
\column{40}{@{}>{\hspre}l<{\hspost}@{}}%
\column{60}{@{}>{\hspre}l<{\hspost}@{}}%
\column{E}{@{}>{\hspre}l<{\hspost}@{}}%
\>[B]{}\Varid{initATM}\mathbin{::}\Conid{L}\;(\mathbin{∃}\Varid{s}.\Conid{ATM}\;\Varid{s}\RLolly\constraintfont{\Conid{Ready}\;\Varid{s}}){}\<[E]%
\\
\>[B]{}\Varid{shutDown}\mathbin{::}\constraintfont{\Conid{Ready}\;\Varid{s}}\Lolly \Conid{ATM}\;\Varid{s}\to \Conid{L}\;(){}\<[E]%
\\
\>[B]{}\Varid{insertCard}\mathbin{::}\constraintfont{\Conid{Ready}\;\Varid{s}}\Lolly \Conid{ATM}\;\Varid{s}\to \Conid{L}\;(()\RLolly\constraintfont{\Conid{CardInserted}\;\Varid{s}}){}\<[E]%
\\
\>[B]{}\Varid{checkPIN}\mathbin{::}{}\<[14]%
\>[14]{}\constraintfont{\Conid{CardInserted}\;\Varid{s}}\Lolly {}\<[E]%
\\
\>[14]{}\Conid{ATM}\;\Varid{s}\to \Conid{Int}\to \Conid{L}\;(\Conid{Either}\;(()\RLolly\Conid{Session}\;\Varid{s})\;{}\<[60]%
\>[60]{}(()\RLolly\constraintfont{\Conid{CardInserted}})){}\<[E]%
\\
\>[B]{}\Varid{dispense}\mathbin{::}\constraintfont{\Conid{Session}\;\Varid{s}}\Lolly {}\<[40]%
\>[40]{}\Conid{ATM}\;\Varid{s}\to \Conid{L}\;(()\RLolly\constraintfont{\Conid{Session}\;\Varid{s}}){}\<[E]%
\\
\>[B]{}\Varid{ejectCard}\mathbin{::}\constraintfont{\Conid{CardInserted}\;\Varid{s}}\Lolly \Conid{L}\;(()\RLolly\constraintfont{\Conid{Ready}\;\Varid{s}}){}\<[E]%
\\
\>[B]{}\Varid{terminateSession}\mathbin{::}\constraintfont{\Conid{Session}\;\Varid{s}}\Lolly \Conid{L}\;(()\RLolly\constraintfont{\Conid{CardInserted}\;\Varid{s}}){}\<[E]%
\\
\>[B]{}\Varid{message}\mathbin{::}\Conid{ATM}\;\Varid{s}\to \Conid{String}\to \Conid{L}\;(){}\<[E]%
\ColumnHook
\end{hscode}\resethooks
One possible use of the above API is the following:
\begin{hscode}\SaveRestoreHook
\column{B}{@{}>{\hspre}l<{\hspost}@{}}%
\column{14}{@{}>{\hspre}l<{\hspost}@{}}%
\column{16}{@{}>{\hspre}l<{\hspost}@{}}%
\column{E}{@{}>{\hspre}l<{\hspost}@{}}%
\>[B]{}\Varid{runATM}\mathbin{::}\Conid{L}\;(){}\<[E]%
\\
\>[B]{}\Varid{runATM}\mathrel{=}\mathbf{do}\;{}\<[14]%
\>[14]{}\Varid{m}\leftarrow \Varid{initATM}{}\<[E]%
\\
\>[14]{}\Varid{insertCard}\;\Varid{m}{}\<[E]%
\\
\>[14]{}\Varid{ok}\leftarrow \Varid{checkPIN}\;\Varid{m}\;\mathrm{1234}{}\<[E]%
\\
\>[14]{}\Varid{message}\;\Varid{m}\;\text{\ttfamily \char34 Checking~PIN\char34}{}\<[E]%
\\
\>[14]{}\mathbf{case}\;\Varid{ok}\;\mathbf{of}{}\<[E]%
\\
\>[14]{}\hsindent{2}{}\<[16]%
\>[16]{}\Conid{Left}\;()\to \Varid{dispense}\;\Varid{m}{}\<[E]%
\\
\>[14]{}\hsindent{2}{}\<[16]%
\>[16]{}\Conid{Right}\;()\to \Varid{terminateSession}\;\Varid{m}{}\<[E]%
\\
\>[14]{}\Varid{ejectCard}\;\Varid{m}{}\<[E]%
\\
\>[14]{}\Varid{shutDown}\;\Varid{m}{}\<[E]%
\ColumnHook
\end{hscode}\resethooks
Compared to the (plain) use of linear types of Brady, linear
constraints allows one to use a single value of type \ensuremath{\Conid{ATM}}, instead of
consuming and recreating for each call.  Note that for Brady,
\ensuremath{\Varid{ejectCard}} can work both in \ensuremath{\constraintfont{\Conid{Session}}} and \ensuremath{\constraintfont{\Conid{CardInserted}}} state. One
idea would be add an extra state \ensuremath{\constraintfont{\Conid{Ejectable}}} which is the union of
those.  This is possible for classes in Haskell; i.e. we can have
instances \ensuremath{\constraintfont{\Conid{Session}}\FatArrow \constraintfont{\Conid{Ejectable}}} and \ensuremath{\constraintfont{\Conid{CardInserted}}\FatArrow \constraintfont{\Conid{Ejectable}}}, but
such instances won't be used by our system to resolve constraints.
\section{A qualified type system for linear constraints}
\label{sec:qualified-type-system}

We now present our design for a qualified type system~\cite{QualifiedTypes} that supports
linear constraints. Our design, based on the work of \citet{OutsideIn}, is compatible with Haskell and \textsc{ghc}.

\subsection{Multiplicities}

The first concept to define formally is that of \emph{multiplicity}. Multiplicity let us generalise over the linear (\(1\)) and unrestricted ($\omega$) cases.
In general, given a set of multiplicities, the desired
(sub)structural rules can be obtained by endowing multiplicities with the appropriate
semiring structure~\citep{abel_unified_2020}. In
this paper, we use the same multiplicity structure as Linear Haskell:\footnote{Even though linear Haskell additionally supports multiplicity
polymorphism, we do not support multiplicity polymorphism on
constraint arguments.  Linear Haskell takes advantage of multiplicity
polymorphism to avoid duplication of higher-order functions. The
prototypical example is \ensuremath{\Varid{map}\mathbin{::}(\Varid{a}\;\mathop{\to_{\multiplicityfont{m}}}\;\Varid{b})\to [\mskip1.5mu \Varid{a}\mskip1.5mu]\;\mathop{\to_{\multiplicityfont{m}}}\;[\mskip1.5mu \Varid{b}\mskip1.5mu]}, where \ensuremath{\mathop{\to_{\multiplicityfont{m}}}} is the notation for a function arrow of
multiplicity \ensuremath{\multiplicityfont{\Varid{m}}}.  First-order functions, on the other
hand, do not need multiplicity polymorphism, because linear functions
can be $\eta$-expanded into unrestricted functions as explained in
\cref{sec:linear-types}. Higher-order functions whose arguments are
themselves constrained functions are rare, so we do not yet see the
need to extend multiplicity polymorphism to apply to constraints.
Futhermore, it is not clear how to extend the constraint solver of
\cref{sec:constraint-solver} to support multiplicity-polymorphic
constraints.}${}^,$\footnote{Here and in the rest of the paper we adopt the convention that
equations defining a function by pattern matching are marked with a
$\left\{\right.$ to their left.}
$$
\begin{array}{lcll}
    \multiplicityfont{ \pi }  ,   \multiplicityfont{ \rho }   & \bnfeq &   \multiplicityfont{ \ottsym{1} }   \bnfor   \multiplicityfont{ \omega }   & \text{Multiplicities}
\end{array}
$$
$$
\begin{array}{c@{\qquad\qquad}c}
\left\{
  \begin{array}{lcl}
      \multiplicityfont{ \pi  \ottsym{+}  \rho }   & = &   \multiplicityfont{ \omega }  
  \end{array}
\right.
&
\left\{
  \begin{array}{lcl}
      \multiplicityfont{  \ottsym{1} {⋅} \pi  }   & = &   \multiplicityfont{ \pi }   \\
      \multiplicityfont{  \omega {⋅} \pi  }   & = &   \multiplicityfont{ \omega }  
  \end{array}
\right.
\end{array}
$$

\subsection{Simple Constraints and Entailment}
\label{sec:constraint-domain}

We call constraints such as
\ensuremath{\constraintfont{\Conid{Read}\;\Varid{n}}} or \ensuremath{\constraintfont{\Conid{Write}\;\Varid{n}}} \emph{atomic
  constraints}. The set of atomic constraints is a parameter of our
qualified type system.

\begin{definition}[Atomic constraints]
  The qualified type system is parameterised by a set, whose elements
  are called \emph{atomic constraints}. We use the variable $  \constraintfont{ \ottnt{q} }  $
  to denote atomic constraints.
\end{definition}
Atomic constraints are assembled into \emph{simple constraints}
$  \constraintfont{ \ottnt{Q} }  $, which play the hybrid role of constraint contexts and
(linear) logic formulae. We define a simple constraint to be a pair of
a set of unrestricted constraints $  \constraintfont{ \ottnt{U} }  $ and a multiset of
linear constraints $  \constraintfont{ \ottnt{L} }  $.  The linear constraints must be
stored in a multiset, because assuming the same constraint twice is
distinct from assuming it only once.
\begin{definition}[Simple constraints]
$$
\begin{array}{rcll}
    \constraintfont{ \ottnt{U} }   & \bnfeq & \ldots & \text{set of atomic constraints $  \constraintfont{ \ottnt{q} }  $} \\
    \constraintfont{ \ottnt{L} }   & \bnfeq & \ldots & \text{multiset of atomic constraints $  \constraintfont{ \ottnt{q} }  $} \\
    \constraintfont{ \ottnt{Q} }   & \bnfeq &   \constraintfont{ \ottsym{(}  \ottnt{U}  \ottsym{,}  \ottnt{L}  \ottsym{)} }   & \text{simple constraints}
\end{array}
$$
\end{definition}

We won't manipulate simple constraints directly, but rather
\textit{via} the following constructions:
\begin{description}
\item[Scaled atomic constraints] $  \constraintfont{   \multiplicityfont{ \pi }  \scale \ottnt{q}  }  $ is a simple constraint,
  where $  \multiplicityfont{ \pi }  $ specifies whether $  \constraintfont{ \ottnt{q} }  $ is to be used linearly
  or not.
\item[Conjunction] Two simple constraints can be paired up
  $  \constraintfont{ \ottnt{Q_{{\mathrm{1}}}}  \qtensor  \ottnt{Q_{{\mathrm{2}}}} }  $. Semantically, this corresponds to the multiplicative
  conjunction of linear logic. Tensor products represent pairs of
  constraints such as \ensuremath{\constraintfont{(\Conid{Read}\;\Varid{n},\Conid{Write}\;\Varid{n})}} from Haskell.
\item[Empty conjunction] The empty conjunction, $  \constraintfont{  \mathbf{\varepsilon}  }  $, is used
  to represent functions which don't require any constraints. It is
  the neutral element of conjunction.
\end{description}

The above constructions admit straightforward definitions:
\begin{definition}
  
$$
\begin{array}{ccc}
  \begin{array}{r@{\;}c@{\;}l}
      \constraintfont{  \mathbf{\varepsilon}  }   &=&   \constraintfont{ \ottsym{(}  \emptyset  \ottsym{,}  \emptyset  \ottsym{)} }  
  \end{array}
  &
    \left\{
    \begin{array}{r@{\;}c@{\;}l}
        \constraintfont{   \multiplicityfont{ \ottsym{1} }  \scale \ottnt{q}  }   &=&   \constraintfont{ \ottsym{(}  \emptyset  \ottsym{,}  \ottnt{q}  \ottsym{)} }   \\
        \constraintfont{   \multiplicityfont{ \omega }  \scale \ottnt{q}  }   &=&   \constraintfont{ \ottsym{(}  \ottnt{q}  \ottsym{,}  \emptyset  \ottsym{)} }  
    \end{array}
                        \right.
  &
    \begin{array}{r@{\;}c@{\;}l}
        \constraintfont{ \ottsym{(}  \ottnt{U}_{{\mathrm{1}}}  \ottsym{,}  \ottnt{L}_{{\mathrm{1}}}  \ottsym{)}  \qtensor  \ottsym{(}  \ottnt{U}_{{\mathrm{2}}}  \ottsym{,}  \ottnt{L}_{{\mathrm{2}}}  \ottsym{)} }   &=&   \constraintfont{ \ottsym{(}   \ottnt{U}_{{\mathrm{1}}}  \cup  \ottnt{U}_{{\mathrm{2}}}   \ottsym{,}   \ottnt{L}_{{\mathrm{1}}}  \uplus  \ottnt{L}_{{\mathrm{2}}}   \ottsym{)} }  
    \end{array}
\end{array}
$$
\end{definition}

They satisfy the following laws:
\begin{align*}
    \constraintfont{ \ottnt{Q_{{\mathrm{1}}}}  \qtensor  \ottnt{Q_{{\mathrm{2}}}} }   & =   \constraintfont{ \ottnt{Q_{{\mathrm{2}}}}  \qtensor  \ottnt{Q_{{\mathrm{1}}}} }   \\
    \constraintfont{ \ottsym{(}  \ottnt{Q_{{\mathrm{1}}}}  \qtensor  \ottnt{Q_{{\mathrm{2}}}}  \ottsym{)}  \qtensor  \ottnt{Q_{{\mathrm{3}}}} }   & =   \constraintfont{ \ottnt{Q_{{\mathrm{1}}}}  \qtensor  \ottsym{(}  \ottnt{Q_{{\mathrm{2}}}}  \qtensor  \ottnt{Q_{{\mathrm{3}}}}  \ottsym{)} }   \\
    \constraintfont{   \multiplicityfont{ \omega }  \scale \ottnt{Q}   \qtensor    \multiplicityfont{ \omega }  \scale \ottnt{Q}  }   & =   \constraintfont{   \multiplicityfont{ \omega }  \scale \ottnt{Q}  }   \\
    \constraintfont{ \ottnt{Q}  \qtensor   \mathbf{\varepsilon}  }   & =   \constraintfont{ \ottnt{Q} }  \\
    \constraintfont{   \multiplicityfont{ \pi }  \scale \ottsym{(}    \multiplicityfont{ \rho }  \scale \ottnt{q}   \ottsym{)}  }   &=   \constraintfont{   \multiplicityfont{ \ottsym{(}   \pi {⋅} \rho   \ottsym{)} }  \scale \ottnt{q}  }  
\end{align*}

We overload the scaling notation to apply not just to atomic, but also simple constraints, and define
$  \constraintfont{   \multiplicityfont{ \pi }  \scale \ottnt{Q}  }  $ as:
$$
                      \left\{
                      \begin{array}{r@{\;}c@{\;}l}
                          \constraintfont{   \multiplicityfont{ \ottsym{1} }  \scale \ottsym{(}  \ottnt{U}  \ottsym{,}  \ottnt{L}  \ottsym{)}  }   &=&   \constraintfont{ \ottsym{(}  \ottnt{U}  \ottsym{,}  \ottnt{L}  \ottsym{)} }   \\
                          \constraintfont{   \multiplicityfont{ \omega }  \scale \ottsym{(}  \ottnt{U}  \ottsym{,}  \ottnt{L}  \ottsym{)}  }   &=&   \constraintfont{ \ottsym{(}   \ottnt{U}  \cup  \ottnt{L}   \ottsym{,}  \emptyset  \ottsym{)} }  
                      \end{array}
                                                    \right.
$$

The semantics of simple constraints (and, indeed, of atomic
constraints) is given by an \emph{entailment relation}. Just like the
set of atomic constraints, the entailment relation is a parameter of
our system.

\begin{figure}
  \maybesmall
  \begin{enumerate}
  \item $ \constraintfont{ \ottnt{Q} }   \Vdash   \constraintfont{ \ottnt{Q} } $.
  \item If $ \constraintfont{ \ottnt{Q_{{\mathrm{1}}}} }   \Vdash   \constraintfont{ \ottnt{Q_{{\mathrm{2}}}} } $  and $ \constraintfont{ \ottnt{Q_{{\mathrm{2}}}} }   \Vdash   \constraintfont{ \ottnt{Q_{{\mathrm{3}}}} } $, then $ \constraintfont{ \ottnt{Q_{{\mathrm{1}}}} }   \Vdash   \constraintfont{ \ottnt{Q_{{\mathrm{3}}}} } $.
  \item If $ \constraintfont{ \ottnt{Q_{{\mathrm{1}}}} }   \Vdash   \constraintfont{ \ottnt{Q'_{{\mathrm{1}}}} } $ and $ \constraintfont{ \ottnt{Q_{{\mathrm{2}}}} }   \Vdash   \constraintfont{ \ottnt{Q'_{{\mathrm{2}}}} } $, then $ \constraintfont{ \ottnt{Q_{{\mathrm{1}}}}  \qtensor  \ottnt{Q_{{\mathrm{2}}}} }   \Vdash   \constraintfont{ \ottnt{Q'_{{\mathrm{1}}}}  \qtensor  \ottnt{Q'_{{\mathrm{2}}}} } $.
  \item If $ \constraintfont{ \ottnt{Q_{{\mathrm{1}}}} }   \Vdash   \constraintfont{ \ottnt{Q_{{\mathrm{2}}}} } $, then $ \constraintfont{   \multiplicityfont{ \pi }  \scale \ottnt{Q_{{\mathrm{1}}}}  }   \Vdash   \constraintfont{   \multiplicityfont{ \pi }  \scale \ottnt{Q_{{\mathrm{2}}}}  } $.
  \item If $ \constraintfont{   \multiplicityfont{ \omega }  \scale \ottnt{Q}  }   \Vdash   \constraintfont{ \ottnt{Q} } $.
  \item If $ \constraintfont{   \multiplicityfont{ \omega }  \scale \ottnt{Q}  }   \Vdash   \constraintfont{  \mathbf{\varepsilon}  } $.
  \item If $  \constraintfont{ \ottnt{q} }   \in  \constraintfont{\mathcal{D} } $, then $ \constraintfont{   \multiplicityfont{ \ottsym{1} }  \scale \ottnt{q}  }   \Vdash   \constraintfont{   \multiplicityfont{ \ottsym{1} }  \scale \ottnt{q}   \qtensor    \multiplicityfont{ \ottsym{1} }  \scale \ottnt{q}  } $.
  \item If $  \constraintfont{ \ottnt{q} }   \in  \constraintfont{\mathcal{D} } $, then $ \constraintfont{   \multiplicityfont{ \ottsym{1} }  \scale \ottnt{q}  }   \Vdash   \constraintfont{  \mathbf{\varepsilon}  } $.
  \end{enumerate}
\caption{Requirements for the entailment relation $ \constraintfont{ \ottnt{Q_{{\mathrm{1}}}} }   \Vdash   \constraintfont{ \ottnt{Q_{{\mathrm{2}}}} } $}
\label{fig:entailment-relation}
\end{figure}
\begin{definition}[Entailment relation]
  \label{def:entailment-relation}
  The qualified type system is parameterised by a relation
  $ \constraintfont{ \ottnt{Q_{{\mathrm{1}}}} }   \Vdash   \constraintfont{ \ottnt{Q_{{\mathrm{2}}}} } $ between two simple constraints, as well as by a
  distinguished set $ \constraintfont{\mathcal{D} } $ of duplicable atomic constraints.

  We write, abusing notation, $  \constraintfont{ \ottnt{Q} }   \in  \constraintfont{\mathcal{D} } $ for a simple constraint
  $  \constraintfont{ \ottnt{Q} }  =  \constraintfont{ \ottsym{(}  \ottnt{U}  \ottsym{,}  \ottnt{L}  \ottsym{)} }  $ if for all $  \constraintfont{ \ottnt{q} }  \in  \constraintfont{ \ottnt{L} }  $ we
  have $  \constraintfont{ \ottnt{q} }   \in  \constraintfont{\mathcal{D} } $. (This isn't a plain inclusion because $  \constraintfont{ \ottnt{L} }  $ is a multiset and $ \constraintfont{\mathcal{D} } $.)

  The entailment relation must obey the laws listed in
  \cref{fig:entailment-relation}. Note that $ \constraintfont{\mathcal{D} } $ being a set of atomic
  constraints, the rules pertaining to $ \constraintfont{\mathcal{D} } $ are in terms of atomic
  constraints rather than of arbitrary simple constraints.
\end{definition}
\info{See Fig 3, p14 of OutsideIn\cite{OutsideIn}.}
The set $ \constraintfont{\mathcal{D} } $ is a set of constraints which can be duplicated and
discarded (see~\cref{fig:entailment-relation}). We use $ \constraintfont{\mathcal{D} } $ to
model the \ensuremath{\constraintfont{\constraintfont{\Conid{Linearly}}}} constraint. Crucially, it is \emph{not
  the case} that $ \constraintfont{   \multiplicityfont{ \ottsym{1} }  \scale \ottnt{q}  }   \Vdash   \constraintfont{   \multiplicityfont{ \omega }  \scale \ottnt{q}  } $ for $  \constraintfont{ \ottnt{q} }   \in  \constraintfont{\mathcal{D} } $. While it may seem
counter-intuitive, there is nothing in linear logic mandating that a
formula that can be duplicated and discarded be (equivalent to) an
unrestricted formula. This observation has been exploited, for
instance, to introduce so-called
subexponentials~\cite{subexponentials}. For our use case, it lets the
typechecker dispatch \ensuremath{\constraintfont{\constraintfont{\Conid{Linearly}}}} constraints (using
duplication), but prevents the result of constrained functions to be used
unrestrictedly.

\subsection{Typing Rules}
\label{sec:typing-rules}

With this material in place, we can now present our type system. The
grammar is given in \cref{fig:declarative:grammar}, which also
includes the definitions of scaling on contexts $   \multiplicityfont{ \pi }  \scale \Gamma  $ and addition
of contexts $ \Gamma_{{\mathrm{1}}}  \ottsym{+}  \Gamma_{{\mathrm{2}}} $. Note that addition on contexts is actually
a partial function, as it requires that, if a variable $ \ottmv{x} $ is bound
in both $ \Gamma_{{\mathrm{1}}} $ and $ \Gamma_{{\mathrm{2}}} $, then $ \ottmv{x} $ is assigned the same
type in both (but perhaps different multiplicities). This partiality
is not a problem in practice, as the required condition for combining
contexts is always satisfied.

\begin{figure}
  \maybesmall
  $$
  \begin{array}[b]{lcll}
     \ottmv{a} ,  \ottmv{b}  & \bnfeq & \ldots & \text{Type variables} \\
     \ottmv{x} ,  \ottmv{y}  & \bnfeq & \ldots & \text{Expression variables} \\
    \ottmv{T} & \bnfeq & \ldots & \text{Type constructors} \\
     \ottmv{K}  & \bnfeq & \ldots & \text{Data constructors} \\
     \sigma  & \bnfeq &   \forall   \overline{\ottmv{a} } .   \constraintfont{ \ottnt{Q} }   \Lolly  \tau   & \text{Type schemes} \\
     \tau ,  \upsilon  & \bnfeq &  \ottmv{a}  \bnfor   \exists   \overline{\ottmv{a} } .  \tau  \RLolly   \constraintfont{ \ottnt{Q} }    \bnfor   \tau_{{\mathrm{1}}}  \to_{  \multiplicityfont{ \pi }  }  \tau_{{\mathrm{2}}}  
                            \bnfor  \ottmv{T} \, \overline{\tau}  & \text{Types} \\
     \Gamma ,  \Delta  & \bnfeq &  ∙  \bnfor  \Gamma  \ottsym{,}   \ottmv{x} {:}_{  \multiplicityfont{ \pi }  } \sigma   &
                                                              \text{Contexts} \\
     \ottnt{e}  & \bnfeq &  \ottmv{x}  \bnfor  \ottmv{K}  \bnfor  \lambda  \ottmv{x}  \ottsym{.}  \ottnt{e}  \bnfor \ottnt{e_{{\mathrm{1}}}} \, \ottnt{e_{{\mathrm{2}}}} \bnfor  \packbox \, \ottnt{e}  & \text{Expressions}\\
                 &\bnfor &  \klet\ \packbox  \ottmv{x}  =  \ottnt{e_{{\mathrm{1}}}}  \  \ottkw{in}  \  \ottnt{e_{{\mathrm{2}}}}  \bnfor   \kcase_  \multiplicityfont{ \pi }   \, \ottnt{e} \, \ottkw{of} \, \ottsym{\{}  \overline{\ottmv{K}_i\ \overline{\ottmv{x}_i } \to \ottnt{e}_i }  \ottsym{\}}  &\\
                 &\bnfor &  \klet_  \multiplicityfont{ \pi }   \, \ottmv{x}  \ottsym{=}  \ottnt{e_{{\mathrm{1}}}} \, \ottkw{in} \, \ottnt{e_{{\mathrm{2}}}} \bnfor   \klet_  \multiplicityfont{ \pi }   \, \ottmv{x}  \ottsym{:}  \sigma  \ottsym{=}  \ottnt{e_{{\mathrm{1}}}} \, \ottkw{in} \, \ottnt{e_{{\mathrm{2}}}}  &
  \end{array}
  $$
  Context scaling $   \multiplicityfont{ \pi }  \scale \Gamma  $ and addition of contexts $ \Gamma_{{\mathrm{1}}}  \ottsym{+}  \Gamma_{{\mathrm{2}}} $ is defined as follows:
  $$
  \begin{array}{cc}
  \left\{
  \begin{array}{r@{\;}c@{\;}l}
     \multiplicityfont{ \pi }  \scale ∙   &=&  ∙  \\
     \multiplicityfont{ \pi }  \scale \ottsym{(}  \Gamma  \ottsym{,}   \ottmv{x} {:}_{  \multiplicityfont{ \rho }  } \sigma   \ottsym{)}   &=&    \multiplicityfont{ \pi }  \scale \Gamma   \ottsym{,}   \ottmv{x} {:}_{  \multiplicityfont{ \ottsym{(}   \pi {⋅} \rho   \ottsym{)} }  } \sigma  
  \end{array}
  \right.
  &
  \left\{
  \begin{array}{r@{\;}c@{\;}lll}
   \ottsym{(}  \Gamma_{{\mathrm{1}}}  \ottsym{,}   \ottmv{x} {:}_{  \multiplicityfont{ \pi }  } \sigma   \ottsym{)}  \ottsym{+}  \Gamma_{{\mathrm{2}}}  &=&  \Gamma_{{\mathrm{1}}}  \ottsym{+}  \Gamma'_{{\mathrm{2}}}  \ottsym{,}   \ottmv{x} {:}_{  \multiplicityfont{ \ottsym{(}  \pi  \ottsym{+}  \rho  \ottsym{)} }  } \sigma   & \text{where} & \Gamma_{{\mathrm{2}}}  \ottsym{=}   \ottsym{\{}   \ottmv{x} {:}_{  \multiplicityfont{ \rho }  } \sigma   \ottsym{\}}  \cup  \Gamma'_{{\mathrm{2}}}  \\
  &&&&  \ottmv{x}  \notin  \Gamma'_{{\mathrm{2}}}  \\
   \ottsym{(}  \Gamma_{{\mathrm{1}}}  \ottsym{,}   \ottmv{x} {:}_{  \multiplicityfont{ \pi }  } \sigma   \ottsym{)}  \ottsym{+}  \Gamma_{{\mathrm{2}}}  &=&  \Gamma_{{\mathrm{1}}}  \ottsym{+}  \Gamma_{{\mathrm{2}}}  \ottsym{,}   \ottmv{x} {:}_{  \multiplicityfont{ \pi }  } \sigma   & \text{where} &  \ottmv{x}  \notin  \Gamma_{{\mathrm{2}}}  \\
   ∙  \ottsym{+}  \Gamma_{{\mathrm{2}}}  &=&  \Gamma_{{\mathrm{2}}} 
  \end{array}
  \right.
  \end{array}
  $$
  \caption{Grammar of the qualified type system}
  \label{fig:declarative:grammar}
\end{figure}

\begin{figure}
  \centering
  \drules[E]{$ \constraintfont{ \ottnt{Q} }   \ottsym{;}  \Gamma  \vdash  \ottnt{e}  \ottsym{:}  \tau$}{Expression
    typing}{Var,Abs,App,Pack,Unpack,Let,LetSig,Case,Sub}
  \caption{Qualified type system}
  \label{fig:typing-rules}
\end{figure}

The typing rules are in \cref{fig:typing-rules}.
A qualified type system~\cite{QualifiedTypes} such as ours introduces a
judgement of the form $ \constraintfont{ \ottnt{Q} }   \ottsym{;}  \Gamma  \vdash  \ottnt{e}  \ottsym{:}  \tau$, where $ \Gamma $ is a standard
type context, and $  \constraintfont{ \ottnt{Q} }  $ is a constraint we have assumed to be true.
$  \constraintfont{ \ottnt{Q} }  $ behaves
much like $ \Gamma $, which will be instrumental for
desugaring in \cref{sec:desugaring}; the main difference is
that $ \Gamma $ is addressed explicitly, whereas $  \constraintfont{ \ottnt{Q} }  $
is used implicitly in \rref{E-Var}.

Because constraints are used implicitly, if there are several instances
of the same $  \constraintfont{ \ottnt{q} }  $, it is non-deterministic which one is used in
which instance of \rref*{E-Var}. How do we
reconcile this non-determinism with the use of linear constraints, in
\cref{sec:memory-ownership} to thread mutations? We certainly don't want type
inference to non-deterministically reorder a \ensuremath{\Varid{readRef}} and a
\ensuremath{\Varid{writeRef}}! The solution is that the \textsc{api} is arranged so that
only a single instance of \ensuremath{\constraintfont{\Conid{RW}\;\Varid{n}}} is ever
provided. Therefore there is a single possible threading of the reads
and writes.

The type system of \cref{fig:typing-rules} is purely
declarative: note, for example, that \rref{E-App} does not describe
how to break the typing assumptions into constraints
$  \constraintfont{ \ottnt{Q_{{\mathrm{1}}}} }  $,$  \constraintfont{ \ottnt{Q_{{\mathrm{2}}}} }  $ and contexts $ \Gamma_{{\mathrm{1}}} $,$ \Gamma_{{\mathrm{2}}} $. We will see
how to infer constraints in \cref{sec:type-inference}. Yet,
this system is our ground truth: a system with a simple enough
definition that programmers can reason about typing. As is standard in
the qualified-type literature (since the seminal work of \citet{QualifiedTypes}), we do not
directly give a dynamic
semantics to this language; instead, we will give it meaning via
desugaring to a simpler core language in \cref{sec:desugaring}.

We survey several distinctive features of our qualified type system below:

\info{See Fig 10, p25 of OutsideIn\cite{OutsideIn}.}
\paragraph*{Linear functions.}
The type of linear functions is written $  \ottmv{a}  \to_{  \multiplicityfont{ \ottsym{1} }  }  \ottmv{b}  $.
  Despite our focus on linear constraints,
  we still need linearity in ordinary arguments.
  Indeed, the linearity of arrows interacts in interesting
  ways with linear constraints: If $\ottmv{f}  \ottsym{:}   \ottmv{a}  \to_{  \multiplicityfont{ \omega }  }  \ottmv{b} $ and
  $\ottmv{x}  \ottsym{:}   \constraintfont{   \multiplicityfont{ \ottsym{1} }  \scale \ottnt{q}  }   \Lolly  \ottmv{a}$, then calling $ \ottmv{f} \, \ottmv{x} $ would actually use $  \constraintfont{ \ottnt{q} }  $
  many times. We must make sure it is impossible to derive
  $ \constraintfont{   \multiplicityfont{ \ottsym{1} }  \scale \ottnt{q}  }   \ottsym{;}   \ottmv{f} {:}_{  \multiplicityfont{ \omega }  }  \ottmv{a}  \to_{  \multiplicityfont{ \omega }  }  \ottmv{b}    \ottsym{,}   \ottmv{x} {:}_{  \multiplicityfont{ \omega }  }  \constraintfont{   \multiplicityfont{ \ottsym{1} }  \scale \ottnt{q}  }   \Lolly  \ottmv{a}   \vdash  \ottmv{f} \, \ottmv{x}  \ottsym{:}  \ottmv{b}$.
  Otherwise we could make, for instance, the \ensuremath{\Varid{overusing}} function from
  \cref{sec:overusing}.
  One can check that $ \constraintfont{   \multiplicityfont{ \ottsym{1} }  \scale \ottnt{q}  }   \ottsym{;}   \ottmv{f} {:}_{  \multiplicityfont{ \omega }  }  \ottmv{a}  \to_{  \multiplicityfont{ \omega }  }  \ottmv{b}    \ottsym{,}   \ottmv{x} {:}_{  \multiplicityfont{ \omega }  }  \constraintfont{   \multiplicityfont{ \ottsym{1} }  \scale \ottnt{q}  }   \Lolly  \ottmv{a}   \vdash  \ottmv{f} \, \ottmv{x}  \ottsym{:}  \ottmv{b}$
  indeed does not
  type check, because the scaling of $  \constraintfont{ \ottnt{Q_{{\mathrm{2}}}} }  $ in \rref{E-App} ensures that
  the constraint would be $  \constraintfont{   \multiplicityfont{ \omega }  \scale \ottnt{q}  }  $ instead. On the other hand,
  it is perfectly fine to have $ \constraintfont{   \multiplicityfont{ \ottsym{1} }  \scale \ottnt{q}  }   \ottsym{;}   \ottmv{f} {:}_{  \multiplicityfont{ \omega }  }  \ottmv{a}  \to_{  \multiplicityfont{ \ottsym{1} }  }  \ottmv{b}    \ottsym{,}   \ottmv{x} {:}_{  \multiplicityfont{ \omega }  }  \constraintfont{   \multiplicityfont{ \ottsym{1} }  \scale \ottnt{q}  }   \Lolly  \ottmv{a}   \vdash  \ottmv{f} \, \ottmv{x}  \ottsym{:}  \ottmv{b}$ when $ \ottmv{f} $ is a linear function.

\paragraph*{Variables.}
As is standard, the \rref{E-Var} rule works in a context containing more
than just the used binding for $ \ottmv{x} $. However, crucially,
our rule allows only
\emph{unrestricted} variables to be discarded; linear variables \emph{must} be
used. We can see this in the rule by noticing that the context has an unrestricted
component $   \multiplicityfont{ \omega }  \scale \Gamma_{{\mathrm{2}}}  $. The $ \Gamma_{{\mathrm{1}}} $ component might be restricted or might not,
allowing this rule to apply both for restricted and unrestricted $ \ottmv{x} $
(leveraging the equality $  \multiplicityfont{ \ottsym{1}  \ottsym{+}  \omega }  =  \multiplicityfont{ \omega }  $).

\paragraph*{Data constructors.}
Data constructors $ \ottmv{K} $ don't have a dedicated typing
rule. Instead they are typed using the \rref{E-Var}, where they are
treated as if they were unrestricted variables.

\paragraph*{Let-bindings.}
Bindings in a \ensuremath{\mathbf{let}} may be for either linear or unrestricted variables.
  We could require all bindings to be linear and to implement unrestricted
  information only using \ensuremath{\Conid{Ur}}, but it is very easy to add a multiplicity
  annotation on \ensuremath{\mathbf{let}}, and so we do.

\paragraph*{Local assumptions.}
\Rref{E-Let} includes support for local
  assumptions. We thus have the ability to generalise a subset of
  the constraints needed by $ \ottnt{e_{{\mathrm{1}}}} $ (but not the type variables---no
  \ensuremath{\mathbf{let}}-generalisation here, though it could be added). The inference algorithm of
  \cref{sec:type-inference} will not make use of this
  possibility. 

\paragraph*{Existentials.}
 We include $  \exists   \overline{\ottmv{a} } .  \tau  \RLolly   \constraintfont{ \ottnt{Q} }   $, as
  introduced in \cref{sec:what-it-looks-like}, together with
  the $\packbox$ constructor. See rules~\rref*{E-Pack} and
  \rref*{E-Unpack}. We use these rules are the type system's counterpart to our
  \ensuremath{\Conid{Linearly}.\mathbf{do}} notation (see \cref{sec:packing-unpacking-do}):
  $$
  \begin{array}{lcl}
    \ensuremath{\Varid{\Conid{Linearly}.return}} &=&  \lambda  \ottmv{x}  \ottsym{.}  \packbox \, \ottmv{x} \\
    \ensuremath{(\Conid{Linearly}{.}\bind )} &=&  \lambda  \ottmv{x}  \ottsym{.}  \lambda  \ottmv{f}  \ottsym{.}   \klet\ \packbox  \ottmv{x'}  =  \ottmv{x}  \  \ottkw{in}  \  \ottmv{f}  \, \ottmv{x'} \\
    \ensuremath{(\Conid{Linearly}{.}\sequ )} &=&  \lambda  \ottmv{x}  \ottsym{.}  \lambda  \ottmv{y}  \ottsym{.}   \klet\ \packbox  \ottmv{x'}  =  \ottmv{x}  \  \ottkw{in}  \   \kcase_  \multiplicityfont{ \ottsym{1} }   \, \ottmv{x} \, \ottkw{of} \, \ottsym{\{}  \ottsym{()}  \to  \ottmv{y}  \ottsym{\}}  
  \end{array}
  $$
\info{No substitution on $  \constraintfont{ \ottnt{Q_{{\mathrm{1}}}} }  $ in the E-Unpack rule, because there is
  only existential quantification.}

\section{Constraint inference}
\label{sec:type-inference}

The type system of \cref{fig:typing-rules} gives a declarative description
of what programs are acceptable. We now present the algorithmic counterpart to
this system. Our algorithm is structured, unsurprisingly, around generating and
solving constraints, broadly following the template of
\citet{essence-of-ml-type-inference}.
That is, our algorithm takes a pass over the abstract syntax tree, generating constraints as it goes. Then, separately, we solve those
constraints (that is, try to satisfy them) in the presence of a set of assumptions,
or we determine that the assumptions do not imply that the constraints hold. In the
latter case, we issue an error to the programmer.

The procedure is responsible for inferring both \emph{types} and \emph{constraints}.
For our type system, type inference can be done independently from constraint
inference. Thus, we focus on the latter, and defer type inference to
an external oracle~\cite[e.g.][]{linear-types-inference}.
That is, we assume an algorithm that produces typing derivations for the
judgement $\Gamma  \vdash  \ottnt{e}  \ottsym{:}  \tau$, ignoring all the constraints. Then, we describe a
constraint generation algorithm that passes over these typing derivations.
%
%
We can make this simplification for two reasons:
\begin{itemize}
\item We do not formalise type equality constraints, and our implementation
  in \textsc{ghc} (\cref{sec:equality-constraints}) takes care to disallow linear equality constraints to influence type inference.
  Indeed, a typical treatment of unification
  would be unsound for linear equalities, because it reuses the same
  equality many times (or none at all). Linear equalities make sense
  (\citet{shulman2018linear} puts linear
  equalities to great use), but they do not seem to lend themselves to
  automation.
\item We do not support, or intend to support, multiplicity
  polymorphism in constraint arrows. That is, the multiplicity of a
  constraint is always syntactically known to be either linear or
  unrestricted. This way, no equality constraints (which might, conceivably,
  relate multiplicity variables) can interfere with
  constraint resolution.
\end{itemize}
%

\subsection{Wanted Constraints}
\label{sec:wanteds}

The constraints $  \constraintfont{ \constraintfont{C} }  $ generated in our system have a richer
logical structure than the simple constraints $  \constraintfont{ \ottnt{Q} }  $, above. Following
\textsc{ghc} and echoing \citet{OutsideIn}, we call these \emph{wanted constraints}:
they are constraints which the constraint solver \emph{wants} to prove.
An unproved wanted constraint results in a type error reported to the programmer.
$$
\begin{array}{lcll}
    \constraintfont{ \constraintfont{C} }   & \bnfeq &   \constraintfont{ \ottnt{Q} }   \bnfor   \constraintfont{ \constraintfont{C_{{\mathrm{1}}}}  \qtensor  \constraintfont{C_{{\mathrm{2}}}} }   \bnfor   \constraintfont{ \constraintfont{C_{{\mathrm{1}}}}  \aand  \constraintfont{C_{{\mathrm{2}}}} }   \bnfor   \constraintfont{   \multiplicityfont{ \pi }  \scale ( \ottnt{Q}  \Lolly  \constraintfont{C} )  }  &
                                                                \text{Wanted constraints}
\end{array}
$$
A simple constraint is a valid wanted constraint, and we have two forms of
conjunction for wanted constraints:
the new
$  \constraintfont{ \constraintfont{C_{{\mathrm{1}}}}  \aand  \constraintfont{C_{{\mathrm{2}}}} }  $ construction (read $  \constraintfont{ \constraintfont{C_{{\mathrm{1}}}} }  $ \emph{with} $  \constraintfont{ \constraintfont{C_{{\mathrm{2}}}} }  $), alongside
the more typical $  \constraintfont{ \constraintfont{C_{{\mathrm{1}}}}  \qtensor  \constraintfont{C_{{\mathrm{2}}}} }  $. These are
connectives from linear logic: $  \constraintfont{ \constraintfont{C_{{\mathrm{1}}}}  \qtensor  \constraintfont{C_{{\mathrm{2}}}} }  $ is the
\emph{multiplicative} conjunction, and $  \constraintfont{ \constraintfont{C_{{\mathrm{1}}}}  \aand  \constraintfont{C_{{\mathrm{2}}}} }  $ is the \emph{additive}
conjunction. Both connectives are conjunctions, but they differ
in meaning. To satisfy $  \constraintfont{ \constraintfont{C_{{\mathrm{1}}}}  \qtensor  \constraintfont{C_{{\mathrm{2}}}} }  $ one consumes the (linear)
assumptions consumed by satisfying $  \constraintfont{ \constraintfont{C_{{\mathrm{1}}}} }  $ and those consumed by $  \constraintfont{ \constraintfont{C_{{\mathrm{2}}}} }  $;
if an assumed linear constraint is needed to prove both $  \constraintfont{ \constraintfont{C_{{\mathrm{1}}}} }  $ and $  \constraintfont{ \constraintfont{C_{{\mathrm{2}}}} }  $,
then $  \constraintfont{ \constraintfont{C_{{\mathrm{1}}}}  \qtensor  \constraintfont{C_{{\mathrm{2}}}} }  $ will not be provable, because that linear assumption cannot
be used twice. On the
other hand, satisfying $  \constraintfont{ \constraintfont{C_{{\mathrm{1}}}}  \aand  \constraintfont{C_{{\mathrm{2}}}} }  $ requires that satisfying $  \constraintfont{ \constraintfont{C_{{\mathrm{1}}}} }  $
and $  \constraintfont{ \constraintfont{C_{{\mathrm{2}}}} }  $ must each
consume the \emph{same} assumptions, which $  \constraintfont{ \constraintfont{C_{{\mathrm{1}}}}  \aand  \constraintfont{C_{{\mathrm{2}}}} }  $ consumes as well.
Thus, if $  \constraintfont{ \constraintfont{C} }  $ is assumed linearly (and we have no other assumptions),
then $  \constraintfont{ \constraintfont{C}  \qtensor  \constraintfont{C} }  $ is not provable, while $  \constraintfont{ \constraintfont{C}  \aand  \constraintfont{C} }  $ is.
The intuition, here, is that in $  \constraintfont{ \constraintfont{C_{{\mathrm{1}}}}  \aand  \constraintfont{C_{{\mathrm{2}}}} }  $, only
one of $  \constraintfont{ \constraintfont{C_{{\mathrm{1}}}} }  $ or $  \constraintfont{ \constraintfont{C_{{\mathrm{2}}}} }  $ will be eventually used. ``With'' constraints
arise from the branches in a $\kcase$-expression.

The last form of wanted constraint $  \constraintfont{ \constraintfont{C} }  $ is an implication
$  \constraintfont{   \multiplicityfont{ \pi }  \scale ( \ottnt{Q}  \Lolly  \constraintfont{C} )  }  $. If $  \multiplicityfont{ \pi }  $ is 1, the rule is unsurprising. The more interesting case is $  \constraintfont{   \multiplicityfont{ \omega }  \scale ( \ottnt{Q}  \Lolly  \constraintfont{C} )  }  $:
to prove $  \constraintfont{   \multiplicityfont{ \omega }  \scale ( \ottnt{Q}  \Lolly  \constraintfont{C} )  }  $, you need to prove $  \constraintfont{ \constraintfont{C} }  $ under the
\emph{linear} assumption $  \constraintfont{ \ottnt{Q} }  $, but without using any other linear
assumptions.

These implications arise when we unpack an existential
package that contains a linear constraint and also when checking a \ensuremath{\mathbf{let}}-binding.
We can define scaling over wanted constraints by recursion as follows, where we
use scaling over simple constraints in the simple-constraint case:
$$
\left\{
  \begin{array}{lcl}
      \constraintfont{   \multiplicityfont{ \pi }  \scale \ottsym{(}  \constraintfont{C_{{\mathrm{1}}}}  \qtensor  \constraintfont{C_{{\mathrm{2}}}}  \ottsym{)}  }   & = &   \constraintfont{   \multiplicityfont{ \pi }  \scale \constraintfont{C_{{\mathrm{1}}}}   \qtensor    \multiplicityfont{ \pi }  \scale \constraintfont{C_{{\mathrm{2}}}}  }   \\
      \constraintfont{   \multiplicityfont{ \ottsym{1} }  \scale \ottsym{(}  \constraintfont{C_{{\mathrm{1}}}}  \aand  \constraintfont{C_{{\mathrm{2}}}}  \ottsym{)}  }   & = &   \constraintfont{ \constraintfont{C_{{\mathrm{1}}}}  \aand  \constraintfont{C_{{\mathrm{2}}}} }   \\
      \constraintfont{   \multiplicityfont{ \omega }  \scale \ottsym{(}  \constraintfont{C_{{\mathrm{1}}}}  \aand  \constraintfont{C_{{\mathrm{2}}}}  \ottsym{)}  }   & = &   \constraintfont{   \multiplicityfont{ \omega }  \scale \constraintfont{C_{{\mathrm{1}}}}   \qtensor    \multiplicityfont{ \omega }  \scale \constraintfont{C_{{\mathrm{2}}}}  }   \\
      \constraintfont{   \multiplicityfont{ \pi }  \scale \ottsym{(}    \multiplicityfont{ \rho }  \scale ( \ottnt{Q}  \Lolly  \constraintfont{C} )   \ottsym{)}  }   & = &   \constraintfont{   \multiplicityfont{ \ottsym{(}   \pi {⋅} \rho   \ottsym{)} }  \scale ( \ottnt{Q}  \Lolly  \constraintfont{C} )  }  
  \end{array}
\right.
$$
For the most part, scaling of wanted constraints is straightforward. The only
peculiar case is when we scale the additive conjunction $  \constraintfont{ \constraintfont{C_{{\mathrm{1}}}}  \aand  \constraintfont{C_{{\mathrm{2}}}} }  $ by
$  \multiplicityfont{ \omega }  $, the result is a multiplicative conjunction. The intuition here is
that when if we have both $  \constraintfont{   \multiplicityfont{ \omega }  \scale \constraintfont{C_{{\mathrm{1}}}}  }  $ and $  \constraintfont{   \multiplicityfont{ \omega }  \scale \constraintfont{C_{{\mathrm{2}}}}  }  $, then
a choice between $  \constraintfont{ \constraintfont{C_{{\mathrm{1}}}} }  $ and $  \constraintfont{ \constraintfont{C_{{\mathrm{2}}}} }  $ can be made $  \multiplicityfont{ \omega }  $ times.

We define an entailment relation over wanteds in \cref{fig:wanted:entailment}.
Note that this relation uses only simple constraints $  \constraintfont{ \ottnt{Q} }  $ as assumptions, as
there is no way to assume the more elaborate $  \constraintfont{ \constraintfont{C} }  $\footnote{Allowing the full wanted-constraint syntax
in assumptions is the subject of work by \citet{quantified-constraints}.}.
\begin{figure}
  \maybesmall
  \centering
  \drules[C]{$ \constraintfont{ \ottnt{Q} }   \vdash   \constraintfont{ \constraintfont{C} } $} {Wanted-constraint entailment}
  {Dom,Id,Tensor,With,Impl}
  \caption{Wanted-constraint entailment}
  \label{fig:wanted:entailment}
\end{figure}

\begin{toappendix}
  \label{sec:appendix:proofs-lemmas}
\end{toappendix}

Before we move on to constraint generation proper, let us highlight a few
technical, yet essential, lemmas about the wanted-constraint
entailment relation. Proofs of these lemmas (and others) appear in \cref{sec:appendix:proofs-lemmas}.

\begin{lemmarep}[Inversion]
  \label{lem:inversion}
  The inference rules of $ \constraintfont{ \ottnt{Q} }   \vdash   \constraintfont{ \constraintfont{C} } $ can be read bottom-up (up to the
  set $ \constraintfont{\mathcal{D} } $) as well
  as top-down, as is required of $ \constraintfont{ \ottnt{Q_{{\mathrm{1}}}} }   \Vdash   \constraintfont{ \ottnt{Q_{{\mathrm{2}}}} } $ in
  \cref{fig:entailment-relation}. That is:
  \begin{itemize}
  \item If $ \constraintfont{ \ottnt{Q} }   \vdash   \constraintfont{ \constraintfont{C_{{\mathrm{1}}}}  \qtensor  \constraintfont{C_{{\mathrm{2}}}} } $, then there exists $  \constraintfont{ \ottnt{Q_{{\mathrm{1}}}} }  $ and $  \constraintfont{ \ottnt{Q_{{\mathrm{2}}}} }  $
    such that $ \constraintfont{ \ottnt{Q_{{\mathrm{1}}}} }   \vdash   \constraintfont{ \constraintfont{C_{{\mathrm{1}}}} } $, $ \constraintfont{ \ottnt{Q_{{\mathrm{2}}}} }   \vdash   \constraintfont{ \constraintfont{C_{{\mathrm{2}}}} } $, and
    $ \constraintfont{ \ottnt{Q} }   \Vdash   \constraintfont{ \ottnt{Q_{{\mathrm{1}}}}  \qtensor  \ottnt{Q_{{\mathrm{2}}}} } $.
  \item If $ \constraintfont{ \ottnt{Q} }   \vdash   \constraintfont{ \constraintfont{C_{{\mathrm{1}}}}  \aand  \constraintfont{C_{{\mathrm{2}}}} } $, then $ \constraintfont{ \ottnt{Q} }   \vdash   \constraintfont{ \constraintfont{C_{{\mathrm{1}}}} } $ and $ \constraintfont{ \ottnt{Q} }   \vdash   \constraintfont{ \constraintfont{C_{{\mathrm{2}}}} } $.
  \item If $ \constraintfont{ \ottnt{Q} }   \vdash   \constraintfont{   \multiplicityfont{ \pi }  \scale ( \ottnt{Q_{{\mathrm{2}}}}  \Lolly  \constraintfont{C} )  } $, then there exists $  \constraintfont{ \ottnt{Q_{{\mathrm{1}}}} }  $ such
    that $ \constraintfont{ \ottnt{Q_{{\mathrm{1}}}}  \qtensor  \ottnt{Q_{{\mathrm{2}}}} }   \vdash   \constraintfont{ \constraintfont{C} } $ and  $ \constraintfont{ \ottnt{Q} }   \Vdash   \constraintfont{   \multiplicityfont{ \pi }  \scale \ottnt{Q_{{\mathrm{1}}}}  } $
  \end{itemize}
\end{lemmarep}
\begin{proof}
  We prove each of these statement by induction
  \begin{itemize}
  \item Suppose $ \constraintfont{ \ottnt{Q} }   \vdash   \constraintfont{ \constraintfont{C_{{\mathrm{1}}}}  \aand  \constraintfont{C_{{\mathrm{2}}}} } $, then there are two cases
    \begin{itemize}
    \item either it is the conclusion of a \rref*{C-With} rule,
      and the result is immediate.
    \item or it is the result of a \rref*{C-Dom} rule, then, there
      exists $  \constraintfont{ \ottnt{Q'} }  $, such that $ \constraintfont{ \ottnt{Q} }   \Vdash   \constraintfont{ \ottnt{Q'} } $ and
      $ \constraintfont{ \ottnt{Q'} }   \vdash   \constraintfont{ \constraintfont{C_{{\mathrm{1}}}}  \aand  \constraintfont{C_{{\mathrm{2}}}} } $.

      By induction $ \constraintfont{ \ottnt{Q'} }   \vdash   \constraintfont{ \constraintfont{C_{{\mathrm{2}}}} } $ and $ \constraintfont{ \ottnt{Q'} }   \vdash   \constraintfont{ \constraintfont{C_{{\mathrm{2}}}} } $, applying
      \rref*{C-Dom} to both gives $ \constraintfont{ \ottnt{Q} }   \vdash   \constraintfont{ \constraintfont{C_{{\mathrm{2}}}} } $ and $ \constraintfont{ \ottnt{Q} }   \vdash   \constraintfont{ \constraintfont{C_{{\mathrm{2}}}} } $ as
      required.
    \end{itemize}

  \item Suppose $ \constraintfont{ \ottnt{Q} }   \vdash   \constraintfont{   \multiplicityfont{ \pi }  \scale ( \ottnt{Q_{{\mathrm{2}}}}  \Lolly  \constraintfont{C} )  } $, then there are two cases
    \begin{itemize}
    \item either it is the conclusion of a \rref*{C-Impl} rule,
      and the result is immediate.
    \item or it is the result of a \rref*{C-Dom} rule, then, there
      exists $  \constraintfont{ \ottnt{Q'} }  $, such that $ \constraintfont{ \ottnt{Q} }   \Vdash   \constraintfont{ \ottnt{Q'} } $ and
      $ \constraintfont{ \ottnt{Q'} }   \vdash   \constraintfont{   \multiplicityfont{ \pi }  \scale ( \ottnt{Q_{{\mathrm{2}}}}  \Lolly  \constraintfont{C} )  } $.

      By induction, there exists $  \constraintfont{ \ottnt{Q_{{\mathrm{1}}}} }  $ such that $ \constraintfont{ \ottnt{Q_{{\mathrm{1}}}}  \qtensor  \ottnt{Q_{{\mathrm{2}}}} }   \vdash   \constraintfont{ \constraintfont{C} } $ and $ \constraintfont{ \ottnt{Q'} }   \Vdash   \constraintfont{   \multiplicityfont{ \pi }  \scale \ottnt{Q_{{\mathrm{1}}}}  } $, then we conclude by \rref*{C-Dom}.
    \end{itemize}
  \item Suppose $ \constraintfont{ \ottnt{Q} }   \vdash   \constraintfont{ \constraintfont{C_{{\mathrm{1}}}}  \qtensor  \constraintfont{C_{{\mathrm{2}}}} } $ we have the following cases:
  \begin{itemize}
  \item either it is the conclusion of a \rref*{C-Tensor} rule, and
    the result is immediate
  \item or it is the result of a \rref*{C-Id} rule, in which case
    $  \constraintfont{ \ottnt{Q} }   =   \constraintfont{ \constraintfont{C_{{\mathrm{1}}}}  \qtensor  \constraintfont{C_{{\mathrm{2}}}} }  $, which proves the result
  \item or it is the result of a \rref*{C-Dom} rule, in which case
    there is $  \constraintfont{ \ottnt{Q'} }  $ such that $ \constraintfont{ \ottnt{Q} }   \Vdash   \constraintfont{ \ottnt{Q'} } $ and $ \constraintfont{ \ottnt{Q'} }   \vdash   \constraintfont{ \constraintfont{C_{{\mathrm{1}}}}  \qtensor  \constraintfont{C_{{\mathrm{2}}}} } $.

    By induction, there exist $  \constraintfont{ \ottnt{Q'_{{\mathrm{1}}}} }  $, and
    $  \constraintfont{ \ottnt{Q'_{{\mathrm{2}}}} }  $, such that $ \constraintfont{ \ottnt{Q'_{{\mathrm{1}}}} }   \vdash   \constraintfont{ \constraintfont{C_{{\mathrm{1}}}} } $,
    $ \constraintfont{ \ottnt{Q'_{{\mathrm{2}}}} }   \vdash   \constraintfont{ \constraintfont{C_{{\mathrm{2}}}} } $, and $ \constraintfont{ \ottnt{Q'} }   \Vdash   \constraintfont{ \ottnt{Q'_{{\mathrm{1}}}}  \qtensor  \ottnt{Q'_{{\mathrm{2}}}} } $.

    By transitivity of entailment, $ \constraintfont{ \ottnt{Q} }   \Vdash   \constraintfont{ \ottnt{Q'_{{\mathrm{1}}}}  \qtensor  \ottnt{Q'_{{\mathrm{2}}}} } $, as required.
  \end{itemize}
  \end{itemize}
\end{proof}

\begin{lemmarep}[Scaling]
  \label{lem:wanted:promote}
  If $ \constraintfont{ \ottnt{Q} }   \vdash   \constraintfont{ \constraintfont{C} } $, then $ \constraintfont{   \multiplicityfont{ \pi }  \scale \ottnt{Q}  }   \vdash   \constraintfont{   \multiplicityfont{ \pi }  \scale \constraintfont{C}  } $.
\end{lemmarep}
\begin{proof}
  By induction on the derivation of $ \constraintfont{ \ottnt{Q} }   \vdash   \constraintfont{ \constraintfont{C} } $
  \begin{itemize}
    \item \Rref*{C-Dom}: there exists $  \constraintfont{ \ottnt{Q'} }  $ such that $ \constraintfont{ \ottnt{Q'} }   \vdash   \constraintfont{ \constraintfont{C} } $
          and $ \constraintfont{ \ottnt{Q} }   \Vdash   \constraintfont{ \ottnt{Q'} } $. By induction $ \constraintfont{   \multiplicityfont{ \pi }  \scale \ottnt{Q'}  }   \vdash   \constraintfont{   \multiplicityfont{ \pi }  \scale \constraintfont{C}  } $, and by
          \cref{fig:entailment-relation} $ \constraintfont{   \multiplicityfont{ \pi }  \scale \ottnt{Q}  }   \Vdash   \constraintfont{   \multiplicityfont{ \pi }  \scale \ottnt{Q'}  } $, then by
          \rref*{C-Dom}, we conclude that $ \constraintfont{   \multiplicityfont{ \pi }  \scale \ottnt{Q}  }   \vdash   \constraintfont{   \multiplicityfont{ \pi }  \scale \constraintfont{C}  } $.
    \item \Rref*{C-Id}: by \rref*{C-Id}, $ \constraintfont{   \multiplicityfont{ \pi }  \scale \ottnt{Q}  }   \vdash   \constraintfont{   \multiplicityfont{ \pi }  \scale \ottnt{Q}  } $.
    \item \Rref*{C-Tensor}: $  \constraintfont{ \ottnt{Q} }  =  \constraintfont{ \ottnt{Q_{{\mathrm{1}}}}  \qtensor  \ottnt{Q_{{\mathrm{2}}}} }  $ and $  \constraintfont{ \constraintfont{C} }  =  \constraintfont{ \constraintfont{C_{{\mathrm{1}}}}  \qtensor  \constraintfont{C_{{\mathrm{2}}}} }  $, with
          $ \constraintfont{ \ottnt{Q_{{\mathrm{1}}}} }   \vdash   \constraintfont{ \constraintfont{C_{{\mathrm{1}}}} } $ and $ \constraintfont{ \ottnt{Q_{{\mathrm{2}}}} }   \vdash   \constraintfont{ \constraintfont{C_{{\mathrm{2}}}} } $. By induction,
          $ \constraintfont{   \multiplicityfont{ \pi }  \scale \ottnt{Q_{{\mathrm{1}}}}  }   \vdash   \constraintfont{   \multiplicityfont{ \pi }  \scale \constraintfont{C_{{\mathrm{1}}}}  } $ and $ \constraintfont{   \multiplicityfont{ \pi }  \scale \ottnt{Q_{{\mathrm{2}}}}  }   \vdash   \constraintfont{   \multiplicityfont{ \pi }  \scale \constraintfont{C_{{\mathrm{2}}}}  } $, then by
          \rref*{C-Tensor},
          $  \constraintfont{   \multiplicityfont{ \pi }  \scale \ottnt{Q}  }  = \constraintfont{   \multiplicityfont{ \pi }  \scale \ottnt{Q_{{\mathrm{1}}}}   \qtensor    \multiplicityfont{ \pi }  \scale \ottnt{Q_{{\mathrm{2}}}}  }   \vdash   \constraintfont{   \multiplicityfont{ \pi }  \scale \constraintfont{C_{{\mathrm{1}}}}   \qtensor    \multiplicityfont{ \pi }  \scale \constraintfont{C_{{\mathrm{2}}}}  } =  \constraintfont{   \multiplicityfont{ \pi }  \scale \constraintfont{C}  }  $.
    \item \Rref*{C-With}: by case of $  \multiplicityfont{ \pi }  $, if $  \multiplicityfont{ \pi }  =  \multiplicityfont{ \ottsym{1} }  $, then the
          result holds by definition. If $  \multiplicityfont{ \pi }  =  \multiplicityfont{ \omega }  $, then and
          $  \constraintfont{ \constraintfont{C} }  =  \constraintfont{ \constraintfont{C_{{\mathrm{1}}}}  \aand  \constraintfont{C_{{\mathrm{2}}}} }  $ with $ \constraintfont{ \ottnt{Q} }   \vdash   \constraintfont{ \constraintfont{C_{{\mathrm{1}}}} } $ and $ \constraintfont{ \ottnt{Q} }   \vdash   \constraintfont{ \constraintfont{C_{{\mathrm{2}}}} } $. By
          induction, $ \constraintfont{   \multiplicityfont{ \omega }  \scale \ottnt{Q}  }   \vdash   \constraintfont{   \multiplicityfont{ \omega }  \scale \constraintfont{C_{{\mathrm{1}}}}  } $ and $ \constraintfont{   \multiplicityfont{ \omega }  \scale \ottnt{Q}  }   \vdash   \constraintfont{   \multiplicityfont{ \omega }  \scale \constraintfont{C_{{\mathrm{2}}}}  } $,
          and we conclude with \rref*{C-Tensor} (not \rref*{C-With}, due to the
          definition of $  \constraintfont{   \multiplicityfont{ \omega }  \scale \ottsym{(}  \constraintfont{C_{{\mathrm{1}}}}  \aand  \constraintfont{C_{{\mathrm{2}}}}  \ottsym{)}  }  $)): 
          $  \constraintfont{   \multiplicityfont{ \omega }  \scale \ottnt{Q}  }  = \constraintfont{   \multiplicityfont{ \omega }  \scale \ottnt{Q}   \qtensor    \multiplicityfont{ \omega }  \scale \ottnt{Q}  }   \vdash   \constraintfont{   \multiplicityfont{ \omega }  \scale \constraintfont{C_{{\mathrm{1}}}}   \qtensor    \multiplicityfont{ \omega }  \scale \constraintfont{C_{{\mathrm{2}}}}  } =  \constraintfont{   \multiplicityfont{ \omega }  \scale \ottsym{(}  \constraintfont{C_{{\mathrm{1}}}}  \aand  \constraintfont{C_{{\mathrm{2}}}}  \ottsym{)}  }  =  \constraintfont{   \multiplicityfont{ \omega }  \scale \constraintfont{C}  }  $.
    \item \Rref*{C-Impl}: straightforward as scaling is built in the rule.
  \end{itemize}
\end{proof}

\begin{lemmarep}[Inversion of scaling]
  \label{lem:wanted:demote}
  If $ \constraintfont{ \ottnt{Q} }   \vdash   \constraintfont{   \multiplicityfont{ \pi }  \scale \constraintfont{C}  } $ then $ \constraintfont{ \ottnt{Q'} }   \vdash   \constraintfont{ \constraintfont{C} } $ and $ \constraintfont{ \ottnt{Q} }   \Vdash   \constraintfont{   \multiplicityfont{ \pi }  \scale \ottnt{Q'}  } $ for some $  \constraintfont{ \ottnt{Q'} }  $.
\end{lemmarep}
\begin{proof}
  By induction on the proof of of $ \constraintfont{ \ottnt{Q} }   \vdash   \constraintfont{   \multiplicityfont{ \pi }  \scale \constraintfont{C}  } $
  \begin{description}
    \item[\Rref*{C-Dom}] We have $ \constraintfont{ \ottnt{Q_{{\mathrm{2}}}} }   \vdash   \constraintfont{   \multiplicityfont{ \pi }  \scale \constraintfont{C}  } $ and $ \constraintfont{ \ottnt{Q} }   \Vdash   \constraintfont{ \ottnt{Q_{{\mathrm{2}}}} } $. By induction
          hypothesis there exists $ \constraintfont{ \ottnt{Q'} }   \vdash   \constraintfont{ \constraintfont{C} } $ and $ \constraintfont{ \ottnt{Q_{{\mathrm{2}}}} }   \Vdash   \constraintfont{ \ottnt{Q'} } $. Then, by
          transitivity of the entailment relation $ \constraintfont{ \ottnt{Q} }   \Vdash   \constraintfont{ \ottnt{Q'} } $ as required.
    \item[\Rref*{C-Id}] We have $  \constraintfont{   \multiplicityfont{ \pi }  \scale \constraintfont{C}  }   =   \constraintfont{   \multiplicityfont{ \pi }  \scale \ottnt{Q'}  }  =  \constraintfont{ \ottnt{Q} }  $, so the result holds by
          reflexivity of the entailment relation.
    \item[\Rref*{C-Tensor}] The conclusion is of the form $  \constraintfont{ \ottnt{Q} }  =  \constraintfont{ \ottnt{Q_{{\mathrm{1}}}}  \qtensor  \ottnt{Q_{{\mathrm{2}}}} }  $,
          $ \constraintfont{ \ottnt{Q_{{\mathrm{1}}}} }   \vdash   \constraintfont{   \multiplicityfont{ \pi }  \scale \constraintfont{C_{{\mathrm{1}}}}   \qtensor    \multiplicityfont{ \pi }  \scale \constraintfont{C_{{\mathrm{2}}}}  } $, and we have $ \constraintfont{ \ottnt{Q_{{\mathrm{1}}}} }   \vdash   \constraintfont{   \multiplicityfont{ \pi }  \scale \constraintfont{C_{{\mathrm{1}}}}  } $ and
          $ \constraintfont{ \ottnt{Q_{{\mathrm{2}}}} }   \vdash   \constraintfont{   \multiplicityfont{ \pi }  \scale \constraintfont{C_{{\mathrm{2}}}}  } $ as the premises.

          By induction we have $ \constraintfont{ \ottnt{Q'_{{\mathrm{1}}}} }   \vdash   \constraintfont{ \constraintfont{C_{{\mathrm{1}}}} } $ and $ \constraintfont{ \ottnt{Q'_{{\mathrm{2}}}} }   \vdash   \constraintfont{ \constraintfont{C_{{\mathrm{2}}}} } $ with
          $ \constraintfont{ \ottnt{Q_{{\mathrm{1}}}} }   \Vdash   \constraintfont{   \multiplicityfont{ \pi }  \scale \ottnt{Q'_{{\mathrm{1}}}}  } $ and $ \constraintfont{ \ottnt{Q_{{\mathrm{2}}}} }   \Vdash   \constraintfont{   \multiplicityfont{ \pi }  \scale \ottnt{Q'_{{\mathrm{2}}}}  } $. by \Rref*{C-Tensor}
          $ \constraintfont{ \ottnt{Q'_{{\mathrm{1}}}}  \qtensor  \ottnt{Q'_{{\mathrm{2}}}} }   \vdash   \constraintfont{ \constraintfont{C_{{\mathrm{1}}}}  \qtensor  \constraintfont{C_{{\mathrm{2}}}} } $ and by tensoring of the entailment relation
          (\cref{fig:entailment-relation}) $ \constraintfont{ \ottnt{Q_{{\mathrm{1}}}}  \qtensor  \ottnt{Q_{{\mathrm{2}}}} }   \Vdash   \constraintfont{   \multiplicityfont{ \pi }  \scale \ottnt{Q_{{\mathrm{1}}}}   \qtensor    \multiplicityfont{ \pi }  \scale \ottnt{Q_{{\mathrm{2}}}}  } $ as
          required.
    \item[\Rref*{C-With}] We have that $  \multiplicityfont{ \pi }  =  \multiplicityfont{ \ottsym{1} }  $ since $  \constraintfont{ \constraintfont{C_{{\mathrm{1}}}}  \aand  \constraintfont{C_{{\mathrm{2}}}} }  $ is
          not in the image of scaling by $  \multiplicityfont{ \omega }  $. Therefore we can choose
          $\ottnt{Q'}=\ottnt{Q}$ and the result holds trivially.
    \item[\Rref*{C-Impl}] In the conclusion of the rule, $\ottnt{Q}$ is already of
          the form $  \multiplicityfont{ \pi }  \scale \ottnt{Q'} $, so the result holds trivially
  \end{description}
\end{proof}

\subsection{Constraint Generation}
\label{sec:constraint-generation}
\label{sec:constraint-generation-soundness}

The process of inferring constraints is split into two parts: generating
constraints, which we do in this section, then solving them in
\cref{sec:constraint-solver}. Constraint generation is described by
the judgement $\Gamma  \vdashi  \ottnt{e}  \ottsym{:}  \tau  \leadsto   \constraintfont{ \constraintfont{C} } $ (defined in
\cref{fig:constraint-generation}) which outputs a constraint $  \constraintfont{ \constraintfont{C} }  $
required to make $\ottnt{e}$ typecheck.
The definition
$\Gamma  \vdashi  \ottnt{e}  \ottsym{:}  \tau  \leadsto   \constraintfont{ \constraintfont{C} } $ is syntax directed, so it can directly be read as an
algorithm, taking as input a \emph{typing derivation} for $\Gamma  \vdash  \ottnt{e}  \ottsym{:}  \tau$
(produced by an external type inference oracle as discussed above). Notably, the
algorithm has access to the context splitting from the (previously computed)
typing derivation, and is
thus indeed syntax directed.
\info{See Fig.13, p39 of OutsideIn~\cite{OutsideIn}}

\info{Not caring about inferences simplifies $\packbox$ quite a bit, we
  are using the pseudo-inferred type to generate constraint. In a real
  system, we would need $\packbox$ to know its type (\emph{e.g.} using
  bidirectional type checking).}
\begin{figure}
  \maybesmall
  \centering
  \drules[G]{$\Gamma  \vdashi  \ottnt{e}  \ottsym{:}  \tau  \leadsto   \constraintfont{ \constraintfont{C} } $}{Constraint generation}{Var, Abs,
    App, Pack, Unpack, Case, Let, LetSig}

  \caption{Constraint generation}
  \label{fig:constraint-generation}
\end{figure}

The rules of \cref{fig:constraint-generation} constitute a mostly
unsurprising translation of the rules of \cref{fig:typing-rules},
except for the following points of interest:

\emph{Case expressions.}
Note the use of $ \aand $ in the conclusion of \rref{G-Case}.
We require that each branch of a $\kcase$ expression use the exact
same (linear) assumptions; this is enforced by combining the
emitted constraints with $ \aand $, not $ \qtensor $.
  This can also be understood in terms of the array example of
  \cref{sec:introduction}:
  if an array is freed in one branch of a $\kcase$, we require it to be freed (or frozen) in
the other branches too.
  Otherwise, the array's state will be unknown to the type system
  after the $\kcase$.

\emph{Implications.} The introduction of constraints local to a
  definition (\rref{G-LetSig}) corresponds to
  emitting an implication constraint.

\emph{Unannotated \ensuremath{\mathbf{let}}.}
 However, the~\rref*{G-Let} rule does not produce an implication
  constraint, as we do not model \ensuremath{\mathbf{let}}-generalisation~\cite{let-should-not-be-generalised}.

\vspace{1ex}
The key property of the constraint-generation algorithm is that,
if the generated constraint is solvable, then we can indeed type the
term in the qualified type system of
\cref{sec:qualified-type-system}. That is,
these rules are simply an implementation of our declarative qualified
type system.

\begin{lemmarep}[Soundness of constraint generation]\label{lem:generation-soundness}
  For all $  \constraintfont{ \ottnt{Q_{\ottmv{g}}} }  $, if $\Gamma  \vdashi  \ottnt{e}  \ottsym{:}  \tau  \leadsto   \constraintfont{ \constraintfont{C} } $ and $ \constraintfont{ \ottnt{Q_{\ottmv{g}}} }   \vdash   \constraintfont{ \constraintfont{C} } $ then
  $ \constraintfont{ \ottnt{Q_{\ottmv{g}}} }   \ottsym{;}  \Gamma  \vdash  \ottnt{e}  \ottsym{:}  \tau$.
\end{lemmarep}
\begin{proof}
  By induction on $\Gamma  \vdashi  \ottnt{e}  \ottsym{:}  \tau  \leadsto   \constraintfont{ \constraintfont{C} } $
  \begin{description}
  \item[\rref*{G-Var}] We have
    \begin{itemize}
    \item $\Gamma_{{\mathrm{1}}}  \ottsym{=}   \ottmv{x} {:}_{  \multiplicityfont{ \ottsym{1} }  }  \forall   \overline{\ottmv{a} } .   \constraintfont{ \ottnt{Q} }   \Lolly  \upsilon  $
    \item $\Gamma_{{\mathrm{1}}}  \ottsym{+}    \multiplicityfont{ \omega }  \scale \Gamma_{{\mathrm{2}}}   \vdashi  \ottmv{x}  \ottsym{:}  \upsilon  \ottsym{[}  \overline{\tau}  \ottsym{/}  \overline{\ottmv{a} }  \ottsym{]}  \leadsto   \constraintfont{ \ottnt{Q}  \ottsym{[}  \overline{\tau}  \ottsym{/}  \overline{\ottmv{a} }  \ottsym{]} } $
    \item $ \constraintfont{ \ottnt{Q_{\ottmv{g}}} }   \vdash   \constraintfont{ \ottnt{Q}  \ottsym{[}  \overline{\tau}  \ottsym{/}  \overline{\ottmv{a} }  \ottsym{]} } $
    \end{itemize}
    Therefore, by rules~\rref*{E-Var} and~\rref*{E-Sub}, it follows
    immediately that $ \constraintfont{ \ottnt{Q_{\ottmv{g}}} }   \ottsym{;}  \Gamma_{{\mathrm{1}}}  \ottsym{+}    \multiplicityfont{ \omega }  \scale \Gamma_{{\mathrm{2}}}   \vdash  \ottmv{x}  \ottsym{:}  \upsilon  \ottsym{[}  \overline{\tau}  \ottsym{/}  \overline{\ottmv{a} }  \ottsym{]}$
  \item[\rref*{G-Abs}] We have
    \begin{itemize}
    \item $\Gamma  \vdashi  \lambda  \ottmv{x}  \ottsym{.}  \ottnt{e}  \ottsym{:}   \tau_{{\mathrm{0}}}  \to_{  \multiplicityfont{ \pi }  }  \tau   \leadsto   \constraintfont{ \constraintfont{C} } $
    \item $ \constraintfont{ \ottnt{Q_{\ottmv{g}}} }   \vdash   \constraintfont{ \constraintfont{C} } $
    \item $\Gamma  \ottsym{,}   \ottmv{x} {:}_{  \multiplicityfont{ \pi }  } \tau_{{\mathrm{0}}}   \vdashi  \ottnt{e}  \ottsym{:}  \tau  \leadsto   \constraintfont{ \constraintfont{C} } $
    \end{itemize}
    By induction hypothesis we have
    \begin{itemize}
    \item $ \constraintfont{ \ottnt{Q_{\ottmv{g}}} }   \ottsym{;}  \Gamma  \ottsym{,}   \ottmv{x} {:}_{  \multiplicityfont{ \pi }  } \tau_{{\mathrm{0}}}   \vdash  \ottnt{e}  \ottsym{:}  \tau$
    \end{itemize}
    From which follows that $ \constraintfont{ \ottnt{Q_{\ottmv{g}}} }   \ottsym{;}  \Gamma  \vdash  \lambda  \ottmv{x}  \ottsym{.}  \ottnt{e}  \ottsym{:}   \tau_{{\mathrm{0}}}  \to_{  \multiplicityfont{ \pi }  }  \tau $.
  \item[\rref*{G-Let}] We have
    \begin{itemize}
    \item $  \multiplicityfont{ \pi }  \scale \Gamma_{{\mathrm{1}}}   \ottsym{+}  \Gamma_{{\mathrm{2}}}  \vdashi   \klet_  \multiplicityfont{ \pi }   \, \ottmv{x}  \ottsym{=}  \ottnt{e_{{\mathrm{1}}}} \, \ottkw{in} \, \ottnt{e_{{\mathrm{2}}}}  \ottsym{:}  \tau  \leadsto   \constraintfont{   \multiplicityfont{ \pi }  \scale \constraintfont{C_{{\mathrm{1}}}}   \qtensor  \constraintfont{C_{{\mathrm{2}}}} } $
    \item $ \constraintfont{ \ottnt{Q_{\ottmv{g}}} }   \vdash   \constraintfont{   \multiplicityfont{ \pi }  \scale \constraintfont{C_{{\mathrm{1}}}}   \qtensor  \constraintfont{C_{{\mathrm{2}}}} } $
    \item $\Gamma_{{\mathrm{2}}}  \ottsym{,}   \ottmv{x} {:}_{  \multiplicityfont{ \pi }  } \tau_{{\mathrm{1}}}   \vdashi  \ottnt{e_{{\mathrm{2}}}}  \ottsym{:}  \tau  \leadsto   \constraintfont{ \constraintfont{C_{{\mathrm{2}}}} } $
    \item $\Gamma_{{\mathrm{1}}}  \vdashi  \ottnt{e_{{\mathrm{1}}}}  \ottsym{:}  \tau_{{\mathrm{1}}}  \leadsto   \constraintfont{ \constraintfont{C_{{\mathrm{1}}}} } $
    \end{itemize}
    By \cref{lem:inversion,lem:wanted:demote}, there exist $  \constraintfont{ \ottnt{Q_{{\mathrm{1}}}} }  $
    and $  \constraintfont{ \ottnt{Q_{{\mathrm{2}}}} }  $ such that
    \begin{itemize}
    \item $ \constraintfont{ \ottnt{Q_{{\mathrm{1}}}} }   \vdash   \constraintfont{ \constraintfont{C_{{\mathrm{1}}}} } $
    \item $ \constraintfont{ \ottnt{Q_{{\mathrm{2}}}} }   \vdash   \constraintfont{ \constraintfont{C_{{\mathrm{2}}}} } $
    \item $ \constraintfont{ \ottnt{Q_{\ottmv{g}}} }   \Vdash   \constraintfont{   \multiplicityfont{ \pi }  \scale \ottnt{Q_{{\mathrm{1}}}}   \qtensor  \ottnt{Q_{{\mathrm{2}}}} } $
    \end{itemize}
    By induction hypothesis we have
    \begin{itemize}
    \item $ \constraintfont{ \ottnt{Q_{{\mathrm{1}}}} }   \ottsym{;}  \Gamma_{{\mathrm{1}}}  \vdash  \ottnt{e_{{\mathrm{1}}}}  \ottsym{:}  \tau_{{\mathrm{1}}}$
    \item $ \constraintfont{ \ottnt{Q_{{\mathrm{2}}}} }   \ottsym{;}  \Gamma_{{\mathrm{2}}}  \ottsym{,}   \ottmv{x} {:}_{  \multiplicityfont{ \pi }  } \tau_{{\mathrm{1}}}   \vdash  \ottnt{e_{{\mathrm{1}}}}  \ottsym{:}  \tau_{{\mathrm{1}}}$
    \end{itemize}
    From which follows (by \rref*{E-Let} and \rref*{E-Sub}) that $ \constraintfont{ \ottnt{Q_{\ottmv{g}}} }   \ottsym{;}    \multiplicityfont{ \pi }  \scale \Gamma_{{\mathrm{1}}}   \ottsym{+}  \Gamma_{{\mathrm{2}}}  \vdash   \klet_  \multiplicityfont{ \pi }   \, \ottmv{x}  \ottsym{=}  \ottnt{e_{{\mathrm{1}}}} \, \ottkw{in} \, \ottnt{e_{{\mathrm{2}}}}  \ottsym{:}  \tau$.
  \item[\rref*{G-LetSig}] We have
    \begin{itemize}
    \item $  \multiplicityfont{ \pi }  \scale \Gamma_{{\mathrm{1}}}   \ottsym{+}  \Gamma_{{\mathrm{2}}}  \vdashi   \klet_  \multiplicityfont{ \pi }   \, \ottmv{x}  \ottsym{:}   \forall   \overline{\ottmv{a} } .   \constraintfont{ \ottnt{Q} }   \Lolly  \tau_{{\mathrm{1}}}   \ottsym{=}  \ottnt{e_{{\mathrm{1}}}} \, \ottkw{in} \, \ottnt{e_{{\mathrm{2}}}}  \ottsym{:}  \tau  \leadsto   \constraintfont{ \constraintfont{C_{{\mathrm{2}}}}  \qtensor    \multiplicityfont{ \pi }  \scale ( \ottnt{Q}  \Lolly  \constraintfont{C_{{\mathrm{1}}}} )  } $
    \item $ \constraintfont{ \ottnt{Q_{\ottmv{g}}} }   \vdash   \constraintfont{ \constraintfont{C_{{\mathrm{2}}}}  \qtensor    \multiplicityfont{ \pi }  \scale ( \ottnt{Q}  \Lolly  \constraintfont{C_{{\mathrm{1}}}} )  } $
    \item $\Gamma_{{\mathrm{1}}}  \vdashi  \ottnt{e_{{\mathrm{1}}}}  \ottsym{:}  \tau_{{\mathrm{1}}}  \leadsto   \constraintfont{ \constraintfont{C_{{\mathrm{1}}}} } $
    \item $\Gamma_{{\mathrm{2}}}  \ottsym{,}   \ottmv{x} {:}_{  \multiplicityfont{ \pi }  }  \forall   \overline{\ottmv{a} } .   \constraintfont{ \ottnt{Q} }   \Lolly  \tau_{{\mathrm{1}}}    \vdashi  \ottnt{e_{{\mathrm{2}}}}  \ottsym{:}  \tau  \leadsto   \constraintfont{ \constraintfont{C_{{\mathrm{2}}}} } $
    \end{itemize}
    By \cref{lem:inversion,lem:wanted:demote}, there exist $  \constraintfont{ \ottnt{Q_{{\mathrm{1}}}} }  $,
    $  \constraintfont{ \ottnt{Q_{{\mathrm{2}}}} }  $ such
    that
    \begin{itemize}
    \item $ \constraintfont{ \ottnt{Q_{{\mathrm{2}}}} }   \vdash   \constraintfont{ \constraintfont{C_{{\mathrm{2}}}} } $
    \item $ \constraintfont{ \ottnt{Q_{{\mathrm{1}}}}  \qtensor  \ottnt{Q} }   \vdash   \constraintfont{ \constraintfont{C_{{\mathrm{1}}}} } $
    \item $ \constraintfont{ \ottnt{Q_{\ottmv{g}}} }   \Vdash   \constraintfont{   \multiplicityfont{ \pi }  \scale \ottnt{Q_{{\mathrm{1}}}}   \qtensor  \ottnt{Q_{{\mathrm{2}}}} } $
    \end{itemize}
    By induction hypothesis
    \begin{itemize}
    \item $ \constraintfont{ \ottnt{Q_{{\mathrm{1}}}}  \qtensor  \ottnt{Q} }   \ottsym{;}  \Gamma_{{\mathrm{1}}}  \vdash  \ottnt{e_{{\mathrm{1}}}}  \ottsym{:}  \tau_{{\mathrm{1}}}$
    \item $ \constraintfont{ \ottnt{Q_{{\mathrm{2}}}} }   \ottsym{;}  \Gamma_{{\mathrm{2}}}  \ottsym{,}   \ottmv{x} {:}_{  \multiplicityfont{ \pi }  }  \forall   \overline{\ottmv{a} } .   \constraintfont{ \ottnt{Q} }   \Lolly  \tau_{{\mathrm{1}}}    \vdash  \ottnt{e_{{\mathrm{2}}}}  \ottsym{:}  \tau$
    \end{itemize}
    Hence (by \rref*{E-LetSig} and \rref*{E-Sub}) $ \constraintfont{ \ottnt{Q_{\ottmv{g}}} }   \ottsym{;}    \multiplicityfont{ \pi }  \scale \Gamma_{{\mathrm{1}}}   \ottsym{+}  \Gamma_{{\mathrm{2}}}  \vdash   \klet_  \multiplicityfont{ \pi }   \, \ottmv{x}  \ottsym{:}   \forall   \overline{\ottmv{a} } .   \constraintfont{ \ottnt{Q} }   \Lolly  \tau_{{\mathrm{1}}}   \ottsym{=}  \ottnt{e_{{\mathrm{1}}}} \, \ottkw{in} \, \ottnt{e_{{\mathrm{2}}}}  \ottsym{:}  \tau$
  \item[\rref*{G-App}] \info{Most of the linearity problems are in the App
      rule. Unpack is also relevant.}
    We have
    \begin{itemize}
    \item $\Gamma_{{\mathrm{1}}}  \ottsym{+}    \multiplicityfont{ \pi }  \scale \Gamma_{{\mathrm{2}}}   \vdashi  \ottnt{e_{{\mathrm{1}}}} \, \ottnt{e_{{\mathrm{2}}}}  \ottsym{:}  \tau  \leadsto   \constraintfont{ \constraintfont{C_{{\mathrm{1}}}}  \qtensor    \multiplicityfont{ \pi }  \scale \constraintfont{C_{{\mathrm{2}}}}  } $
    \item $ \constraintfont{ \ottnt{Q_{\ottmv{g}}} }   \vdash   \constraintfont{ \constraintfont{C_{{\mathrm{1}}}}  \qtensor    \multiplicityfont{ \pi }  \scale \constraintfont{C_{{\mathrm{2}}}}  } $
    \item $\Gamma_{{\mathrm{1}}}  \vdashi  \ottnt{e_{{\mathrm{1}}}}  \ottsym{:}   \tau_{{\mathrm{2}}}  \to_{  \multiplicityfont{ \pi }  }  \tau   \leadsto   \constraintfont{ \constraintfont{C_{{\mathrm{1}}}} } $
    \item $\Gamma_{{\mathrm{2}}}  \vdashi  \ottnt{e_{{\mathrm{2}}}}  \ottsym{:}  \tau_{{\mathrm{2}}}  \leadsto   \constraintfont{ \constraintfont{C_{{\mathrm{2}}}} } $
    \end{itemize}
    By \cref{lem:inversion,lem:wanted:demote}, there exist
    $  \constraintfont{ \ottnt{Q_{{\mathrm{1}}}} }  $, $  \constraintfont{ \ottnt{Q_{{\mathrm{2}}}} }  $ such that
    \begin{itemize}
    \item $ \constraintfont{ \ottnt{Q_{{\mathrm{1}}}} }   \vdash   \constraintfont{ \constraintfont{C_{{\mathrm{1}}}} } $
    \item $ \constraintfont{ \ottnt{Q_{{\mathrm{2}}}} }   \vdash   \constraintfont{ \constraintfont{C_{{\mathrm{2}}}} } $
    \item $ \constraintfont{ \ottnt{Q_{\ottmv{g}}} }   \Vdash   \constraintfont{ \ottnt{Q_{{\mathrm{1}}}}  \qtensor    \multiplicityfont{ \pi }  \scale \ottnt{Q_{{\mathrm{2}}}}  } $
    \end{itemize}
    By induction hypothesis
    \begin{itemize}
    \item $ \constraintfont{ \ottnt{Q_{{\mathrm{1}}}} }   \ottsym{;}  \Gamma_{{\mathrm{1}}}  \vdash  \ottnt{e_{{\mathrm{1}}}}  \ottsym{:}   \tau_{{\mathrm{2}}}  \to_{  \multiplicityfont{ \pi }  }  \tau $
    \item $ \constraintfont{ \ottnt{Q_{{\mathrm{2}}}} }   \ottsym{;}  \Gamma_{{\mathrm{2}}}  \vdash  \ottnt{e_{{\mathrm{2}}}}  \ottsym{:}  \tau_{{\mathrm{2}}}$
    \end{itemize}
    Hence (by \rref*{E-App} and \rref*{E-Sub}) $ \constraintfont{ \ottnt{Q_{\ottmv{g}}} }   \ottsym{;}  \Gamma_{{\mathrm{1}}}  \ottsym{+}    \multiplicityfont{ \pi }  \scale \Gamma_{{\mathrm{2}}}   \vdash  \ottnt{e_{{\mathrm{1}}}} \, \ottnt{e_{{\mathrm{2}}}}  \ottsym{:}  \tau$.
  \item[\rref*{G-Pack}] We have
    \begin{itemize}
    \item $\Gamma  \vdashi  \packbox \, \ottnt{e}  \ottsym{:}   \exists   \overline{\ottmv{a} } .  \tau  \RLolly   \constraintfont{ \ottnt{Q} }    \leadsto   \constraintfont{ \constraintfont{C}  \qtensor  \ottnt{Q}  \ottsym{[}  \overline{\upsilon}  \ottsym{/}  \overline{\ottmv{a} }  \ottsym{]} } $
    \item $ \constraintfont{ \ottnt{Q_{\ottmv{g}}} }   \vdash   \constraintfont{ \constraintfont{C}  \qtensor  \ottnt{Q}  \ottsym{[}  \overline{\upsilon}  \ottsym{/}  \overline{\ottmv{a} }  \ottsym{]} } $
    \item $\Gamma  \vdashi  \ottnt{e}  \ottsym{:}  \tau  \ottsym{[}  \overline{\upsilon}  \ottsym{/}  \overline{\ottmv{a} }  \ottsym{]}  \leadsto   \constraintfont{ \constraintfont{C} } $
    \end{itemize}
    By \cref{lem:inversion}, there exist $  \constraintfont{ \ottnt{Q_{{\mathrm{1}}}} }  $, $  \constraintfont{ \ottnt{Q_{{\mathrm{2}}}} }  $
    such that
    \begin{itemize}
    \item $ \constraintfont{ \ottnt{Q_{{\mathrm{1}}}} }   \vdash   \constraintfont{ \constraintfont{C} } $
    \item $ \constraintfont{ \ottnt{Q_{{\mathrm{2}}}} }   \vdash   \constraintfont{ \ottnt{Q}  \ottsym{[}  \overline{\upsilon}  \ottsym{/}  \overline{\ottmv{a} }  \ottsym{]} } $
    \item $ \constraintfont{ \ottnt{Q_{\ottmv{g}}} }   \Vdash   \constraintfont{ \ottnt{Q_{{\mathrm{1}}}}  \qtensor  \ottnt{Q_{{\mathrm{2}}}} } $
    \end{itemize}
    By induction hypothesis
    \begin{itemize}
    \item $ \constraintfont{ \ottnt{Q_{{\mathrm{1}}}} }   \ottsym{;}  \Gamma  \vdash  \ottnt{e}  \ottsym{:}  \tau  \ottsym{[}  \overline{\upsilon}  \ottsym{/}  \overline{\ottmv{a} }  \ottsym{]}$
    \end{itemize}
    So, by \rref{E-Pack}, we have $ \constraintfont{ \ottnt{Q_{{\mathrm{1}}}}  \qtensor  \ottnt{Q}  \ottsym{[}  \overline{\upsilon}  \ottsym{/}  \overline{\ottmv{a} }  \ottsym{]} }   \ottsym{;}  \Gamma  \vdash  \packbox \, \ottnt{e}  \ottsym{:}   \exists   \overline{\ottmv{a} } .  \tau  \RLolly   \constraintfont{ \ottnt{Q} }  $. By \rref{E-Sub}, we conclude
    $ \constraintfont{ \ottnt{Q_{\ottmv{g}}} }   \ottsym{;}    \multiplicityfont{ \omega }  \scale \Gamma   \vdash  \packbox \, \ottnt{e}  \ottsym{:}   \exists   \overline{\ottmv{a} } .  \tau  \RLolly   \constraintfont{ \ottnt{Q} }  $.
  \item[\rref*{G-Unpack}] We have
    \begin{itemize}
    \item $\Gamma_{{\mathrm{1}}}  \ottsym{+}  \Gamma_{{\mathrm{2}}}  \vdashi   \klet\ \packbox  \ottmv{x}  =  \ottnt{e_{{\mathrm{1}}}}  \  \ottkw{in}  \  \ottnt{e_{{\mathrm{2}}}}   \ottsym{:}  \tau  \leadsto   \constraintfont{ \constraintfont{C_{{\mathrm{1}}}}  \qtensor    \multiplicityfont{ \ottsym{1} }  \scale ( \ottnt{Q'}  \Lolly  \constraintfont{C_{{\mathrm{2}}}} )  } $
    \item $ \constraintfont{ \ottnt{Q_{\ottmv{g}}} }   \vdash   \constraintfont{ \constraintfont{C_{{\mathrm{1}}}}  \qtensor    \multiplicityfont{ \ottsym{1} }  \scale ( \ottnt{Q'}  \Lolly  \constraintfont{C_{{\mathrm{2}}}} )  } $
    \item $\Gamma_{{\mathrm{1}}}  \vdashi  \ottnt{e_{{\mathrm{1}}}}  \ottsym{:}   \exists   \overline{\ottmv{a} } .  \tau_{{\mathrm{1}}}  \RLolly   \constraintfont{ \ottnt{Q'} }    \leadsto   \constraintfont{ \constraintfont{C_{{\mathrm{1}}}} } $
    \item $\Gamma_{{\mathrm{2}}}  \ottsym{,}   \ottmv{x} {:}_{  \multiplicityfont{ \pi }  } \tau_{{\mathrm{1}}}   \vdashi  \ottnt{e_{{\mathrm{2}}}}  \ottsym{:}  \tau  \leadsto   \constraintfont{ \constraintfont{C_{{\mathrm{2}}}} } $
    \end{itemize}
    By \cref{lem:inversion}, there exist $  \constraintfont{ \ottnt{Q_{{\mathrm{1}}}} }  $, $  \constraintfont{ \ottnt{Q_{{\mathrm{2}}}} }  $
    such that
    \begin{itemize}
    \item $ \constraintfont{ \ottnt{Q_{{\mathrm{1}}}} }   \vdash   \constraintfont{ \constraintfont{C_{{\mathrm{1}}}} } $
    \item $ \constraintfont{ \ottnt{Q_{{\mathrm{2}}}}  \qtensor  \ottnt{Q'} }   \vdash   \constraintfont{ \constraintfont{C_{{\mathrm{2}}}} } $
    \item $ \constraintfont{ \ottnt{Q_{\ottmv{g}}} }   \Vdash   \constraintfont{ \ottnt{Q_{{\mathrm{1}}}}  \qtensor  \ottnt{Q_{{\mathrm{2}}}} } $
    \end{itemize}
    By induction hypothesis
    \begin{itemize}
    \item $ \constraintfont{ \ottnt{Q_{{\mathrm{1}}}} }   \ottsym{;}  \Gamma_{{\mathrm{1}}}  \vdash  \ottnt{e_{{\mathrm{1}}}}  \ottsym{:}   \exists   \overline{\ottmv{a} } .  \tau_{{\mathrm{1}}}  \RLolly   \constraintfont{ \ottnt{Q'} }  $
    \item $ \constraintfont{ \ottnt{Q_{{\mathrm{2}}}}  \qtensor  \ottnt{Q} }   \ottsym{;}  \Gamma_{{\mathrm{2}}}  \vdash  \ottnt{e_{{\mathrm{2}}}}  \ottsym{:}  \tau$
    \end{itemize}
    Therefore (by \rref*{E-Unpack} and \rref*{E-Sub}) $ \constraintfont{ \ottnt{Q_{\ottmv{g}}} }   \ottsym{;}  \Gamma_{{\mathrm{1}}}  \ottsym{+}  \Gamma_{{\mathrm{2}}}  \vdash   \klet\ \packbox  \ottmv{x}  =  \ottnt{e_{{\mathrm{1}}}}  \  \ottkw{in}  \  \ottnt{e_{{\mathrm{2}}}}   \ottsym{:}  \tau$.
  \item[\rref*{G-Case}] We have
    \begin{itemize}
    \item $  \multiplicityfont{ \pi }  \scale \Gamma   \ottsym{+}  \Delta  \vdashi   \kcase_  \multiplicityfont{ \pi }   \, \ottnt{e} \, \ottkw{of} \, \ottsym{\{}  \overline{\ottmv{K}_i\ \overline{\ottmv{x}_i } \to \ottnt{e}_i }  \ottsym{\}}  \ottsym{:}  \tau  \leadsto   \constraintfont{   \multiplicityfont{ \pi }  \scale \constraintfont{C}   \qtensor  \bigaand  \constraintfont{C_{\ottmv{i}}} } $
    \item $ \constraintfont{ \ottnt{Q_{\ottmv{g}}} }   \vdash   \constraintfont{   \multiplicityfont{ \pi }  \scale \constraintfont{C}   \qtensor  \bigaand  \constraintfont{C_{\ottmv{i}}} } $
    \item $\Gamma  \vdashi  \ottnt{e}  \ottsym{:}  \ottmv{T} \, \overline{\sigma}  \leadsto   \constraintfont{ \constraintfont{C} } $
    \item For each $i$, $\Delta  \ottsym{,}   \overline{  \ottmv{x_{\ottmv{i}}} {:}_{  \multiplicityfont{ \ottsym{(}   \pi {⋅} \pi_{\ottmv{i}}   \ottsym{)} }  } \upsilon_{\ottmv{i}}  \ottsym{[}  \overline{\sigma}  \ottsym{/}  \overline{\ottmv{a} }  \ottsym{]}  }   \vdashi  \ottnt{e_{\ottmv{i}}}  \ottsym{:}  \tau  \leadsto   \constraintfont{ \constraintfont{C_{\ottmv{i}}} } $
    \end{itemize}
    By repeated uses of \cref{lem:inversion} as well as \cref{lem:wanted:demote}, there exist
    $  \constraintfont{ \ottnt{Q} }  $, $  \constraintfont{ \ottnt{Q'} }  $ such that
    \begin{itemize}
    \item $ \constraintfont{ \ottnt{Q} }   \vdash   \constraintfont{ \constraintfont{C} } $
    \item For each $i$, $ \constraintfont{ \ottnt{Q'} }   \vdash   \constraintfont{ \constraintfont{C_{\ottmv{i}}} } $
    \item $ \constraintfont{ \ottnt{Q_{\ottmv{g}}} }   \Vdash   \constraintfont{   \multiplicityfont{ \pi }  \scale \ottnt{Q}   \qtensor  \ottnt{Q'} } $
    \end{itemize}
    By induction hypothesis
    \begin{itemize}
    \item $ \constraintfont{ \ottnt{Q} }   \ottsym{;}  \Gamma  \vdash  \ottnt{e}  \ottsym{:}  \ottmv{T} \, \overline{\sigma}$
    \item For each $i$, $ \constraintfont{ \ottnt{Q'} }   \ottsym{;}  \Delta  \ottsym{,}   \overline{  \ottmv{x_{\ottmv{i}}} {:}_{  \multiplicityfont{ \ottsym{(}   \pi {⋅} \pi_{\ottmv{i}}   \ottsym{)} }  } \upsilon_{\ottmv{i}}  \ottsym{[}  \overline{\sigma}  \ottsym{/}  \overline{\ottmv{a} }  \ottsym{]}  }   \vdash  \ottnt{e_{\ottmv{i}}}  \ottsym{:}  \tau$
    \end{itemize}
    Therefore (by \rref*{E-Case} and \rref*{E-Sub}) $ \constraintfont{ \ottnt{Q_{\ottmv{g}}} }   \ottsym{;}    \multiplicityfont{ \pi }  \scale \Gamma   \ottsym{+}  \Delta  \vdash   \kcase_  \multiplicityfont{ \pi }   \, \ottnt{e} \, \ottkw{of} \, \ottsym{\{}  \overline{\ottmv{K}_i\ \overline{\ottmv{x}_i } \to \ottnt{e}_i }  \ottsym{\}}  \ottsym{:}  \tau$.
  \end{description}
\end{proof}

\subsection{Constraint Solving}
\label{sec:constraint-solver}

In this section, we build a \emph{constraint solver} that decides whether or not
$ \constraintfont{ \ottnt{Q_{\ottmv{g}}} }   \vdash   \constraintfont{ \constraintfont{C} } $ holds, as required by \cref{lem:generation-soundness}.

So far, the constraint domain has been kept abstract. For the sake of
concreteness let us fix a particular constraint domain for the duration of this
section. As atomic constraints, we take an arbitrary countable set. We will only
compare atomic constraints for equality, so the only structure that we require
is that equality of atomic constraints is decidable equality. Furthermore, we fix an
arbitrary decidable subset $ \constraintfont{\mathcal{D} } $ of duplicable atomic constraints (to which
the \ensuremath{\constraintfont{\Conid{Linearly}}} constraint must belong). This is one of the simplest possible
domain but it is sufficient to support our examples throughout this paper. The
entailment relation for this domain is given by a simple recursive definition
shown in \cref{fig:concrete-domain}. The intention for our definition relation
is to mean that a constraint $  \constraintfont{ \ottnt{q} }  $ is entailed if and only if it is already
assumed; with the appropriate linearity restrictions.

\begin{figure}
  \maybesmall
  \centering
  \drules[Q]{$\ottnt{QF_{{\mathrm{1}}}}  \Vdash  \ottnt{QF_{{\mathrm{2}}}}$}{Entailment relation}{Linear, DupOne, DupNone, Ur}

  \caption{Concrete constraint entailment}
  \label{fig:concrete-domain}
\end{figure}

The constraint solver itself is represented by the following judgement:
$$
\vdashs   \constraintfont{ \constraintfont{C} }   \leadsto   \constraintfont{ \ottnt{Q} } 
$$
The judgement doesn't take any context as an input. Instead the judgement outputs a simple constraint
$  \constraintfont{ \ottnt{Q} }  $, which we want to read as a context which entails $  \constraintfont{ \constraintfont{C} }  $ (see
\cref{lem:solver-soundness}).

Returning a context rather than taking it as an argument lets us effectively
count the number of occurrences of each atomic constraint in the solution, which we can then
reject if it is not compatible with the multiplicity of the given constraint.
This is essentially the same technique that \textsc{ghc} uses to implement
linearity of regular terms in its typechecker. Note that we return the whole
context here, whereas \textsc{ghc}'s typechecker only returns multiplicity,
which is more suited for a concrete implementation.

If the constraint solver finds a solution then $\vdashs   \constraintfont{ \constraintfont{C} }   \leadsto   \constraintfont{ \ottnt{Q} } $ then
$  \constraintfont{ \ottnt{Q} }  $ entails $  \constraintfont{ \constraintfont{C} }  $. In order to check that $ \constraintfont{ \ottnt{Q_{\ottmv{g}}} }   \vdash   \constraintfont{ \constraintfont{C} } $, it therefore suffices to
check that $\vdashs   \constraintfont{   \multiplicityfont{ \ottsym{1} }  \scale ( \ottnt{Q_{\ottmv{g}}}  \Lolly  \constraintfont{C} )  }   \leadsto   \constraintfont{  \mathbf{\varepsilon}  } $.

\begin{toappendix}
  \begin{lemma}[Soundness of $\vdashd   \constraintfont{ \ottnt{Q_{\ottmv{i}}} }   \setminus   \constraintfont{ \ottnt{Q_{\ottmv{b}}} }   \leadsto   \constraintfont{ \ottnt{Q_{\ottmv{o}}} } $]\label{l:diff-soundness}
  If $\vdashd   \constraintfont{ \ottnt{Q_{\ottmv{i}}} }   \setminus   \constraintfont{ \ottnt{Q_{\ottmv{b}}} }   \leadsto   \constraintfont{ \ottnt{Q_{\ottmv{o}}} } $ is derivable (see \cref{fig:constraint-solver})),
  then $ \constraintfont{ \ottnt{Q_{\ottmv{o}}}  \qtensor  \ottnt{Q_{\ottmv{b}}} }   \Vdash   \constraintfont{ \ottnt{Q_{\ottmv{i}}} } $.
\end{lemma}
\begin{nestedproof}
  By induction on the derivation of $\vdashd   \constraintfont{ \ottnt{Q_{\ottmv{i}}} }   \setminus   \constraintfont{ \ottnt{Q_{\ottmv{b}}} }   \leadsto   \constraintfont{ \ottnt{Q_{\ottmv{o}}} } $.
  \begin{description}
    \item[\rref*{D-Linear}] We have
          \begin{itemize}
            \item $  \constraintfont{ \ottnt{Q_{\ottmv{b}}} }  =  \constraintfont{ \ottsym{(}  \ottnt{U}_{\ottmv{b}}  \ottsym{,}   \ottnt{L}_{\ottmv{b}}  \uplus  \ottnt{q}   \ottsym{)} }  $
            \item $\vdashd   \constraintfont{ \ottnt{Q_{\ottmv{i}}} }   \setminus   \constraintfont{ \ottsym{(}  \ottnt{U}_{\ottmv{b}}  \ottsym{,}  \ottnt{L}_{\ottmv{b}}  \ottsym{)} }   \leadsto   \constraintfont{ \ottnt{Q_{\ottmv{o}}}  \qtensor    \multiplicityfont{ \ottsym{1} }  \scale \ottnt{q}  } $ hence, by
            induction hypothesis $ \constraintfont{ \ottnt{Q_{\ottmv{o}}}  \qtensor    \multiplicityfont{ \ottsym{1} }  \scale \ottnt{q}   \qtensor  \ottsym{(}  \ottnt{U}_{\ottmv{b}}  \ottsym{,}  \ottnt{L}_{\ottmv{b}}  \ottsym{)} }   \Vdash   \constraintfont{ \ottnt{Q_{\ottmv{i}}} } $
          \end{itemize}
          We need to prove that $ \constraintfont{ \ottnt{Q_{\ottmv{o}}}  \qtensor  \ottsym{(}  \ottnt{U}_{\ottmv{b}}  \ottsym{,}   \ottnt{L}_{\ottmv{b}}  \uplus  \ottnt{q}   \ottsym{)} }   \Vdash   \constraintfont{ \ottnt{Q_{\ottmv{i}}} } $, but
          by definition of the tensor product
          $  \constraintfont{ \ottsym{(}  \ottnt{U}_{\ottmv{b}}  \ottsym{,}   \ottnt{L}_{\ottmv{b}}  \uplus  \ottnt{q}   \ottsym{)} }   =   \constraintfont{   \multiplicityfont{ \ottsym{1} }  \scale \ottnt{q}   \qtensor  \ottsym{(}  \ottnt{U}_{\ottmv{b}}  \ottsym{,}  \ottnt{L}_{\ottmv{b}}  \ottsym{)} }  $, so the
          induction hypothesis proves this case.
    \item[\rref*{D-Dup}] We have
          \begin{itemize}
            \item $  \constraintfont{ \ottnt{q} }   \in  \constraintfont{\mathcal{D} } $
            \item $  \constraintfont{ \ottnt{Q_{\ottmv{b}}} }  =  \constraintfont{ \ottsym{(}  \ottnt{U}_{\ottmv{b}}  \ottsym{,}   \ottnt{L}_{\ottmv{b}}  \uplus  \ottnt{q}   \ottsym{)} }  $
            \item $\vdashd   \constraintfont{ \ottnt{Q_{\ottmv{i}}} }   \setminus   \constraintfont{ \ottsym{(}  \ottnt{U}_{\ottmv{b}}  \ottsym{,}  \ottnt{L}_{\ottmv{b}}  \ottsym{)} }   \leadsto   \constraintfont{ \ottnt{Q_{\ottmv{o}}}  \qtensor   \overline{   \multiplicityfont{ \ottsym{1} }  \scale \ottnt{q}  }  } $ hence, by
            induction hypothesis $ \constraintfont{ \ottnt{Q_{\ottmv{o}}}  \qtensor   \overline{   \multiplicityfont{ \ottsym{1} }  \scale \ottnt{q}  }   \qtensor  \ottsym{(}  \ottnt{U}_{\ottmv{b}}  \ottsym{,}  \ottnt{L}_{\ottmv{b}}  \ottsym{)} }   \Vdash   \constraintfont{ \ottnt{Q_{\ottmv{i}}} } $
          \end{itemize}
          We need to prove that $ \constraintfont{ \ottnt{Q_{\ottmv{o}}}  \qtensor  \ottsym{(}  \ottnt{U}_{\ottmv{b}}  \ottsym{,}   \ottnt{L}_{\ottmv{b}}  \uplus  \ottnt{q}   \ottsym{)} }   \Vdash   \constraintfont{ \ottnt{Q_{\ottmv{i}}} } $,
          which, by definition of the tensor product is the same as
          $ \constraintfont{ \ottnt{Q_{\ottmv{o}}}  \qtensor    \multiplicityfont{ \ottsym{1} }  \scale \ottnt{q}   \qtensor  \ottsym{(}  \ottnt{U}_{\ottmv{b}}  \ottsym{,}  \ottnt{L}_{\ottmv{b}}  \ottsym{)} }   \Vdash   \constraintfont{ \ottnt{Q_{\ottmv{i}}} } $. By \rref{D-Dom} it is hence
          sufficient to prove that $ \constraintfont{   \multiplicityfont{ \ottsym{1} }  \scale \ottnt{q}  }   \Vdash   \constraintfont{  \overline{   \multiplicityfont{ \ottsym{1} }  \scale \ottnt{q}  }  } $, which is precisely
          what $  \constraintfont{ \ottnt{q} }   \in  \constraintfont{\mathcal{D} } $ gives us.
    \item[\rref*{D-Ur}] We have
          \begin{itemize}
            \item $  \constraintfont{ \ottnt{Q_{\ottmv{b}}} }   =   \constraintfont{ \ottsym{(}   \ottnt{U}_{\ottmv{b}}  \cup  \ottnt{q}   \ottsym{,}  \ottnt{L}_{\ottmv{b}}  \ottsym{)} }  $
            \item
            $\vdashd   \constraintfont{ \ottnt{Q_{\ottmv{i}}} }   \setminus   \constraintfont{ \ottsym{(}  \ottnt{U}_{\ottmv{b}}  \ottsym{,}  \ottnt{L}_{\ottmv{b}}  \ottsym{)} }   \leadsto   \constraintfont{ \ottnt{Q_{\ottmv{o}}}  \qtensor    \multiplicityfont{ \pi_{{\mathrm{1}}} }  \scale \ottnt{q}   \qtensor \, ... \, \qtensor    \multiplicityfont{ \pi_{\ottmv{n}} }  \scale \ottnt{q}  } $
            hence, by induction hypothesis $ \constraintfont{ \ottnt{Q_{\ottmv{o}}}  \qtensor    \multiplicityfont{ \pi_{{\mathrm{1}}} }  \scale \ottnt{q}   \qtensor \, ... \, \qtensor    \multiplicityfont{ \pi_{\ottmv{n}} }  \scale \ottnt{q}   \qtensor  \ottsym{(}  \ottnt{U}_{\ottmv{b}}  \ottsym{,}  \ottnt{L}_{\ottmv{b}}  \ottsym{)} }   \Vdash   \constraintfont{ \ottnt{Q_{\ottmv{i}}} } $ 
          \end{itemize}
          We need to prove that $ \constraintfont{ \ottnt{Q_{\ottmv{o}}}  \qtensor  \ottsym{(}   \ottnt{U}_{\ottmv{b}}  \cup  \ottnt{q}   \ottsym{,}  \ottnt{L}_{\ottmv{b}}  \ottsym{)} }   \Vdash   \constraintfont{ \ottnt{Q_{\ottmv{i}}} } $,
          which, by definition of the tensor product is the same as
          $ \constraintfont{ \ottnt{Q_{\ottmv{o}}}  \qtensor    \multiplicityfont{ \omega }  \scale \ottnt{q}   \qtensor  \ottsym{(}  \ottnt{U}_{\ottmv{b}}  \ottsym{,}  \ottnt{L}_{\ottmv{b}}  \ottsym{)} }   \Vdash   \constraintfont{ \ottnt{Q_{\ottmv{i}}} } $. By \rref{D-Dom} it is
          hence sufficient to prove that $ \constraintfont{   \multiplicityfont{ \omega }  \scale \ottnt{q}  }   \Vdash   \constraintfont{   \multiplicityfont{ \pi_{{\mathrm{1}}} }  \scale \ottnt{q}   \qtensor \, ... \, \qtensor    \multiplicityfont{ \pi_{\ottmv{n}} }  \scale \ottnt{q}  } $
          which follows from the conditions in \cref{fig:entailment-relation}.
  \end{description}
\end{nestedproof}
\begin{lemma}[Soundness of $  \constraintfont{ \ottnt{Q_{{\mathrm{1}}}}  \wedge  \ottnt{Q_{{\mathrm{2}}}} }  $]\label{l:meet-soundness}
  For any $  \constraintfont{ \ottnt{Q_{{\mathrm{1}}}} }  $, $  \constraintfont{ \ottnt{Q_{{\mathrm{2}}}} }  $, we have $ \constraintfont{ \ottnt{Q_{{\mathrm{1}}}}  \wedge  \ottnt{Q_{{\mathrm{2}}}} }   \Vdash   \constraintfont{ \ottnt{Q_{{\mathrm{1}}}} } $ and $ \constraintfont{ \ottnt{Q_{{\mathrm{1}}}}  \wedge  \ottnt{Q_{{\mathrm{2}}}} }   \Vdash   \constraintfont{ \ottnt{Q_{{\mathrm{2}}}} } $
\end{lemma}
\begin{nestedproof}
  By induction on the derivation of $ \constraintfont{ \ottnt{Q_{{\mathrm{1}}}} }   \wedge   \constraintfont{ \ottnt{Q_{{\mathrm{2}}}} }   \ottsym{=}   \constraintfont{ \ottnt{Q} } $
  \begin{description}
    \item[\rref*{Inf-Ur}] We have
          \begin{itemize}
            \item $  \constraintfont{ \ottnt{Q_{{\mathrm{1}}}} }   =   \constraintfont{ \ottsym{(}  \ottnt{U}_{{\mathrm{1}}}  \ottsym{,}  \emptyset  \ottsym{)} }  $ in other words,
            $  \constraintfont{ \ottnt{Q_{{\mathrm{1}}}} }   =   \constraintfont{   \multiplicityfont{ \omega }  \scale \ottnt{Q'_{{\mathrm{1}}}}  }  $ for some $  \constraintfont{ \ottnt{Q'_{{\mathrm{1}}}} }  $.
            \item $  \constraintfont{ \ottnt{Q_{{\mathrm{2}}}} }   =   \constraintfont{ \ottsym{(}  \ottnt{U}_{{\mathrm{2}}}  \ottsym{,}  \emptyset  \ottsym{)} }  $ in other words,
            $  \constraintfont{ \ottnt{Q_{{\mathrm{2}}}} }   =   \constraintfont{   \multiplicityfont{ \omega }  \scale \ottnt{Q'_{{\mathrm{2}}}}  }  $ for some $  \constraintfont{ \ottnt{Q'_{{\mathrm{2}}}} }  $.
            \item
            $  \constraintfont{ \ottnt{Q_{{\mathrm{1}}}}  \wedge  \ottnt{Q_{{\mathrm{2}}}} }   =   \constraintfont{ \ottsym{(}   \ottnt{U}_{{\mathrm{1}}}  \cup  \ottnt{U}_{{\mathrm{2}}}   \ottsym{,}  \emptyset  \ottsym{)} }   =   \constraintfont{   \multiplicityfont{ \omega }  \scale \ottnt{Q'_{{\mathrm{1}}}}   \qtensor    \multiplicityfont{ \omega }  \scale \ottnt{Q'_{{\mathrm{2}}}}  }  $.
          \end{itemize}
          It is the case, as required, that
          $ \constraintfont{   \multiplicityfont{ \omega }  \scale \ottnt{Q'_{{\mathrm{1}}}}   \qtensor    \multiplicityfont{ \omega }  \scale \ottnt{Q'_{{\mathrm{2}}}}  }   \Vdash   \constraintfont{   \multiplicityfont{ \omega }  \scale \ottnt{Q'_{{\mathrm{1}}}}  } $ and
          $ \constraintfont{   \multiplicityfont{ \omega }  \scale \ottnt{Q'_{{\mathrm{1}}}}   \qtensor    \multiplicityfont{ \omega }  \scale \ottnt{Q'_{{\mathrm{2}}}}  }   \Vdash   \constraintfont{   \multiplicityfont{ \omega }  \scale \ottnt{Q'_{{\mathrm{2}}}}  } $ by the rules of
          \cref{fig:entailment-relation}.
    \item[\rref*{D-Match}] We have
          \begin{itemize}
            \item $  \constraintfont{ \ottnt{Q_{{\mathrm{1}}}} }   =   \constraintfont{ \ottsym{(}  \ottnt{U}_{{\mathrm{1}}}  \ottsym{,}   \ottnt{L}_{{\mathrm{1}}}  \uplus  \ottnt{q}   \ottsym{)} }   =   \constraintfont{ \ottnt{Q'_{{\mathrm{1}}}}  \qtensor    \multiplicityfont{ \ottsym{1} }  \scale \ottnt{q}  }  $
            \item $  \constraintfont{ \ottnt{Q_{{\mathrm{2}}}} }   =   \constraintfont{ \ottsym{(}  \ottnt{U}_{{\mathrm{2}}}  \ottsym{,}   \ottnt{L}_{{\mathrm{2}}}  \uplus  \ottnt{q}   \ottsym{)} }   =   \constraintfont{ \ottnt{Q'_{{\mathrm{2}}}}  \qtensor    \multiplicityfont{ \ottsym{1} }  \scale \ottnt{q}  }  $
            \item $  \constraintfont{ \ottnt{Q_{{\mathrm{1}}}}  \wedge  \ottnt{Q_{{\mathrm{2}}}} }   =   \constraintfont{ \ottsym{(}  \ottnt{Q'_{{\mathrm{1}}}}  \wedge  \ottnt{Q'_{{\mathrm{2}}}}  \ottsym{)}  \qtensor    \multiplicityfont{ \ottsym{1} }  \scale \ottnt{q}  }  $
          \end{itemize}
          By induction hypothesis, $ \constraintfont{ \ottnt{Q'_{{\mathrm{1}}}}  \wedge  \ottnt{Q'_{{\mathrm{2}}}} }   \Vdash   \constraintfont{ \ottnt{Q'_{{\mathrm{1}}}} } $ and
          $ \constraintfont{ \ottnt{Q'_{{\mathrm{1}}}}  \wedge  \ottnt{Q'_{{\mathrm{2}}}} }   \Vdash   \constraintfont{ \ottnt{Q'_{{\mathrm{2}}}} } $. Which, tensoring with $  \constraintfont{   \multiplicityfont{ \ottsym{1} }  \scale \ottnt{q}  }  $ according
          to the rules of \cref{fig:entailment-relation}, gives
          $ \constraintfont{ \ottsym{(}  \ottnt{Q'_{{\mathrm{1}}}}  \wedge  \ottnt{Q'_{{\mathrm{2}}}}  \ottsym{)}  \qtensor    \multiplicityfont{ \ottsym{1} }  \scale \ottnt{q}  }   \Vdash   \constraintfont{ \ottnt{Q'_{{\mathrm{1}}}}  \qtensor    \multiplicityfont{ \ottsym{1} }  \scale \ottnt{q}  } $ and
          $ \constraintfont{ \ottsym{(}  \ottnt{Q'_{{\mathrm{1}}}}  \wedge  \ottnt{Q'_{{\mathrm{2}}}}  \ottsym{)}  \qtensor    \multiplicityfont{ \ottsym{1} }  \scale \ottnt{q}  }   \Vdash   \constraintfont{ \ottnt{Q'_{{\mathrm{2}}}}  \qtensor    \multiplicityfont{ \ottsym{1} }  \scale \ottnt{q}  } $. We conclude by definition of
          $  \constraintfont{ \ottnt{Q_{{\mathrm{1}}}}  \wedge  \ottnt{Q_{{\mathrm{2}}}} }  $.
    \item[\rref*{Inf-DiffL} and \rref*{Inf-DiffR}] These two cases are
          symmetric, so we'll prove the case of \rref*{Inf-DiffL},
          \rref*{Inf-DiffL} is similar. That is we have
          \begin{itemize}
            \item $  \constraintfont{ \ottnt{Q_{{\mathrm{1}}}} }   =   \constraintfont{ \ottsym{(}  \ottnt{U}_{{\mathrm{1}}}  \ottsym{,}   \ottnt{L}_{{\mathrm{1}}}  \uplus  \ottnt{q}   \ottsym{)} }   =   \constraintfont{ \ottnt{Q'_{{\mathrm{1}}}}  \qtensor    \multiplicityfont{ \ottsym{1} }  \scale \ottnt{q}  }  $
            \item $  \constraintfont{ \ottnt{Q_{{\mathrm{1}}}}  \wedge  \ottnt{Q_{{\mathrm{2}}}} }   =   \constraintfont{ \ottsym{(}  \ottnt{Q'_{{\mathrm{1}}}}  \wedge  \ottnt{Q_{{\mathrm{2}}}}  \ottsym{)}  \qtensor    \multiplicityfont{ \omega }  \scale \ottnt{q}  }  $
          \end{itemize}
          By induction hypothesis $ \constraintfont{ \ottnt{Q'_{{\mathrm{1}}}}  \wedge  \ottnt{Q_{{\mathrm{2}}}} }   \Vdash   \constraintfont{ \ottnt{Q'_{{\mathrm{1}}}} } $ and
          $ \constraintfont{ \ottnt{Q'_{{\mathrm{1}}}}  \wedge  \ottnt{Q_{{\mathrm{2}}}} }   \Vdash   \constraintfont{ \ottnt{Q_{{\mathrm{2}}}} } $. Now $ \constraintfont{ \ottsym{(}  \ottnt{Q'_{{\mathrm{1}}}}  \wedge  \ottnt{Q_{{\mathrm{2}}}}  \ottsym{)}  \qtensor    \multiplicityfont{ \omega }  \scale \ottnt{q}  }   \Vdash   \constraintfont{ \ottnt{Q_{{\mathrm{2}}}} } $ by the
          weakening rule of \cref{fig:entailment-relation}. And
          $ \constraintfont{ \ottsym{(}  \ottnt{Q'_{{\mathrm{1}}}}  \wedge  \ottnt{Q_{{\mathrm{2}}}}  \ottsym{)}  \qtensor    \multiplicityfont{ \omega }  \scale \ottnt{q}  }   \Vdash   \constraintfont{ \ottnt{Q'_{{\mathrm{1}}}}  \qtensor    \multiplicityfont{ \ottsym{1} }  \scale \ottnt{q}  } $ by the rules of
          \cref{fig:entailment-relation}.
    \item[\rref*{Inf-DiffDL} and \rref*{Inf-DiffDR}] These two cases are
          symmetric so we'll only prove the case of \rref*{Inf-DiffDL},
          \rref*{Inf-DiffDR} is similar. That is we have
          \begin{itemize}
            \item $  \constraintfont{ \ottnt{q} }   \in  \constraintfont{\mathcal{D} } $
            \item $  \constraintfont{ \ottnt{Q_{{\mathrm{1}}}} }   =   \constraintfont{ \ottsym{(}  \ottnt{U}_{{\mathrm{1}}}  \ottsym{,}   \ottnt{L}_{{\mathrm{1}}}  \uplus  \ottnt{q}   \ottsym{)} }   =   \constraintfont{ \ottnt{Q'_{{\mathrm{1}}}}  \qtensor    \multiplicityfont{ \ottsym{1} }  \scale \ottnt{q}  }  $
            \item $  \constraintfont{ \ottnt{Q_{{\mathrm{1}}}}  \wedge  \ottnt{Q_{{\mathrm{2}}}} }   =   \constraintfont{ \ottsym{(}  \ottnt{Q'_{{\mathrm{1}}}}  \wedge  \ottnt{Q_{{\mathrm{2}}}}  \ottsym{)}  \qtensor    \multiplicityfont{ \ottsym{1} }  \scale \ottnt{q}  }  $
          \end{itemize}
          By induction hypothesis $ \constraintfont{ \ottnt{Q'_{{\mathrm{1}}}}  \wedge  \ottnt{Q_{{\mathrm{2}}}} }   \Vdash   \constraintfont{ \ottnt{Q'_{{\mathrm{1}}}} } $ and
          $ \constraintfont{ \ottnt{Q'_{{\mathrm{1}}}}  \wedge  \ottnt{Q_{{\mathrm{2}}}} }   \Vdash   \constraintfont{ \ottnt{Q_{{\mathrm{2}}}} } $. Now $ \constraintfont{ \ottsym{(}  \ottnt{Q'_{{\mathrm{1}}}}  \wedge  \ottnt{Q_{{\mathrm{2}}}}  \ottsym{)}  \qtensor    \multiplicityfont{ \ottsym{1} }  \scale \ottnt{q}  }   \Vdash   \constraintfont{ \ottnt{Q_{{\mathrm{2}}}} } $ because
          $  \constraintfont{ \ottnt{q} }   \in  \constraintfont{\mathcal{D} } $ therefore $ \constraintfont{   \multiplicityfont{ \ottsym{1} }  \scale \ottnt{q}  }   \Vdash   \constraintfont{  \mathbf{\varepsilon}  } $. Finally
          $ \constraintfont{ \ottsym{(}  \ottnt{Q'_{{\mathrm{1}}}}  \wedge  \ottnt{Q_{{\mathrm{2}}}}  \ottsym{)}  \qtensor    \multiplicityfont{ \ottsym{1} }  \scale \ottnt{q}  }   \Vdash   \constraintfont{ \ottnt{Q'_{{\mathrm{1}}}}  \qtensor    \multiplicityfont{ \ottsym{1} }  \scale \ottnt{q}  } $ by the rules of
          \cref{fig:entailment-relation}.
  \end{description}
\end{nestedproof}

\end{toappendix}

\begin{lemmarep}[Constraint solver soundness]
  \label{lem:solver-soundness}
  If $\vdashs   \constraintfont{ \constraintfont{C} }   \leadsto   \constraintfont{ \ottnt{Q} } $, then $ \constraintfont{ \ottnt{Q} }   \vdash   \constraintfont{ \constraintfont{C} } $
\end{lemmarep}
\begin{proof}
  By induction on $\vdashs   \constraintfont{ \constraintfont{C_{\ottmv{w}}} }   \leadsto   \constraintfont{ \ottnt{Q} } $
  \begin{description}
  \item[\rref*{S-Atom}] We need to prove that $ \constraintfont{   \multiplicityfont{ \pi }  \scale \ottnt{q}  }   \vdash   \constraintfont{   \multiplicityfont{ \pi }  \scale \ottnt{q}  } $ which holds by definition
  \item[\rref*{S-ImplOne}] We have
          \begin{itemize}
            \item $\vdashs   \constraintfont{ \constraintfont{C} }   \leadsto   \constraintfont{ \ottnt{Q_{\ottmv{i}}} } $, hence $ \constraintfont{ \ottnt{Q_{\ottmv{i}}} }   \vdash   \constraintfont{ \constraintfont{C} } $ by induction hypothesis
            \item $\vdashd   \constraintfont{ \ottnt{Q_{\ottmv{i}}} }   \setminus   \constraintfont{ \ottnt{Q_{\ottmv{b}}} }   \leadsto   \constraintfont{ \ottnt{Q_{\ottmv{o}}} } $
          \end{itemize}
          We need to prove that $ \constraintfont{   \multiplicityfont{ \pi }  \scale \ottnt{Q_{\ottmv{o}}}  }   \vdash   \constraintfont{   \multiplicityfont{ \pi }  \scale ( \ottnt{Q_{\ottmv{b}}}  \Lolly  \constraintfont{C} )  } $.

          By \rref{C-Impl}, it suffices that  $ \constraintfont{ \ottnt{Q_{\ottmv{o}}}  \qtensor  \ottnt{Q_{\ottmv{b}}} }   \vdash   \constraintfont{ \constraintfont{C} } $. Since
          $\vdashd   \constraintfont{ \ottnt{Q_{\ottmv{i}}} }   \setminus   \constraintfont{ \ottnt{Q_{\ottmv{b}}} }   \leadsto   \constraintfont{ \ottnt{Q_{\ottmv{o}}} } $, by \cref{l:diff-soundness},
          $ \constraintfont{ \ottnt{Q_{\ottmv{o}}}  \qtensor  \ottnt{Q_{\ottmv{b}}} }   \Vdash   \constraintfont{ \ottnt{Q_{\ottmv{i}}} } $, together with the induction
          hypothesis and \rref{C-Dom}, we get $ \constraintfont{ \ottnt{Q_{\ottmv{o}}}  \qtensor  \ottnt{Q_{\ottmv{b}}} }   \vdash   \constraintfont{ \constraintfont{C} } $ as required.
  \item[\rref*{S-Tensor}] We have
          \begin{itemize}
            \item $\vdashs   \constraintfont{ \constraintfont{C_{{\mathrm{1}}}} }   \leadsto   \constraintfont{ \ottnt{Q_{{\mathrm{1}}}} } $, hence $ \constraintfont{ \ottnt{Q_{{\mathrm{1}}}} }   \vdash   \constraintfont{ \constraintfont{C_{{\mathrm{1}}}} } $ by induction hypothesis
            \item $\vdashs   \constraintfont{ \constraintfont{C_{{\mathrm{2}}}} }   \leadsto   \constraintfont{ \ottnt{Q_{{\mathrm{2}}}} } $, hence $ \constraintfont{ \ottnt{Q_{{\mathrm{2}}}} }   \vdash   \constraintfont{ \constraintfont{C_{{\mathrm{2}}}} } $ by induction hypothesis
          \end{itemize}

          We need to prove that $ \constraintfont{ \ottnt{Q_{{\mathrm{1}}}}  \qtensor  \ottnt{Q_{{\mathrm{2}}}} }   \vdash   \constraintfont{ \constraintfont{C_{{\mathrm{1}}}}  \qtensor  \constraintfont{C_{{\mathrm{2}}}} } $, which follows from
          \rref{C-Tensor}.
    \item[\rref*{S-With}] We have
          \begin{itemize}
            \item $\vdashs   \constraintfont{ \constraintfont{C_{{\mathrm{1}}}} }   \leadsto   \constraintfont{ \ottnt{Q_{{\mathrm{1}}}} } $, hence $ \constraintfont{ \ottnt{Q_{{\mathrm{1}}}} }   \vdash   \constraintfont{ \constraintfont{C_{{\mathrm{1}}}} } $ by induction hypothesis
            \item $\vdashs   \constraintfont{ \constraintfont{C_{{\mathrm{2}}}} }   \leadsto   \constraintfont{ \ottnt{Q_{{\mathrm{2}}}} } $, hence $ \constraintfont{ \ottnt{Q_{{\mathrm{2}}}} }   \vdash   \constraintfont{ \constraintfont{C_{{\mathrm{2}}}} } $ by induction hypothesis
          \end{itemize}

          We need to prove that $ \constraintfont{ \ottnt{Q_{{\mathrm{1}}}}  \wedge  \ottnt{Q_{{\mathrm{2}}}} }   \vdash   \constraintfont{ \constraintfont{C_{{\mathrm{1}}}}  \aand  \constraintfont{C_{{\mathrm{2}}}} } $. By \rref{C-With}, it
          suffices that both
          \begin{itemize}
            \item $ \constraintfont{ \ottnt{Q_{{\mathrm{1}}}}  \wedge  \ottnt{Q_{{\mathrm{2}}}} }   \vdash   \constraintfont{ \constraintfont{C_{{\mathrm{1}}}} } $, and
            \item $ \constraintfont{ \ottnt{Q_{{\mathrm{1}}}}  \wedge  \ottnt{Q_{{\mathrm{2}}}} }   \vdash   \constraintfont{ \constraintfont{C_{{\mathrm{2}}}} } $
          \end{itemize}

          By \cref{l:meet-soundness}, we know that $ \constraintfont{ \ottnt{Q_{{\mathrm{1}}}}  \wedge  \ottnt{Q_{{\mathrm{2}}}} }   \Vdash   \constraintfont{ \ottnt{Q_{{\mathrm{1}}}} } $ and
          $ \constraintfont{ \ottnt{Q_{{\mathrm{1}}}}  \wedge  \ottnt{Q_{{\mathrm{2}}}} }   \Vdash   \constraintfont{ \ottnt{Q_{{\mathrm{2}}}} } $. Therefore by \rref{C-Dom}
          $ \constraintfont{ \ottnt{Q_{{\mathrm{1}}}}  \wedge  \ottnt{Q_{{\mathrm{2}}}} }   \vdash   \constraintfont{ \constraintfont{C_{{\mathrm{1}}}} } $ and $ \constraintfont{ \ottnt{Q_{{\mathrm{1}}}}  \wedge  \ottnt{Q_{{\mathrm{2}}}} }   \vdash   \constraintfont{ \constraintfont{C_{{\mathrm{2}}}} } $ follow from
          $ \constraintfont{ \ottnt{Q_{{\mathrm{1}}}} }   \vdash   \constraintfont{ \constraintfont{C_{{\mathrm{1}}}} } $ and $ \constraintfont{ \ottnt{Q_{{\mathrm{2}}}} }   \vdash   \constraintfont{ \constraintfont{C_{{\mathrm{2}}}} } $ respectively.
  \end{description}
\end{proof}


Building on this atomic-constraint solver, we use a linear proof
search algorithm based on the recipe given
by~\citet{resource-management-for-ll-proof-search}, with one caveat that we
discuss in \cref{sec:solver-design}. \Cref{fig:constraint-solver}
presents the rules of the constraint solver.

\begin{figure}
  \maybesmall
  \centering
  \drules[S]{$\vdashs   \constraintfont{ \constraintfont{C_{\ottmv{w}}} }   \leadsto   \constraintfont{ \ottnt{Q} } $}{Constraint solving}{Atom, Tensor, Impl, With}
  \drules[D]{$\vdashd   \constraintfont{ \ottnt{Q_{\ottmv{i}}} }   \setminus   \constraintfont{ \ottnt{Q_{\ottmv{b}}} }   \leadsto   \constraintfont{ \ottnt{Q_{\ottmv{o}}} } $}{Check bounds}{Linear, Dup, Ur}
  \drules[Inf]{$ \constraintfont{ \ottnt{Q_{{\mathrm{1}}}} }   \wedge   \constraintfont{ \ottnt{Q_{{\mathrm{2}}}} }   \ottsym{=}   \constraintfont{ \ottnt{Q} } $}{Meet}{Ur, Match, DiffL, DiffR, DiffDL, DiffDR}
  \caption{Constraint solver}
  \label{fig:constraint-solver}
\end{figure}

\begin{itemize}
  \item The~\rref*{S-Tensor} rule proceeds by solving both sides then combining
  the output context with a linear product.
  \item The~\rref*{S-With} rule handles additive conjunction. Additive
        conjunction is generated from $\kcase$ expressions, only one of them is
        actually going to be executed, therefore both branches must consume
        exactly the same resources. This is handled by the \emph{join}
        $  \constraintfont{ \ottnt{Q_{{\mathrm{1}}}}  \wedge  \ottnt{Q_{{\mathrm{2}}}} }  $, which computes a context which is compatible with
        $  \constraintfont{ \ottnt{Q_{{\mathrm{1}}}} }  $ and $  \constraintfont{ \ottnt{Q_{{\mathrm{2}}}} }  $. That is $ \constraintfont{ \ottnt{Q_{{\mathrm{1}}}}  \wedge  \ottnt{Q_{{\mathrm{2}}}} }   \Vdash   \constraintfont{ \ottnt{Q_{{\mathrm{1}}}} } $ and
        $ \constraintfont{ \ottnt{Q_{{\mathrm{1}}}}  \wedge  \ottnt{Q_{{\mathrm{2}}}} }   \Vdash   \constraintfont{ \ottnt{Q_{{\mathrm{2}}}} } $.
  \item The~\rref*{S-Impl} rule handles implications. What implications
        $  \constraintfont{   \multiplicityfont{ \pi }  \scale ( \ottnt{Q_{\ottmv{b}}}  \Lolly  \constraintfont{C} )  }  $ do is limit the scope of the constraints in
        $  \constraintfont{ \ottnt{Q_{\ottmv{b}}} }  $. This is achieved via the judgement
        $\vdashd   \constraintfont{ \ottnt{Q_{\ottmv{i}}} }   \setminus   \constraintfont{ \ottnt{Q_{\ottmv{b}}} }   \leadsto   \constraintfont{ \ottnt{Q_{\ottmv{o}}} } $, the rules for which are quite formalistic,
        but are easily expressed in English: for each atomic constraint
        $  \constraintfont{ \ottnt{q} }  $ in $  \constraintfont{ \ottnt{Q_{\ottmv{b}}} }  $, there must be a single occurrence of $  \constraintfont{ \ottnt{q} }  $
        in $  \constraintfont{ \ottnt{Q_{\ottmv{b}}} }  $, and all occurrences of $  \constraintfont{ \ottnt{q} }  $ must then removed from
        $  \constraintfont{ \ottnt{Q_{\ottmv{i}}} }  $ to yield $  \constraintfont{ \ottnt{Q_{\ottmv{o}}} }  $. Then a linear atomic constraint
        $ \constraintfont{   \multiplicityfont{ \ottsym{1} }  \scale \ottnt{q}  }   \notin \, \constraintfont{\mathcal{D} }$ removes a single linear constraint from
        $  \constraintfont{ \ottnt{Q_{\ottmv{i}}} }  $, a linear atomic constraint $  \constraintfont{   \multiplicityfont{ \ottsym{1} }  \scale \ottnt{q}  }   \in  \constraintfont{\mathcal{D} } $ removes any
        number of linear constraints (but no unrestricted constraints), and an
        unrestricted constraint $  \constraintfont{   \multiplicityfont{ \omega }  \scale \ottnt{q}  }  $ removes any number of
        constraints with arbitrary multiplicity. If this fails to remove all
        occurrences of $  \constraintfont{ \ottnt{q} }  $, then the solver fails.

\end{itemize}

The solving strategy is mainly embodied by the~\rref*{S-Impl} rule. It makes
strong assumptions in order to avoid guesses. Avoiding guesses is a key property
of \textsc{ghc}'s solver~\cite[Section~6.4]{OutsideIn}, one we must maintain if
we are to be compatible with \textsc{ghc}.

One consequence of our strategy is that the solver can only ever use
the occurrence of $  \constraintfont{   \multiplicityfont{ \ottsym{1} }  \scale \ottnt{q}  }  $ introduced by the innermost implication.
Even though there could be valid solutions using more outer
implications. In particular the solver will fail on constraints of the form
$  \constraintfont{   \multiplicityfont{ \ottsym{1} }  \scale (   \multiplicityfont{ \ottsym{1} }  \scale \ottnt{q}   \Lolly    \multiplicityfont{ \ottsym{1} }  \scale (   \multiplicityfont{ \ottsym{1} }  \scale \ottnt{q}   \Lolly    \multiplicityfont{ \ottsym{1} }  \scale \ottnt{q}   \qtensor    \multiplicityfont{ \ottsym{1} }  \scale \ottnt{q}  )  )  }  $ even though it has a valid
solution. For instance, the following would fail to typecheck with our solver.

\begin{hscode}\SaveRestoreHook
\column{B}{@{}>{\hspre}l<{\hspost}@{}}%
\column{3}{@{}>{\hspre}l<{\hspost}@{}}%
\column{E}{@{}>{\hspre}l<{\hspost}@{}}%
\>[B]{}\mathbf{class}\;\constraintfont{\Conid{C}}{}\<[E]%
\\
\>[B]{}\Varid{giveC}\mathbin{::}(\constraintfont{\Conid{C}}\Lolly \Conid{Int})\to \Conid{Int}{}\<[E]%
\\[\blanklineskip]%
\>[B]{}\Varid{counting}\mathbin{::}(\constraintfont{(\Conid{C},\Conid{C})}\Lolly \Conid{Int})\to \Conid{Int}{}\<[E]%
\\
\>[B]{}\Varid{counting}\;\Varid{take\char95 two}\mathrel{=}{}\<[E]%
\\
\>[B]{}\hsindent{3}{}\<[3]%
\>[3]{}\Varid{giveC}\mathbin{\$}{}\<[E]%
\\
\>[B]{}\hsindent{3}{}\<[3]%
\>[3]{}\Varid{giveC}\mathbin{\$}\mbox{\onelinecomment  Multiplicity error: this linear \ensuremath{\constraintfont{\Conid{C}}} is used several times.}{}\<[E]%
\\
\>[B]{}\hsindent{3}{}\<[3]%
\>[3]{}\Varid{take\char95 two}{}\<[E]%
\ColumnHook
\end{hscode}\resethooks

Likewise, if there are several occurrences of the same $  \constraintfont{   \multiplicityfont{ \ottsym{1} }  \scale \ottnt{q}  }  $ in
$  \constraintfont{ \ottnt{Q_{\ottmv{b}}} }  $, then the solver considers this an ambiguity and fails. For instance

\begin{hscode}\SaveRestoreHook
\column{B}{@{}>{\hspre}l<{\hspost}@{}}%
\column{3}{@{}>{\hspre}l<{\hspost}@{}}%
\column{5}{@{}>{\hspre}l<{\hspost}@{}}%
\column{E}{@{}>{\hspre}l<{\hspost}@{}}%
\>[B]{}\mathbf{class}\;\constraintfont{\Conid{C}}{}\<[E]%
\\
\>[B]{}\Varid{giveTwoC}\mathbin{::}(\constraintfont{(\Conid{C},\Conid{C})}\Lolly \Conid{Int})\to \Conid{Int}{}\<[E]%
\\[\blanklineskip]%
\>[B]{}\Varid{repeating}\mathbin{::}(\constraintfont{\Conid{C}}\Lolly \Conid{Int})\to \Conid{Int}{}\<[E]%
\\
\>[B]{}\Varid{repeating}\;\Varid{take\char95 one}\mathrel{=}{}\<[E]%
\\
\>[B]{}\hsindent{3}{}\<[3]%
\>[3]{}\Varid{giveTwoC}\mathbin{\$}{}\<[E]%
\\
\>[3]{}\hsindent{2}{}\<[5]%
\>[5]{}\Varid{take\char95 one}\mathbin{+}\Varid{take\char95 one}\mbox{\onelinecomment  Ambiguity error: there are two different \ensuremath{\constraintfont{\Conid{C}}} to choose from}{}\<[E]%
\ColumnHook
\end{hscode}\resethooks

Both examples illustrate ways in which the solver is incomplete, which it has to
be to be guess free. Indeed, they can be ascribed a type in the declarative type
system of \cref{sec:typing-rules}. This solver is, nevertheless, enough to
handle all our motivating examples (\cref{sec:memory-ownership})

\section{Desugaring}
\label{sec:desugaring}

The semantics of our language is given by desugaring it into
a simpler core language: a variant of the $λ^q$
calculus~\cite{LinearHaskell}. We
define the core language's type system here; its operational semantics
is the same, \emph{mutatis mutandis}, as that of Linear Haskell.

\subsection{The Core Calculus}
\label{sec:core-calculus}
\label{sec:ds:inferred-constraints}

\begin{figure}
  \maybesmall
  \centering
  $$
  \begin{array}{lcll}
     \ottmv{a} ,  \ottmv{b}  & \bnfeq & \ldots & \text{Type variables} \\
     \ottmv{x} ,  \ottmv{y}  & \bnfeq & \ldots & \text{Expression variables} \\
    \ottmv{T} & \bnfeq & \ldots & \text{Type constructors} \\
     \ottmv{K}  & \bnfeq & \ldots & \text{Data constructors} \\
     \sigma  & \bnfeq &   \forall   \overline{\ottmv{a} } .  \tau   & \text{Type schemes} \\
     \tau ,  \upsilon  & \bnfeq &  \ottmv{a}  \bnfor   \exists   \overline{\ottmv{a} } .  \tau  \otimes  \upsilon   \bnfor   \tau_{{\mathrm{1}}}  \to_{  \multiplicityfont{ \pi }  }  \tau_{{\mathrm{2}}}  
                            \bnfor  \ottmv{T} \, \overline{\tau}  & \text{Types} \\
     \Gamma ,  \Delta  & \bnfeq &  ∙  \bnfor  \Gamma  \ottsym{,}   \ottmv{x} {:}_{  \multiplicityfont{ \pi }  } \sigma   &
                                                              \text{Contexts} \\
     \ottnt{e}  & \bnfeq &  \ottmv{x}  \bnfor  \ottmv{K}  \bnfor  \lambda  \ottmv{x}  \ottsym{.}  \ottnt{e}  \bnfor \ottnt{e_{{\mathrm{1}}}} \, \ottnt{e_{{\mathrm{2}}}} \bnfor  \packbox \, \ottsym{(}  \ottnt{e_{{\mathrm{1}}}}  \ottsym{,}  \ottnt{e_{{\mathrm{2}}}}  \ottsym{)}  & \text{Expressions}\\
                 &\bnfor &  \klet\ \packbox (  \ottmv{y}  ,  \ottmv{x}  ) =  \ottnt{e_{{\mathrm{1}}}}  \ \kin \  \ottnt{e_{{\mathrm{2}}}}  \bnfor   \kcase_  \multiplicityfont{ \pi }   \, \ottnt{e} \, \ottkw{of} \, \ottsym{\{}  \overline{\ottmv{K}_i\ \overline{\ottmv{x}_i } \to \ottnt{e}_i }  \ottsym{\}}  &\\
                 &\bnfor &  \klet_  \multiplicityfont{ \pi }   \, \ottmv{x}  \ottsym{=}  \ottnt{e_{{\mathrm{1}}}} \, \ottkw{in} \, \ottnt{e_{{\mathrm{2}}}} \bnfor   \klet_  \multiplicityfont{ \pi }   \, \ottmv{x}  \ottsym{:}  \sigma  \ottsym{=}  \ottnt{e_{{\mathrm{1}}}} \, \ottkw{in} \, \ottnt{e_{{\mathrm{2}}}}  &
  \end{array}
  $$

  \drules[L]{$\Gamma  \vdash  \ottnt{e}  \ottsym{:}  \tau$}{Core language
    typing}{Var,Abs,App,Pack,Unpack,Let,Case}
  \caption{Core calculus (subset)}
  \label{fig:core-typing-rules}\label{fig:core-grammar}
\end{figure}

The core calculus is a variant of the type system defined in
\cref{sec:qualified-type-system}, but without constraints. That is, the evidence for constraints is passed explicitly in this core calculus.
Following $λ^q$, we assume the existence of the following data types:
\begin{itemize}
\item $ \tau_{{\mathrm{1}}}  \otimes  \tau_{{\mathrm{2}}} $ with sole constructor
  $ ({,})   \ottsym{:}    \forall   \ottmv{a} \, \ottmv{b} .  \ottmv{a}   \to_{  \multiplicityfont{ \ottsym{1} }  }   \ottmv{b}  \to_{  \multiplicityfont{ \ottsym{1} }  }  \ottmv{a}    \otimes  \ottmv{b}$. We will write $\ottsym{(}  \ottnt{e_{{\mathrm{1}}}}  \ottsym{,}  \ottnt{e_{{\mathrm{2}}}}  \ottsym{)}$ for $  ({,})  \, \ottnt{e_{{\mathrm{1}}}} \, \ottnt{e_{{\mathrm{2}}}} $.
\item $  \mathbf{1}  $ with sole constructor $\ottsym{()}  \ottsym{:}   \mathbf{1} $.
\item $ \ottkw{Ur} \, \tau $ with sole constructor $\ottkw{Ur}  \ottsym{:}   \forall   \ottmv{a} .   \ottmv{a}  \to_{  \multiplicityfont{ \omega }  }  \ottkw{Ur} \, \ottmv{a}  $
\end{itemize}
The core calculus is described in \cref{fig:core-typing-rules}.
The main differences with the qualified system are as follows:
\begin{itemize}
  \item Type schemes $ \sigma $ do not support qualified types.
  \item Existentially quantified types ($  \exists   \overline{\ottmv{a} } .  \tau  \RLolly   \constraintfont{ \ottnt{Q} }   $) are now represented as an (existentially quantified, linear) pair of values ($  \exists   \overline{\ottmv{a} } .  \tau_{{\mathrm{2}}}  \otimes  \tau_{{\mathrm{1}}}  $).
Accordingly, $\packbox$ operates on pairs.
\end{itemize}
The differences between our core calculus and $λ^q$ are as follows:
\begin{itemize}
\item We do not support multiplicity polymorphism.
\item On the other hand, we do include type polymorphism.
\item Polymorphism is implicit rather than explicit. This is not an
  essential difference, but it simplifies the presentation. We could,
for example, include more details in the terms in order to make type-checking
more obvious; this amounts essentially to an encoding of typing derivations
in the terms\footnote{See, for example, \citet{weirich-icfp17} and their
comparison between an implicit core language D and an explicit one DC.}.
\item We have existential types. These can be realised in regular Haskell as a
  family of datatypes.
\end{itemize}

Using \cref{lem:generation-soundness} together with
\cref{lem:solver-soundness} we know that if
$\Gamma  \vdashi  \ottnt{e}  \ottsym{:}  \tau  \leadsto   \constraintfont{ \constraintfont{C} } $ and $\vdashs   \constraintfont{ \constraintfont{C} }   \leadsto   \constraintfont{ \ottnt{Q} } $, then
$ \constraintfont{ \ottnt{Q} }   \ottsym{;}  \Gamma  \vdash  \ottnt{e}  \ottsym{:}  \tau$.
As such we only need to desugar derivations of $ \constraintfont{ \ottnt{Q} }   \ottsym{;}  \Gamma  \vdash  \ottnt{e}  \ottsym{:}  \tau$ into the
core calculus.

\subsection{From Qualified to Core}
\label{sec:ds:from-qualified-core}

\subsubsection{Evidence}
In order to desugar derivations of the qualified system to the core calculus,
we pass evidence explicitly\footnote{This technique is also often called
  dictionary-passing style \citep{type-classes-impl} because, in the case of type classes, evidences are
  dictionaries, and because type classes were the original form of constraints
  in Haskell.}.
To do so, we require some more material from
constraints. Namely, we assume a type $  \dsevidence{\constraintfont{ \ottnt{q} } }  $ for each atomic
constraint $  \constraintfont{ \ottnt{q} }  $,
defined in \cref{fig:evidence}.  The $ \dsevidence{\constraintfont{ \_ } } $ operation
extends to simple constraints as
$  \dsevidence{\constraintfont{ \ottnt{Q} } }  $.
Furthermore, we require that for every $  \constraintfont{ \ottnt{Q_{{\mathrm{1}}}} }  $ and $  \constraintfont{ \ottnt{Q_{{\mathrm{2}}}} }  $
such that $ \constraintfont{ \ottnt{Q_{{\mathrm{1}}}} }   \Vdash   \constraintfont{ \ottnt{Q_{{\mathrm{2}}}} } $, there is a (linear) function
$ \dsevidence{\constraintfont{ \ottnt{Q_{{\mathrm{1}}}} }  \Vdash  \constraintfont{ \ottnt{Q_{{\mathrm{2}}}} } }   \ottsym{:}    \dsevidence{\constraintfont{ \ottnt{Q_{{\mathrm{1}}}} } }   \to_{  \multiplicityfont{ \ottsym{1} }  }   \dsevidence{\constraintfont{ \ottnt{Q_{{\mathrm{2}}}} } }  $.

Let us now define a family of functions $  \dstype{ \_ }  $ to translate
the type schemes, types, contexts, and typing derivations of the qualified system into the
types, type schemes, contexts, and terms of the core calculus.

\subsubsection{Translating Types}
Type schemes $ \sigma $ are translated by turning the implicit argument $  \constraintfont{ \ottnt{Q} }  $
into an explicit one of type $  \dsevidence{\constraintfont{ \ottnt{Q} } }  $. Translating types $ \tau $
and contexts $ \Gamma $ proceeds as
expected.

\begin{minipage}{0.5\linewidth}
$$
\left\{
  \begin{array}{lcl}
      \dstype{  \forall   \overline{\ottmv{a} } .   \constraintfont{ \ottnt{Q} }   \Lolly  \tau  }   & = &    \forall   \overline{\ottmv{a} } .   \dsevidence{\constraintfont{ \ottnt{Q} } }    \to_{  \multiplicityfont{ \ottsym{1} }  }   \dstype{ \tau }    \\
  \end{array}
\right.
$$
$$
\left\{
  \begin{array}{lcl}
      \dstype{  \tau_{{\mathrm{1}}}  \to_{  \multiplicityfont{ \pi }  }  \tau_{{\mathrm{2}}}  }   & = &    \dstype{ \tau_{{\mathrm{1}}} }   \to_{  \multiplicityfont{ \pi }  }   \dstype{ \tau_{{\mathrm{2}}} }    \\
      \dstype{  \exists   \overline{\ottmv{a} } .  \tau  \RLolly   \constraintfont{ \ottnt{Q} }   }   & = &   \exists   \overline{\ottmv{a} } .   \dstype{ \tau }   \otimes   \dsevidence{\constraintfont{ \ottnt{Q} } }   
  \end{array}
\right.
$$
\end{minipage}%
\begin{minipage}{0.5\linewidth}
$$
\left\{
  \begin{array}{lcl}
      \dstype{ ∙ }   &=&  ∙  \\
      \dstype{ \Gamma  \ottsym{,}   \ottmv{x} {:}_{  \multiplicityfont{ \pi }  } \tau  }   &=&   \dstype{ \Gamma }   \ottsym{,}   \ottmv{x} {:}_{  \multiplicityfont{ \pi }  }  \dstype{ \tau }   
  \end{array}
\right.
$$
\end{minipage}

\subsubsection{Translating Terms}
Given a derivation $ \constraintfont{ \ottnt{Q} }   \ottsym{;}  \Gamma  \vdash  \ottnt{e}  \ottsym{:}  \tau$, we can build an expression
$ \dsterm{ \ottmv{z} }{  \constraintfont{ \ottnt{Q} }   \ottsym{;}  \Gamma  \vdash  \ottnt{e}  \ottsym{:}  \tau } $, such that
$ \dstype{ \Gamma }   \ottsym{,}   \ottmv{z} {:}_{  \multiplicityfont{ \ottsym{1} }  }  \dsevidence{\constraintfont{ \ottnt{Q} } }    \vdash   \dsterm{ \ottmv{z} }{  \constraintfont{ \ottnt{Q} }   \ottsym{;}  \Gamma  \vdash  \ottnt{e}  \ottsym{:}  \tau }   \ottsym{:}   \dstype{ \tau } $ (for some fresh variable
$ \ottmv{z} $). Even though we abbreviate the derivation as only its
concluding judgement, the translation is defined recursively on the
whole typing derivation: in particular, we have access to typing rule
premises in the body of the definition.
We present some of the interesting cases in \cref{fig:desugaring}, the complete
definition of desugaring can be found in \cref{sec:appendix:desugaring}.
\begin{figure}
    \maybesmall
\centering
  \begin{subfigure}{0.3\linewidth}%
$$
\left\{
  \begin{array}{lcl}
      \dsevidence{\constraintfont{   \multiplicityfont{ \ottsym{1} }  \scale \ottnt{q}  } }   & = &   \dsevidence{\constraintfont{ \ottnt{q} } }   \\
      \dsevidence{\constraintfont{   \multiplicityfont{ \omega }  \scale \ottnt{q}  } }   & = &  \ottkw{Ur} \, \ottsym{(}   \dsevidence{\constraintfont{ \ottnt{q} } }   \ottsym{)}  \\
      \dsevidence{\constraintfont{  \mathbf{\varepsilon}  } }   & = &   \mathbf{1}   \\
      \dsevidence{\constraintfont{ \ottnt{Q_{{\mathrm{1}}}}  \qtensor  \ottnt{Q_{{\mathrm{2}}}} } }   & = &   \dsevidence{\constraintfont{ \ottnt{Q_{{\mathrm{1}}}} } }   \otimes   \dsevidence{\constraintfont{ \ottnt{Q_{{\mathrm{2}}}} } }  
  \end{array}
\right.
$$
  \caption{Evidence passing}
  \label{fig:evidence}
  \end{subfigure}\hfill
  \begin{subfigure}{0.7\linewidth}%
$$
\left\{
  \;
  \begin{minipage}{0.5\linewidth}
\begin{hscode}\SaveRestoreHook
\column{B}{@{}>{\hspre}l<{\hspost}@{}}%
\column{3}{@{}>{\hspre}l<{\hspost}@{}}%
\column{5}{@{}>{\hspre}l<{\hspost}@{}}%
\column{22}{@{}>{\hspre}c<{\hspost}@{}}%
\column{22E}{@{}l@{}}%
\column{25}{@{}>{\hspre}l<{\hspost}@{}}%
\column{52}{@{}>{\hspre}c<{\hspost}@{}}%
\column{52E}{@{}l@{}}%
\column{E}{@{}>{\hspre}l<{\hspost}@{}}%
\>[B]{}\dsterm{\ottmv{\Varid{z}}}{  \constraintfont{ \ottnt{Q} }  ; \Gamma  \vdash \Varid{x}\mathbin{:} \upsilon [ \overline{\tau} / \overline{\ottmv{a} } ]}\mathrel{=}\Varid{x}\;\Varid{z}{}\<[E]%
\\
\>[B]{}\dsterm{\ottmv{\Varid{z}}}{  \constraintfont{ \ottnt{Q} }   \qtensor   \constraintfont{ \ottnt{Q_{{\mathrm{1}}}} }  [ \overline{\upsilon} / \overline{\ottmv{a} } ]; \Gamma  \vdash \packbox\Varid{e}\mathbin{:}\exists\; \overline{\ottmv{a} } .\tau\RLolly  \constraintfont{ \ottnt{Q_{{\mathrm{1}}}} }  }\mathrel{=}{}\<[E]%
\\
\>[B]{}\hsindent{3}{}\<[3]%
\>[3]{} \kcase_  \multiplicityfont{ \ottsym{1} }  \;\Varid{z}\;\mathbf{of}\;\{\mskip1.5mu (\Varid{z'},\Varid{z''})\to {}\<[E]%
\\
\>[3]{}\hsindent{2}{}\<[5]%
\>[5]{}\packbox(\Varid{z''},\dsterm{\ottmv{\Varid{z'}}}{  \constraintfont{ \ottnt{Q} }  ; \Gamma  \vdash \Varid{e}\mathbin{:}\tau[ \overline{\upsilon} / \overline{\ottmv{a} } ]})\mskip1.5mu\}{}\<[E]%
\\
\>[B]{}\dsterm{\ottmv{\Varid{z}}}{  \constraintfont{ \ottnt{Q_{{\mathrm{1}}}} }   \qtensor   \constraintfont{ \ottnt{Q_{{\mathrm{2}}}} }  ; \Gamma_{{\mathrm{1}}} \ottsym{+} \Gamma_{{\mathrm{2}}}  \vdash \klet\ \packbox \Varid{x}\mathrel{=} \ottnt{e_{{\mathrm{1}}}} \;\mathbf{in}\; \ottnt{e_{{\mathrm{2}}}} \mathbin{:}\tau}{}\<[52]%
\>[52]{}\mathrel{=}{}\<[52E]%
\\
\>[B]{}\hsindent{3}{}\<[3]%
\>[3]{} \kcase_  \multiplicityfont{ \ottsym{1} }  \;\Varid{z}\;\mathbf{of}\;\{\mskip1.5mu ( \ottmv{z_{{\mathrm{1}}}} , \ottmv{z_{{\mathrm{2}}}} )\to {}\<[E]%
\\
\>[3]{}\hsindent{2}{}\<[5]%
\>[5]{}\klet\ \packbox \Varid{z'},\Varid{x}\mathrel{=}\dsterm{\ottmv{ \ottmv{z_{{\mathrm{1}}}} }}{  \constraintfont{ \ottnt{Q_{{\mathrm{1}}}} }  ; \Gamma_{{\mathrm{1}}}  \vdash  \ottnt{e_{{\mathrm{1}}}} \mathbin{:}\exists\; \overline{\ottmv{a} } . \tau_{{\mathrm{1}}} \RLolly  \constraintfont{ \ottnt{Q} }  }\;\mathbf{in}{}\<[E]%
\\
\>[3]{}\hsindent{2}{}\<[5]%
\>[5]{} \klet_  \multiplicityfont{ \ottsym{1} }  \; \ottmv{z_{{\mathrm{2}}}} '\mathrel{=}( \ottmv{z_{{\mathrm{2}}}} ,\Varid{z'})\;\mathbf{in}{}\<[E]%
\\
\>[3]{}\hsindent{2}{}\<[5]%
\>[5]{}\dsterm{\ottmv{ \ottmv{z_{{\mathrm{2}}}} '}}{  \constraintfont{ \ottnt{Q_{{\mathrm{2}}}} }   \qtensor   \constraintfont{ \ottnt{Q} }  ; \Gamma_{{\mathrm{2}}} ,\Varid{x}\mathop{:_{1}} \tau_{{\mathrm{1}}}  \vdash  \ottnt{e_{{\mathrm{2}}}} \mathbin{:}\tau}\mskip1.5mu\}{}\<[E]%
\\
\>[B]{}\dsterm{\ottmv{\Varid{z}}}{  \constraintfont{ \ottnt{Q} }  ; \Gamma  \vdash \Varid{e}\mathbin{:}\tau}{}\<[22]%
\>[22]{}\mathrel{=}{}\<[22E]%
\>[25]{}\mbox{\onelinecomment  \rref{E-Sub}}{}\<[E]%
\\
\>[B]{}\hsindent{3}{}\<[3]%
\>[3]{} \klet_  \multiplicityfont{ \ottsym{1} }  \;\Varid{z'}\mathrel{=}\dsevidence{  \constraintfont{ \ottnt{Q} }   \Vdash   \constraintfont{ \ottnt{Q_{{\mathrm{1}}}} }  }\;\Varid{z}\;\mathbf{in}\;\dsterm{\ottmv{\Varid{z'}}}{  \constraintfont{ \ottnt{Q_{{\mathrm{1}}}} }  ; \Gamma  \vdash \Varid{e}\mathbin{:}\tau}{}\<[E]%
\\
\>[B]{}\mathbin{...}{}\<[E]%
\ColumnHook
\end{hscode}\resethooks
  \end{minipage}
\right.
$$
  \caption{Desugaring (subset)}
  \label{fig:desugaring}
  \end{subfigure}
  \caption{Evidence passing and desugaring}
\end{figure}

The cases correspond to the~\rref*{E-Var},~\rref*{E-Unpack}\footnote{The attentive
reader may note that the case for $\kunpack$ extracts out $  \constraintfont{ \ottnt{Q_{{\mathrm{1}}}} }  $ and $  \constraintfont{ \ottnt{Q_{{\mathrm{2}}}} }  $
from the provided simple constraint. Given that simple constraints $  \constraintfont{ \ottnt{Q} }  $ have no
internal ordering and allow duplicates (in the non-linear component), this splitting
is not well defined. To fix this, an implementation would have to \emph{name} individual
components of $  \constraintfont{ \ottnt{Q} }  $, and then the typing derivation can indicate which constraints
go with which sub-expression. Happily, \textsc{ghc} \emph{already} names its constraints,
and so this approach fits easily in the implementation. We could also augment our formalism
here with these details, but they add clutter with little insight.}, and~\rref*{E-Sub} rules, respectively.
Variables are stored with qualified types in the environment, so they get
translated to functions that take the evidence as argument. Accordingly, the evidence
is inserted by passing $ \ottmv{z} $ as an argument.
Handling \rref*{E-Unpack} requires splitting the context into two: $ \ottnt{e_{{\mathrm{1}}}} $ is desugared as a pair, and the evidence
it contains is passed to $ \ottnt{e_{{\mathrm{2}}}} $. Finally, subsumption summons the function corresponding to the entailment relation $ \constraintfont{ \ottnt{Q} }   \Vdash   \constraintfont{ \ottnt{Q_{{\mathrm{1}}}} } $
and applies it to $ \ottmv{z} $ : $  \dsevidence{\constraintfont{ \ottnt{Q} } }  $ then proceeds to desugar $ \ottnt{e} $ with the resulting evidence for $  \constraintfont{ \ottnt{Q_{{\mathrm{1}}}} }  $.
Crucially, since $  \dsterm{ \ottmv{z} }{ \_ }  $ is defined on \emph{derivations}, we can access the premises used in the rule.
Namely, $  \constraintfont{ \ottnt{Q_{{\mathrm{1}}}} }  $ is available in this last case from the~\rref*{E-Sub} rule's premise.

It is straightforward by induction, to verify that desugaring is correct:
\begin{theorem}[Desugaring]
If $ \constraintfont{ \ottnt{Q} }   \ottsym{;}  \Gamma  \vdash  \ottnt{e}  \ottsym{:}  \tau$, then
$ \dstype{ \Gamma }   \ottsym{,}   \ottmv{z} {:}_{  \multiplicityfont{ \ottsym{1} }  }  \dsevidence{\constraintfont{ \ottnt{Q} } }    \vdash   \dsterm{ \ottmv{z} }{  \constraintfont{ \ottnt{Q} }   \ottsym{;}  \Gamma  \vdash  \ottnt{e}  \ottsym{:}  \tau }   \ottsym{:}   \dstype{ \tau } $, for any fresh
variable $ \ottmv{z} $.
\end{theorem}

Thanks to the desugaring machinery, the semantics of a language with linear
constraints can be understood in terms of a simple core language with linear
types, such as $λ^q$, or indeed, \textsc{ghc} Core.

\section{Integrating into \textsc{ghc}}

One of the guiding principles behind our design was ease of integration with
modern Haskell. In this section we describe some of the particulars of adding
linear constraints to \textsc{ghc}.

\subsection{Implementation}
\label{sec:implementation}

We have written a prototype implementation~\citep{prototype} of linear constraints on top of \textsc{ghc} 9.1, a version that already
ships with the \text{\ttfamily LinearTypes} extension. Function arrows (\ensuremath{\to }) and context arrows
(\ensuremath{\FatArrow }) share the same internal representation in the typechecker, differentiated
only by a boolean flag. Thus, the \text{\ttfamily LinearTypes} implementation effort has already
laid down the bureaucratic ground work of annotating these arrows with
multiplicity information.

The key changes affect constraint generation and constraint solving. Constraints
are now annotated with a multiplicity, according to the context from
which they arise. With \text{\ttfamily LinearTypes}, \textsc{ghc} already scales the usage
of term variables. We simply modified the scaling function to capture all the
generated constraints and re-emit a scaled version of them, which is a fairly local
change.

The constraint solver maintains a set of given constraints (the \emph{inert set}
in \textsc{ghc} jargon), which corresponds to the $  \constraintfont{ \ottnt{U} }  $, $  \constraintfont{ \ottnt{D} }  $, and $  \constraintfont{ \ottnt{L} }  $
contexts in our solver judgements in \cref{sec:constraint-solver}. When
the solver goes under an implication, the assumptions of the implication are
added to set of givens. When a new given is added, we record the \emph{level} of
the implication (how many implications deep the constraint arises from) along
with the constraint. So that in case there are multiple matching
givens, the constraint solver selects the innermost one
(in \cref{sec:constraint-solver} we use an ordered list
for this purpose).

As constraint solving proceeds, the compiler pipeline constructs a
term in a typed language known as \textsc{ghc} Core~\citep{system-fc}.
In Core, type class constraints are turned into explicit evidence (see
\cref{sec:desugaring}). Thanks to being fully annotated, Core has
decidable typechecking, which is used to find and fix bugs in
the compiler (the Haskell type checker finds mistakes in user
programs). Thus, the Core typechecker verifies that the desugaring
procedure produced a linearity-respecting program before code
generation occurs.

\subsection{Interaction with Other Features}

Since constraints play an important role in \textsc{ghc}'s type system, we must
pay close attention to the interaction of linearity with other language features
related to constraints. Of these, we point out two that require some extra care.
\info{There isn't room to properly explain how we can implement
\ensuremath{\constraintfont{\Conid{Linearly}}} constraints. We don't already speak of
desugaring in the implementation section, so I'd need to give a bit of
context about wrappers or something.

The right approach is to count the number of `Linearly` used. We need
wrappers at the toplevel of definitions, and at each branch of a match
(even if it doesn't introduce an implication, as we may need to adjust
the number of `Linearly` in some branches (presumably by
weakening)). And we do the appropriate amount of
duplications/discards in these wrappers.}

\subsubsection{Superclasses}

Haskell's type classes can have \emph{superclasses}, which place constraints on
all of the instances of that class. For example, the \ensuremath{\Conid{Ord}} class is defined as
\begin{hscode}\SaveRestoreHook
\column{B}{@{}>{\hspre}l<{\hspost}@{}}%
\column{E}{@{}>{\hspre}l<{\hspost}@{}}%
\>[B]{}\mathbf{class}\;\constraintfont{\Conid{Eq}\;\Varid{a}}\FatArrow \constraintfont{\Conid{Ord}\;\Varid{a}}\;\mathbf{where}\mathbin{...}{}\<[E]%
\ColumnHook
\end{hscode}\resethooks
which means that every ordered type must also support equality. Such superclass
declarations extend the entailment relation: if we know that a type is ordered,
we also know that it supports equality. Namely, such a superclass declaration
adds an axiom $\ensuremath{\constraintfont{\Conid{Ord}\;\Varid{a}}} \Vdash \ensuremath{\constraintfont{\Conid{Eq}\;\Varid{a}}}$ to the entailment
relation, and a corresponding axiom in the solver.

In practice, the solver saturates the given constraints with the superclasses of
the declared given constraints. This is a departure from the otherwise
goal-oriented discipline of the solver.

Saturation isn't really compatible with linear constraints, however. Since if we
have a linear \ensuremath{\Conid{Ord}\;\Varid{a}} we would have to choose whether to leave it alone or
replace it with an \ensuremath{\Conid{Eq}\;\Varid{a}}. Short of backtracking, the constraint solver
needs to make a guess, which \textsc{ghc} never does.

Our solution is to consider that superclass declaration add axioms of the form

$$
  \multiplicityfont{ \omega }  .\ensuremath{\constraintfont{\Conid{Ord}\;\Varid{a}}} \Vdash   \multiplicityfont{ \omega }  .\ensuremath{\constraintfont{\Conid{Eq}\;\Varid{a}}}
$$

That is, only unrestricted given constraints expose their superclasses. Saturation
is restricted to unrestricted given constraints avoiding any complication.

\subsubsection{Equality Constraints}
\label{sec:equality-constraints}

In \cref{sec:type-inference} we argued that \emph{type} inference and
\emph{constraint} inference can be performed independently. However, this is not
the case for \textsc{ghc}'s constraint domain, because it supports equality
constraints, which allows unification problems to be deferred, and potentially
be solvable only after solving other constraints first.

They key is to represent unification problems as \emph{unrestricted} equality
constraints, so a given linear equality constraint cannot be used during type
inference.  This way, linear equalities require no
special treatment, and are harmless.

\subsection{Inferring Packing and Unpacking}
\label{sec:implicit-existentials}

Recent work~\cite{existentials} describes an algorithm (call it \textsc{edwl}, after the
authors' names) that
can infer the location of the pack and unpack annotations (our $\packbox$ and $\kunpack$)
in a program.%
\footnote{Actually, \citet{existentials} use an $\ottkw{open}$ construct instead of $\kunpack$
to access the contents of an existential package, but that distinction does
not affect our usage of existentials with linear constraints.}
In Section~9.2 of that paper,
the authors extend their system to include class constraints,
much as we allow our existential packages to carry linear constraints.

Accordingly, \textsc{edwl} would work well for us here and remove the need for
these annotations in the type system of \cref{sec:typing-rules}.
The \textsc{edwl} algorithm is only a small change on the way some types are treated during
bidirectional type-checking. Though the presentation of linear constraints is not
written using a bidirectional algorithm, our implementation in \textsc{ghc}
is indeed bidirectional (as \textsc{ghc}'s existing type inference algorithm
is bidirectional, as described by \citet{practical-type-inference} and
\citet{visible-type-application}) and produces constraints much like we
have presented here, formally. None of this would change in adapting \textsc{edwl}.
Indeed, it would seem that the two extensions are orthogonal in implementation,
though avoiding the need for explicit packing and unpacking would
make linear constraints easier to use.

On the Haskell side, \textsc{edwl} would correspondingly remove the need for the \ensuremath{\Conid{Linearly}.\mathbf{do}}
notation as well. Since the types involved in the \ensuremath{\Conid{Linearly}.\mathbf{do}} notation cannot
properly be written in Haskell, it would alleviate the need for \emph{ad hoc} extensions.

\section{Design considerations}
\label{sec:design-considerations}

In this section we discuss the designs described throughout this paper. Why we
chose these particular designs and what could be potential alternative.
Concretely, we discuss the motivations behind and alternatives designs for the
\ensuremath{\constraintfont{\Conid{Linearly}}} constraint (\cref{sec:linearly-vs-sticky}), nested freezing
(\cref{sec:freezing-in-details}), release functions in borrowing
(\cref{sec:unrestricted-release}), and our constraint solving algorithm
(\cref{sec:solver-design}).

\subsection{Linearly vs sticky scopes}
\label{sec:linearly-vs-sticky}

\unsure{Maybe this should cite the blog post, for completeness \url{https://www.tweag.io/blog/2023-03-23-linear-constraints-linearly/}}
In \cref{sec:Unique-constraint}, we introduced the \ensuremath{\constraintfont{\Conid{Linearly}}} constraint as a
means of creating new linear values (in our case linear constraints). The
\ensuremath{\constraintfont{\Conid{Linearly}}} constraint is the motivation for introducing the set $ \constraintfont{\mathcal{D} } $ of
duplicable atomic constraints. So it is natural to ask whether it is worth the
complexity.

Let's compare our \ensuremath{\Varid{new}} function with a \ensuremath{\Varid{newScoped}} function in the style
advertised in the Linear Haskell paper~\cite{LinearHaskell}, that is \ensuremath{\Varid{newScoped}}
takes as an additional argument a continuation which represents the scope in
which the new array is live.
\begin{hscode}\SaveRestoreHook
\column{B}{@{}>{\hspre}l<{\hspost}@{}}%
\column{6}{@{}>{\hspre}l<{\hspost}@{}}%
\column{E}{@{}>{\hspre}l<{\hspost}@{}}%
\>[B]{}\Varid{new}{}\<[6]%
\>[6]{}\mathbin{::}\constraintfont{\constraintfont{\Conid{Linearly}}}\Lolly \Conid{Int}\to \exists\;\Varid{n}.\Conid{Ur}\;(\Conid{UArray}\;\Varid{a}\;\Varid{n})\RLolly\constraintfont{\Conid{RW}\;\Varid{n}}{}\<[E]%
\\
\>[B]{}\Varid{newScoped}\mathbin{::}\Conid{Int}\to (\forall\;\Varid{n}.\constraintfont{\Conid{RW}\;\Varid{n}}\Lolly \Conid{UArray}\;\Varid{a}\;\Varid{n}\to \Conid{Ur}\;\Varid{r})⊸\Conid{Ur}\;\Varid{r}{}\<[E]%
\ColumnHook
\end{hscode}\resethooks

A first observation, maybe, is that \ensuremath{\Varid{new}} is more primitive. It is easy to
implement \ensuremath{\Varid{newScoped}} in terms of \ensuremath{\Varid{new}}:
\begin{hscode}\SaveRestoreHook
\column{B}{@{}>{\hspre}l<{\hspost}@{}}%
\column{3}{@{}>{\hspre}l<{\hspost}@{}}%
\column{E}{@{}>{\hspre}l<{\hspost}@{}}%
\>[B]{}\Varid{newScoped}\mathbin{::}\Conid{Int}\to (\forall\;\Varid{n}.\constraintfont{\Conid{RW}\;\Varid{n}}\Lolly \Conid{UArray}\;\Varid{a}\;\Varid{n}\to \Conid{Ur}\;\Varid{r})⊸\Conid{Ur}\;\Varid{r}{}\<[E]%
\\
\>[B]{}\Varid{newScoped}\;\Varid{lgth}\;\Varid{scope}\mathrel{=}\Varid{linearly}\mathbin{\$}\mathbf{do}{}\<[E]%
\\
\>[B]{}\hsindent{3}{}\<[3]%
\>[3]{}\Conid{Ur}\;\Varid{arr}\leftarrow \Varid{new}\;\Varid{lgth}{}\<[E]%
\\
\>[B]{}\hsindent{3}{}\<[3]%
\>[3]{}\Varid{scope}\;\Varid{arr}{}\<[E]%
\ColumnHook
\end{hscode}\resethooks

Using scope functions like \ensuremath{\Varid{newScoped}} is going to be more verbose whenever
there's more than one scope involved. Scopes, by themselves, aren't a problem,
the verbosity may be acceptable. In the Rust programming language, by
comparison, values are scoped as well. Scopes are represented with the
\text{\ttfamily \char123{}\char46{}\char46{}\char46{}\char125{}} syntax, instead of functions, but in a first approximation these
are the same.

Where the difference is significant, is that with the \ensuremath{\constraintfont{\Conid{Linearly}}} constraint and
the \ensuremath{\Varid{new}} function there can be several arrays in one scope, like in Rust,
rather than a single one as in \ensuremath{\Varid{newScoped}}. With \ensuremath{\Varid{new}} the \ensuremath{\Varid{linearly}} function is
the only one in charge of introducing scopes.

However, there is a way in which the scoped style of introducing linear values
is harmful to expressiveness. It is due to how we prevent linear constraints
from escaping scopes by making sure that the scope returns an \ensuremath{\Conid{Ur}\;\Varid{r}}.

Suppose we're given a function
\begin{hscode}\SaveRestoreHook
\column{B}{@{}>{\hspre}l<{\hspost}@{}}%
\column{E}{@{}>{\hspre}l<{\hspost}@{}}%
\>[B]{}\Varid{sum}\mathbin{::}\constraintfont{\Conid{RW}\;\Varid{n}}\Lolly \Conid{UArray}\;\Conid{Int}\;\Varid{n}\to \Conid{Ur}\;\Conid{Int}{}\<[E]%
\ColumnHook
\end{hscode}\resethooks
And we want to use it to compute the sum of the first two elements of an array.
It would look a like this:
\begin{hscode}\SaveRestoreHook
\column{B}{@{}>{\hspre}l<{\hspost}@{}}%
\column{4}{@{}>{\hspre}l<{\hspost}@{}}%
\column{6}{@{}>{\hspre}l<{\hspost}@{}}%
\column{E}{@{}>{\hspre}l<{\hspost}@{}}%
\>[B]{}\Varid{newScoped}\;\mathrm{57}\mathbin{\$}\lambda \Varid{arr}\to {}\<[E]%
\\
\>[B]{}\hsindent{4}{}\<[4]%
\>[4]{}\mathbin{...}\mbox{\onelinecomment  some initialisation/computation involving arr}{}\<[E]%
\\
\>[B]{}\hsindent{4}{}\<[4]%
\>[4]{}\mathbf{let}!(\Conid{Ur}\;\Varid{s})\mathrel{=}\Varid{newScoped}\;\mathrm{2}\mathbin{\$}\lambda \Varid{prefix}\to \mathbf{do}{}\<[E]%
\\
\>[4]{}\hsindent{2}{}\<[6]%
\>[6]{}\Conid{Ur}\;\Varid{a0}\leftarrow \Varid{read}\;\Varid{arr}\;\mathrm{0}{}\<[E]%
\\
\>[4]{}\hsindent{2}{}\<[6]%
\>[6]{}\Conid{Ur}\;\Varid{a}_{1}\leftarrow \Varid{read}\;\Varid{arr}\;\mathrm{1}{}\<[E]%
\\
\>[4]{}\hsindent{2}{}\<[6]%
\>[6]{}\Varid{write}\;\Varid{prefix}\;\mathrm{0}\;\Varid{a0}{}\<[E]%
\\
\>[4]{}\hsindent{2}{}\<[6]%
\>[6]{}\Varid{write}\;\Varid{prefix}\;\mathrm{1}\;\Varid{a}_{1}{}\<[E]%
\\
\>[4]{}\hsindent{2}{}\<[6]%
\>[6]{}\Varid{sum}\;\Varid{prefix2}{}\<[E]%
\\
\>[B]{}\hsindent{4}{}\<[4]%
\>[4]{}\mathbf{in}{}\<[E]%
\\
\>[B]{}\hsindent{4}{}\<[4]%
\>[4]{}\mathbin{...}\mbox{\onelinecomment  more writing to arr}{}\<[E]%
\ColumnHook
\end{hscode}\resethooks

However, this example doesn't typecheck. The read and write capabilities for \ensuremath{\Varid{arr}} are captured
by the scope of the \ensuremath{\Varid{prefix}} array, and they're not allowed to escape that
scope. That's because \ensuremath{\Conid{Ur}} is a blunt instrument: if it prevents some linear
constraints from escaping (the capabilities for \ensuremath{\Varid{prefix}}) then it prevents all
capabilities from escaping.

So the \ensuremath{\Varid{prefix}} scope demands that \ensuremath{\Varid{arr}} be frozen before exiting the scope, and
we aren't allowed to write after exiting the scope. The only solution to make
this example work is to extend the scope of \ensuremath{\Varid{prefix}} all the way to the end of
the scope of \ensuremath{\Varid{arr}}.

So the end of \ensuremath{\Conid{Ur}}-terminated scope tend to stick together, which is
detrimental to the modularity of programs. Duplicable constraints is a
low price to pay to lift this limitation.

\subsection{Further considerations on freezing}
\label{sec:freezing-in-details}

\Cref{sec:o1-freeze} showed how linear constraints can be used to make a system
of nested mutable data structures which can be safely frozen. That is, which can
be turned into in an immutable data structure in $O(1)$.

Linear types on their own already give a solution~\citep{LinearHaskell} to safe
freezing of non-nested data structures, the \textsc{api} is the following
variant of \cref{fig:linear-interface}.

\begin{hscode}\SaveRestoreHook
\column{B}{@{}>{\hspre}l<{\hspost}@{}}%
\column{6}{@{}>{\hspre}l<{\hspost}@{}}%
\column{E}{@{}>{\hspre}l<{\hspost}@{}}%
\>[B]{}\Varid{new}{}\<[6]%
\>[6]{}\mathbin{::}\Conid{Int}\to (\Conid{MArray}\;\Varid{a}⊸\Conid{Ur}\;\Varid{r})⊸\Conid{Ur}\;\Varid{r}{}\<[E]%
\\
\>[B]{}\Varid{write}\mathbin{::}\Conid{MArray}\;\Varid{a}⊸\Conid{Int}\to \Varid{a}\to \Conid{MArray}\;\Varid{a}{}\<[E]%
\\
\>[B]{}\Varid{read}\mathbin{::}\Conid{MArray}\;\Varid{a}⊸\Conid{Int}\to (\Conid{MArray}\;\Varid{a},\Conid{Ur}\;\Varid{a}){}\<[E]%
\\
\>[B]{}\Varid{freeze}\mathbin{::}\Conid{MArray}\;\Varid{a}⊸\Conid{Ur}\;(\Conid{Array}\;\Varid{a}){}\<[E]%
\\
\>[B]{}\Varid{readI}\mathbin{::}\Conid{Array}\;\Varid{a}\to \Conid{Int}\to \Varid{a}{}\<[E]%
\ColumnHook
\end{hscode}\resethooks

In this \textsc{api}, we can only read and write unrestricted values into an
array. It has to be that way, otherwise it would be possible to turn any
linear value into an unrestricted linear value:

\begin{hscode}\SaveRestoreHook
\column{B}{@{}>{\hspre}l<{\hspost}@{}}%
\column{3}{@{}>{\hspre}l<{\hspost}@{}}%
\column{5}{@{}>{\hspre}l<{\hspost}@{}}%
\column{E}{@{}>{\hspre}l<{\hspost}@{}}%
\>[B]{}\Varid{write'}\mathbin{::}\Conid{MArray}\;\Varid{a}⊸\Conid{Int}\to \Varid{a}⊸\Conid{MArray}\;\Varid{a}{}\<[E]%
\\
\>[B]{}\Varid{read'}\mathbin{::}\Conid{MArray}\;\Varid{a}⊸\Conid{Int}\to (\Conid{MArray}\;\Varid{a},\Varid{a}){}\<[E]%
\\[\blanklineskip]%
\>[B]{}\Varid{badUr}\mathbin{::}\Varid{a}⊸\Conid{Ur}\;\Varid{a}{}\<[E]%
\\
\>[B]{}\Varid{badUr}\;\Varid{x}\mathrel{=}\Varid{newArray}\;\mathrm{1}\mathbin{\$}\lambda \Varid{arr}\to {}\<[E]%
\\
\>[B]{}\hsindent{3}{}\<[3]%
\>[3]{}\mathbf{case}\;\Varid{freeze}\;(\Varid{write}\;\mathrm{0}\;\Varid{x})\;\mathbf{of}{}\<[E]%
\\
\>[3]{}\hsindent{2}{}\<[5]%
\>[5]{}\Conid{Ur}\;\Varid{arr'}\to \Conid{Ur}\;(\Varid{readI}\;\Varid{arr'}\;\mathrm{0}){}\<[E]%
\ColumnHook
\end{hscode}\resethooks

In particular there's no hope of ever having an \ensuremath{\Conid{MArray}\;(\Conid{MArray}\;\Varid{a})} value: we
can construct this type, but it is not inhabited.
In order to build a nested mutable array, we seem to have to give up \ensuremath{\Varid{freeze}}.

At the heart of this problem is the solution to change the type upon freezing.
This is a rather standard solution in Haskell: mutable arrays and immutable
arrays are distinguished by having different types. For \ensuremath{\Conid{ST}} arrays, for
instance, you can read of an immutable array in pure code, while mutable array
reads can only be done in the \ensuremath{\Conid{ST}} monad. In the linearly typed solution,
similarly, immutable arrays must always be used linearly, while immutable arrays
are typically unrestricted.

So if we are to implement a freeze function for nested linear arrays, presumably
it should change the type of the elements of the array too.

\begin{hscode}\SaveRestoreHook
\column{B}{@{}>{\hspre}l<{\hspost}@{}}%
\column{E}{@{}>{\hspre}l<{\hspost}@{}}%
\>[B]{}\Varid{freezeNested0}\mathbin{::}\Conid{MArray}\;\Varid{a}\to \Conid{Ur}\;(\Conid{Array}\,\mathbin{??}){}\<[E]%
\ColumnHook
\end{hscode}\resethooks

There's no uniform way to choose the type \ensuremath{\mathbin{??}}, but we could use a type class.
We can take our inspiration from Haskell's safe \ensuremath{\Conid{Coerce}} type
class~\citep{safe-coercions}: \ensuremath{\Conid{Coerce}} is a type class to safely convert the
type of a value to another type with the same runtime representation,
implemented as a no-op. This is what we want our freeze function to do, it just
needs more restriction. Here's what it could look like

\begin{hscode}\SaveRestoreHook
\column{B}{@{}>{\hspre}l<{\hspost}@{}}%
\column{E}{@{}>{\hspre}l<{\hspost}@{}}%
\>[B]{}\mathbf{class}\;\Conid{Freeze}\;\Varid{a}\;\Varid{b}{}\<[E]%
\\
\>[B]{}\mathbf{instance}\;\Conid{Freeze}\;(\Conid{Ur}\;\Varid{a})\;(\Conid{Ur}\;\Varid{a}){}\<[E]%
\\
\>[B]{}\mathbf{instance}\;\constraintfont{\Conid{Freeze}\;\Varid{a}\;\Varid{b}}\FatArrow \Conid{Freeze}\;(\Conid{MArray}\;\Varid{a})\;(\Conid{Array}\;\Varid{b}){}\<[E]%
\\[\blanklineskip]%
\>[B]{}\Varid{freeze}\mathbin{::}\constraintfont{\Conid{Freeze}\;\Varid{a}\;\Varid{b}}\FatArrow \Varid{a}\to \Conid{Ur}\;\Varid{b}{}\<[E]%
\ColumnHook
\end{hscode}\resethooks

Note that, since \ensuremath{\Conid{Freeze}} cannot be reflexive, a nested array structure must
have leaves of type \ensuremath{\Conid{Ur}\;\Varid{a}}. Here, \ensuremath{\Conid{Ur}} plays the same role as \ensuremath{\Conid{Val}} in
\cref{sec:nested-arrays,sec:narray-impl,sec:o1-freeze}.

An issue with this approach is that since \ensuremath{\Conid{Freeze}} must be abstract, for safety,
it is hard to extend this type class. For instance, if one would make a
\begin{hscode}\SaveRestoreHook
\column{B}{@{}>{\hspre}l<{\hspost}@{}}%
\column{E}{@{}>{\hspre}l<{\hspost}@{}}%
\>[B]{}\mathbf{newtype}\;\Conid{Matrix}\;\Varid{a}\mathrel{=}\Conid{MkMatrix}\;(\Conid{MArray}\;(\Conid{MArray}\;\Varid{a})){}\<[E]%
\ColumnHook
\end{hscode}\resethooks
then \ensuremath{\Conid{Matrix}} would not implement \ensuremath{\Conid{Freeze}}.

Instead of going down this path, linear constraints suggest to rethink the
initial postulate: that freezing an array requires changing its type. Instead we
have a single type of arrays whose usage is constrained by a capability. For a
nested data structure the capability applies to the entire array. Then, only
this capability's type needs to change to freeze the entire nested data
structure. Since the capability is computationally trivial, freezing is $O(1)$
by construction. This solution is akin to Rust's ownership, where freezing
corresponds to storing an owned value \text{\ttfamily A} inside an \text{\ttfamily RC\char60{}A\char62{}} or a similar
structure. This is precisely what we propose in \cref{sec:o1-freeze}: the key is
that freezing arrays doesn't change their type.

\subsection{Variations on release operators}
\label{sec:unrestricted-release}

In \cref{sec:borrowing-slices}, we introduced how linear constraints can be
leveraged to create an ownership discipline where sub-structures (there array
slices) can be borrowed. Borrowing functions return the borrowed sub-structure
together with what we've been calling a release operator. \emph{E.g.} in

\begin{hscode}\SaveRestoreHook
\column{B}{@{}>{\hspre}l<{\hspost}@{}}%
\column{14}{@{}>{\hspre}l<{\hspost}@{}}%
\column{E}{@{}>{\hspre}l<{\hspost}@{}}%
\>[B]{}\Varid{restrict}\mathbin{::}{}\<[14]%
\>[14]{}\constraintfont{\Conid{RW}\;\Varid{n}}\Lolly \Conid{UArray}\;\Varid{a}\;\Varid{n}\to \Conid{Int}\to \Conid{Int}\to {}\<[E]%
\\
\>[14]{}\exists\;\Varid{p}.(\Conid{Ur}\;(\Conid{UArray}\;\Varid{a}\;\Varid{p}),(\constraintfont{\Conid{RW}\;\Varid{p}}\Lolly ()\RLolly\constraintfont{\Conid{RW}\;\Varid{n}}))\RLolly\constraintfont{\Conid{RW}\;\Varid{p}}{}\<[E]%
\ColumnHook
\end{hscode}\resethooks

We borrow an \ensuremath{\Conid{UArray}} with a new type-level name \ensuremath{\Varid{p}}, and we also get a release
operator of type \ensuremath{\constraintfont{\Conid{RW}\;\Varid{p}}\Lolly ()\RLolly\constraintfont{\Conid{RW}\;\Varid{n}}}. In this section
we explore variations on their design.

\subsubsection{Unrestricted release operator}

We've made it so that the release operator is linear, so it must be called
exactly once. It doesn't change the design a lot to return an unrestricted
release operator \ensuremath{\Conid{Ur}\;(\constraintfont{\Conid{RW}\;\Varid{p}}\Lolly ()\RLolly\constraintfont{\Conid{RW}\;\Varid{n}})}. After all,
the \ensuremath{\Conid{RW}\;\Varid{p}} constraint is linear as well, so you always must use the release
operator exactly once whether it is linear or unrestricted. Maybe it would be
more convenient for the release operators to be unrestricted.

Except, while this holds for the simple \ensuremath{\Conid{UArray}}, it is not entirely true of the
nested \ensuremath{\Conid{NArray}} where there's another way to consume an \ensuremath{\Conid{RW}\;\Varid{p}} constraint in our
\textsc{api}: the \ensuremath{\Varid{write}} function.

\begin{hscode}\SaveRestoreHook
\column{B}{@{}>{\hspre}l<{\hspost}@{}}%
\column{E}{@{}>{\hspre}l<{\hspost}@{}}%
\>[B]{}\Varid{borrowNA}\mathbin{::}\constraintfont{\Conid{RW}\;\Varid{n}}\Lolly \Conid{NArray}\;\Varid{a}\;\Varid{n}\to \Conid{Int}\to \exists\;\Varid{p}.(\Conid{Ur}\;(\Varid{a}\;\Varid{p}),(\constraintfont{\Conid{RW}\;\Varid{p}}\Lolly ()\RLolly\constraintfont{\Conid{RW}\;\Varid{n}}))\RLolly\constraintfont{\Conid{RW}\;\Varid{p}}{}\<[E]%
\\
\>[B]{}\Varid{writeNA}\mathbin{::}\constraintfont{(\Conid{RW}\;\Varid{n},\Conid{RW}\;\Varid{p})}\Lolly \Conid{NArray}\;\Varid{a}\;\Varid{n}\to \Conid{Int}\to \Varid{a}\;\Varid{p}\to ()\RLolly\constraintfont{\Conid{RW}\;\Varid{n}}{}\<[E]%
\ColumnHook
\end{hscode}\resethooks

If the release operator returned by \ensuremath{\Varid{borrowNA}} was unrestricted, then it would
be possible to borrow a slice and never return it. Instead it could be stored in
another \ensuremath{\Conid{NArray}}.

This is semantically dubious, as it means never giving control back to the
original owner, which remains unavailable for the rest of the computation. If
the owner has responsibilities (\emph{e.g.} freeing memory), then they won't be
performed. So linear release operator, or rather lack of a linear release
operator, acts as a proof of ownership.

In the case of our arrays there is no problem as significant as not freeing
memory, however, in \cref{sec:nested-arrays}, linear release operators prevent
writing borrowed slice inside deep borrows, which would be a runtime error.

So if we were to have release operator be unrestricted, we would presumably need
to keep track of ownership another way. For instance by adding another
linear-constraint capability for owned structure which would be consumed by
\ensuremath{\Varid{writeNA}} but borrowed slice wouldn't have.

\subsubsection{Release function as a constraint}

An orthogonal variation that we could make to the release operator is turn it
into a constraint. That is borrowing functions would look like

\begin{hscode}\SaveRestoreHook
\column{B}{@{}>{\hspre}l<{\hspost}@{}}%
\column{14}{@{}>{\hspre}l<{\hspost}@{}}%
\column{E}{@{}>{\hspre}l<{\hspost}@{}}%
\>[B]{}\Varid{restrict}\mathbin{::}{}\<[14]%
\>[14]{}\constraintfont{\Conid{RW}\;\Varid{n}}\Lolly \Conid{UArray}\;\Varid{a}\;\Varid{n}\to \Conid{Int}\to \Conid{Int}\to {}\<[E]%
\\
\>[14]{}\exists\;\Varid{p}.\Conid{Ur}\;(\Conid{UArray}\;\Varid{a}\;\Varid{p})\RLolly\constraintfont{(\Conid{RW}\;\Varid{p},(\constraintfont{\Conid{RW}\;\Varid{p}}\Lolly \constraintfont{\Conid{RW}\;\Varid{n}}))}{}\<[E]%
\ColumnHook
\end{hscode}\resethooks

The constraint solver that we've described in \cref{sec:constraint-solver},
cannot handle such implication given constraints. But \textsc{ghc} can, via a
feature called ``quantified constraints'' \citep{quantified-constraints}. In
this article we didn't include quantified constraints because they add quite a
bit of complication which we considered to be a distraction for the sake of this
article. But we expect linear constraint and quantified constraints to integrate
without an issue, and an integration into \textsc{ghc} would automatically
have qualified linear constraint.

If we were to use such release constraints instead of release operators, then
release constraints would be called implicitly by the typechecker where they
would be needed. For instance, the following snippet from \cref{sec:valiant}

\begin{hscode}\SaveRestoreHook
\column{B}{@{}>{\hspre}l<{\hspost}@{}}%
\column{3}{@{}>{\hspre}l<{\hspost}@{}}%
\column{5}{@{}>{\hspre}l<{\hspost}@{}}%
\column{7}{@{}>{\hspre}l<{\hspost}@{}}%
\column{E}{@{}>{\hspre}l<{\hspost}@{}}%
\>[3]{}\mathbf{if}\;\Varid{size}\;\Varid{p}\mathbin{>}\Varid{size}\;\Varid{q}{}\<[E]%
\\
\>[3]{}\hsindent{2}{}\<[5]%
\>[5]{}\mathbf{then}\;\Conid{Linearly}.\mathbf{do}{}\<[E]%
\\
\>[5]{}\hsindent{2}{}\<[7]%
\>[7]{}\mathbf{let}\;\Varid{i}\mathrel{=}\Varid{size}\;\Varid{p}\mathbin{\Varid{`div`}}\mathrm{2}{}\<[E]%
\\
\>[5]{}\hsindent{2}{}\<[7]%
\>[7]{}(\Conid{Ur}\;(\Varid{a},\Varid{x},\Varid{b}),\Varid{release}_{\Varid{p}})\leftarrow \Varid{splitUpperMatrix}\;\Varid{p}\;\Varid{i}{}\<[E]%
\\
\>[5]{}\hsindent{2}{}\<[7]%
\>[7]{}(\Conid{Ur}\;(\Varid{y},\Varid{z}),\Varid{release}_{\Varid{n}})\leftarrow \Varid{sliceMatrixV}\;\Varid{n}\;\Varid{i}{}\<[E]%
\\
\>[5]{}\hsindent{2}{}\<[7]%
\>[7]{}\Varid{w}\;\Varid{b}\;\Varid{z}\;\Varid{q}{}\<[E]%
\\
\>[5]{}\hsindent{2}{}\<[7]%
\>[7]{}\Varid{multiplyMatricesInto}\;\Varid{y}\;\Varid{x}\;\Varid{z}{}\<[E]%
\\
\>[5]{}\hsindent{2}{}\<[7]%
\>[7]{}\Varid{w}\;\Varid{a}\;\Varid{y}\;\Varid{q}{}\<[E]%
\\
\>[5]{}\hsindent{2}{}\<[7]%
\>[7]{}\Varid{release}_{\Varid{p}};\Varid{release}_{\Varid{n}}{}\<[E]%
\\
\>[3]{}\hsindent{2}{}\<[5]%
\>[5]{}\mathbf{else}{}\<[E]%
\ColumnHook
\end{hscode}\resethooks
would instead look like
\begin{hscode}\SaveRestoreHook
\column{B}{@{}>{\hspre}l<{\hspost}@{}}%
\column{3}{@{}>{\hspre}l<{\hspost}@{}}%
\column{5}{@{}>{\hspre}l<{\hspost}@{}}%
\column{7}{@{}>{\hspre}l<{\hspost}@{}}%
\column{E}{@{}>{\hspre}l<{\hspost}@{}}%
\>[3]{}\mathbf{if}\;\Varid{size}\;\Varid{p}\mathbin{>}\Varid{size}\;\Varid{q}{}\<[E]%
\\
\>[3]{}\hsindent{2}{}\<[5]%
\>[5]{}\mathbf{then}\;\Conid{Linearly}.\mathbf{do}{}\<[E]%
\\
\>[5]{}\hsindent{2}{}\<[7]%
\>[7]{}\mathbf{let}\;\Varid{i}\mathrel{=}\Varid{size}\;\Varid{p}\mathbin{\Varid{`div`}}\mathrm{2}{}\<[E]%
\\
\>[5]{}\hsindent{2}{}\<[7]%
\>[7]{}\Conid{Ur}\;(\Varid{a},\Varid{x},\Varid{b})\leftarrow \Varid{splitUpperMatrix}\;\Varid{p}\;\Varid{i}{}\<[E]%
\\
\>[5]{}\hsindent{2}{}\<[7]%
\>[7]{}\Conid{Ur}\;(\Varid{y},\Varid{z})\leftarrow \Varid{sliceMatrixV}\;\Varid{n}\;\Varid{i}{}\<[E]%
\\
\>[5]{}\hsindent{2}{}\<[7]%
\>[7]{}\Varid{w}\;\Varid{b}\;\Varid{z}\;\Varid{q}{}\<[E]%
\\
\>[5]{}\hsindent{2}{}\<[7]%
\>[7]{}\Varid{multiplyMatricesInto}\;\Varid{y}\;\Varid{x}\;\Varid{z}{}\<[E]%
\\
\>[5]{}\hsindent{2}{}\<[7]%
\>[7]{}\Varid{w}\;\Varid{a}\;\Varid{y}\;\Varid{q}{}\<[E]%
\\
\>[3]{}\hsindent{2}{}\<[5]%
\>[5]{}\mathbf{else}{}\<[E]%
\ColumnHook
\end{hscode}\resethooks

Which, arguably, is a little cleaner. However, while the main reason for not
including releases as constraints was that we didn't think it was worth treating
quantified constraints in this article, it is worth mentioning that these authors
have had discussions among ourselves as to whether release constraints were not
a little too magical. Maybe clarifying the ends of scope with explicit calls to
release operators is desirable.

\subsection{Design of the constraint solver}
\label{sec:solver-design}

The judgements of our constraint solver (\cref{sec:constraint-solver}) are of the
form
$$
\vdashs   \constraintfont{ \constraintfont{C} }   \leadsto   \constraintfont{ \ottnt{Q} } 
$$
In which $  \constraintfont{ \constraintfont{C} }  $ is an input and $  \constraintfont{ \ottnt{Q} }  $ is an output. This is not the
typical shape of linear-logic proof-search algorithms, in particular
\cite{hh-ll} and \cite{resource-management-for-ll-proof-search}, whose
proof-search algorithm our constraint solver is based on, do not have a judgement
of this form (yet we still claim that our algorithm is essentially the same). In
fact, even the earlier, shorter, version of this article \citep{SpiwackKBWE22}
used the more typical form of judgement:
$$
  \constraintfont{ \ottnt{U} }   ;   \constraintfont{ \ottnt{L}_{\ottmv{i}} }    \vdashs    \constraintfont{ \constraintfont{C} }    \leadsto    \constraintfont{ \ottnt{L}_{\ottmv{r}} }  
$$
In this judgement $  \constraintfont{ \ottnt{L}_{\ottmv{i}} }  $ is an input context of linear atomic constraints
($  \constraintfont{ \ottnt{U} }  $ are unrestricted constraints) available so solve $  \constraintfont{ \constraintfont{C} }  $. The
judgement then outputs the \emph{leftover} constraints $  \constraintfont{ \ottnt{L}_{\ottmv{r}} }  $, those
linear atomic constraints from $  \constraintfont{ \ottnt{L}_{\ottmv{i}} }  $ which haven't been used in the
resolution of $  \constraintfont{ \constraintfont{C} }  $ and can (and, in fact, must) be used to solve other
constraints. For instance the rule to solve $  \constraintfont{ \constraintfont{C_{{\mathrm{1}}}}  \qtensor  \constraintfont{C_{{\mathrm{2}}}} }  $ is the following:
$$
\inferrule
{   \constraintfont{ \ottnt{U} }   ;   \constraintfont{ \ottnt{L}_{\ottmv{i}} }    \vdashs    \constraintfont{ \constraintfont{C_{{\mathrm{1}}}} }    \leadsto    \constraintfont{ \ottnt{L}_{\ottmv{r}} }  \\
    \constraintfont{ \ottnt{U} }   ;   \constraintfont{ \ottnt{L}_{\ottmv{r}} }    \vdashs    \constraintfont{ \constraintfont{C_{{\mathrm{2}}}} }    \leadsto    \constraintfont{ \ottnt{L}_{\ottmv{o}} }  
}
{  \constraintfont{ \ottnt{U} }   ;   \constraintfont{ \ottnt{L}_{\ottmv{i}} }    \vdashs    \constraintfont{ \constraintfont{C_{{\mathrm{1}}}}  \qtensor  \constraintfont{C_{{\mathrm{2}}}} }    \leadsto    \constraintfont{ \ottnt{L}_{\ottmv{o}} }  }
$$
we see how the leftovers of $  \constraintfont{ \constraintfont{C_{{\mathrm{1}}}} }  $ are then used as the input of $\constraintfont{C_{{\mathrm{2}}}}$.

In reference to Haskell's classic monads, we will call our presentation a
\emph{writer-style} solver and the classic presentation a \emph{state-style}
solver.

The proof-search literature prefers state-style solvers because they backtrack
less. For instance, when solving $  \constraintfont{   \multiplicityfont{ \ottsym{1} }  \scale ( \ottnt{Q}  \Lolly    \multiplicityfont{ \ottsym{1} }  \scale ( \ottnt{Q}  \Lolly  \ottnt{Q}  \qtensor  \ottnt{Q} )  )  }  $, in the writer style,
both branches of the tensor could use, say, the inner $  \constraintfont{ \ottnt{Q} }  $ hypothesis,
which violates the linearity condition and will cause backtracking, while in the
state style, both branch will coordinate and necessary pick different $  \constraintfont{ \ottnt{Q} }  $
and find a proof without backtracking at all.

But our solving algorithm must not backtrack, so we are not concerned
with this consideration. However, the constraint
$  \constraintfont{   \multiplicityfont{ \ottsym{1} }  \scale ( \ottnt{Q}  \Lolly    \multiplicityfont{ \ottsym{1} }  \scale ( \ottnt{Q}  \Lolly  \ottnt{Q}  \qtensor  \ottnt{Q} )  )  }  $ is an example where the backtracking-free
state-style solver can solve a constraint while the writer-style
solver cannot.

On the other hand, for the state-style solver not to be sensitive to the
ordering of the constraints, we had to make a concession and consider that
having the same atomic constraint as an unrestricted and a linear given was an
ambiguity, so that the state-style solver of \citet{SpiwackKBWE22} cannot solve
$  \constraintfont{   \multiplicityfont{ \ottsym{1} }  \scale (   \multiplicityfont{ \omega }  \scale \ottnt{q}   \Lolly    \multiplicityfont{ \ottsym{1} }  \scale (   \multiplicityfont{ \ottsym{1} }  \scale \ottnt{q}   \Lolly    \multiplicityfont{ \ottsym{1} }  \scale \ottnt{q}  )  )  }  $. Yet the writer-style solver of \cref{fig:constraint-solver} can solve this example.

So really, the state-style and writer-style constraint solvers are incomparable.
Why choose one over the other? Here are the considerations that went into this
choice:

\begin{itemize}
  \item Writer-style solvers are generally simpler (both to present and to
        implement). The use of state-style solver for proof-search sacrifices
        simplicity for algorithmic considerations which don't apply to this
        paper.
  \item When solving $  \constraintfont{ \constraintfont{C_{{\mathrm{1}}}}  \qtensor  \constraintfont{C_{{\mathrm{2}}}} }  $, $  \constraintfont{ \constraintfont{C_{{\mathrm{1}}}} }  $ and $  \constraintfont{ \constraintfont{C_{{\mathrm{2}}}} }  $ are perfectly
        symmetric, making typechecking robust to reordering the source program.
        This isn't the case in state-passing style where $  \constraintfont{ \constraintfont{C_{{\mathrm{1}}}} }  $ and
        $  \constraintfont{ \constraintfont{C_{{\mathrm{2}}}} }  $ are checked in different contexts. Which requires additional
        care to make sure that reordering $  \constraintfont{ \constraintfont{C_{{\mathrm{1}}}} }  $ and $  \constraintfont{ \constraintfont{C_{{\mathrm{2}}}} }  $ won't affect
        the resolution. The writer style preserves this symmetry.
  \item Solving $  \constraintfont{ \constraintfont{C_{{\mathrm{1}}}}  \qtensor  \constraintfont{C_{{\mathrm{2}}}} }  $ with a writer-style solver processes $  \constraintfont{ \constraintfont{C_{{\mathrm{1}}}} }  $
        and $  \constraintfont{ \constraintfont{C_{{\mathrm{2}}}} }  $ in parallel, whereas the state-passing style
        sequentialises $  \constraintfont{ \constraintfont{C_{{\mathrm{1}}}} }  $ and $  \constraintfont{ \constraintfont{C_{{\mathrm{2}}}} }  $. This matters for integrating
        linear constraints into \textsc{ghc}'s constraint solver. \textsc{Ghc}
        likes to suspend constraints it doesn't know how to solve yet: if
        \textsc{ghc} don't know how to solve $  \constraintfont{ \constraintfont{C_{{\mathrm{1}}}} }  $, it will pause solving
        $  \constraintfont{ \constraintfont{C_{{\mathrm{1}}}} }  $ and start solving $  \constraintfont{ \constraintfont{C_{{\mathrm{2}}}} }  $ hoping that the resolution of
        $  \constraintfont{ \constraintfont{C_{{\mathrm{2}}}} }  $ will give clues to help solve $  \constraintfont{ \constraintfont{C_{{\mathrm{1}}}} }  $. Suspension,
        however, isn't compatible with the sequentialised state-passing style:
        it is not possible to wait until $  \constraintfont{ \constraintfont{C_{{\mathrm{2}}}} }  $ is solved to solve
        $  \constraintfont{ \constraintfont{C_{{\mathrm{1}}}} }  $ if solving $  \constraintfont{ \constraintfont{C_{{\mathrm{2}}}} }  $ requires the output of $  \constraintfont{ \constraintfont{C_{{\mathrm{1}}}} }  $.
  \item Incidentally, counting down from a context like the state-passing style
        does doesn't generalise to arbitrary semiring in the style of
        \cite{abel_unified_2020}. While there is no current plan to extend
        \text{ghc} with a richer semiring than multiplicities, the
        writer-style is more future proof.
\end{itemize}

A final consideration is that the errors of the writer-style solver seems to
match programmers expectations better.
Consider the following example:
\begin{hscode}\SaveRestoreHook
\column{B}{@{}>{\hspre}l<{\hspost}@{}}%
\column{7}{@{}>{\hspre}l<{\hspost}@{}}%
\column{9}{@{}>{\hspre}l<{\hspost}@{}}%
\column{11}{@{}>{\hspre}l<{\hspost}@{}}%
\column{14}{@{}>{\hspre}c<{\hspost}@{}}%
\column{14E}{@{}l@{}}%
\column{18}{@{}>{\hspre}l<{\hspost}@{}}%
\column{E}{@{}>{\hspre}l<{\hspost}@{}}%
\>[B]{}\Varid{f}\mathrel{=}{}\<[7]%
\>[7]{}\Varid{linearly}\mathbin{\$}\Conid{Linearly}.\mathbf{do}{}\<[E]%
\\
\>[7]{}\hsindent{2}{}\<[9]%
\>[9]{}\Conid{Ur}\;\Varid{arr}\leftarrow \Varid{new}\;\mathrm{10}{}\<[E]%
\\
\>[7]{}\hsindent{2}{}\<[9]%
\>[9]{}\Varid{fr}{}\<[14]%
\>[14]{}\mathbin{::}{}\<[14E]%
\>[18]{}\constraintfont{\Conid{RW}\;\Varid{n}}\Lolly ()\RLolly()\mbox{\onelinecomment  fails here with \ensuremath{\constraintfont{\Conid{RW}\;\Varid{n}}} used more than once}{}\<[E]%
\\
\>[7]{}\hsindent{2}{}\<[9]%
\>[9]{}\Varid{fr}{}\<[14]%
\>[14]{}\mathrel{=}{}\<[14E]%
\>[18]{}\Conid{Linearly}.\mathbf{do}{}\<[E]%
\\
\>[9]{}\hsindent{2}{}\<[11]%
\>[11]{}\Varid{free}\;\Varid{arr}{}\<[E]%
\\
\>[9]{}\hsindent{2}{}\<[11]%
\>[11]{}\Varid{free}\;\Varid{arr}{}\<[E]%
\\
\>[7]{}\hsindent{2}{}\<[9]%
\>[9]{}()\leftarrow \Varid{fr}\mbox{\onelinecomment  A state-style solver would fail here saying that there is no \ensuremath{\constraintfont{\Conid{RW}\;\Varid{n}}} left to use}{}\<[E]%
\\
\>[7]{}\hsindent{2}{}\<[9]%
\>[9]{}\Varid{\Conid{Linearly}.return}\mathbin{\$}\Conid{Ur}\;(){}\<[E]%
\ColumnHook
\end{hscode}\resethooks
In this example, it is possible to type \ensuremath{\Varid{fr}} in the declarative type system.
However, \ensuremath{\Varid{fr}} is incorrect as it frees the array several time, which isn't
permitted. A state-style solver would let \ensuremath{\Varid{fr}} typecheck by using the capability
generated by \ensuremath{\Varid{new}} (on the second line) to free the array a second time,
ultimately pointing at the wrong location for a mistake. Whereas the
writer-style solver points at the actual mistake.

\section{Related work}
\label{sec:related-work}

\paragraph*{OutsideIn}
\label{sec:outsidein}

Our aim is to integrate the present work in \textsc{ghc}, and
accordingly the qualified type system in
\cref{sec:qualified-type-system} and the constraint inference
algorithm in \cref{sec:type-inference} follow a similar
presentation to that of OutsideIn~\cite{OutsideIn}, \textsc{ghc}'s
constraint solver algorithm.  Even though our presentation is
self-contained, we outline some of the differences from that work.

The solver judgement in OutsideIn takes the following form:
\[\mathcal{Q}\ ;\ Q_{\mathit{given}}\ ;\ \overline{\alpha}_{\mathit{tch}} \overset{\mathit{solv}}{\mapsto} C_{\mathit{wanted}} \leadsto Q_{\mathit{residual}}\  ; \ \theta\]
The main differences between OutsideIn's solver judgement and our solver judgements in \cref{sec:constraint-solver} are:
\begin{itemize}
  \item OutsideIn's judgement includes top-level axioms schemes separately
($\mathcal{Q}$), which we have omitted for the sake of simplicity.
  \item OutsideIn's judgement has an input context $Q_{\mathit{given}}$ whereas
        we don't. Having such a context is more realistic from an implementation
        point of view: in practice given constraints are named in
        $Q_{\mathit{given}}$, the evidence for solving a constraint $[{{C}}]$
        can then refer to those names. Since we avoid names in our presentation,
        we don't need such a context. The context $Q_{\mathit{given}}$ is also
        useful for quantified constraints \cite{quantified-constraints}, where
        proving a given is allowed to be a pattern that matches the goal.
  \item In addition to constraint inference, OutsideIn performs type
inference, requiring additional bookkeeping in the solver judgment. The solver
takes as input a set of \emph{touchable} variables $\overline{\alpha}_{tch}$
which record the type variables that can be unified at any given time, and
produces a type substitution $\theta$ as an output.
As discussed in \cref{sec:type-inference}, we do not perform type
inference, only constraint inference. Therefore, our solver need not return a
type assignment.
  \item Both OutsideIn and our solver output a set of constraints,
        $Q_{\mathit{residual}}$ and $  \constraintfont{ \ottnt{Q_{\ottmv{r}}} }  $ respectively. They are similar
        in spirit, but not used exactly the same. Note in particular that while
        OutsideIn has both an input and output context, in the terminology of
        \cref{sec:solver-design}, OutsideIn is a combination of writer-style and
        reader-style, not a state-style solver.

        Where they differ is that $Q_{\mathit{residual}}$ is used for type
        inference (when a function has no type signature), and as such only
        contains the constraints which are not in $Q_{\mathit{given}}$. In our
        solver it is essential that $  \constraintfont{ \ottnt{Q_{\ottmv{r}}} }  $ returns all the hypotheses that
        have been needed, except those who have been discharged by an
        implication, so that we can check that their usage respect the linearity
        constraints.

\item Finally, while OutsideIn has a single kind of conjunction, our constraint
language requires two: $  \constraintfont{ \ottnt{Q_{{\mathrm{1}}}}  \qtensor  \ottnt{Q_{{\mathrm{2}}}} }  $ and $  \constraintfont{ \ottnt{Q_{{\mathrm{1}}}}  \aand  \ottnt{Q_{{\mathrm{2}}}} }  $. This shows up when
generating constraints for $\kcase$ expressions in the~\rref{G-Case} rule.
OutsideIn accumulates constraints across branches (taking the union of each
branch), whereas we need to make sure that each branch of a $\kcase$-expression
consumes the same constraints.
\end{itemize}

\paragraph*{Ownership}

%
Ownership and borrowing are the key features of Rust's safe memory management model.
In \cref{sec:memory-ownership} we show how linear constraints can be used to
implement such an ownership model as a library.
Although linear constraints do not have the
convenience of Rust's syntax, we expect that they will support a
greater variety of abstractions.

Clean is another language with built-in ownership typing. Like Haskell
it is a lazy language. Mutation is performed by returning a new
reference, like in Linear Haskell without linear constraints.

\paragraph*{Languages with capabilities}

The idea of using capabilities to enforce high-level resource usage protocols is not new~\citep{DBLP:conf/pldi/DeLineF01},
and as such has been applied in practical programming languages before.
Both Mezzo~\cite{mezzo-permissions} and
\textsc{ats}~\cite{AtsLinearViews} served as inspiration for the
design of linear constraints. Of the two, Mezzo is more specialised,
being entirely built around its system of capabilities.  \textsc{Ats}
is the closest to our system because it appeals explicitly to linear
logic, and because the capabilities (known as \emph{stateful views})
are not tied to a particular use case. However,
\textsc{ats} does not have full inference of capabilities.

Other than that, the two systems have a lot of similarities. They have a
finer-grained capability system than is expressible in Rust (or our
encoding of it in \cref{sec:memory-ownership}) which makes it possible to change
the type of a reference cell upon write (though linear constraints could be used to implement such type-changing references too). They also eschew scoped
borrowing in favour of more traditional read and write capabilities.

Linear constraints are more general than either Mezzo or \textsc{ats},
while maintaining a considerably simpler inference algorithm, and at
the same time supporting a richer set of constraints (such as \textsc{gadt}s). This
simplicity is a benefit of abstracting over the simple-constraint
domain. In fact, it should be possible to see Mezzo or \textsc{ats} as
particular instantiations of the simple-constraint domain, with linear
constraints providing the general inference mechanism.


\paragraph*{Linearly typed languages}

Affe~\cite{kindly-bent} is a linearly typed \textsc{ml}-style core
language with mutable references and arrays, augmented with a notion
of borrowing. It has dedicated syntax for the scope of borrows. In
contrast, we represent scopes as functions. Affe is presented as a
fully integrated solution, while linear constraints is a small layer
on top of Linear Haskell.

\paragraph*{Logic programming}

There are a lot of commonalities between \textsc{ghc}'s constraint and logic
programs. Traditional type classes can be seen as Horn clause programs, much
like Prolog programs. \textsc{ghc} puts further restrictions in order to
avoid backtracking for speed and predictability.

The recent addition of quantified
constraints~\cite{quantified-constraints} extends type class
resolution to Hereditary Harrop~\cite{hereditary-harrop} programs. A generalisation of the
Hereditary Harrop fragment to linear logic, described by~\citet{hh-ll},
is the foundation of the Lolli language~\cite{hodas-thesis-lolli}.
The authors also coin the notion of \emph{uniform} proof. A fragment where
uniform proofs are complete supports goal-oriented proof search, like
Prolog does.

Completeness of uniform proofs is equivalent to
\cref{lem:inversion}, which, in turn, is used in the proof of the
soundness \cref{lem:generation-soundness}. Therefore our linear
constraints are compatible with quantified constraints: we simply need
to adapt~\cref{lem:inversion}.

It is interesting that goal-oriented search is baked into the
definition of OutsideIn. It is not only used as the constraint solving
strategy, but it seems to required for the soundness of the constraint
generation algorithm. Or, if they are not required, uniform proofs are
at least an effective strategy to prove soundness.



\section{Conclusion}
\label{sec:conclusion}

We showed how a simple linear type system like that of Linear
Haskell can be extended with an inference mechanism which lets the
compiler manage some of the additional complexity of linear types
instead of the programmer. Linear constraints narrow the gap between linearly
typed languages and dedicated linear-like typing disciplines such as Rust's,
Mezzo's, or \textsc{ats}'s.

\bibliographystyle{ACM-Reference-Format}
\bibliography{bibliography}
\clearpage

\appendix

\renewcommand{\appendix}{}
\renewcommand{\appendixprelim}{}

\section{Complete Desugaring}
\label{sec:appendix:desugaring}

The complete definition of the desugaring function from
\cref{sec:desugaring} can be found in
\cref{fig:full:desugaring}.

For the sake of concision, we allow ourselves to write nested patterns
in $\kcase$ expressions of the core language. Desugaring nested patterns
into atomic $\kcase$ expression is routine.

In the complete description, we use a device which was omitted in the
main body of the article. Namely, we'll need a way to turn every
$  \dsevidence{\constraintfont{   \multiplicityfont{ \omega }  \scale \ottnt{Q}  } }  $ into an $ \ottkw{Ur} \, \ottsym{(}   \dsevidence{\constraintfont{ \ottnt{Q} } }   \ottsym{)} $. For any
$\ottnt{e}  \ottsym{:}   \dsevidence{\constraintfont{   \multiplicityfont{ \omega }  \scale \ottnt{Q}  } } $, we shall write $ \underline{ \ottnt{e} }_{  \constraintfont{ \ottnt{Q} }  }   \ottsym{:}  \ottkw{Ur} \, \ottsym{(}   \dsevidence{\constraintfont{   \multiplicityfont{ \omega }  \scale \ottnt{Q}  } }   \ottsym{)}$. As a shorthand, particularly useful in nested
patterns, we will write $  \kcase_  \multiplicityfont{ \pi }   \, \ottnt{e} \, \ottkw{of} \, \ottsym{\{}   \underline{ \ottmv{x} }_{  \constraintfont{ \ottnt{Q} }  }   \to  \ottnt{e'}  \ottsym{\}} $ for
$  \kcase_  \multiplicityfont{ \pi }   \,  \underline{ \ottnt{e} }_{  \constraintfont{ \ottnt{Q} }  }  \, \ottkw{of} \, \ottsym{\{}  \ottkw{Ur} \, \ottmv{x}  \to  \ottnt{e'}  \ottsym{\}} $.
$$
\left\{
  \begin{array}{lcl}
      \underline{ \ottnt{e} }_{  \constraintfont{  \mathbf{\varepsilon}  }  }  & = &   \kcase_  \multiplicityfont{ \ottsym{1} }   \, \ottnt{e} \, \ottkw{of} \, \ottsym{\{}  \ottsym{()}  \to  \ottkw{Ur} \, \ottsym{()}  \ottsym{\}}  \\
      \underline{ \ottnt{e} }_{  \constraintfont{   \multiplicityfont{ \ottsym{1} }  \scale \ottnt{q}  }  }   & = &  \ottnt{e}  \\
      \underline{ \ottnt{e} }_{  \constraintfont{   \multiplicityfont{ \omega }  \scale \ottnt{q}  }  }   & = &   \kcase_  \multiplicityfont{ \ottsym{1} }   \, \ottnt{e} \, \ottkw{of} \, \ottsym{\{}  \ottkw{Ur} \, \ottmv{x}  \to  \ottkw{Ur} \, \ottsym{(}  \ottkw{Ur} \, \ottmv{x}  \ottsym{)}  \ottsym{\}}  \\
      \underline{ \ottnt{e} }_{  \constraintfont{ \ottnt{Q_{{\mathrm{1}}}}  \qtensor  \ottnt{Q_{{\mathrm{2}}}} }  }   & = &   \kcase_  \multiplicityfont{ \ottsym{1} }   \, \ottnt{e} \, \ottkw{of} \, \ottsym{\{}  \ottsym{(}   \underline{ \ottmv{x} }_{  \constraintfont{ \ottnt{Q_{{\mathrm{1}}}} }  }   \ottsym{,}   \underline{ \ottmv{y} }_{  \constraintfont{ \ottnt{Q_{{\mathrm{2}}}} }  }   \ottsym{)}  \to  \ottkw{Ur} \, \ottsym{(}  \ottmv{x}  \ottsym{,}  \ottmv{y}  \ottsym{)}  \ottsym{\}} 
  \end{array}
\right.
$$
We will omit the $  \constraintfont{ \ottnt{Q} }  $ in $  \underline{ \ottnt{e} }_{  \constraintfont{ \ottnt{Q} }  }  $ and write
$  \underline{ \ottnt{e} }  $ when it can be easily inferred from the context.

\begin{figure}
  \small
  \centering

$$
\left\{\;
\begin{minipage}{0.8\linewidth}
\begin{hscode}\SaveRestoreHook
\column{B}{@{}>{\hspre}l<{\hspost}@{}}%
\column{3}{@{}>{\hspre}l<{\hspost}@{}}%
\column{5}{@{}>{\hspre}l<{\hspost}@{}}%
\column{7}{@{}>{\hspre}l<{\hspost}@{}}%
\column{29}{@{}>{\hspre}c<{\hspost}@{}}%
\column{29E}{@{}l@{}}%
\column{53}{@{}>{\hspre}c<{\hspost}@{}}%
\column{53E}{@{}l@{}}%
\column{67}{@{}>{\hspre}c<{\hspost}@{}}%
\column{67E}{@{}l@{}}%
\column{E}{@{}>{\hspre}l<{\hspost}@{}}%
\>[B]{}\dsterm{\ottmv{\Varid{z}}}{  \constraintfont{ \ottnt{Q} }  ; \Gamma  \vdash \Varid{x}\mathbin{:} \upsilon \;[\mskip1.5mu  \overline{\tau} \mathbin{/} \overline{\ottmv{a} } \mskip1.5mu]}{}\<[29]%
\>[29]{}\mathrel{=}{}\<[29E]%
\\
\>[B]{}\hsindent{5}{}\<[5]%
\>[5]{}\Varid{x}\;\Varid{z}{}\<[E]%
\\
\>[B]{}\dsterm{\ottmv{\Varid{z}}}{  \constraintfont{ \ottnt{Q} }  ; \Gamma  \vdash \lambda \Varid{x}.\Varid{e}\mathbin{:} \tau_{{\mathrm{1}}} \to_{\pi} \tau_{{\mathrm{2}}} }\mathrel{=}{}\<[E]%
\\
\>[B]{}\hsindent{3}{}\<[3]%
\>[3]{}\lambda \Varid{x}.\dsterm{\ottmv{\Varid{z}}}{  \constraintfont{ \ottnt{Q} }  ; \Gamma ,\Varid{x}\mathop{:_{  \multiplicityfont{ \pi }  }} \tau_{{\mathrm{1}}}  \vdash \Varid{e}\mathbin{:} \tau_{{\mathrm{2}}} }{}\<[E]%
\\
\>[B]{}\dsterm{\ottmv{\Varid{z}}}{  \constraintfont{ \ottnt{Q_{{\mathrm{1}}}} }   \qtensor   \constraintfont{ \ottnt{Q_{{\mathrm{2}}}} }  ; \Gamma_{{\mathrm{1}}} \mathbin{+} \Gamma_{{\mathrm{2}}}  \vdash  \ottnt{e_{{\mathrm{1}}}} \; \ottnt{e_{{\mathrm{2}}}} \mathbin{:} \tau }\mathrel{=}{}\<[E]%
\\
\>[B]{}\hsindent{3}{}\<[3]%
\>[3]{} \kcase_  \multiplicityfont{ \ottsym{1} }  \;\Varid{z}\;\mathbf{of}\;\{\mskip1.5mu ( \ottmv{z_{{\mathrm{1}}}} , \ottmv{z_{{\mathrm{2}}}} )\to {}\<[E]%
\\
\>[3]{}\hsindent{2}{}\<[5]%
\>[5]{}(\dsterm{\ottmv{ \ottmv{z_{{\mathrm{1}}}} }}{  \constraintfont{ \ottnt{Q_{{\mathrm{1}}}} }  ; \Gamma_{{\mathrm{1}}}  \vdash  \ottnt{e_{{\mathrm{1}}}} \mathbin{:} \tau_{{\mathrm{1}}} \to_{1} \tau })\;(\dsterm{\ottmv{ \ottmv{z_{{\mathrm{2}}}} }}{  \constraintfont{ \ottnt{Q_{{\mathrm{2}}}} }  ; \Gamma_{{\mathrm{2}}}  \vdash  \ottnt{e_{{\mathrm{2}}}} \mathbin{:} \tau_{{\mathrm{1}}} })\mskip1.5mu\}{}\<[E]%
\\
\>[B]{}\dsterm{\ottmv{\Varid{z}}}{  \constraintfont{ \ottnt{Q_{{\mathrm{1}}}} }   \qtensor   \multiplicityfont{ \omega }  \mathbin{⋅}  \constraintfont{ \ottnt{Q_{{\mathrm{2}}}} }  ; \Gamma_{{\mathrm{1}}} \mathbin{+}  \multiplicityfont{ \omega }  \mathbin{⋅} \Gamma_{{\mathrm{2}}}  \vdash  \ottnt{e_{{\mathrm{1}}}} \; \ottnt{e_{{\mathrm{2}}}} \mathbin{:} \tau }\mathrel{=}{}\<[E]%
\\
\>[B]{}\hsindent{3}{}\<[3]%
\>[3]{} \kcase_  \multiplicityfont{ \ottsym{1} }  \;\Varid{z}\;\mathbf{of}\;\{\mskip1.5mu ( \ottmv{z_{{\mathrm{1}}}} ,\underline{ \ottmv{z_{{\mathrm{2}}}} })\to {}\<[E]%
\\
\>[3]{}\hsindent{2}{}\<[5]%
\>[5]{}(\dsterm{\ottmv{ \ottmv{z_{{\mathrm{1}}}} }}{  \constraintfont{ \ottnt{Q_{{\mathrm{1}}}} }  ; \Gamma_{{\mathrm{1}}}  \vdash  \ottnt{e_{{\mathrm{1}}}} \mathbin{:} \tau_{{\mathrm{1}}} \to_{\omega} \tau })\;(\dsterm{\ottmv{ \ottmv{z_{{\mathrm{2}}}} }}{  \constraintfont{ \ottnt{Q_{{\mathrm{2}}}} }  ; \Gamma_{{\mathrm{2}}}  \vdash  \ottnt{e_{{\mathrm{2}}}} \mathbin{:} \tau_{{\mathrm{1}}} })\mskip1.5mu\}{}\<[E]%
\\
\>[B]{}\dsterm{\ottmv{\Varid{z}}}{  \constraintfont{ \ottnt{Q} }   \qtensor   \constraintfont{ \ottnt{Q_{{\mathrm{1}}}} }  [ \overline{\upsilon} / \overline{\ottmv{a} } ]; \Gamma  \vdash \packbox\Varid{e}\mathbin{:}\exists\; \overline{\ottmv{a} } . \tau \RLolly  \constraintfont{ \ottnt{Q_{{\mathrm{1}}}} }  }\mathrel{=}{}\<[E]%
\\
\>[B]{}\hsindent{3}{}\<[3]%
\>[3]{} \kcase_  \multiplicityfont{ \ottsym{1} }  \;\Varid{z}\;\mathbf{of}\;\{\mskip1.5mu (\Varid{z'},\Varid{z''})\to {}\<[E]%
\\
\>[3]{}\hsindent{2}{}\<[5]%
\>[5]{}\packbox(\Varid{z''},\dsterm{\ottmv{\Varid{z'}}}{  \constraintfont{ \ottnt{Q} }  ; \Gamma  \vdash \Varid{e}\mathbin{:} \tau [ \overline{\upsilon} / \overline{\ottmv{a} } ]})\mskip1.5mu\}{}\<[E]%
\\
\>[B]{}\dsterm{\ottmv{\Varid{z}}}{  \constraintfont{ \ottnt{Q_{{\mathrm{1}}}} }   \qtensor   \constraintfont{ \ottnt{Q_{{\mathrm{2}}}} }  ; \Gamma_{{\mathrm{1}}} \mathbin{+} \Gamma_{{\mathrm{2}}}  \vdash \klet\ \packbox \Varid{x}\mathrel{=} \ottnt{e_{{\mathrm{1}}}} \;\mathbf{in}\; \ottnt{e_{{\mathrm{2}}}} \mathbin{:} \tau }\mathrel{=}{}\<[E]%
\\
\>[B]{}\hsindent{3}{}\<[3]%
\>[3]{} \kcase_  \multiplicityfont{ \ottsym{1} }  \;\Varid{z}\;\mathbf{of}\;\{\mskip1.5mu ( \ottmv{z_{{\mathrm{1}}}} , \ottmv{z_{{\mathrm{2}}}} )\to {}\<[E]%
\\
\>[3]{}\hsindent{2}{}\<[5]%
\>[5]{}\klet\ \packbox \Varid{z'},\Varid{x}\mathrel{=}\dsterm{\ottmv{ \ottmv{z_{{\mathrm{1}}}} }}{  \constraintfont{ \ottnt{Q_{{\mathrm{1}}}} }  ; \Gamma_{{\mathrm{1}}}  \vdash  \ottnt{e_{{\mathrm{1}}}} \mathbin{:}\exists\; \overline{\ottmv{a} } . \tau_{{\mathrm{1}}} \RLolly  \constraintfont{ \ottnt{Q} }  }\;\mathbf{in}{}\<[E]%
\\
\>[3]{}\hsindent{2}{}\<[5]%
\>[5]{} \klet_  \multiplicityfont{ \ottsym{1} }  \; \ottmv{z_{{\mathrm{2}}}} '\mathrel{=}( \ottmv{z_{{\mathrm{2}}}} ,\Varid{z'})\;\mathbf{in}{}\<[E]%
\\
\>[3]{}\hsindent{2}{}\<[5]%
\>[5]{}\dsterm{\ottmv{ \ottmv{z_{{\mathrm{2}}}} '}}{  \constraintfont{ \ottnt{Q_{{\mathrm{2}}}} }   \qtensor   \constraintfont{ \ottnt{Q} }  ; \Gamma_{{\mathrm{2}}} ,\Varid{x}\mathop{:_{1}} \tau_{{\mathrm{1}}}  \vdash  \ottnt{e_{{\mathrm{2}}}} \mathbin{:} \tau }\mskip1.5mu\}{}\<[E]%
\\
\>[B]{}\dsterm{\ottmv{\Varid{z}}}{  \constraintfont{ \ottnt{Q_{{\mathrm{1}}}} }   \qtensor   \constraintfont{ \ottnt{Q_{{\mathrm{2}}}} }  ; \Gamma_{{\mathrm{1}}} \mathbin{+} \Gamma_{{\mathrm{2}}}  \vdash  \klet_  \multiplicityfont{ \ottsym{1} }  \;\Varid{x}\mathrel{=} \ottnt{e_{{\mathrm{1}}}} \;\mathbf{in}\; \ottnt{e_{{\mathrm{2}}}} \mathbin{:} \tau }\mathrel{=}{}\<[E]%
\\
\>[B]{}\hsindent{3}{}\<[3]%
\>[3]{} \kcase_  \multiplicityfont{ \ottsym{1} }  \;\Varid{z}\;\mathbf{of}\;\{\mskip1.5mu ( \ottmv{z_{{\mathrm{1}}}} , \ottmv{z_{{\mathrm{2}}}} )\to {}\<[E]%
\\
\>[3]{}\hsindent{2}{}\<[5]%
\>[5]{} \klet_  \multiplicityfont{ \ottsym{1} }  \;\Varid{x}\mathbin{:}\dsevidence{  \constraintfont{ \ottnt{Q} }  }\to_{1} \tau_{{\mathrm{1}}} \mathrel{=}\dsterm{\ottmv{ \ottmv{z_{{\mathrm{1}}}} }}{  \constraintfont{ \ottnt{Q_{{\mathrm{1}}}} }   \qtensor   \constraintfont{ \ottnt{Q} }  ; \Gamma_{{\mathrm{1}}}  \vdash  \ottnt{e_{{\mathrm{1}}}} \mathbin{:} \tau_{{\mathrm{1}}} }{}\<[E]%
\\
\>[3]{}\hsindent{2}{}\<[5]%
\>[5]{}\mathbf{in}\;\dsterm{\ottmv{ \ottmv{z_{{\mathrm{2}}}} }}{  \constraintfont{ \ottnt{Q_{{\mathrm{2}}}} }  ; \Gamma_{{\mathrm{2}}} ,\Varid{x}\mathop{:_{1}} \tau_{{\mathrm{1}}}  \vdash  \ottnt{e_{{\mathrm{2}}}} \mathbin{:} \tau }\mskip1.5mu\}{}\<[E]%
\\
\>[B]{}\dsterm{\ottmv{\Varid{z}}}{  \multiplicityfont{ \omega }  \mathbin{⋅}  \constraintfont{ \ottnt{Q_{{\mathrm{1}}}} }   \qtensor   \constraintfont{ \ottnt{Q_{{\mathrm{2}}}} }  ;  \multiplicityfont{ \omega }  \mathbin{⋅} \Gamma_{{\mathrm{1}}} \mathbin{+} \Gamma_{{\mathrm{2}}}  \vdash  \klet_  \multiplicityfont{ \omega }  \;\Varid{x}\mathrel{=} \ottnt{e_{{\mathrm{1}}}} \;\mathbf{in}\; \ottnt{e_{{\mathrm{2}}}} \mathbin{:} \tau }\mathrel{=}{}\<[E]%
\\
\>[B]{}\hsindent{3}{}\<[3]%
\>[3]{} \kcase_  \multiplicityfont{ \ottsym{1} }  \;\Varid{z}\;\mathbf{of}\;\{\mskip1.5mu (\underline{ \ottmv{z_{{\mathrm{1}}}} }, \ottmv{z_{{\mathrm{2}}}} )\to {}\<[E]%
\\
\>[3]{}\hsindent{2}{}\<[5]%
\>[5]{} \klet_  \multiplicityfont{ \omega }  \;\Varid{x}\mathbin{:}\dsevidence{  \constraintfont{ \ottnt{Q} }  }\to_{1} \tau_{{\mathrm{1}}} \mathrel{=}\dsterm{\ottmv{ \ottmv{z_{{\mathrm{1}}}} }}{  \constraintfont{ \ottnt{Q_{{\mathrm{1}}}} }   \qtensor   \constraintfont{ \ottnt{Q} }  ; \Gamma_{{\mathrm{1}}}  \vdash  \ottnt{e_{{\mathrm{1}}}} \mathbin{:} \tau_{{\mathrm{1}}} }\;\mathbf{in}{}\<[E]%
\\
\>[3]{}\hsindent{2}{}\<[5]%
\>[5]{}\dsterm{\ottmv{ \ottmv{z_{{\mathrm{2}}}} }}{  \constraintfont{ \ottnt{Q_{{\mathrm{2}}}} }  ; \Gamma_{{\mathrm{2}}} ,\Varid{x}\mathop{:_{  \multiplicityfont{ \omega }  }} \tau_{{\mathrm{1}}}  \vdash  \ottnt{e_{{\mathrm{2}}}} \mathbin{:} \tau }\mskip1.5mu\}{}\<[E]%
\\
\>[B]{}\dsterm{\ottmv{\Varid{z}}}{  \constraintfont{ \ottnt{Q_{{\mathrm{1}}}} }   \qtensor   \constraintfont{ \ottnt{Q_{{\mathrm{2}}}} }  ; \Gamma_{{\mathrm{1}}} \mathbin{+} \Gamma_{{\mathrm{2}}}  \vdash  \klet_  \multiplicityfont{ \ottsym{1} }  \;\Varid{x}\mathbin{:}\forall\; \overline{\ottmv{a} } .  \constraintfont{ \ottnt{Q} }  \Lolly  \tau_{{\mathrm{1}}} \mathrel{=} \ottnt{e_{{\mathrm{1}}}} \;\mathbf{in}\; \ottnt{e_{{\mathrm{2}}}} \mathbin{:} \tau }\mathrel{=}{}\<[E]%
\\
\>[B]{}\hsindent{3}{}\<[3]%
\>[3]{} \kcase_  \multiplicityfont{ \ottsym{1} }  \;\Varid{z}\;\mathbf{of}\;\{\mskip1.5mu ( \ottmv{z_{{\mathrm{1}}}} , \ottmv{z_{{\mathrm{2}}}} )\to {}\<[E]%
\\
\>[3]{}\hsindent{2}{}\<[5]%
\>[5]{} \klet_  \multiplicityfont{ \ottsym{1} }  \;\Varid{x}\mathbin{:}\forall\; \overline{\ottmv{a} } .\dsevidence{  \constraintfont{ \ottnt{Q} }  }\to_{1} \tau_{{\mathrm{1}}} \mathrel{=}\dsterm{\ottmv{ \ottmv{z_{{\mathrm{1}}}} }}{  \constraintfont{ \ottnt{Q_{{\mathrm{1}}}} }   \qtensor   \constraintfont{ \ottnt{Q} }  ; \Gamma_{{\mathrm{1}}}  \vdash  \ottnt{e_{{\mathrm{1}}}} \mathbin{:} \tau_{{\mathrm{1}}} }\;\mathbf{in}{}\<[E]%
\\
\>[3]{}\hsindent{2}{}\<[5]%
\>[5]{}\dsterm{\ottmv{ \ottmv{z_{{\mathrm{2}}}} }}{  \constraintfont{ \ottnt{Q_{{\mathrm{2}}}} }  ; \Gamma_{{\mathrm{2}}} ,\Varid{x}\mathop{:_{1}}\forall\; \overline{\ottmv{a} } .  \constraintfont{ \ottnt{Q} }  \Lolly  \tau_{{\mathrm{1}}}  \vdash  \ottnt{e_{{\mathrm{2}}}} \mathbin{:} \tau }\mskip1.5mu\}{}\<[E]%
\\
\>[B]{}\dsterm{\ottmv{\Varid{z}}}{  \multiplicityfont{ \omega }  \mathbin{⋅}  \constraintfont{ \ottnt{Q_{{\mathrm{1}}}} }   \qtensor   \constraintfont{ \ottnt{Q_{{\mathrm{2}}}} }  ;  \multiplicityfont{ \omega }  \mathbin{⋅} \Gamma_{{\mathrm{1}}} \mathbin{+} \Gamma_{{\mathrm{2}}}  \vdash  \klet_  \multiplicityfont{ \omega }  \;\Varid{x}\mathbin{:}\forall\; \overline{\ottmv{a} } .  \constraintfont{ \ottnt{Q} }  \Lolly  \tau_{{\mathrm{1}}} \mathrel{=} \ottnt{e_{{\mathrm{1}}}} \;\mathbf{in}\; \ottnt{e_{{\mathrm{2}}}} \mathbin{:} \tau }\mathrel{=}{}\<[E]%
\\
\>[B]{}\hsindent{3}{}\<[3]%
\>[3]{} \kcase_  \multiplicityfont{ \ottsym{1} }  \;\Varid{z}\;\mathbf{of}\;\{\mskip1.5mu (\underline{ \ottmv{z_{{\mathrm{1}}}} }, \ottmv{z_{{\mathrm{2}}}} )\to {}\<[E]%
\\
\>[3]{}\hsindent{2}{}\<[5]%
\>[5]{} \klet_  \multiplicityfont{ \omega }  \;\Varid{x}\mathbin{:}\forall\; \overline{\ottmv{a} } .\dsevidence{  \constraintfont{ \ottnt{Q} }  }\to_{1} \tau_{{\mathrm{1}}} \mathrel{=}\dsterm{\ottmv{ \ottmv{z_{{\mathrm{1}}}} }}{  \constraintfont{ \ottnt{Q_{{\mathrm{1}}}} }   \qtensor   \constraintfont{ \ottnt{Q} }  ; \Gamma_{{\mathrm{1}}}  \vdash  \ottnt{e_{{\mathrm{1}}}} \mathbin{:} \tau_{{\mathrm{1}}} }\;\mathbf{in}{}\<[E]%
\\
\>[3]{}\hsindent{2}{}\<[5]%
\>[5]{}\dsterm{\ottmv{ \ottmv{z_{{\mathrm{2}}}} }}{  \constraintfont{ \ottnt{Q_{{\mathrm{2}}}} }  ; \Gamma_{{\mathrm{2}}} ,\Varid{x}\mathop{:_{  \multiplicityfont{ \omega }  }} \tau_{{\mathrm{1}}}  \vdash  \ottnt{e_{{\mathrm{2}}}} \mathbin{:} \tau }\mskip1.5mu\}{}\<[E]%
\\
\>[B]{}\dsterm{\ottmv{\Varid{z}}}{  \multiplicityfont{ \omega }  \mathbin{⋅}  \constraintfont{ \ottnt{Q_{{\mathrm{1}}}} }   \qtensor   \constraintfont{ \ottnt{Q_{{\mathrm{2}}}} }  ;  \multiplicityfont{ \omega }  \mathbin{⋅} \Gamma_{{\mathrm{1}}} \mathbin{+} \Gamma_{{\mathrm{2}}}  \vdash  \kcase_  \multiplicityfont{ \ottsym{1} }  \;\Varid{e}\;\mathbf{of}\;\{\mskip1.5mu \overline{\ottmv{K}_i\ \overline{\ottmv{x}_i } \to \ottnt{e}_i }\mskip1.5mu\}\mathbin{:} \tau }{}\<[67]%
\>[67]{}\mathrel{=}{}\<[67E]%
\\
\>[B]{}\hsindent{3}{}\<[3]%
\>[3]{} \kcase_  \multiplicityfont{ \ottsym{1} }  \;\Varid{z}\;\mathbf{of}\;\{\mskip1.5mu (\underline{ \ottmv{z_{{\mathrm{1}}}} }, \ottmv{z_{{\mathrm{2}}}} )\to {}\<[E]%
\\
\>[3]{}\hsindent{2}{}\<[5]%
\>[5]{} \kcase_  \multiplicityfont{ \ottsym{1} }  \;(\dsterm{\ottmv{ \ottmv{z_{{\mathrm{1}}}} }}{  \constraintfont{ \ottnt{Q_{{\mathrm{1}}}} }  ; \Gamma_{{\mathrm{1}}}  \vdash \Varid{e}\mathbin{:}\Conid{T}\; \overline{\tau} })\;\mathbf{of}{}\<[E]%
\\
\>[5]{}\hsindent{2}{}\<[7]%
\>[7]{}\{\mskip1.5mu \overline{\Conid{K}\; \overline{\ottmv{x} }_{\ottmv{i}} \to \dsterm{\ottmv{ \ottmv{z_{{\mathrm{2}}}} }}{  \constraintfont{ \ottnt{Q_{{\mathrm{2}}}} }  ; \Gamma_{{\mathrm{2}}} ,\overline{ \ottmv{x_{\ottmv{i}}} \mathop{:_{(  \multiplicityfont{ \pi }  \mathbin{⋅}  \multiplicityfont{ \pi_{\ottmv{i}} }  )}} \upsilon_{\ottmv{i}} [ \overline{\tau} / \overline{\ottmv{a} } ]} \vdash  \ottnt{e_{\ottmv{i}}} \mathbin{:} \tau }}\mskip1.5mu\}\mskip1.5mu\}{}\<[E]%
\\
\>[B]{}\dsterm{\ottmv{\Varid{z}}}{  \constraintfont{ \ottnt{Q_{{\mathrm{1}}}} }   \qtensor   \constraintfont{ \ottnt{Q_{{\mathrm{2}}}} }  ; \Gamma_{{\mathrm{1}}} \mathbin{+} \Gamma_{{\mathrm{2}}}  \vdash  \kcase_  \multiplicityfont{ \omega }  \;\Varid{e}\;\mathbf{of}\;\{\mskip1.5mu \overline{\ottmv{K}_i\ \overline{\ottmv{x}_i } \to \ottnt{e}_i }\mskip1.5mu\}\mathbin{:} \tau }{}\<[53]%
\>[53]{}\mathrel{=}{}\<[53E]%
\\
\>[B]{}\hsindent{3}{}\<[3]%
\>[3]{} \kcase_  \multiplicityfont{ \ottsym{1} }  \;\Varid{z}\;\mathbf{of}\;\{\mskip1.5mu ( \ottmv{z_{{\mathrm{1}}}} , \ottmv{z_{{\mathrm{2}}}} )\to {}\<[E]%
\\
\>[3]{}\hsindent{2}{}\<[5]%
\>[5]{} \kcase_  \multiplicityfont{ \omega }  \;(\dsterm{\ottmv{ \ottmv{z_{{\mathrm{1}}}} }}{  \constraintfont{ \ottnt{Q_{{\mathrm{1}}}} }  ; \Gamma_{{\mathrm{1}}}  \vdash \Varid{e}\mathbin{:}\Conid{T}\; \overline{\tau} })\;\mathbf{of}{}\<[E]%
\\
\>[5]{}\hsindent{2}{}\<[7]%
\>[7]{}\{\mskip1.5mu \overline{\Conid{K}\; \overline{\ottmv{x} }_{\ottmv{i}} \to \dsterm{\ottmv{ \ottmv{z_{{\mathrm{2}}}} }}{  \constraintfont{ \ottnt{Q_{{\mathrm{2}}}} }  ; \Gamma_{{\mathrm{2}}} ,\overline{ \ottmv{x_{\ottmv{i}}} \mathop{:_{(  \multiplicityfont{ \pi }  \mathbin{⋅}  \multiplicityfont{ \pi_{\ottmv{i}} }  )}} \upsilon_{\ottmv{i}} [ \overline{\tau} / \overline{\ottmv{a} } ]} \vdash  \ottnt{e_{\ottmv{i}}} \mathbin{:} \tau }}\mskip1.5mu\}\mskip1.5mu\}{}\<[E]%
\ColumnHook
\end{hscode}\resethooks
\end{minipage}
\right.
$$

  \caption{Desugaring}
  \label{fig:full:desugaring}
\end{figure}

\end{document}